\documentclass[DIV=calc,paper=a4,fontsize=11pt,twocolumn]{scrartcl} 

\usepackage[english]{babel}
\usepackage[protrusion=true,expansion=true]{microtype}
\usepackage{amsmath,amsfonts,amsthm}
\usepackage[final]{graphicx}
\usepackage{xcolor}
\usepackage[normal,small,hypcap,up,labelfont=bf,textfont=it]{caption}
\usepackage{epstopdf}
\usepackage{subfig}
\usepackage{booktabs}
\usepackage{fix-cm}
\usepackage{amssymb,amsfonts}
\usepackage{dsfont}
\usepackage{bbm}
\usepackage{pstricks}
\usepackage{cite}
\usepackage[utf8]{inputenc}
\usepackage[perpage,symbol*]{footmisc}
\usepackage[varg]{txfonts}
\usepackage{balance}
\usepackage{fancyhdr}
\PassOptionsToPackage{hyphens}{url}\usepackage[pdfencoding=auto]{hyperref}
\usepackage{bookmark}
\usepackage{verbatim}
\usepackage{fontenc}
\usepackage{cuted}

\usepackage{bm}
\usepackage{mathrsfs}

\theoremstyle{definition}
\newtheorem{definition}{Definition}[section]
\newtheorem{example}[definition]{Example}

\theoremstyle{plain}
\newtheorem{theorem}[definition]{Theorem}
\newtheorem{lemma}[definition]{Lemma}
\newtheorem{corollary}[definition]{Corollary}

\newcommand{\RKet}[1]{\ensuremath{\left \vert #1 \right )}}
\newcommand{\Ket}[1]{\ensuremath{\left \vert #1 \right \rangle}}
\newcommand{\Bra}[1]{\ensuremath{\left \langle #1 \right \vert}}
\newcommand{\BraKet}[2]{\ensuremath{\left \langle #1 \middle \vert #2
\right \rangle}}
\newcommand{\KetBra}[2]{\ensuremath{\left \vert #1 \middle \rangle
\middle \langle #2 \right \vert}}
\newcommand{\Proj}[1]{\ensuremath{\left [ #1 \right ]}}
\newcommand{\Hilb}[1][]{\ensuremath{\mathcal{H}_{#1}}}
\newcommand{\Tr}[2][]{\ensuremath{\text{Tr}_{#1} \left ( #2 \right )}}

\DeclareCaptionFont{mycolor}{\color[HTML]{0000FF}}
\usepackage{sectsty}													
\allsectionsfont{
\color[HTML]{31ADF3}\usefont{OT1}{phv}{b}{n}
}

\sectionfont{
\color[HTML]{31ADF3}\usefont{OT1}{phv}{b}{n}
}

\usepackage{fancyhdr}												
\usepackage{lettrine}
\newcommand{\initial}[1]{%
\lettrine[lines=3,lhang=0.3,nindent=0em]{
\color[HTML]{31ADF3}
{\textsf{#1}}}{}}

\usepackage{titling}															

\newcommand{\HorRule}{\color[HTML]{31ADF3}
\rule{\linewidth}{1pt}%
}

\pretitle{\vspace{-30pt} \begin{flushleft} \HorRule
\fontsize{34}{34} \usefont{OT1}{phv}{b}{n} \color[HTML]{31ADF3} \selectfont
}
\title{Is the Quantum State Real?\\An Extended Review of $\psi$-ontology Theorems}					
\posttitle{\par\end{flushleft}\vskip 0.5em}

\preauthor{\begin{flushleft}\large \lineskip 0.5em \usefont{OT1}{phv}{b}{sl} \color[HTML]{31ADF3}}
\author{Matthew Saul Leifer\\[8pt]}											
\postauthor{\footnotesize \usefont{OT1}{phv}{m}{sl} \color[HTML]{000000}
Perimeter Institute for Theoretical Physics, 31 Caroline St. N., Waterloo, ON, Canada. E-mail: \href{mailto:matt@mattleifer.info}{matt@mattleifer.info}\\[10pt]		
\scriptsize\usefont{OT1}{phv}{m}{n} \color[HTML]{31ADF3}{\textbf{Editors: \emph{Klaas Landsman}, \emph{Roger Colbeck} \& \emph{Terry Rudolph}} }\\[5pt]
\color[HTML]{000000}{Article history: Submitted on December 15, 2013;  Accepted on October 12, 2014; Published on November 5, 2014.}
\par\end{flushleft}\HorRule}

\date{}																				

\begin{document}
\maketitle
\thispagestyle{fancy} 			
\initial{T}\textbf{owards the end of 2011, Pusey, Barrett and Rudolph derived a theorem that aimed to show that the quantum state must be ontic (a state of reality) in a broad class of realist approaches to quantum theory. This result attracted a lot of attention and controversy. The aim of this review article is to review the background to the Pusey--Barrett--Rudolph Theorem, to provide a clear presentation of the theorem itself, and to review related work that has appeared since the publication of the Pusey--Barrett--Rudolph paper. In particular, this review:
Explains what it means for the quantum state to be ontic or epistemic (a state of knowledge);
Reviews arguments for and against an ontic interpretation of the quantum state as they existed prior to the Pusey--Barrett--Rudolph Theorem;
Explains why proving the reality of the quantum state is a very strong constraint on realist theories in that it would imply many of the known no-go theorems, such as Bell's Theorem and the need for an exponentially large ontic state space;
Provides a comprehensive presentation of the Pusey--Barrett--Rudolph Theorem itself, along with subsequent improvements and criticisms of its assumptions;
Reviews two other arguments for the reality of the quantum state: the first due to Hardy and the second due to Colbeck and Renner, and explains why their assumptions are less compelling than those of the Pusey--Barrett--Rudolph Theorem;
Reviews subsequent work aimed at ruling out stronger notions of what it means for the quantum state to be epistemic and points out open questions in this area. The overall aim is not only to provide the background needed for the novice in this area to understand the current status, but also to discuss often overlooked subtleties that should be of interest to the experts.\\ Quanta 2014; 3: 67--155.}
\begin{figure}[b!]
\rule{245 pt}{0.5 pt}\\[3pt]
\raisebox{-0.2\height}{\includegraphics[width=5mm]{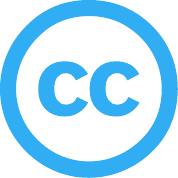}}\raisebox{-0.2\height}{\includegraphics[width=5mm]{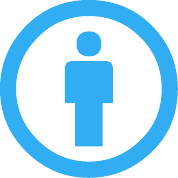}}
\footnotesize{This is an open access article distributed under the terms of the Creative Commons Attribution License \href{http://creativecommons.org/licenses/by/3.0/}{CC-BY-3.0}, which permits unrestricted use, distribution, and reproduction in any medium, provided the original author and source are credited.}
\end{figure}
\section{Introduction}

\label{Intro}

In 1964, John Bell fundamentally changed the way that we think about
quantum theory \cite{Bell1964}.  Abner Shimony famously referred to tests of
Bell's Theorem as ``experimental metaphysics'' \cite{Shimony1984}, but
I disagree with this characterization.  What Bell's Theorem really
shows us is that the foundations of quantum theory is a bona fide
field of physics, in which questions are to be resolved by rigorous
argument and experiment, rather than remaining the subject of
open-ended debate.  In other words, it is a mistake to crudely divide
quantum theory into its practical part and its interpretation, and to
think of the latter as metaphysics, experimental or otherwise.

In the wake of Bell's Theorem, the study of entanglement and
nonlocality has become a mainstream field of physics, particularly in
light of its practical applications in quantum information science,
but Bell's broader lesson---that the interpretation of quantum theory
should be approached as a rigorous science---has rather been missed.
This is nowhere more evident than in debates about the status of the
quantum state.  The question of just what type of thing the quantum
state, or wavefunction, represents, has been with us since the
beginnings of quantum theory.  The likes of de Broglie and
Schr{\"o}dinger initially wanted to view the wavefunction as a real
physical wave, just like the waves of classical field theory, with
perhaps some additional structure to account for particle-like or
``quantum'' properties \cite{Bacciagaluppi2009}.  In contrast,
following Born's introduction of the probability rule \cite{Born1983},
the Copenhagen interpretation advocated by Bohr, Heisenberg, Pauli
et.\ al.\ came to view the wavefunction as a ``probability wave'' and
denied the need for a more fundamental reality to underlie it
\cite{Faye2008}.  In modern terms, most realist interpretations of
quantum theory; such as many-worlds \cite{Everett1957, DeWitt1973,
Wallace2012}, de Broglie--Bohm theory \cite{Broglie2009, Bohm1952,
Bohm1952a, Duerr2009}, spontaneous collapse theories
\cite{Ghirardi1986, Bassi2013}, and modal interpretations
\cite{Lombardi2013}; view the wavefunction as part of reality, whereas
those that follow more Copenhagenish lines \cite{Fuchs2010,
Fuchs2010a, Fuchs2013, Pitowsky2006, Bub2010, Healey2012,
Rovelli1996, Muynck2002, Mermin2013, Peres1995, Brukner2003} tend to
view it as a representation of knowledge, information, or belief.  The
big advantage of the latter view is that the notorious collapse of the
wavefunction can be explained as the effect of acquiring new
information, no more serious than the updating of classical
probabilities in the light of new data, rather than as an anomalous
physical process that needs to be eliminated or explained away.

The question then is whether this is a necessary dichotomy.  Is the
only way to avoid having this weird multidimensional object as part of
reality to give up on reality altogether, or can we reach a compromise
in which there is a well-founded reality, but one in which the
wavefunction only represents knowledge?  This seems like a question
that is ripe for attacking with the kind of conceptual rigor that Bell
brought to nonlocality, and indeed Pusey, Barrett and Rudolph have
recently proven a theorem to the effect that the wavefunction must be
ontic (i.e.\ a state of reality), as opposed to epistemic (i.e.\ a
state of knowledge) in a broad class of realist approaches to quantum
theory \cite{Pusey2012}.

Since then, there has been much discussion and criticism of the
Pusey--Barrett--Rudolph Theorem in both formal \cite{Bub2012,
  Bub2012a, Drezet2012, Drezet2012a, Emerson2013, Hall2011,
  Hofmann2011, Mansfield2013, Miller2013, Schlosshauer2013,
  Schlosshauer2012, Wallden2013} and informal venues
\cite{Aaronson2011, Griffiths2012, Leifer2011, Leifer2012a,
  Leifer2011a, Motl2011, Reich2012, Reich2011, Spekkens2012,
  Wallace2011}, as well as a couple of attempts to derive the same
conclusion as Pusey--Barrett--Rudolph from different assumptions
\cite{Colbeck2012, Hardy2013}.  The Pusey--Barrett--Rudolph Theorem
and its successors all employ auxiliary assumptions of varying degrees
of reasonableness.  Without these assumptions, it has been shown that
the wavefunction may be epistemic \cite{Lewis2012}.  Therefore, there
has also been subsequent work aimed at ruling out stronger notions of
what it means for the wavefunction to be epistemic, without using such
auxiliary assumptions \cite{Aaronson2013, Barrett2013, Branciard2014,
  Leifer2013c, Leifer2014, Maroney2012, Maroney2012a, Patra2013a}.
The aim of this review article is to provide the background necessary
for understanding these results, to provide a comprehensive
presentation and criticism of them, and to explain their implications.

One of the most intriguing things about proving that the wavefunction
must be ontic is that it would imply a large number of existing no-go
results, including Bell's Theorem \cite{Bell1964} and excess baggage
theorems \cite{Hardy2004, Montina2008, Montina2011}(i.e.\ showing that
the size of the ontic state space must be infinite and must scale
exponentially with the number of systems).  Therefore, even apart from
its foundational significance, proving the reality of the wavefunction
could potentially provide a powerful unification of the known no-go
theorems, and may have applications in quantum information theory.

My aim is that this review should be accessible to as wide an audience
as possible, but I have made three decisions about how to present the
material that make my treatment somewhat more involved than those
found elsewhere in the literature.  Firstly, I adopt rigorous measure
theoretic probability theory.  It is common in the literature to
specialize to finite sample spaces or to adopt a less rigorous
approach to continuous spaces, which basically involves proving all
results as if you were dealing with smooth and continuous probability
densities and then hoping everything still works when you throw in a
bunch of Dirac delta functions.  Although a measure theoretic approach
may reduce accessibility, there are important reasons for adopting it.
It would be odd to attempt to prove the reality of the wavefunction
within a framework that does not admit a model in which the
wavefunction is real in the first place.  Since the wavefunction
involves continuous parameters, this means that there is no option of
specializing to finite sample spaces.  Furthermore, there are several
special cases of interest for which the optimistic non-rigorous
approach simply does not work, including the case where the
wavefunction, and only the wavefunction, is the state of reality.
Therefore, in order to cover all the cases of interest, there is
really no option other than taking a measure theoretic approach.  As
an aid to accessibility, I outline how the main definitions and
arguments specialize to the case of a finite sample space, which
should be sufficient for those who do not wish to get embroiled in the
technical details.

Secondly, it is common in the literature to assume that we are
interested in modeling all pure states and all projective
measurements on some finite dimensional Hilbert space, and to
specialize results to that context.  However, some results apply
equally to subsets of states and measurements, which I call
\emph{fragments} of quantum theory.  In addition, it is known that
some fragments of quantum theory, have natural models in which the
wavefunction is epistemic \cite{Spekkens2007, Schreiber,
Bartlett2012}.  Therefore, I think it is important to emphasize the
minimal structures in which the various results can be proved, rather
than just assuming that we are trying to model all states and
measurements on some Hilbert space.

The third presentation decision concerns my treatment of
\emph{preparation contextuality} (see \S\ref{PC} for the formal
definition).  The main issue we are interested in is whether pure
quantum states must be ontic, since it is uncontroversial that mixed
states can at least sometimes represent lack of knowledge about which
of a set of pure states was prepared.  It is common in the literature
to assume that each pure quantum state is represented by a unique
probability measure over the possible states of reality, but I do not
make this assumption.  In a preparation contextual model, different
methods of preparing a quantum state may lead to different probability
measures.  In fact, this must occur for mixed states
\cite{Spekkens2005}, so it seems sensible to allow for the possibility
that it might occur for pure states as well.  In addition, some of the
intermediate results to be discussed hold equally well for mixed
states, but this can only be established by adopting definitions that
are broad enough to encompass mixed states, which are necessarily
preparation contextual.

These three presentation decisions mean that the definitions,
statements of results, and proofs that appear in this review often
differ from those in the existing literature.  Generally, this is just
a matter of making a few obvious generalizations, without
substantively changing the ideas.  For this reason, I do not
explicitly point out where such generalizations have been made.

The review is divided into three parts.  \hyperref[OED]{Part I} is a general
review of the distinction between ontic and epistemic interpretations
of the quantum state.  It discusses the arguments that had been given
for ontic and epistemic interpretations of the wavefunction prior to
the discovery of the Pusey--Barrett--Rudolph Theorem.  My aim in this part is to convince
you that there is some merit to the epistemic interpretation and that
previous arguments for the reality of the quantum state are
unconvincing.  In this part, I also give a formal definition of the
class of models to which the Pusey--Barrett--Rudolph Theorem and related results apply,
and define what it means for the quantum state to be ontic or
epistemic within this class of models.  Following this, I give a
detailed discussion of the other no-go theorems that would follow as
corollaries of proving the reality of the wavefunction.

\hyperref[SPON]{Part II} reviews the three theorems that attempt to prove the
reality of the wavefunction: the Pusey--Barrett--Rudolph Theorem, Hardy's Theorem, and the
Colbeck--Renner Theorem.  The treatment of the Pusey--Barrett--Rudolph Theorem is the most
detailed of the three, since it has attracted the largest literature
and has been subject to the largest amount of confusion and criticism.
In my view, it makes the strongest case of the three theorems for the
reality of the wavefunction, although it is still not bulletproof, so
I go to some lengths to sort the silly criticisms from the substantive
ones.  The assumptions behind the Hardy and Colbeck--Renner Theorems
receive a more critical treatment, but the theorems are still
presented in detail because they are interestingly related to other
arguments about realist interpretations of quantum theory.

\hyperref[Beyond]{Part III} deals with attempts to go beyond the rigid
distinction between epistemic and ontic interpretations of the
wavefunction by positing stronger constraints on epistemic
interpretations.  One of the aims of doing this is to remove the
problematic auxiliary assumptions needed to prove the three main
theorems, whilst still arriving at a conclusion that is morally
similar.  This part is shorter than the other two and mostly just
summarizes the known results without proof.  The reason for this is
that many of the results are only preliminary and will likely be
superseded by the time this review is published.  My main aim in this
part is to point out the most promising directions for future
research.

\section*{Part I. The $\psi$-ontic/epistemic distinction\label{OED}}

The results reviewed in this paper aim to show that the quantum state
must be \emph{ontic} (a state of reality) rather than \emph{epistemic}
(a state of knowledge).  What does this mean and why is it important?
The word ``ontology'' derives from the Greek word for ``being'' and
refers to the branch of metaphysics that concerns the character of
things that exist.  In the present context, an \emph{ontic state}
refers to something that objectively exists in the world,
independently of any observer or agent.  In other words, ontic states
are the things that would still exist if all intelligent beings were
suddenly wiped out from the universe.  On the other hand,
``epistemology'' is the branch of philosophy that studies of the
nature and scope of knowledge.  An \emph{epistemic state} is therefore
a description of what an observer currently knows about a physical
system.  It is something that exists in the mind of the observer
rather than in the external physical world.
\begin{figure}[t!]
\centering
\subfloat[\color{blue} An ontic state is a point in phase space.]{\includegraphics[width=75mm]{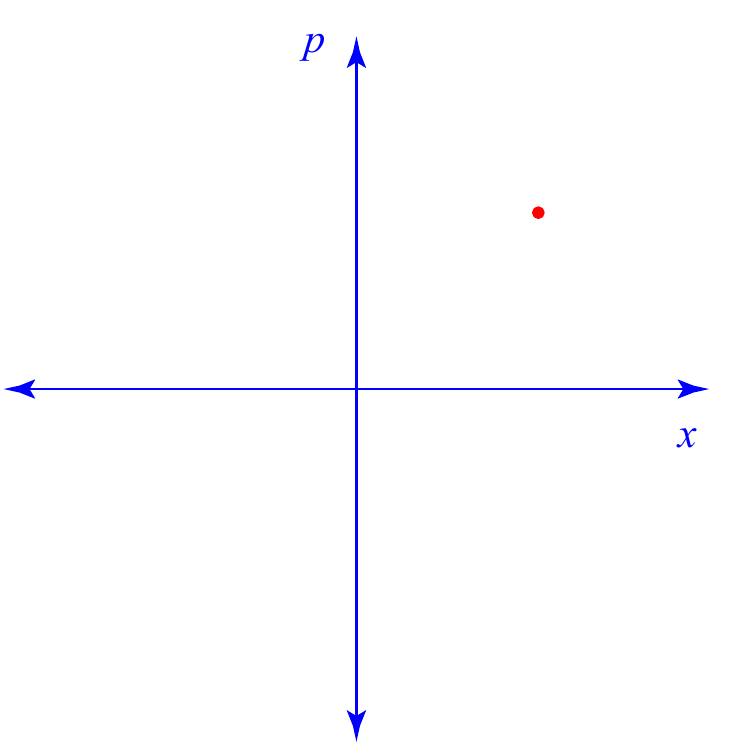}}\\
\subfloat[\color{blue} An epistemic state is a probability density on phase space. Contours indicate lines of equal probability density.]{\includegraphics[width=75mm]{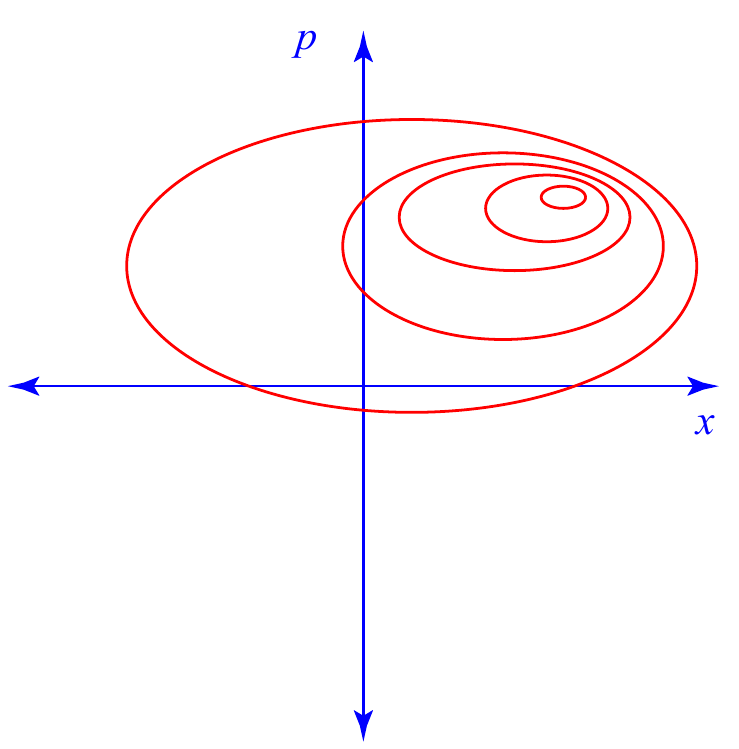}}\\
\subfloat[\color{blue} An ontic state (cross) is deemed possible in more than one epistemic state ($f_1$ and $f_2$).  Phase space has been schematically collapsed down to one dimension for illustrative purposes.]{\includegraphics[width=75mm]{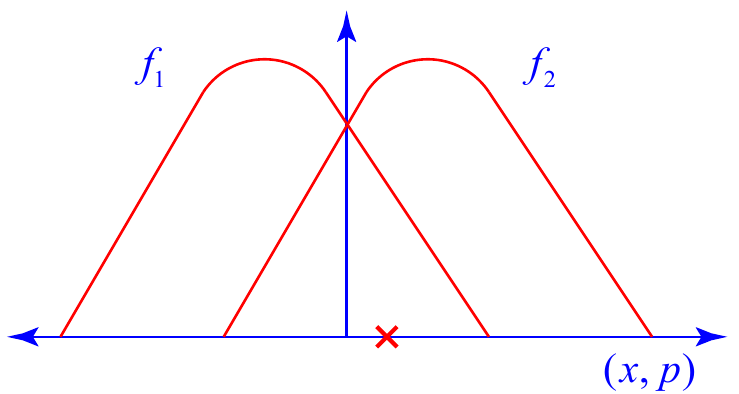}}
\caption{\color[HTML]{0000FF}{\label{fig:OEQ:Ontic-Epistemic}The distinction between ontic and epistemic states in
single particle classical mechanics.}}
\end{figure}

In classical mechanics, the distinction between ontic and epistemic
states is fairly clear.  A single Newtonian particle in one dimension
has a position $x$ and a momentum $p$ and these are objective
properties of the particle that exist independently of us.  All other
objective properties of the particle are functions of $x$ and $p$.
The ontic state of the particle is therefore the phase space point
$(x,p)$.  This evolves according to Hamilton's equations
\begin{align}
\frac{d p}{d t} & = -\frac{\partial H}{\partial x} & \frac{d x}{d t}
& = \frac{\partial H}{\partial p},
\end{align}
where $H$ is the Hamiltonian.  On the other hand, if we do not know
the exact position and momentum of the particle then our knowledge
about its ontic state is represented by a probability density $f(x,p)$
over phase space.  By applying Hamilton's equations to the individual
phase space points on which $f(x,p)$ is supported, it can be shown
that $f(x,p)$ evolves according to Liouville's equation
\begin{equation}
\frac{\partial f}{\partial t} = \frac{\partial H}{\partial
x}\frac{\partial f}{\partial p} - \frac{\partial H}{\partial
p}\frac{\partial f}{\partial x}.
\end{equation}
The probability density $f(x,p)$ is our epistemic state.  See
Fig.~\ref{fig:OEQ:Ontic-Epistemic} for an illustration of the
distinction between classical ontic and epistemic states.  For other
types of classical system the situation is analogous, the only
difference being the dimension of the phase space, e.g. $6N$
dimensions for $N$ particles in $3$ dimensional space or a continuum
for field systems.  The phase space point is still the ontic state and
a probability density over phase space is the epistemic state.

Note that calling a probability density ``epistemic'' is controversial
in some circles.  It presupposes a broadly Bayesian interpretation of
probability theory in which probabilities represent an agent's
knowledge, information, or beliefs.  Fortunately, the issue at stake
does not really depend on this as it also appears in other
interpretations of probability under a different name.  What is
important is that the states dubbed ``epistemic'' only have
probabilistic import so they cannot be regarded as intrinsic
properties of individual physical systems.  The key property that this
implies is that a given ontic state is deemed possible in more than
one epistemic state.

On the Bayesian reading, this is due to the fact
that different agents may have different knowledge about one and the
same physical system.  For example, perhaps Alice knows the position
of a classical particle exactly but nothing about its momentum, whilst
Bob knows the momentum precisely but nothing about its position.
Alice and Bob would then assign different probability distributions to
the system, with the crucial property that they would overlap on the
ontic state actually occupied by the system.

Other interpretations of probability exhibit the same property in a
different way.  For example, on a frequentist account of probability,
probabilities represent the relative frequencies of occurrence of some
property in an ensemble of independently and identically prepared
systems.  In this context, we would talk about a state being
``statistical'' rather than ``epistemic''.  The statistical state of a
system depends upon the choice of ensemble that the individual system
is regarded as being a part of.  For example, suppose a classical
particle occupies the phase space point $(1~m,1~kgms^{-1})$.  If we
regard it as part of an ensemble of particles that all have positive
position, but some have negative momentum, then it will be assigned a
different probability distribution than if we regard it as part of an
ensemble of particles all of which have positive momentum, but some
have negative position.  In the former case, the probability
distribution will have support on negative momentum phase space points
and in the latter case it will have support on negative position phase
space points.  The point is that ensembles consist of more than one
individual system and the same ontic state may occur as a part of many
different ensembles.  A frequentist will not be lead astray by
substituting the word ``statistical'' for every occurrence of the word
``epistemic'' in this article, but the latter terminology is used here
because it has become standard.

Interpretations of probability that involve single-case objective
chances present more of a challenge for the ontic/epistemic
distinction, since they imply that probabilities can at least
sometimes be ontic.  Nevertheless, I believe that an appropriate
distinction can still be made in most of these theories.  This
discussion is deferred to Appendix~\S\ref{App:Chance} since it is
mainly of interest to those concerned with the philosophy of
probability.  However, it is worth mentioning that many of those who
have felt the need to introduce objective chances have been motivated
in part by the role that probability plays in physics, and in quantum
theory in particular.  Since quantum probabilities are functions of
the wavefunction, they only present a novel issue for the
interpretation of probability if the wavefunction itself is ontic
because only then would quantum probabilities need to have a more
objective status than they do in classical physics.  Since the status
of the wavefunction is precisely the question at issue, it is perhaps
wise to defer judgment on the necessity of objective chances until the
reality of the wavefunction is decided.

What is at stake then is the following question: When a quantum state
$\Ket{\psi}$ is assigned to a physical system, does this mean that
there is some independently existing property of the individual system
that is in one-to-one correspondence with $\Ket{\psi}$ (up to a global
phase), or is $\Ket{\psi}$ simply a mathematical tool for determining
probabilities, existing only in the minds and calculations of quantum
theorists?  This is perhaps the most hotly debated issue in all of
quantum foundations.  I refer to it as \emph{the
$\psi$-ontic/epistemic distinction} and use the terms
$\psi$-ontic/$\psi$-epistemic to describe interpretations that adopt
an ontic/epistemic view of the quantum state.  Holders of the
$\psi$-ontic view have been dubbed $\psi$-ontologists by Christopher
Granade (a masters student in Rob Spekkens' quantum foundations course
at Perimeter Institute in 2010) and, continuing in this vein, I refer
to the reality of the quantum state as $\psi$-ontology and to theorems
that attempt to establish this view as $\psi$-ontology theorems.

To avoid misunderstanding, note that the $\psi$-ontic/epistemic
distinction is not about whether quantum states are ontic
independently of whether quantum theory is exactly true.  It is not
about whether the ultimate final theory of physics, if indeed such a
thing exists, will feature quantum states as part of its ontology.  We
have little idea of what such a final theory might look like and
consequently we have little idea of what reality is actually made of
at the most fundamental level.  Nevertheless, we can still ask what
quantum theory itself \emph{says} about reality.  In other words, we
are imagining a hypothetical world in which quantum theory is in fact
a completely correct theory of physics, and asking whether quantum
states would have to be ontological in that world.  That world is very
unlikely to be our actual world, so the question is really about the
internal structure of quantum theory.  More specifically, it is about
what kinds of explanation are compatible with quantum theory.  For
example, a $\psi$-ontic view implies that we should draw analogies
between quantum states and phase space points when comparing quantum
and classical physics, and between the Schr{\"o}dinger equation and
Hamilton's equations, whereas a $\psi$-epistemic view says that the
appropriate analogies are between quantum states and probability
distributions, and between the Schr{\"o}dinger equation and
Liouville's equation.  If nothing else, this strongly impacts how we
are to understand the classical limit of quantum theory (e.g.\ see
\cite{Ballentine1994, Emerson2001, Emerson2001a, Emerson2001b}).  So,
whilst the ontic/epistemic question may at first sight seem abstract
and philosophical, it does in fact have concrete implications for
physics.

The remainder of this part is structured as follows.  \S\ref{PsiEp}
discusses arguments in favor of the $\psi$-epistemic view, with the
aim of convincing you that $\psi$-ontology theorems are telling us
something deep and surprising.  For those that remain unconvinced,
\S\ref{Ontic} reviews the main arguments for the reality of quantum
states that were given prior to the discovery of $\psi$-ontology
theorems.  In my view, none of these are particularly compelling, so
even someone who is already convinced of the reality of quantum states
needs something like a $\psi$-ontology theorem if they aspire to
defend their position with the same sort of conceptual force with
which Bell derived nonlocality.  Following this, \S\ref{Form}
introduces the framework of ontological models, in which
$\psi$-ontology theorems are proven, and gives the rigorous definition
of the $\psi$-ontic/$\psi$-epistemic distinction.  Finally,
\S\ref{Imp} discusses the implications of proving $\psi$-ontology, by
showing that several existing no-go theorems can be derived from it.

\section{Arguments for a $\psi$-epistemic interpretation}

\label{PsiEp}

Before getting into the details of $\psi$-epistemic explanations, it
is important to distinguish two kinds of $\psi$-epistemic
interpretation.  The most popular type are those variously described
as anti-realist, instrumentalist, or positivist.  Since these labels
are often intended as terms of abuse, I prefer to call these
approaches \emph{neo-Copenhagen} in order to avoid implications for
the philosophy of science that go way beyond how we choose to
understand quantum theory.  All such interpretations bear a family
resemblance to the Copenhagen interpretation in that they are both
$\psi$-epistemic and they deny the need for any deeper description of
reality beyond quantum theory.  Here, by ``Copenhagen'' I mean the
views of Bohr, Heisenberg, Pauli et.\ al.\ (see e.g. \cite{Faye2008}),
which are clearly $\psi$-epistemic, rather than the view often found
under this name in textbooks, which is actually due to Dirac
\cite{Dirac1958} and von Neumann \cite{Neumann1955} and is more
ambiguous about whether the wavefunction is real.  If asked what
quantum states represent knowledge about, neo-Copenhagenists are
likely to answer that they represent knowledge about the outcomes of
future measurements, rather than knowledge of some underlying
observer-independent reality.  Modern neo-Copenhagen views include the
Quantum Bayesianism of Caves, Fuchs and Schack \cite{Fuchs2010,
Fuchs2010a, Fuchs2013}, the views of of Bub and Pitowsky
\cite{Pitowsky2006}, the quantum pragmatism of Healey
\cite{Healey2012}, the relational quantum mechanics of Rovelli
\cite{Rovelli1996}, the empiricist interpretation of W. M. de Muynck
\cite{Muynck2002}, as well as the views of David Mermin
\cite{Mermin2013}, Asher Peres \cite{Peres1995}, and Brukner and
Zeilinger \cite{Brukner2003}.  Some may quibble about whether all
these interpretations resemble Copenhagen enough to be called
neo-Copenhagen, but for present purposes all that matters is that
these authors do not view the quantum state as an intrinsic property
of an individual system and they do not believe that a deeper reality
is required to make sense of quantum theory.

The second type of $\psi$-epistemic interpretation are those that are
realist, in the sense that they do posit some underlying ontology.
They just deny that the wavefunction is part of that ontology.
Instead, the wavefunction is to be understood as representing our
knowledge of the underlying reality, in the same way that a
probability distribution on phase space represents our knowledge of
the true phase space point occupied by a classical particle.  There is
evidence that Einstein's view was of this type \cite{Harrigan2010}.
Ballentine's statistical interpretation \cite{Ballentine1970} is also
compatible with this view in that he leaves open the possibility that
hidden variables exist and only insists that, if they do exist, the
wavefunction remains statistical (as a frequentist, Ballentine uses
the term ``statistical'' rather than ``epistemic'').  More recently,
Spekkens has been a strong advocate of this point of view
\cite{Spekkens2007}.

Neo-Copenhagen and realist $\psi$-epistemic interpretations share much
of the same explanatory structure, since they both view probability
measures as the correct classical analogy for the wavefunction.  Many
of the arguments for adopting a $\psi$-epistemic interpretation apply
equally to both of them.  On the other hand, $\psi$-ontology theorems
only apply to realist interpretations.  This is to be expected as it
would be difficult to prove that the wavefunction must be ontic in a
framework that does not admit the existence of ontic states in the
first place.  Because of this, $\psi$-epistemicists always have the
option of becoming neo-Copenhagen in the face of $\psi$-ontology
theorems.

Realist $\psi$-epistemic interpretations are already strongly
constrained by existing no-go theorems, such as Bell's Theorem
\cite{Bell1964} and the Kochen--Specker Theorem \cite{Kochen1967},
which go some way to explaining why not many concrete $\psi$-epistemic
models have been proposed.  However, there is no reason to view these
results as decisive against realist $\psi$-epistemic interpretations
any more than they are decisive against realist $\psi$-ontic
interpretations.  For example, Bohmian mechanics and spontaneous
collapse theories still attract considerable support despite
displaying nonlocality and contextuality, as the existing no-go
theorems imply they must.  Thus, we would be guilty of a
double-standard if we ruled out realist $\psi$-epistemic
interpretations on the basis of these results but still admitted the
possibility of $\psi$-ontic ones.  What is needed is a theorem that
explicitly addresses the $\psi$-ontic/epistemic distinction, and this
is the gap that $\psi$-ontology theorems are intended to fill.

In the remainder of this section, the main arguments in favor of
$\psi$-epistemic interpretations are reviewed.  Because we do not have
a fully worked out realist $\psi$-epistemic model that covers the
whole of quantum theory, it is helpful to introduce toy models that
are similar to quantum theory in some respects, but in which the
analogous notion to the quantum state is clearly epistemic.  These are
intended to demonstrate the kinds of explanation that are possible in
$\psi$-epistemic theories.  Spekkens' toy theory \cite{Spekkens2007},
which reinvigorated interest in realist $\psi$-epistemic models in
recent years, is reviewed in \S\ref{Spek}.  There are also
$\psi$-epistemic models that cover fragments of quantum theory,
e.g.\ just pure state preparations and projective measurements of a
single qubit or just continuous variable systems when restricted to
Gaussian states and operations.  These are reviewed in \S\ref{Frag}.
Finally, I review three further arguments for the $\psi$-epistemic
view based on the fact that quantum theory can be viewed as a
generalization of classical probability theory in \S\ref{GenProb}, on
the collapse of the wavefunction in \S\ref{Collapse}, and on the size
of the quantum state space in \S\ref{Excess}.

\subsection{Spekkens toy bit}

\label{Spek}

Spekkens introduced a toy theory \cite{Spekkens2007} that
qualitatively reproduces the physics of spin-$1/2$ particles (or any
other instantiation of qubits) when they are prepared and measured in
the $x$, $y$ and $z$ bases.  The full version of Spekkens theory
incorporates dynamics and composite systems, including reproducing
some of the phenomena associated with entangled states but, for
illustrative purposes, we restrict attention to the simplest case of a
single \emph{toy bit}.  The toy theory is meant to demonstrate the
explanatory power of $\psi$-epistemic interpretations by providing
natural explanations of many quantum phenomena that are puzzling if
the quantum state is ontic.
\begin{figure}[t!]
\centering
\includegraphics[width=79mm]{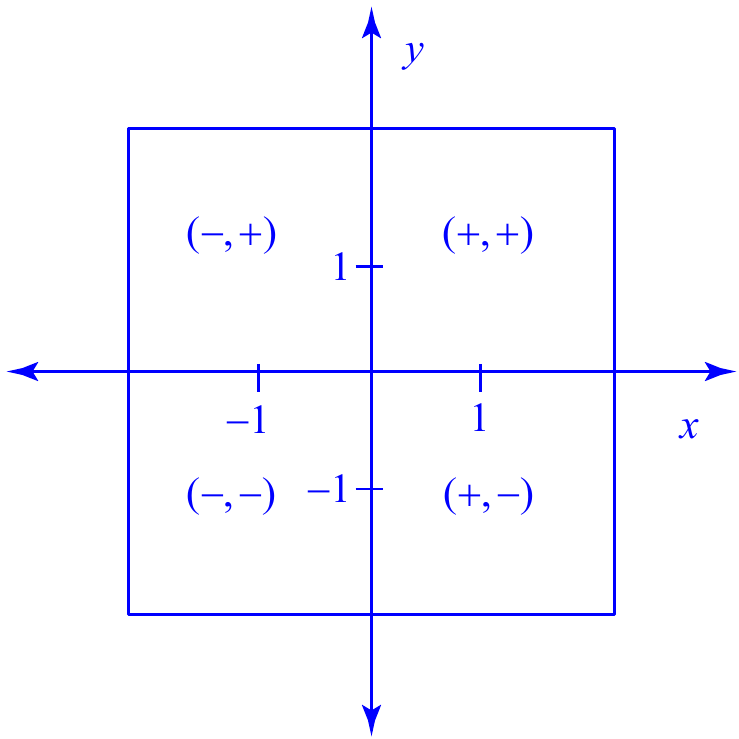}
\caption{\color[HTML]{0000FF}{\label{fig:toy-ontic}The ontic state space of Spekkens
toy-bit. The ontic states are labeled $(x,y)$ where $x$ and $y$
take values $\pm 1$, which are abbreviated to $\pm$ for
compactness.}}
\end{figure}
A toy bit consists of a system that can be in one of four states,
labeled $(-,-), (-,+), (+,-)$ and $(+,+)$.  In
Fig.~\ref{fig:toy-ontic}, these are depicted laid out as a grid in the
$x-y$ plane, with the origin lying at the center of the grid.  The
ontic states can then be thought of as representing the coordinates of
the centers of the grid cells, given by $(x,y)$, with $\pm$ short for
$\pm 1$.  For concreteness, one can imagine that the cells of the grid
represent four boxes and that the system is a ball that can be in one
of them.  The ontic state $(x,y)$ then represents the state of affairs
in which the ball is in the box centered on the coordinates $(x,y)$.

The most fine grained description of the toy bit is always its ontic
state, but we might not know exactly which of the ontic states is
occupied.  In general, our knowledge of the system is described by a
probability distribution over the four ontic states, and this
probability distribution is our epistemic state.  Spekkens imagines
that there is a restriction on the set of epistemic states that may be
assigned to the system, called the \emph{knowledge-balance principle},
which is in some ways analogous to the uncertainty principle.  Roughly
speaking, the knowledge balance principle states that at most half of
the information needed to specify the ontic state can be known at any
given time.  This means, for example, that if we know the
$x$-coordinate with certainty then we cannot know anything about the
$y$-coordinate.  Given this restriction, there are six possible states
of maximal knowledge, termed \emph{pure states}, as shown in the left
hand side of Fig.~\ref{fig:toy-state-meas}.  The pure states are
denoted $\RKet{x \pm}, \RKet{y \pm}, \RKet{z \pm}$ in analogy to the
quantum states $\Ket{x \pm}, \Ket{y \pm}, \Ket{z \pm}$ of a spin-$1/2$
particle.
\begin{figure}[t!]
\centering
\includegraphics[width=80mm]{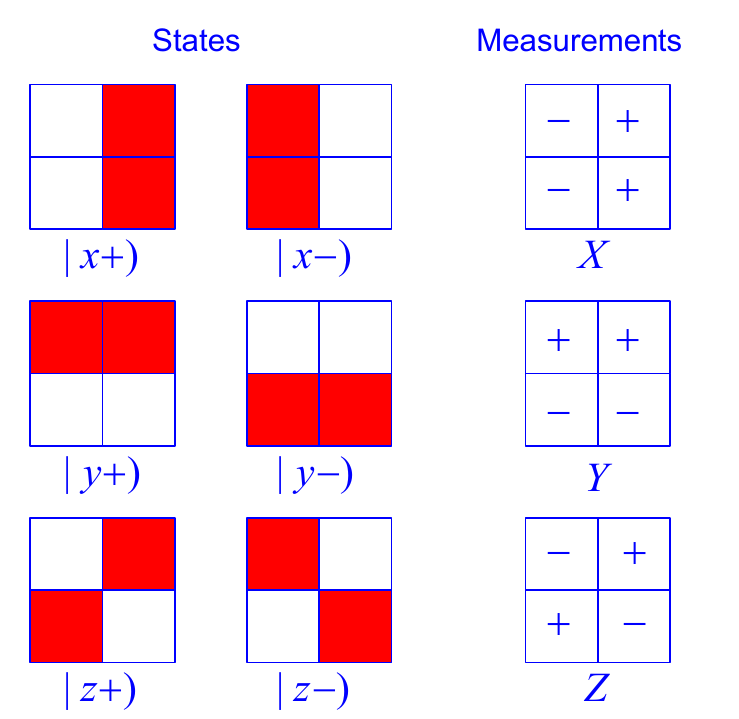}
\caption{\color[HTML]{0000FF}{\label{fig:toy-state-meas}The allowed states and
measurements of Spekkens' toy bit.  For the states, a red square
indicates that the corresponding ontic state has probability $1/2$
and a white square indicates probability $0$.  For the
measurements, a square labeled $+$ gives the $+1$ outcome with
certainty and a square labeled $-$ gives the $-1$ outcome with
certainty.}}
\end{figure}
Note that the epistemic states $\RKet{z \pm}$ are not states with a
definite value of the $z$-coordinate.  The system is two dimensional
so it does not have a third coordinate.  Instead, $\RKet{z+}$ is the
state in which we know only that the $x$ and $y$ coordinates are equal
and $\RKet{z-}$ is the state in which we know only that they are
different.  Defining $z=xy$, this is equivalent to saying that
$\RKet{z \pm}$ is the state in which we know only that $z = \pm 1$.

Although the knowledge balance principle has been imposed by hand, it
is easy to imagine that it could arise from a lack of fine-grained
control over the system.  For example, imagine a preparation device
that pushes the ball to the left along the $x$-axis, but that same
device also causes a random disturbance to the $y$-coordinate, such
that the best we can do after operating the device is to assign the
state $\RKet{x-}$.

Having described the epistemic states of the theory, the next task is
to describe the measurements.  Spekkens requires that measurements be
\emph{repeatable}, which means that if a measurement is repeated twice
in succession then it should yield the same outcome both times.  Also,
the measurement should respect the knowledge balance principle, so
that our epistemic state after the measurement contains at most half
of the information required to specify the ontic state.  In order to
satisfy this second requirement, the measurement must necessarily
cause a disturbance to the ontic state, since otherwise we could end
up in a situation in which we know the ontic state exactly.  For
example, if a measurement of the $x$-coordinate could be implemented
without disturbance then measuring the $x$-coordinate followed by
measuring the $y$-coordinate would tell us the exact ontic state of
the system.

There are three nontrivial measurements that can be implemented in
such a way that they satisfy the two requirements: the $X$ measurement
reveals the $x$ coordinate, the $Y$ measurement reveals the $y$
coordinate, and the $Z$ measurement reveals the value of $z=xy$.
These are illustrated on the right hand side of
Fig.~\ref{fig:toy-state-meas}.  Each of these measurements causes a
random exchange between the pairs of ontic states that give the same
outcome in the measurement.  For example, if we perform an $X$
measurement and get the $+1$ outcome, then with probability $1/2$
nothing happens and with probability $1/2$ the states $(+,-)$ and
$(+,+)$ are exchanged.  This ensures that we always end up in an
epistemic state that satisfies the knowledge-balance principle at the
end of the measurement, in this case $\RKet{x+}$.  It is easy to see
that this is the only type of disturbance that is compatible with both
repeatability and the knowledge-balance principle.  For example, for
an $X$-measurement the random disturbance cannot exchange ontic states
that have different values of the $x$-coordinate, e.g. $(+,+)$ and
$(-,+)$, since this would violate repeatability.

The theory described so far makes exactly the same predictions as
quantum theory for sequences of measurements in the $x$, $y$ and $z$
directions of spin-$1/2$ particles prepared in one of the states
$\Ket{x \pm}$, $\Ket{y \pm}$ and $\Ket{z \pm}$ if we identify these
six states with $\RKet{x \pm}$, $\RKet{y \pm}$ and $\RKet{z \pm}$ and
the Pauli observables $\sigma_x$, $\sigma_y$ and $\sigma_z$ with $X$,
$Y$ and $Z$.  It can thus be regarded as a hidden variable theory for
this kind of experiment.  Further, the quantum states are epistemic in
this representation, as they are each represented by probability
distributions that have support on two ontic states and nonorthogonal
states overlap, e.g. $\RKet{x+}$ and $\RKet{y+}$ both assign
probability $1/2$ to the ontic state $(+,+)$.

Several features of quantum theory that are puzzling on the
$\psi$-ontic view are present in this theory and have very natural
explanations.  Firstly, consider the fact that nonorthogonal pure
states cannot be perfectly distinguished by a measurement, e.g.\ if
either the state $\Ket{x+}$ or the state $\Ket{y+}$ is prepared, and
you do not know which, then there is no measurement that will enable
you to deduce this information with certainty.  If quantum states are
ontic then the two preparations correspond to distinct states of
reality and it is puzzling that we cannot detect this difference.  On
the other hand, the toy theory states $\RKet{x+}$ and $\RKet{y+}$
overlap on the ontic state $(+,+)$ and this will be occupied by the
system $50\%$ of the time whenever $\RKet{x+}$ or $\RKet{y+}$ is
prepared.  When this does happen, there is nothing about the ontic
state of the system that could possibly tell you whether $\RKet{x+}$
or $\RKet{y+}$ was prepared.  Therefore, we must fail to distinguish
the two preparations at least $50\%$ of the time.  The overlap of the
two epistemic states accounts for their indistinguishability.

Another feature of quantum theory that is easily accounted for in
Spekkens' model is the no-cloning theorem.  In quantum theory, there
is no transformation that copies both of two nonorthogonal states.
For example, there is no device that operates with certainty and
outputs both $\Ket{x+}\otimes\Ket{x+}$ when $\Ket{x+}$ is input and
$\Ket{y+}\otimes\Ket{y+}$ when $\Ket{y+}$ is input.  On the
$\psi$-ontic view this is puzzling because the two states represent
distinct states of reality, so one might expect that this distinctness
could be detected and then copied over to another system.  Again, this
is easily explained in Spekkens' model in terms of the overlap between
the epistemic states $\RKet{x+}$ and $\RKet{y+}$.
Fig.~\ref{fig:Spek:Clone} shows the inputs and outputs of the
hypothetical toy-theory cloning machine.  The two input states overlap
on the ontic state $(+,+)$ and this occurs $50\%$ of the time
regardless of which input state is prepared.  Since the cloning
machine only has access to the ontic state, it must do the same thing
to the state $(+,+)$, regardless of whether it occurs because
$\RKet{x+}$ was prepared or because $\RKet{y+}$ was prepared.
Therefore, $50\%$ of the time, the input must get mapped to the same
set of ontic states, with the same probabilities, regardless of which
state was prepared, so there must be at least a $50\%$ overlap of the
output states of any physically possible device for these two input
states.  In contrast, the output states of the hypothetical cloning
machine only overlap on the ontic state $((+,+),(+,+))$ and this must
only occur $25\%$ of the time at the output for either input state.
Therefore, the cloning machine is impossible.
\begin{figure}[t!]
\centering
\includegraphics[width=85mm]{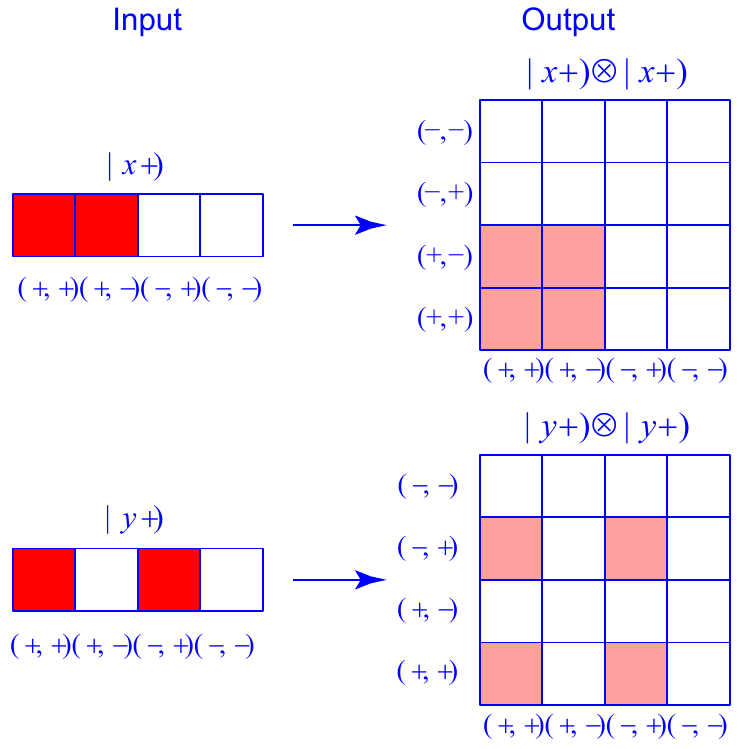}
\caption{\color[HTML]{0000FF}{\label{fig:Spek:Clone}Inputs and outputs of a hypothetical toy bit cloning machine.  In order to represent a two toy bit state, the ontic state space of a single toy bit is represented along one dimension.  For the outputs, the horizontal axis represents the first toy bit and the vertical axis represents the second toy bit.  Dark red represents ontic states occupied with $50\%$ probability and light red represents those occupied with $25\%$ probability. The inputs overlap on an ontic state that they both assign $50\%$ probability, but the outputs only overlap on an ontic state that they both assign $25\%$ probability.}}
\end{figure}
In Spekkens' toy theory, both indistinguishability and no-cloning
follow from the more general fact that a stochastic map cannot
decrease the overlap of two probability distributions. In quantum
theory, there is a similar result that no transformation that can be
implemented with certainty can decrease the inner product between two
pure states \cite{Chefles1998}.  This suggests that the inner product
of two quantum states is analogous to the overlap between two
probability distributions, and this analogy would be most easily
explained if quantum states with nonzero inner product were literally
represented by overlapping probability distributions on some ontic
state space, i.e.\ by a realist $\psi$-epistemic interpretation.

Finally, consider the fact that mixed states in quantum theory have
more than one decomposition into a convex sum of pure states.  For
example, the maximally mixed state of a spin-$1/2$ particle is $I/2$,
where $I$ is the identity operator, and this can be written
alternatively as
\begin{align}
\frac{I}{2} & = \frac{1}{2} \Ket{x+}\Bra{x+} + \frac{1}{2}
\Ket{x-}\Bra{x-} \\
& = \frac{1}{2} \Ket{y+}\Bra{y+} + \frac{1}{2} \Ket{y-}\Bra{y-} \\
& = \frac{1}{2} \Ket{z+}\Bra{z+} + \frac{1}{2} \Ket{z-}\Bra{z-}.
\end{align}
Physically speaking, this means that if we prepare a spin-$1/2$
particle in the $\Ket{x+}$ state with probability $1/2$ and in the
$\Ket{x-}$ state with probability $1/2$ then no experiment can tell
the difference between this ensemble and that formed by preparing it
in the $\Ket{y+}$ state with probability $1/2$ and in the $\Ket{y-}$
state with probability $1/2$ (and similarly for $\Ket{z \pm}$).  On a
$\psi$-ontic view this is puzzling because the $\Ket{x \pm}$ states
are ontologically distinct from the $\Ket{y \pm}$ (and $\Ket{z \pm}$)
states so this difference should be detectable.  However, in Spekkens'
theory this non-uniqueness of decomposition is easily explained because
preparing $\RKet{x+}$ with probability $1/2$ and $\RKet{x-}$ with
probability $1/2$ leads to exactly the same distribution over ontic
states as preparing $\RKet{y+}$ with probability $1/2$ and $\RKet{y-}$
with probability $1/2$ (and similarly for $\RKet{z \pm}$).  This is
illustrated in Fig.~\ref{fig:toy-multi-decomp}.  Note that this is
only possible because the distributions corresponding to nonorthogonal
quantum states overlap.
\begin{figure}[t!]
\centering
\includegraphics[width=70mm]{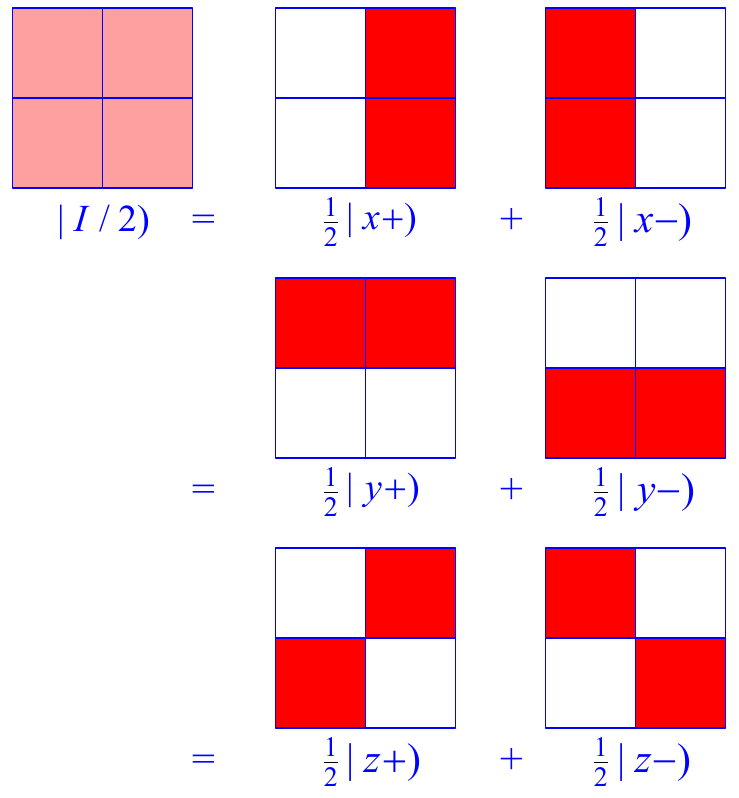}
\caption{\color[HTML]{0000FF}{\label{fig:toy-multi-decomp}Multiple pure-state
decompositions of a mixed state in Spekkens toy theory.  The
maximally mixed state $\RKet{I/2}$ can be written as a $50/50$
mixture in three different ways. Light red indicates a
probability of $1/4$.}}
\end{figure}

What has been presented in this section is just a small fraction of
the quantum phenomena that are accounted for in Spekkens' toy model.
Many more can be found in \cite{Spekkens2007}, but I hope the present
discussion has conveyed a flavor of the type of explanation that is
possible in a realist $\psi$-epistemic theory.

\subsection{Models for fragments of quantum theory}

\label{Frag}

Spekkens' toy model is qualitatively similar to the stabilizer
fragment of quantum theory, which consists of the set of states that
are joint eigenstates of maximal commutative subgroups of the Pauli
group (i.e.\ the group generated by tensor products of the identity and
the three Pauli operators) and has dynamics given by unitaries that
map the Pauli group to itself (see \cite{Gottesman1999} for details of
the stabilizer formalism and \cite{Pusey2012a} for a presentation of
Spekkens' toy theory that closely resembles it).  The stabilizer
fragment is important in quantum information theory as it contains all
the states and operations needed for quantum error correction, as well
as a number of other quantum protocols.  Spekkens' model does not
exactly reproduce the stabilizer fragment when dynamics and
entanglement are taken into account, but other models have been
proposed that do reproduce fragments of quantum theory exactly in a
$\psi$-epistemic manner.

First of all, Spekkens' toy theory has been generalized to larger
dimensions \cite{Schreiber} and to continuous variable systems
\cite{Bartlett2012}.  It turns out that for odd dimensional Hilbert
spaces, Spekkens' model reproduces the stabilizer fragment of quantum
theory exactly.  For continuous variable systems, Spekkens' model
reproduces the Gaussian fragment of quantum theory, in which all
states are Gaussian and the transformations and measurements preserve
the Gaussian nature of the states.  A Gaussian state is one that has a
Gaussian Wigner function.  For a single particle, the Wigner function
is defined in terms of the density operator $\rho$ as $W(x,p) =
\int_{-\infty}^{+\infty} ds e^{ips/\hbar} \Bra{x - \frac{s}{2}} \rho
\Ket{x + \frac{s}{2}}$ and is a pseudo-probability distribution, i.e.\
it is normalized to $1$ but it does not have to be positive.  Gaussian
functions are in fact positive so in this case $W(x,p)$ can be
regarded as a probability distribution and unsurprisingly these are
the epistemic states in Spekkens' continuous variable theory, with the
ontic states being the phase space points $(x,p)$.

Kochen and Specker gave a model for a single qubit that is
$\psi$-epistemic \cite{Kochen1967}.  They were not actually trying to
generate a $\psi$-epistemic model, but rather to provide a
counterexample to their eponymous theorem in $2$-dimensions in order
to show that the theorem requires a Hilbert space of $\geq 3$
dimensions for its proof.  Nevertheless, their model is a paradigmatic
example of a $\psi$-epistemic theory.  The details of this model are
presented in \S\ref{EOM} after we have introduced the formalism for
realist $\psi$-epistemic models more rigorously.  Along similar lines,
Lewis et.\ al.\ \cite{Lewis2012} and Aaronson et.\ al.\
\cite{Aaronson2013} have constructed $\psi$-epistemic models that work
for all finite dimensional systems.  These models were developed as
technical counterexamples to certain conjectures about $\psi$-ontology
theorems and as such they are not very elegant or plausible.  They are
discussed in context in \S\ref{NPIP}.

\subsection{Generalized probability theory}

\label{GenProb}

Apart from specific models, there are also qualitative arguments in
favor of the $\psi$-epistemic view.  The first of these is that
quantum theory can be viewed as a noncommutative generalization of
classical probability theory.  Classically, consider the algebra
$\mathcal{A}$ of random variables on a sample space under pointwise
addition and multiplication.  A probability distribution can then be
regarded as a positive functional $\mu:\mathcal{A} \rightarrow
\mathbb{R}$ that assigns to each random variable its expectation
value.  The quantum generalization of this is to replace the
commutative algebra $\mathcal{A}$ by the noncommutative algebra
$\mathfrak{B} \left ( \Hilb \right )$ of bounded operators on a
Hilbert space $\Hilb$.  A quantum state $\rho$ is isomorphic to a
positive functional $f_{\rho}$ on $\mathfrak{B}(\Hilb)$ given by
$f_{\rho}(M) = \Tr{M \rho}$.  In fact, by a theorem of von Neumann
\cite{Neumann1955}, all positive linear functionals on $\mathfrak{B}
\left ( \Hilb \right )$ that are normalized such that $f(I) = 1$,
where $I$ is the identity operator, are of this form.

Both $\mathcal{A}$ and $\mathfrak{B} \left ( \Hilb \right )$ are
examples of von Neumann algebras, and a generalization of classical
measure theoretic probability can be developed by defining generalized
probability distributions to be positive normalized functionals on
such algebras \cite{Redei2007, Petz2008}.  This generalized theory has
both classical probability theory and quantum theory as special cases.
In this theory, quantum states are playing the same role in the
quantum case that probability measures play in the classical case, and
so it is natural to interpret quantum states and classical
probabilities as the same kind of entity.  Since classical
probabilities are usually interpreted epistemically, it is natural to
interpret quantum states in the same way.

This line of argument would not be too convincing if noncommutative
probability theory were just a formal mathematical generalization with
no practical applications.  However, the theory has a rich array of
applications in quantum statistical mechanics, and especially in
quantum information theory.  The full machinery of von Neumann
algebras is not often needed in quantum information, as we are usually
dealing with finite dimensional systems.  Nevertheless, whenever the
analogy is made between classical probability distributions and
density operators, and between stochastic maps and quantum operations,
generalized probability is at play in the background.  For example, in
quantum compression theory \cite{Schumacher1995}, a density operator
on a finite Hilbert space is viewed as the correct generalization of a
classical information source with finite alphabet, which would be
described by a classical probability distribution.  Similarly, a
quantum channel is described by a quantum operation, and this is
viewed as generalizing a classical channel, which would be modeled as
a stochastic map.

In fact, it is difficult to find any area of quantum information and
computing in which probabilities are not viewed as the correct
classical analogs of quantum states, and this includes areas that
concern themselves exclusively with pure states and unitary
transformations.  For example, the standard circuit model of quantum
computing \cite{Nielsen2000} only employs pure states and unitaries,
but quantum computational complexity classes are most often defined as
generalizations of classical probabilistic complexity classes (see
\cite{Zoo} for definitions of the complexity classes mentioned in this
section.).  The class BQP, usually thought of as the set of problems
that can be solved efficiently on a quantum computer, is sometimes
loosely described as the quantum version of P, the class of problems
that can be solved in polynomial time on a deterministic classical
computer, but in fact it is the generalization of BPP, the set of
problems that can be solved in polynomial time on a probabilistic
classical computer with probability $> 2/3$.  All over quantum
computing theory we find the analogy made to classical probabilistic
computing, and not to classical deterministic computing.

It seems then, that if we take quantum information and computing
seriously, we must take generalized probability theory seriously as
well.  On these and other grounds, I have argued elsewhere
\cite{Leifer2013a, Leifer2013b} that quantum theory is indeed best
viewed as a generalization of probability theory.  The details of this
would take us too far afield, but suffice to say there are good
reasons for viewing quantum states as analogous to probability
distributions and, if we do that, we should try to interpret them both
in the same sort of way.

\subsection{The collapse of the wavefunction}

\label{Collapse}

A straightforward resolution of the collapse of the wavefunction, the
measurement problem, Schr{\"o}dinger's cat and friends is one of the
main advantages of $\psi$-epistemic interpretations.  Recall that the
measurement problem stems from the fact that there are two different
ways of propagating a quantum state forward in time.  When the system
is isolated and not being observed, the quantum state is evolved
smoothly and continuously according to the Schr{\"o}dinger equation.
On the other hand, when a measurement is made on the system, the
quantum state must be updated according to the projection postulate,
leading to the instantaneous and discontinuous collapse of the
wavefunction.  Since a measurement is presumably just some type of
physical interaction between system and apparatus, this poses the
problem of why it is not also modeled by Schr{\"o}dinger evolution.
However, doing so leads to seemingly absurd situations, such as
Schr{\"o}dinger's eponymous cat ending up in a superposition of being
alive and dead at the same time.

The measurement problem is not so much resolved by $\psi$-epistemic
interpretations as it is dissolved by them.  It is revealed as a
pseudo-problem that we were wrong to have placed so much emphasis on
in the first place.  This is because the measurement problem is only
well-posed if we have already established that the quantum state is
ontic, i.e.\ that it is a direct representation of reality.  Only then
does a superposition of dead and alive cats necessarily represent a
distinct physical state of affairs from a definitely alive or
definitely dead cat.  On the other hand, if the quantum state only
represents what we know about reality then the cat may perfectly well
be definitely dead or alive before we look, and the fact that we
describe it by a superposition may simply reflect the fact that we do
not know which possibility has occurred yet.

\subsection{Excess baggage}

\label{Excess}

According to the $\psi$-ontologist, a single qubit contains an
infinite amount of information because a pure state of a qubit is
specified by two continuous complex parameters (ignoring
normalization).  For example, Alice could encode an arbitrarily long
bit string as the decimal expansion of the amplitude of the $\Ket{0}$
state.  However, according to the Holevo bound \cite{Holevo1973}, only
a single bit of classical information can be encoded in a qubit in
such a way that it can be reliably retrieved.  If the quantum state
truly exists in reality, it is puzzling that we cannot detect all of
this extra information.  Hardy has coined the term ``ontological
excess baggage'' to refer to this phenomenon \cite{Hardy2004}.  It
seems that $\psi$-ontologists are attributing a lot more information
to the state of reality than required to explain our observations.

The $\psi$-epistemic response to this is to note that a classical
probability distribution is also specified by continuous parameters.
A probability distribution over a single classical bit requires two
real parameters (again ignoring normalization).  If probabilities were
intrinsic properties of individual systems then this would present a
similar puzzle as there would be an infinite amount of information in
a single bit.  However, classical bits are in fact always either in
the state zero or one and the probabilities simply represent our
knowledge about that value.  In reality, there is just as much
information in a classical bit as we can extract from it, namely one
bit.  If the quantum state is epistemic, then the same resolution is
available to the problem of excess baggage.  The continuous parameters
required to specify the state of a qubit simply represent our
knowledge about it, and the actual ontic state of the qubit, whatever
that may be, might only contain a finite amount of information.

The excess baggage problem is exacerbated by considering how the state
space scales with the number of qubits.  A pure state of $n$ qubits is
specified by $2^n$ complex parameters, but only $n$ bits can be
reliably encoded according to the Holevo bound.  However, the number
of parameters required to specify a probability distribution over $n$
bits also scales exponentially, so the $\psi$-epistemic resolution of
the problem is still available.

In response to this, $\psi$-ontologists might be inclined to point out
that the number of bits that can be reliably encoded in $n$ qubits
depends on how exactly the communication task is defined.  If Alice
and Bob have pre-shared entanglement then Alice can send $2n$ bits to
Bob in $n$ qubits via superdense coding \cite{Bennett1992}.
Similarly, qubits perform better than classical bits in random access
coding \cite{Ambainis1999}, wherein Bob is not required to reliably
retrieve all of the bits that Alice sends, but only a limited number
of them of his choice.  However, the amount of information that Alice
can send to Bob does not scale exponentially with the number of qubits
in any of these protocols, so there is still an excess baggage
problem.

\section{Arguments for a $\psi$-ontic interpretation}

\label{Ontic}

Having reviewed the arguments in favor of $\psi$-epistemic
interpretations, we now look at those that had been put forward in
favor of the reality of quantum states prior to the discovery of
$\psi$-ontology theorems.  Despite receiving a good deal of support, I
hope to convince you that they are far from compelling.  Thus, even
those who are already convinced of the reality of the quantum state
should be interested in establishing their claim rigorously via
$\psi$-ontology theorems.

A big difficulty in extracting arguments for $\psi$-ontology from the
literature is that the majority of authors neglect the possibility of
\emph{realist} $\psi$-epistemic theories.  Instead, they seem to think
that either the wavefunction must be real, or else we must adopt some
kind of neo-Copenhagen approach.  Thus, many purported arguments for
the reality of the wavefunction are really just arguments for the
reality of something, regardless of whether that thing is the
wavefunction.  Since realist $\psi$-epistemic interpretations already
accept the need for an objective reality, such arguments can be
dismissed in the present context.  From amongst these arguments, I
have attempted to sift out those that say something more substantive
about the wavefunction specifically.  I have found four broad classes
of argument, each of which is discussed in turn in this section.
\S\ref{Int} discusses the argument from interference, \S\ref{EEL}
discusses the argument from the eigenvalue-eigenstate link,
\S\ref{ERI} discusses the argument from existing realist
interpretations of quantum theory, and finally \S\ref{QC} discusses
the argument from quantum computation.

\subsection{Interference}

\label{Int}

\begin{quote}
We choose to examine a phenomenon which is impossible,
\emph{absolutely} impossible, to explain in any classical way, and
which has in it the heart of quantum mechanics.  In reality, it
contains the \emph{only} mystery. --- R. P. Feynman
\cite{Feynman2010} [Emphasis in original]
\end{quote}

Following Feynman, single particle interference phenomena, such as the
double slit experiment, are often viewed as containing the essential
mystery of quantum theory.  The problem of explaining the double slit
experiment is usually presented as a dichotomy between explaining it
in terms of a classical wave that spreads out and travels through both
slits or in terms of a classical particle that travels along a
definite trajectory that goes through only one slit.  Neither of these
explanations can account for both the interference pattern and the
fact that it is built out of discrete localized detection events.  A
wave would not produce discrete detection events and a classical
particle would not be affected by whether or not the other slit is
open.  This is taken as evidence that no classical description can
work, and that something more Copenhagen-like must be at work.

Of course, the dichotomy between either classical waves or particles
is a false one.  If we allow the state of reality to be something more
general, i.e.\ some sort of quantum stuff that we do not necessarily
understand yet, then many additional explanations of the experiment
become available.  For example, there is the Bohmian picture in which
both a wave and a particle exist, and the motion of the particle is
guided by the wave.  The wave then explains the interference fringes,
whilst the particle explains the discrete detection events.  This is
by no means the only possibility, but it does highlight the gap in the
usual argument.  Nevertheless, in a realist picture, it seems that
something wavelike needs to exist in order to explain the interference
fringes, and the obvious candidate is the wavefunction.

However, in order to arrive at the conclusion that the wavefunction
must be real, greater leeway has been given in determining what the
ontic state might be like compared to the original argument, which
intended to rule out both particles and waves.  Given this, we should
be careful to rule out other possibilities rigorously, rather than
jumping to the conclusion that the wavefunction must be real.  In this
broader context, the only thing that the double slit experiment
definitively establishes is that there must be some sort influence
that travels through both slits in order to generate the interference
pattern.  It does not establish that this influence must be a
wavefunction.

In fact, interference phenomena occur in some of the previously
discussed $\psi$-epistemic models, so the inference from interference
to the reality of the wavefunction is incorrect.  In Spekkens' toy
theory, a notion of coherent superposition can be introduced such
that, for example, $\RKet{y+}$ is a coherent superposition of
$\RKet{x+}$ and $\RKet{x-}$.  Such superpositions are preserved under
dynamical evolution, so there is a superposition principle in the
theory (see \cite{Spekkens2007} for details).  Further, since all
two-dimensional Hilbert spaces are created equal, there is nothing
special about the interpretation of the toy bit in terms of a
spin-$1/2$ particle.  It could equally well be a model of any other
two-dimensional system.  For example, consider the two dimensional
subspace of an optical mode spanned by the vacuum state $\Ket{0}$ and
the state $\Ket{1}$ where it contains one photon.  The toy bit state
$\RKet{x+}$ can be reinterpreted as $\Ket{0}$ and $\RKet{x-}$ as
$\Ket{1}$, and by doing so a whole host of Mach-Zehnder interferometry
experiments can be qualitatively reproduced by the theory
\cite{Martin2014}.  This includes not only basic interferometry, but
also such seemingly paradoxical effects as the delayed choice
experiment \cite{Wheeler1978} and the Elitzur-Vaidman bomb test
\cite{Elitzur1993}.  In this theory, there is always a fact of the
matter about which arm of the interferometer the photon travels along,
but it does not fall afoul of the standard waves vs.\ particles
argument because the vacuum state has structure.  For example, the
situation in which the photon travels along the left arm of a
Mach-Zehnder interferometer would be represented in quantum theory by
$\Ket{1}_L \otimes \Ket{0}_R$.  The $\Ket{0}_R$ factor would be
represented by the epistemic state $\RKet{x+}$ in the toy theory,
which is compatible with two possible ontic states $(+,+)$ and
$(+,-)$, and these ontic states travel along the right arm of the
interferometer.  Hence, when a photon is in the left arm of an
interferometer, and no photon is in the right arm, there is still a
bit of information traveling along the right arm of the
interferometer, corresponding to whether the ontic state is $(+,+)$ or
$(+,-)$, that can be used to convey information about whether or not
its path was blocked.  There is an influence that travels through both
arms, but that influence is not a wavefunction.

Interference phenomena also occur in all of the models discussed in
\S\ref{Frag} simply because they reproduce fragments of quantum theory
exactly and those fragments contain coherent superpositions.  It is
arguable whether the mechanisms explaining interference in all these
models are plausible, but the main point is that the direct inference
from interference to the reality of the wavefunction is blocked by
them.  If there is an argument from interference to be made then it
will need to employ further assumptions.  Hardy's $\psi$-ontology
theorem, discussed in \S\ref{Hardy}, can be viewed as an attempt at
doing this, but, in light of the way that interference is modeled in
Spekkens' toy theory, its assumptions do not seem all that plausible.

Ultimately, the intuition behind the argument from interference stems
from an analogy with classical fields.  Because wavefunctions can be
superposed, they exhibit interference.  Prior to the discovery of
quantum theory, the only entities in physics that obeyed a
superposition principle and exhibited interference were classical
fields, and these were definitely intended to be taken as real.  For
example, the value of the electromagnetic field at some point in
space-time is an objective property that can be measured by observing
the motion of test charges.  The interference of wavefunctions is then
taken as evidence that they should be interpreted as something similar
to classical fields.

However, the analogy between wavefunctions and fields is only exact
for a single spinless particle, for which the wavefunction is
essentially just a field on ordinary three dimensional space.  This
breaks down for more than one particle, due to the possibility of
entanglement.  The size of the quantum state space scales
exponentially with the number of systems, leading to the previously
discussed excess baggage problem.  The wavefunction can no longer be
viewed as field on ordinary three-dimensional space, so the analogy
with a classical field should be viewed with skepticism.  In
combination with the fact that interference phenomena can be modeled
$\psi$-epistemically, the argument from interference is far from
compelling.

\subsection{The eigenvalue-eigenstate link}

\label{EEL}

The eigenvalue-eigenstate link refers to the tenet of orthodox
quantum theory that when a system is in an eigenstate $\Ket{m}$ of an
observable $M$ with eigenvalue $m$ then $M$ is a property of the
system that has value $m$.  Conversely, when the state is not an
eigenstate of $M$ then $M$ is not a property of the system.  In other
words, the properties of a system consist of all the observables of
which the quantum state is an eigenstate and nothing else.  These
properties are taken to be objectively real, independently of the
observer.

This leads to an argument for the reality of the wavefunction because
the quantum state of a system is determined uniquely by the set of
observables of which it is an eigenstate.  Indeed, it is determined
uniquely by just a single observable, since $\Ket{\psi}$ is an
eigenstate of the projector $\KetBra{\psi}{\psi}$ with eigenvalue $1$
and (up to a global phase) it is the only state in the $+1$ eigenspace
of $\KetBra{\psi}{\psi}$.  The argument is then that, if a system has
a set of definite properties, and those properties uniquely determine
the wavefunction, then the wavefunction itself must be real.

Roger Penrose is perhaps the most prominent advocate of this argument,
so here it is in his own words.

\begin{quotation}
One of the most powerful reasons for rejecting such a subjective
viewpoint concerning the reality of $\Ket{\psi}$ comes from the fact
that whatever $\Ket{\psi}$ might be, there is always---in principle,
at least--a primitive measurement whose \textbf{YES} space consists
of the Hilbert space ray determined by $\Ket{\psi}$. The point is
that the physical state $\Ket{\psi}$ (determined by the ray of
complex multiples of $\Ket{\psi}$) is uniquely determined by the
fact that the outcome \textbf{YES}, for this state, is certain. No
other physical state has this property. For any other state, there
would merely be some probability, short of certainty, that the
outcome will be \textbf{YES}, and an outcome of \textbf{NO} might
occur. Thus, although there is no measurement which will tell us
what $\Ket{\psi}$ actually is, the physical state $\Ket{\psi}$ is
uniquely determined by what it asserts must be the result of a
measurement that might be performed on it\ldots

To put the point a little more forcefully, imagine that a quantum
system has been set up in a known state, say $\Ket{\phi}$, and it is
computed that after a time $t$ the state will have evolved, under
the action of $U$, into another state $\Ket{\psi}$. For example,
$\Ket{\phi}$ might represent the state `spin up' ($\Ket{\phi} =
\Ket{\uparrow}$) of an atom of spin $\frac{1}{2}$, and we can
suppose that it has been put in that state by the action of some
previous measurement. Let us assume that our atom has a magnetic
moment aligned with its spin (i.e.\ it is a little magnet pointing
to the spin direction).  When the atom is placed in a magnetic
field, the spin direction will precess in a well-defined way, that
can be accurately computed as the action of $U$, to give some new
state, say $\Ket{\psi} = \Ket{\rightarrow}$, after a time $t$. Is
this computed state to be taken seriously as part of physical
reality? It is hard to see how this can be denied. For $\Ket{\psi}$
has to be prepared for the possibility that we might choose to
measure it with the primitive measurement referred to above, namely
that whose \textbf{YES} space consists precisely of the multiples of
$\Ket{\psi}$. Here, this is the spin measurement in the direction
$\rightarrow$. The system has to know to give the answer
\textbf{YES}, with certainty for that measurement, whereas no spin
state of the atom other than $\Ket{\psi} = \Ket{\rightarrow}$ could
guarantee this. --- Roger Penrose, quoted in \cite{Fuchs2011}.
\end{quotation}

This argument can be easily countered using any of the existing
$\psi$-epistemic models, to which the same reasoning would apply.  For
example, consider the epistemic state $\RKet{x+}$ in Spekkens' toy
theory.  This state assigns the definite value $+1$ to the $X$
measurement and indeed it is the only allowed epistemic state in the
theory that does this.  In fact, all of the pure states of a toy bit
$\RKet{x \pm}, \RKet{y \pm}$ and $\RKet{z \pm}$ are uniquely
determined by the definite value that they assign to one of the
measurements.  Following Penrose's reasoning, we would then conclude
that these states are objective properties of the system.  However,
this is not the case since the objective properties of the system are
those that are determined by the ontic state, and each ontic state is
compatible with more than one epistemic state.  For example, the ontic
state $(+,+)$ is compatible with $\RKet{x+}$, $\RKet{y+}$ and
$\RKet{z+}$.  Given the complete specification of reality, the
epistemic state is underdetermined.

One way of exposing the error in the eigenvalue-eigenstate argument is
to note that, in the toy theory, the fact that observables uniquely
determine epistemic states is a consequence of the knowledge-balance
principle and not a fundamental fact about reality.  For example,
without the knowledge-balance principle, it would be permissible to
have an epistemic state that assigns probability $2/3$ to $(+,+)$ and
$1/3$ to $(+,-)$.  Just like $\RKet{x+}$, this state assigns
probability $1$ to the $+1$ outcome of the $X$ measurement and this is
the only measurement that is assigned a definite value.  If this state
were allowed then it would no longer be possible to mistake
$\RKet{x+}$ for an objective property of the system.  Penrose has
mistaken the set of states that it is possible to prepare with current
experiments for the set of all logically possible states.

Another way of exposing the error is to look at the restrictions on
measurements in the toy theory.  The measurements only reveal
coarse-grained information about the ontic state.  Without the
knowledge-balance principle it would be permissible to conceive of a
more fine-grained measurement that reveals the ontic state exactly.
This measurement reveals a definite property of the system because it
is determined uniquely by the ontic state.  However, specifying this
observable no longer uniquely determines the epistemic state.  For
example, if we learn that the ontic state is $(+,+)$ then this is
compatible with $\RKet{x+}$, $\RKet{y+}$ and $\RKet{z+}$.

In conclusion, the mistake in the eigenvalue-eigenstate argument is to
assume that the observables that we can actually measure in
experiments form the sum total of all the properties of the system and
to assume that the set of states that we can actually prepare are the
sum total of all logically conceivable states.  Without these
assumptions, the argument is simply false.

\subsection{Existing realist interpretations}

\label{ERI}

There are a handful of fully worked out realist interpretations of
quantum theory, including many-worlds \cite{Everett1957, DeWitt1973,
Wallace2012}, de Broglie--Bohm theory \cite{Broglie2009, Bohm1952,
Bohm1952a, Duerr2009}, spontaneous collapse theories
\cite{Ghirardi1986, Bassi2013} and modal interpretations
\cite{Lombardi2013}.  In each of these interpretations the
wavefunction is part of the ontic state, so there is an argument from
lack of imagination to be made: since all the interpretations of
quantum theory that we have managed to come up with that are
uncontroversially realist have a real wavefunction, then the
wavefunction must be real.

I admit that it behooves the realist $\psi$-epistemicist to try to
construct a fully worked out interpretation.  However, absence of
evidence is not the same thing as evidence of absence.  Nevertheless,
I have frequently heard this argument made in private conversations.
Some people seem to think that since we have a bunch of well worked
out interpretations, we ought to simply pick one of them and not
bother thinking about other possibilities.  Ever since the inception
of quantum theory we have been beset by the problem of quantum jumps,
by which I mean that quantum theorists are liable to jump to
conclusions.

Despite the obvious weakness of this argument, there is a more subtle
point to be made.  John Bell was motivated to work on his eponymous
theorem by noting that de Broglie--Bohm theory exhibited nonlocality.
He wanted to know if this was just a quirk of de Broglie--Bohm theory
or an inescapable property of any realist interpretation of quantum
theory.  With this motivation, he ended up proving the latter.  The
lesson of this is that if we find that all realist interpretations of
quantum theory share a property that some find objectionable then we
ought to determine whether or not this is a necessary property.
However, this is a motivation for developing $\psi$-ontology theorems,
rather than regarding the matter as settled a priori.

\subsection{Quantum computation}

\label{QC}

The final argument I want to consider is due to David Deutsch, who put
it forward as an argument in favor of the many-worlds interpretation.
However, I think the argument can be adapted, more generally, into an
argument for the reality of the wavefunction.  Here is the argument in
Deutsch's own words.

\begin{quote}
To predict that future quantum computers, made to a given
specification, will work in the ways I have described, one need only
solve a few uncontroversial equations. But to explain exactly
\emph{how} they will work, some form of multiple-universe language
is unavoidable. Thus quantum computers provide irresistible evidence
that the multiverse is real. One especially convincing argument is
provided by quantum algorithms [\ldots] which calculate more
intermediate results in the course of a single computation than
there are atoms in the visible universe. When a quantum computer
delivers the output of such a computation, we shall know that those
intermediate results must have been computed somewhere, because they
were needed to produce the right answer. So I issue this challenge
to those who still cling to a single-universe world view: \emph{if
the universe we see around us is all there is, where are quantum
computations performed?} I have yet to receive a plausible
reply. --- David Deutsch \cite{Deutsch1998} [Emphasis in original].
\end{quote}

Quantum algorithms that offer exponential improvement over existing
classical algorithms, such as Shor's factoring algorithm
\cite{Shor1997}, start by putting a quantum system in a superposition
of all possible input strings.  Then, some computation is done on each
of the strings before using interference effects between them to
elicit the answer to the computation.  If each of the branches of the
wavefunction were not individually real, whether or not they are
interpreted in a many-worlds sense, then where does the computation
get done?

This is not exactly an argument for the reality of the wavefunction,
but it is at least an argument that the size of ontic state space
should scale exponentially with the number of qubits, and that the
ontic state should contain pieces that look like the branches of a
wavefunction.  However, whilst I agree with Deutsch that an
interpretation of quantum theory should offer an explanation of how
quantum computations work, it is not at all obvious that the
explanation must be a direct translation of what happens to the
wavefunction.  The argument would perhaps be more compelling if there
were known exponential speedups for problems where we think that the
best we can do classically is to just search through an exponentially
large set of solutions, since we could then argue that a quantum
computer must be doing just that.  This would be the case if we had
such a speedup for the traveling salesman problem, or any other NP
complete problem.  The sort of problems for which we do have
exponential speedup, such as factoring, are more subtle than this.
They lie in NP, but are not NP complete.  If we were to find an
efficient classical algorithm for these problems then it would not
cause the whole structure of computational complexity theory to come
crashing down.  If such an algorithm exists, then whatever deeper
theory underlies quantum theory may be exploiting this same structure
to perform the quantum computation.

Even if such a scenario does not play out, Deutsch's argument is not
decisive against realist $\psi$-epistemic interpretations.  Since we
have not yet constructed a viable interpretation of this sort that
covers the whole of quantum theory, who knows what explanations such a
theory might provide?  Therefore, as Deutsch says, explaining quantum
computation ought to be viewed as a challenge for the $\psi$-epistemic
program rather than an argument against it.

\section{Formalizing the $\psi$-ontic/epistemic distinction}

\label{Form}

Hopefully, by this point I have convinced you that it is worth trying
to settle the question of the reality of the wavefunction rigorously.
The aim of this section is to provide a formal definition of what it
means for the quantum state to be ontic or epistemic within a realist
model of quantum theory.  This is usually done within the framework of
\emph{ontological models}.  This is really no different from the
framework that Bell used to prove his eponymous theorem, and an
ontological model is sometimes alternatively known as a hidden
variable theory.  However, I prefer the term ``ontological model''
because there is a lot of confusion about the meaning of the term
``hidden variables''.  Following the example of Bohmian mechanics, a
hidden variable theory is often thought to be a theory in which some
additional variables are posited alongside the wavefunction, which is
itself conceived of as ontic from the start.  Since the reality of the
wavefunction is precisely the point at issue, we definitely want to
include models in which it is not assumed to be real within our
framework.  In addition, in order to cover orthodox quantum theory, we
want our framework to include models in which the wavefunction is the
\emph{only} thing that is real, i.e.\ there are no additional hidden
variables.  A further confusion is the commonly held view that a
hidden variable theory must restore determinism, whereas we want to
allow for the possibility that nature might be genuinely stochastic.
For these reasons, I prefer to use the term ``ontological model''.  It
is either the same thing as a hidden variable theory or more general,
depending on how general you thought hidden variable theories were in
the first place.

Whilst the Hardy and Colbeck--Renner Theorems involve assumptions about
how dynamics are represented in an ontological model, the Pusey--Barrett--Rudolph Theorem
only involves prepare-and-measure experiments, i.e.\ a system is
prepared in some quantum state and is then immediately measured and
discarded.  Therefore, we deal with prepare-and-measure experiments
first and defer discussion of dynamics until it is needed.  A
$\psi$-ontology theorem aims at proving that any ontological model
that reproduces quantum theory must have ontic quantum states.  This
does not apply to arbitrary fragments of quantum theory, since we have
seen in \S\ref{Frag} that there are fragments that can be modeled
with epistemic quantum states.  In order to understand both cases, we
need to define ontological models for fragments of quantum theory
rather than just for quantum theory as a whole.  The formal definition
of a prepare-and-measure fragment of quantum theory is given in
\S\ref{PME} and then \S\ref{OM} explains how these are represented in
ontological models, with examples given in \S\ref{EOM}.  Based on
this, \S\ref{POEM} gives the formal definition of what it means for
the quantum state to be ontic or epistemic.

\subsection{Prepare and measure experiments}

\label{PME}

In a prepare-and-measure experiment, the experimenter performs a
preparation of some physical system and then immediately measures it,
records the outcome, and discards the system.  The experimenter can
repeat the whole process of preparing and measuring as many times as
she likes in order to build up frequency statistics for comparison
with the probabilities predicted by some physical theory.  Each run of
the experiment is assumed to be statistically independent of the
others and it is assumed that the experimenter can choose which
measurement to perform independently of the choice of preparation.

For completeness, we consider the most general type of quantum
state---a density operator---and the most general type of observable
---a Positive Operator Valued Measure (POVM), although we restrict
attention to POVMs with a finite number of outcomes.  Readers
unfamiliar with these concepts should consult a standard textbook,
such as \cite{Nielsen2000} or \cite{Heinosaari2011}.

\begin{definition}
\label{def:Form:PM}
A \emph{prepare-and-measure (PM) fragment} $\mathfrak{F} = \langle
\Hilb, \mathcal{P}, \mathcal{M} \rangle$ of quantum theory consists
of a Hilbert space $\mathcal{H}$, a set $\mathcal{P}$ of density
operators on $\mathcal{H}$, and a set $\mathcal{M}$ of POVMs on
$\mathcal{H}$.  The probability of obtaining the outcome
corresponding to a POVM element $E \in M$ when performing a
measurement $M \in \mathcal{M}$ on a system prepared in the state
$\rho \in \mathcal{P}$ is given by the Born rule
\begin{equation}
\label{eq:Form:qprob}
\text{Prob}(E|\rho, M) = \Tr{E \rho}.
\end{equation}
\end{definition}

For many of the results reviewed here, the PM fragment under
consideration is the one in which $\mathcal{P}$ is the set of all pure
states on $\Hilb$ and $\mathcal{M}$ consists of measurements in a set
of complete orthonormal bases.  The formalism of PM fragments allows
the sets of states and measurements required to make $\psi$-ontology
theorems go through to be made explicit.  Additionally, many of the
intermediate results used in proving $\psi$-ontology theorems apply to
PM fragments in general, including those that feature mixed states and
POVMs, so it is worth introducing fragments at this level of
generality.

When only a single PM fragment is under consideration, it is assumed
to be denoted $\langle \Hilb, \mathcal{P}, \mathcal{M} \rangle$ so
that the notations $\rho \in \mathcal{P}$ and $M \in \mathcal{M}$ can
be used without first explicitly writing down the triple.

\subsection{Ontological models}

\label{OM}

The idea of an ontological model of a PM fragment is that there is
some set $\Lambda$ of \emph{ontic states} that give a complete
specification of the properties of the physical system as they exist
in reality.  When a quantum system is prepared in a state $\rho \in
\mathcal{P}$, what really happens is that the system occupies one of
the ontic states $\lambda \in \Lambda$.  However, the preparation
procedure may not completely control the ontic state, so our knowledge
of the ontic state is described by a probability measure $\mu$ over
$\Lambda$.  This means that $\Lambda$ needs to be a measurable space,
with a $\sigma$-algebra $\Sigma$, and that $\mu:\Sigma \rightarrow
[0,1]$ is a $\sigma$-additive function satisfying $\mu(\Lambda) = 1$.
For those unfamiliar with measure-theoretic probability, for a finite
space $\Sigma$ would just be the set of all subsets of $\Lambda$ and
$\sigma$-additivity reduces to $\mu(\Omega_1\cup\Omega_2) =
\mu(\Omega_1) + \mu(\Omega_2)$ for all disjoint subsets $\Omega_1$ and
$\Omega_2$ of $\Lambda$.

In general, different methods of preparing the same quantum state may
result in different probability measures over the ontic states.  This
is especially true of mixed states, since they do not have unique
decompositions into a convex mixture of pure states.  If one prepares
a mixed state by choosing randomly from one of the pure states in such
a decomposition, then one can prove that the probability measure must
in general depend on the choice of decomposition.  This is known as
\emph{preparation contextuality}, and is discussed in more detail in
\S\ref{PC}.  For this reason, a quantum state $\rho$ is associated
with a set $\Delta_{\rho}$ of probability measures rather than just a
single unique measure.  Note that it is possible to find models in
which pure states correspond to unique measures, and much of the
literature implicitly assumes this type of model.  However, it turns
out that this assumption is not necessary, so we allow for the
possibility that even a pure state is represented by a set of
probability measures for the sake of generality.

Turning now to measurements, the outcome of a measurement $M \in
\mathcal{M}$ might not reveal $\lambda$ exactly but depend on it only
probabilistically.  This could be because nature is fundamentally
stochastic, but it could also arise in a deterministic theory if the
response of the measuring device depends not only on $\lambda$ but
also on degrees of freedom within the measuring device that are not
under the experimenter's control (see \cite{Harrigan2007} for a
discussion of this).  To account for this, each POVM $M$ is
represented by a conditional probability distribution
$\text{Pr}(E|M,\lambda)$ over $M$.  In a bit more detail, when a
measurement $M \in \mathcal{M}$ is performed, each $\lambda$ must give
rise to a well defined probability distribution over the outcomes, so,
for any fixed $\lambda$, we must have $\text{Pr}(E|M,\lambda) \geq 0$
for all $E \in M$ and $\sum_{E \in M} \text{Pr}(E|M,\lambda) = 1$.
Further, in order to calculate the probabilities that the model
predicts we will observe, we are going to have to average over our
ignorance about $\lambda$ so, for any fixed $E$,
$\text{Pr}(E|M,\lambda)$ must be a measurable function of $\lambda$.
See Fig.~\ref{fig:OM:response} for an illustration of the conditional
probability distribution corresponding to a measurement.
\begin{figure}[t!]
\centering
\includegraphics[width=85mm]{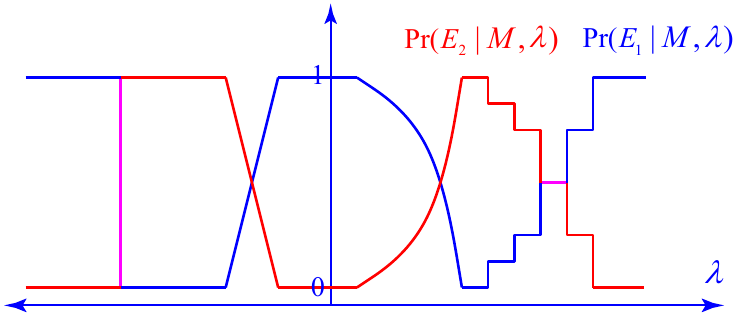}
\caption{\color[HTML]{0000FF}{\label{fig:OM:response}A possible conditional probability
distribution $\text{Pr}(E_j|M,\lambda)$ for a two outcome
measurement $M = \{E_1,E_2\}$ in an ontological model.
$\text{Pr}(E_j|M,\lambda)$ is the conditional probability of
obtaining the outcome $E_j$ when $M$ is measured and the ontic
state is $\lambda$.  As a result, it must satisfy
$\text{Pr}(E_1|M,\lambda) + \text{Pr}(E_2|M,\lambda) = 1$ for all
values of $\lambda$.  For illustrative purposes, the ontic state
space is represented as a $1$-dimensional line, but it may
actually be an arbitrary measurable space.}}
\end{figure}

Again, different methods of implementing the same measurement may
result in different conditional probability distributions. This is
known as \emph{measurement contextuality} and it must occur in certain
types of model due to the Kochen--Specker theorem \cite{Kochen1967}.
Measurement contextuality is discussed further in
Appendix~\ref{App:KS}.  For this reason, a measurement $M$ is
associated with a set $\Xi_M$ of conditional probability distributions
rather than just a single one.

In order to compute the probabilities that the ontological model
predicts for the observable outcomes of measurements, the conditional
probabilities $\text{Pr}(E|M,\lambda)$ have to be averaged over our
ignorance of the true ontic state as specified by $\mu$.  This
gives,
\begin{equation}
\text{Pr}(E|\rho,M) = \int_{\Lambda} \text{Pr}(E|M,\lambda)
d\mu(\lambda).
\end{equation}
Note that, the notation $\text{Pr}$ is used for the probabilities in
an ontological model to distinguish them from the probabilities
$\text{Prob}$ predicted by quantum theory.

Finally, if the ontological model is to reproduce the predictions of
quantum theory, then; for each $\rho \in \mathcal{P}$, $M \in
\mathcal{M}$, each $\mu \in \Delta_{\rho}$, $\text{Pr} \in \Xi_M$ must
satisfy
\begin{equation}
\label{eq:Form:Rep}
\text{Pr}(E|\rho,M) = \text{Prob}(E|\rho,M),
\end{equation}
or in other words
\begin{equation}
\label{eq:Form:RepQ}
\int_{\Lambda} \text{Pr}(E|M,\lambda) d\mu(\lambda) = \Tr{E \rho}.
\end{equation}

Summarizing, we have
\begin{definition}
\label{def:Form:OM}
An \emph{ontological model} of a PM fragment is a
quadruple $(\Lambda, \Sigma, \Delta, \Xi)$ that consists of:
\begin{itemize}
\item A measurable space $(\Lambda, \Sigma)$, where $\Lambda$ is
called the \emph{ontic state space}.
\item A function $\Delta$ that maps each quantum state $\rho \in
\mathcal{P}$ to a set of probability measures $\Delta[\rho] =
\Delta_{\rho}$ on $(\Lambda,\Sigma)$.
\item A function $\Xi$ that maps each POVM $M \in \mathcal{M}$ to a
set of conditional probability distributions over $M$, $\Xi[M] =
\Xi_M$, i.e.\ each $\text{Pr} \in \Xi_M$ is a function from $M
\times \Lambda$ to $\mathbb{R}$ that is measurable as a function
of $\lambda \in \Lambda$ and satisfies, for all $\lambda \in
\Lambda$,
\begin{equation}
\forall E \in M, \,\,\,\, \text{Pr}(E|M,\lambda) \geq 0
\end{equation}
and
\begin{equation}
\sum_{E \in M} \text{Pr}(E|M,\lambda) = 1.
\end{equation}
\end{itemize}

The ontological model \emph{reproduces the quantum predictions} if,
for all $\rho \in \mathcal{P}$ and $M \in \mathcal{M}$, each $\mu
\in \Delta_{\rho}$ and $\text{Pr} \in \Xi_M$ satisfies
\begin{equation}
\label{eq:Form:Rep2}
\forall E \in M, \,\,\,\, \int_{\Lambda} \text{Pr}(E|M,\lambda) d
\mu(\lambda) = \Tr{E \rho}.
\end{equation}
\end{definition}

In what follows, it is convenient to use projectors
$\KetBra{\psi}{\psi}$ to represent pure states instead of vectors
$\Ket{\psi}$ to avoid the global phase ambiguity.  The set of
projectors onto the pure states of $\mathcal{H}$ is known as the
\emph{projective Hilbert space} of $\mathcal{H}$.  The notation
$\Proj{\psi} = \KetBra{\psi}{\psi}$ and $\Delta_{\psi} =
\Delta_{\Proj{\psi}}$ is used to reduce clutter.  Similarly, if a
measurement $M$ consists of projectors onto pure states $\Proj{\phi}$
then we use the notation $\text{Pr}(\phi|M,\lambda)$ as shorthand for
$\text{Pr}(\Proj{\phi}|M,\lambda)$.

\subsection{Examples of Ontological Models}

\label{EOM}

\begin{example}[Spekkens' toy bit]
\label{exa:EOM:spek}
Spekkens' toy bit can be recast as an ontological model of the
PM fragment $\langle \mathbb{C}^2, \mathcal{P},
\mathcal{M} \rangle$ where
\begin{equation}
\mathcal{P} = \left \{\Proj{x+},\Proj{x-}, \Proj{y+},\Proj{y-},
\Proj{z+},\Proj{z-},I/2 \right \},
\end{equation}
and
\begin{equation}
\mathcal{M} = \left \{ \left \{ \Proj{x+},\Proj{x-} \right \},
\left \{\Proj{y+},\Proj{y-} \right \}, \left \{
\Proj{z+},\Proj{z-} \right \} \right \}.
\end{equation}
The ontic state space is $\Lambda = \{(+,+),(+,-),(-,+),(-,-)\}$.
Since $\Lambda$ is a finite set, the integral in
Eq.~\eqref{eq:Form:Rep2} is just a sum
\begin{equation}
\sum_{\lambda \in \Lambda} \text{Pr}(E|M,\lambda)\mu(\lambda) =
\Tr{E \rho},
\end{equation}
where, abusing notation slightly, we write $\mu(\lambda)$ for
$\mu(\{\lambda\})$.

In this model, each quantum state is represented by a unique
probability measure.  The six probability functions $\RKet{x \pm},
\RKet{y \pm}, \RKet{z \pm}$ illustrated in
Fig.~\ref{fig:toy-state-meas} represent the states $\Proj{x \pm},
\Proj{y \pm}, \Proj{z \pm}$, and $\RKet{I/2}$ from
Fig.~\ref{fig:toy-multi-decomp} represents the maximally mixed state
$I/2$.  Each measurement is associated with a unique conditional
probability distribution as described in
Fig.~\ref{fig:toy-state-meas}.  For example, for the $X =
\{\Proj{x+},\Proj{x-}\}$ measurement we have
\begin{equation}
\begin{array}{lllll}
\text{Pr}(x+|X,(+,+)) & = & \text{Pr}(x+|X,(+,-)) & = & 1 \\
\text{Pr}(x+|X,(-,+)) & = & \text{Pr}(x+|X,(-,-)) & = & 0 \\
\text{Pr}(x-|X,(+,+)) & = & \text{Pr}(x-|X,(+,-)) & = & 0 \\
\text{Pr}(x-|X,(-,+)) & = & \text{Pr}(x-|X,(-,-)) & = & 1.
\end{array}
\end{equation}
It is easy to check that this reproduces the quantum predictions for
the fragment.
\end{example}

\begin{example}[The Beltrametti--Bugajski model \cite{Beltrametti1995}]
This model is essentially a translation of the orthodox
interpretation of quantum theory into the language of ontological
models.  The PM fragment is $\langle \mathbb{C}^d, \mathcal{P},
\mathcal{M}\rangle$ where $\mathcal{P}$ contains all the pure states
on $\mathbb{C}^d$ and $\mathcal{M}$ consists of all POVMs on
$\mathbb{C}^d$.  The idea of the model is that the quantum state,
and only the quantum state, represents reality.  Therefore, the
ontic state space $\Lambda$ is just the set of pure states
$\Proj{\lambda}$ for $\Ket{\lambda} \in \mathbb{C}^d$, i.e.\ states
differing by a global phase are identified so $\Lambda$ is the
projective Hilbert space of $\mathbb{C}^d$.  This space carries a
natural topology induced by the inner product, and we take $\Sigma$
to be the Borel $\sigma$-algebra of this topology (see
\cite{Stulpe2007} for details).

A pure quantum state $\Proj{\psi}$ is then represented by the point
measure
\begin{equation}
\mu(\Omega) = \delta_{\psi}(\Omega) = \begin{cases} 1 &
\text{if}\,\, \Proj{\psi} \in \Omega \\
0 & \text{if}\,\, \Proj{\psi} \notin \Omega. \end{cases}
\end{equation}

Each POVM is represented by a unique conditional probability
distribution and, since the quantum state is the ontic state in this
model, it should come as no surprise that the response functions
simply specify the quantum probabilities, i.e.\ if a POVM $M \in
\mathcal{M}$ contains the operator $E$ then
\begin{equation}
\text{Pr}(E|M,\Proj{\lambda}) = \Tr{E \Proj{\lambda}}.
\end{equation}

We then trivially have
\begin{eqnarray}
\int_{\Lambda} \text{Pr}(E|M,\Proj{\lambda}) d\mu(\Proj{\lambda})  &=&
\int_{\Lambda} \Tr{E \Proj{\lambda}} d
\delta_{\psi}(\Proj{\lambda})\nonumber\\
&=& \Tr{E \Proj{\psi}},
\end{eqnarray}
so the model reproduces the quantum predictions.

The model can be extended to mixed states by writing them as convex
combinations of pure states.  For example, the maximally mixed state
of a spin-$1/2$ particle can be written as
\begin{equation}
\frac{I}{2} = \frac{1}{2} \Proj{x+} + \frac{1}{2} \Proj{x-}.
\end{equation}
This means that $I/2$ can be prepared by flipping a coin, preparing
$\Proj{x+}$ if the outcome is heads or $\Proj{x-}$ if the outcome is
tails, and then forgetting or erasing the result of the coin flip.
Let's call this preparation procedure $P$.  It follows that the
predictions of $I/2$ can be reproduced by the measure
\begin{equation}
\mu = \frac{1}{2}\delta_{x+} + \frac{1}{2}\delta_{x-}.
\end{equation}
Note however that this is a preparation contextual model because the
maximally mixed state can also be prepared by mixing a different set
of pure states, e.g.\ the decomposition
\begin{equation}
\frac{I}{2} = \frac{1}{2} \Proj{y+} + \frac{1}{2} \Proj{y-},
\end{equation}
yields the measure
\begin{equation}
\mu = \frac{1}{2}\delta_{y+} + \frac{1}{2}\delta_{y-},
\end{equation}
which reproduces the quantum predictions just as well.  In general,
$\Delta_{\rho}$ consists of one convex combination of point measures
for each of the different ways of writing $\rho$ as a mixture of
pure states.
\end{example}

\begin{example}[The Bell Model \cite{Bell1966}]
\label{exa:EOM:Bell}
In his review of no-go theorems for hidden variable theories
\cite{Bell1966}, Bell introduced an ontological model for
measurements in orthonormal bases on systems with a two-dimensional
Hilbert space.  The obvious generalization to arbitrary finite
dimensional systems is presented here, as it is needed in
\S\ref{NPIP}.  Bell intended his model as a pedagogical device to
point out the flaws in previous no-go theorems, and he was primarily
interested in whether a deterministic theory could reproduce the
quantum predictions.  In modern terms, the Bell model can be thought
of as a minimal modification of the Beltrametti--Bugajski model,
intended to make it deterministic.

The PM fragment of interest is $\mathfrak{F} = \langle \mathbb{C}^d,
\mathcal{P}, \mathcal{M} \rangle$, where $\mathcal{P}$ consists of
all pure states and $\mathcal{M}$ consists of all measurements of
the form $\mathcal{M} = \left \{ \Proj{\phi_j} \right \}_{j=0}^{d-1}$,
where $\left \{ \Ket{\phi_j} \right \}_{j=0}^{d-1}$ is an orthonormal
basis for $\mathbb{C}^d$.

The generalized Bell model employs an ontic state space $\Lambda =
\Lambda_1 \times \Lambda_2$ that is the Cartesian product of two
state spaces.  As in Beltrametti--Bugajski, $\Lambda_1$ is the
projective Hilbert space of $\mathbb{C}^d$, with Borel
$\sigma$-algebra $\Sigma_1$.  $\Lambda_2$ is the unit interval
$[0,1]$, with Borel $\sigma$-algebra $\Sigma_2$, representing an
additional hidden variable.  The measurable space is then
$(\Lambda_1 \times \Lambda_2, \Sigma_1 \otimes \Sigma_2)$, where the
tensor product $\sigma$-algebra $\Sigma_1 \otimes \Sigma_2$ is the
$\sigma$-algebra generated by sets of the form $\Omega_1 \times
\Omega_2$ with $\Omega_1 \in \Sigma_1$, $\Omega_2 \in \Sigma_2$.

The quantum state $\Proj{\psi}$ is represented by a product measure
\begin{equation}
\mu(\Omega) = \int_{\Lambda_2} \mu_1(\Omega_{\lambda_2})
d\mu_2(\lambda_2),
\end{equation}
where $\mu_1$ and $\mu_2$ are probability measures on $\Lambda_1$
and $\Lambda_2$ respectively, $\Omega \in \Sigma_1 \otimes
\Sigma_2$, and
\begin{equation}
\Omega_{\lambda_2} = \left \{ \Proj{\lambda_1} \in
\Lambda_1 \middle | \left ( \Proj{\lambda_1}, \lambda_2 \right )
\in \Omega \right \}.
\end{equation}
As in the Beltrametti--Bugajski model, $\Proj{\lambda_1}$ represents
the quantum state so $\mu_1 = \delta_{\psi}$.  The other variable
$\lambda_2$ is uniformly distributed so $\mu_2$ is just the uniform
measure on $[0,1]$.

The outcome of a measurement $M = \{\Proj{\phi_j}\}_{j=0}^{d-1}$ is
determined as follows.  For each $\Proj{\lambda_1}$, the unit
interval is divided into $d$ subsets of length
$\Tr{\Proj{\phi_j}\Proj{\lambda_1}}$.  If $\lambda_2$ is in the
subset corresponding to $\Proj{\phi_j}$ then the $\Proj{\phi_j}$
outcome occurs with certainty.  It does not matter how we choose the
subsets so long as they are disjoint and of the right length.  One
way of doing it is to pick the $j$th set to be the interval
\begin{equation}
\sum_{k=0}^{j-1} \Tr{\Proj{\phi_k}\Proj{\lambda_1}} \leq \lambda_2
<  \sum_{k=0}^j \Tr{\Proj{\phi_k}\Proj{\lambda_1}},
\end{equation}
as illustrated in Fig.~\ref{fig:EOM:Bell}.
\begin{figure}[t!]
\centering
\includegraphics[width=86mm]{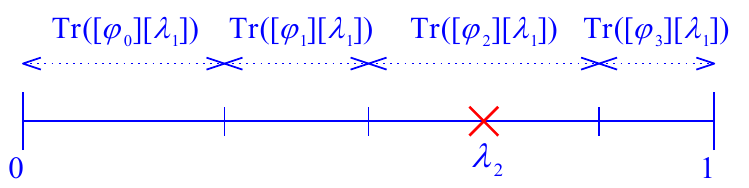}
\caption{\color[HTML]{0000FF}{\label{fig:EOM:Bell}Example of the conditional
probabilities in the Bell model for a measurement in
$\mathbb{C}^4$.  The unit interval is divided into $4$
subintervals of length $\Tr{\Proj{\phi_j}\Proj{\lambda_1}}$.  If
$\lambda_2$ is in the $j$th interval then the outcome will be
$j$.  In the case depicted, the outcome is $\Proj{\phi_2}$.}}
\end{figure}

This corresponds to the conditional probabilities
\begin{equation}
\label{eq:EOM:Bellresponse}
\text{Pr}(\phi_j|M,\Proj{\lambda_1}, \lambda_2)
= \begin{cases} 1 & \text{if} \,\, \begin{array}{c}
\sum_{k=0}^{j-1}\limits \Tr{\Proj{\phi_k}\Proj{\lambda_1}} \leq \lambda_2 ;\\
\lambda_2 < \sum_{k=0}^j\limits \Tr{\Proj{\phi_k}\Proj{\lambda_1}}
\end{array}\\
0 & \text{otherwise}.\end{cases}
\end{equation}
Note that, in this model, measurement-outcome pairs represented by
the same projector generally correspond to different response
functions.  This is because the way that the unit interval is
divided up depends on which other projectors are present in
$\{\Proj{\phi_j}\}_{j=0}^{d-1}$.  This is an example of measurement
contextuality.

It is straightforward to see that this model reproduces the quantum
predictions.  If the state prepared is $\Proj{\psi}$ then the point
measure $\mu_1$ implies that $\Proj{\lambda_1} = \Proj{\psi}$.
Therefore, the length of the $j$th subset of the unit interval will
be $\Tr{\Proj{\phi_j}\Proj{\psi}}$.  Since $\mu_2$ is the
uniform measure on $[0,1]$, the probability of $\lambda_2$ being in
a subset of $[0,1]$ is just the total length of the subset, so the
observed probability of obtaining outcome $j$ is
$\Tr{\Proj{\phi_j}\Proj{\psi}}$, as required.
\end{example}

\begin{example}[The Kochen--Specker model \cite{Kochen1967}]
\label{exa:EOM:KS}
Kochen and Specker introduced an ontological model for measurements
in orthonormal bases on a two-dimensional system.  The
PM fragment is $\langle \mathbb{C}^2, \mathcal{P},
\mathcal{M} \rangle$ where $\mathcal{P}$ consists of all pure states
on $\mathbb{C}^2$ and $\mathcal{M}$ consists of all measurements of
the form $M = \{\Proj{\phi},\Proj{\phi^{\perp}}\}$ with
$\{\Ket{\phi},\Ket{\phi^{\perp}}\}$ an orthonormal basis.

A pure state $\Proj{\psi}$ in a two dimensional Hilbert space can be
represented as a point $\vec{\psi} = (\sin \vartheta \cos \varphi,
\sin \vartheta \sin \varphi, \allowbreak \cos \vartheta)$ on the
surface of a unit $2$-sphere $S_2$ by choosing a representative
vector $\Ket{\psi}$ via
\begin{equation}
\label{eq:EOM:Bloch}
\Ket{\psi} = \cos(\frac{\vartheta}{2}) \Ket{z+} +
e^{\imath \varphi}\sin(\frac{\vartheta}{2}) \Ket{z-},
\end{equation}
where $0 \leq \vartheta < \pi$ and $-\pi < \varphi \leq \pi$.  This
is known as the Bloch sphere representation (see
Fig.~\ref{fig:EOM:bloch}).
\begin{figure}[t!]
\centering
\includegraphics[width=85mm]{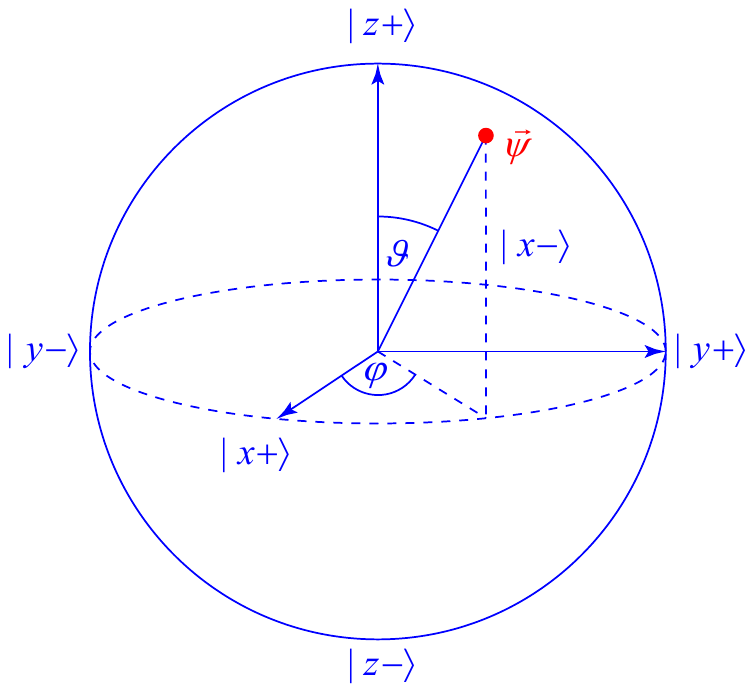}
\caption{\color[HTML]{0000FF}{\label{fig:EOM:bloch}The Bloch sphere representation of a
qubit. The quantum state is represented as a point $\vec{\psi}$ on the surface of a unit 2-sphere.}}
\end{figure}

In the Kochen--Specker model, the ontic state space is the unit
sphere $\Lambda = S_2$, which can be thought of as the Bloch sphere.
The quantum state $\Proj{\psi}$ corresponds to a measure $\mu$ that
can be written as
\begin{equation}
\mu(\Omega) = \int_{\Omega} p(\vec{\lambda}) \sin \vartheta d
\vartheta d \varphi,
\end{equation}
where $\vec{\lambda} = (\sin \vartheta \cos \varphi, \sin \vartheta
\sin \varphi, \cos \vartheta)$ and the density $p$ is given by
\begin{equation}
p(\vec{\lambda}) = \frac{1}{\pi} H(\vec{\psi} \cdot
\vec{\lambda}) \vec{\psi} \cdot \vec{\lambda},
\end{equation}
where $H$ is the Heaviside step function
\begin{equation}
H(x) = \begin{cases}1 & \text{if} \,\, x>0 \\ 0 & \text{if}
\,\, x\leq 0. \end{cases}
\end{equation}

This density is nonzero on the hemisphere defined by all vectors
that subtend an angle of less than $\pi/2$ with $\vec{\psi}$ and it
takes values proportional to the cosine of the angle between
$\vec{\lambda}$ and $\vec{\psi}$.

For a measurement $M = \left \{\Proj{\phi},\Proj{\phi^{\perp}}\right
\}$, the outcome $\Proj{\phi}$ is obtained if the angle between
$\vec{\lambda}$ and $\vec{\phi}$ is smaller than the angle between
$\vec{\lambda}$ and $\vec{\phi^{\perp}}$.  Otherwise the outcome
$\Proj{\phi^{\perp}}$ is obtained.  This corresponds to the
conditional probabilities
\begin{align}
\text{Pr}(\phi|M,\vec{\lambda}) & = H (\vec{\phi} \cdot
\vec{\lambda}) \\
\text{Pr}(\phi^{\perp}|M,\lambda) & = 1 -
\text{Pr}(\phi|M,\vec{\lambda}).
\end{align}

A proof that this reproduces the quantum predictions is given in
Appendix~\ref{App:BMKS}.
\end{example}

\subsection{Defining $\psi$-ontic/epistemic models}

\label{POEM}

We are now in a position to formally define what it means for the
quantum state to be real within an ontological model.  The definition
presented here is the one used by Pusey--Barrett--Rudolph \cite{Pusey2012}, which is a
slightly more rigorous version of a definition originally introduced
by Harrigan and Spekkens \cite{Harrigan2010}.  It is uncontroversial
that mixed states at least sometimes represent knowledge about which
of a set of pure states was prepared, so they are at least partially
epistemic.  For this reason, $\psi$-ontology theorems are only
concerned with proving the reality of pure quantum states.  A
$\psi$-ontic model is then one in which, if the pure state
$\Proj{\psi}$ is prepared, then $\Proj{\psi}$ is part of the ontic
state of the system.  In other words, the ontic state space can be
thought of as being composed of the set of pure quantum states along
with possibly some extra hidden variables.  This will be the case if
the measures corresponding to distinct pure state preparations do not
overlap with one another (see Fig.~\ref{fig:POEM:relabel}).
Conversely, the $\psi$-epistemic explanations of quantum phenomena
discussed in \S\ref{Spek} depend crucially on having overlap between
the measures representing different quantum states, so whether or not
there is overlap is the key issue.

\begin{figure}[t!]
\centering
\subfloat[\color{blue} Ontic model with nonoverlapping quantum
states.]{\includegraphics[width=75mm]{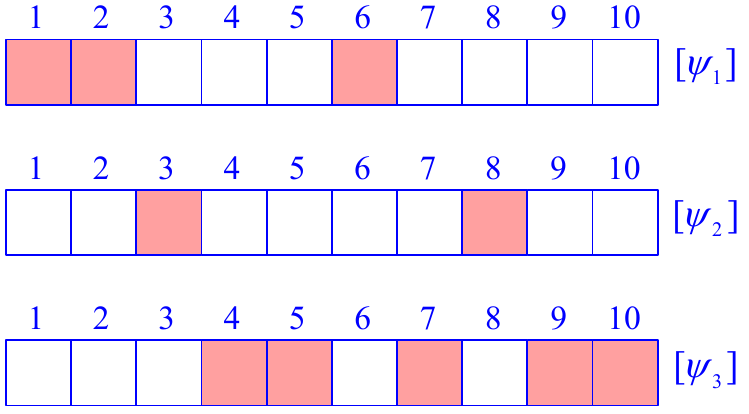}\label{fig:POEM:relabel:start}}\\
\subfloat[\color{blue} Relabeling the ontic
states.]{\label{fig:POEM:relabel:relabel}
\raisebox{15ex}{
\color{blue} \begin{tabular}{|r|r|}
\hline
Old label & New label \\
\hline
$1$ & $(\psi_1,1)$ \\
$2$ & $(\psi_1,2)$ \\
$3$ & $(\psi_2,1)$ \\
$4$ & $(\psi_3,1)$ \\
$5$ & $(\psi_3,2)$ \\
$6$ & $(\psi_1,3)$ \\
$7$ & $(\psi_3,3)$ \\
$8$ & $(\psi_2,2)$ \\
$9$ & $(\psi_3,4)$ \\
$10$ & $(\psi_3,5)$ \\
\hline
\end{tabular}}}\\
\subfloat[\color{blue} Equivalent model in which the quantum state is explicitly
part of the ontology.]{\label{fig:POEM:relabel:end}\includegraphics[width=75mm]{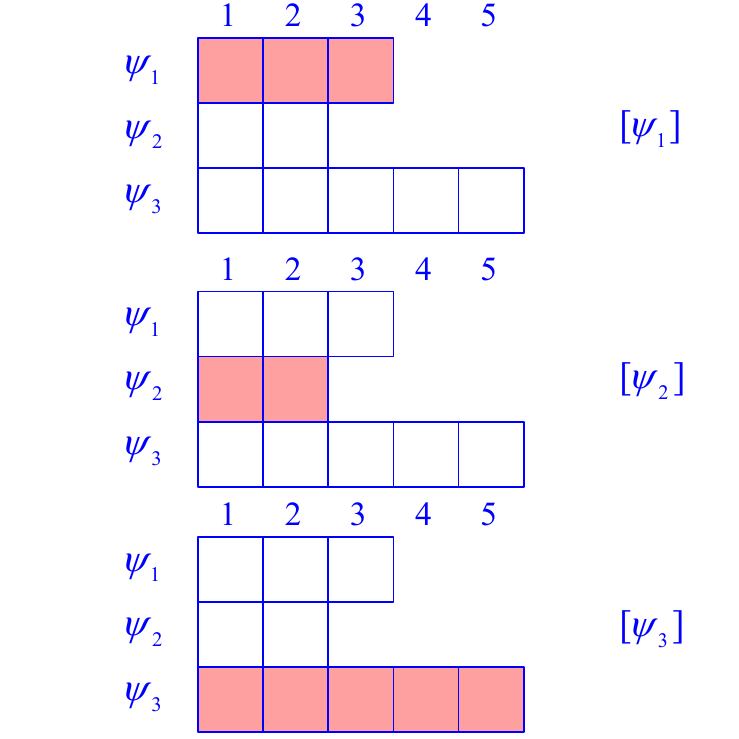}}
\caption{\color[HTML]{0000FF}{\label{fig:POEM:relabel}In an ontological model, if the
probability measures corresponding to distinct quantum states do
not overlap, then the ontic states can be relabeled such that the
quantum state is explicitly part of the ontology.
Panel~\protect\subref{fig:POEM:relabel:start} depicts the measures
corresponding to three quantum states on a discrete ontic state
space consisting of the integers from $1$ to $10$.  Red boxes
indicate the ontic states on which the measures have support.
Panel~\protect\subref{fig:POEM:relabel:relabel} shows a one-to-one
map to a new ontic state space in which the quantum states are an
explicit part of the ontology.  The result is shown in
Panel~\protect\subref{fig:POEM:relabel:end}.}}
\end{figure}

The notion of when two probability measures overlap can be formalized
using the variational distance.
\begin{definition}
\label{def:POEM:vd}
The \emph{variational distance} between two probability measures
$\mu$ and $\nu$ on a measurable space $(\Lambda, \Sigma)$ is
\begin{equation}
\label{eq:POEM:vd}
D(\mu,\nu) = \sup_{\Omega \in \Sigma} \left | \mu(\Omega) -
\nu(\Omega) \right |.
\end{equation}
\end{definition}

Note that taking the absolute value in Eq.~\eqref{eq:POEM:vd} is
optional, since if $\nu(\Omega) > \mu(\Omega)$ for some $\Omega \in
\Sigma$ then
\begin{equation}
\mu(\Lambda \backslash \Omega) - \nu(\Lambda \backslash \Omega) =
1 - \mu(\Omega) - 1 + \nu(\Omega) = \nu(\Omega) - \mu(\Omega),
\end{equation}
so there is always another measurable set for which the difference of
the measures is the same, but $\mu$ is larger than $\nu$.  Thus, the
variational distance can equivalently be defined as
\begin{equation}
D(\mu,\nu) = \sup_{\Omega \in \Sigma} \left ( \mu(\Omega) -
\nu(\Omega) \right ),
\end{equation}
or indeed
\begin{equation}
D(\mu,\nu) = \sup_{\Omega \in \Sigma} \left ( \nu(\Omega) -
\mu(\Omega) \right ).
\end{equation}

The variational distance is a metric on the set of probability
measures and it has the following operational interpretation.  Suppose
that a system is prepared according to one of two preparation
procedures, $P_1$ or $P_2$, where $P_1$ corresponds to the measure
$\mu$ and the $P_2$ to the measure $\nu$, each case having an equal a
priori probability.  You are then told the actual value of $\lambda$
and you wish to make the best possible guess as to which of the two
preparation procedures was used.  Your best strategy will be to choose
a set $\Omega \subseteq \Lambda$ and guess $P_1$ if $\lambda \in
\Omega$ and $P_2$ if $\lambda \notin \Omega$.  More generally, you
could use a probabilistic procedure, but convexity implies that this
cannot increase your probability of success.  Your probability of
success is then
\begin{align}
&&\text{Prob}(P_1)\text{Prob}(\Omega|P_1) + \text{Prob}(P_2)
\text{Prob}(\Lambda \backslash \Omega|P_2)\nonumber\\
&&=\frac{1}{2}\mu(\Omega) + \frac{1}{2}\nu(\Lambda \backslash \Omega)
=\frac{1}{2} \left ( 1 + \mu(\Omega) - \nu(\Omega) \right ),
\end{align}
so the maximum probability of success over all such strategies is
$\frac{1}{2} \left ( 1 + D(\mu,\nu) \right )$.  Success can occur with
unit probability only when $D(\mu,\nu) = 1$, so in this case $\lambda$
effectively determines which probability density was prepared
uniquely.

More rigorously, $D(\mu,\nu) = 1$ is equivalent to the existence of a
measurable set $\Omega \in \Sigma$ such that $\mu(\Omega) = 1$ and
$\nu(\Omega) = 0$.  An optimal guessing strategy then consists of
guessing $\mu$ if $\lambda \in \Omega$ and $\nu$ if $\lambda \notin
\Omega$.

If $\mu$ and $\nu$ are dominated by a third measure $m$, i.e.
\begin{align}
\mu(\Omega) & = \int_{\Omega} p(\lambda)dm(\lambda) \\
\nu(\Omega) & = \int_{\Omega} q(\lambda) dm(\lambda),
\end{align}
for some $m$-measurable densities $p$ and $q$, then
\begin{equation}
\label{eq:POEM:intvar}
D(\mu, \nu) = \frac{1}{2} \int_{\Lambda} \left | p(\lambda) -
q(\lambda) \right | dm(\lambda),
\end{equation}
which is often the most convenient form for computation.

For those unfamiliar with measure-theoretic probability, a measure $m$
dominates a measure $\mu$ if, whenever $m(\Omega) = 0$ then
$\mu(\Omega) = 0$.  If this holds then $\mu$ can be written as a
density with respect to $m$.  The measure $m = \frac{1}{2} \left (\mu
+ \nu \right )$ dominates both $\mu$ and $\nu$, and a similarly for
any finite set of probability measures.  For an uncountable set of
measures there need not exist a measure that dominates all of them.
For example, there is no measure that dominates all of the point
measures $\delta_{\psi}$ occurring in the Beltrametti--Bugajski model,
which is one of the reasons why we work with probability measures
instead of probability densities in the present treatment.

Based on the variational distance, we can define what it means for two
quantum states to have no overlap in an ontic model.
\begin{definition}
A pair of quantum states $\rho, \sigma \in \mathcal{P}$ are
\emph{ontologically distinct} in an ontological model if, for all
$\mu \in \Delta_{\rho}, \nu \in \Delta_{\sigma}$,
\begin{equation}
D(\mu,\nu) = 1,
\end{equation}
otherwise they are \emph{ontologically indistinct}.
\end{definition}
The idea here is that, if $\rho$ and $\sigma$ are ontologically
distinct then, regardless of how they are prepared, they can be
perfectly distinguished given knowledge of the exact ontic state.  In
order to take preparation contextuality into account, all possible
pairs of preparations that are represented by $\rho$ and $\sigma$ must
correspond to measures that have zero overlap.

If a set of preparations is operationally distinguishable, i.e.\ there
exists a measurement that perfectly distinguishes them, then they must
also be ontologically distinct.  The intuition behind this is
straightforward.  Consider the case of a finite ontic state space.  If
the measures corresponding to two operationally distinguishable
preparations did have nontrivial overlap then there would be a finite
probability of an ontic state occurring that is assigned a finite
probability according to both of them, and if this happens then the
ontic state cannot be used to deduce which preparation was performed
with certainty.  Hence, such a model could not reproduce the
statistics of the distinguishing measurement, which does allow this
deduction.  More formally,
\begin{definition}
A finite set of quantum states $\left \{ \rho_j \right \} \subseteq
\mathcal{P}$ is \emph{operationally distinguishable} if there exists
a POVM $\{E_j\} \in \mathcal{M}$ such that
\begin{equation}
\Tr{E_j \rho_k} = \delta_{jk}.
\end{equation}
\end{definition}
\begin{theorem}
\label{prop:POEM:distinguish}
Let $\mathcal{D} \subseteq \mathcal{P}$ be an operationally
distinguishable set of states.  Then, every pair $\rho, \sigma \in
\mathcal{D}$, $\rho \neq \sigma$, is ontologically distinct in any
ontological model that reproduces the operational predictions.
\end{theorem}
\begin{proof}
Let $M \in \mathcal{M}$ be a measurement that distinguishes the
states in $\mathcal{D}$.  Then, for $\rho, \sigma \in \mathcal{D}$,
$\rho \neq \sigma$, there exists a POVM element $E \in M$ such that
\begin{equation}
\Tr{E\rho} = 1 \qquad \text{and} \qquad \Tr{E\sigma} = 0.
\end{equation}
In an ontological model that reproduces the quantum predictions,
this implies that
\begin{align}
\int_{\Lambda} \text{Pr}(E|M,\lambda) d\mu(\lambda) & =
1 \label{eq:POEM:Distinct}\\
\int_{\Lambda} \text{Pr}(E|M,\lambda)d\nu(\lambda) & =
0, \label{eq:POEM:Distinct2}
\end{align}
for all $\mu \in \Delta_{\rho}$, $\nu \in \Delta_{\sigma}$, and for
any $\text{Pr} \in \Xi_M$.

Assume that $\rho$ and $\sigma$ are ontologically indistinct.  Then,
there exist $\mu \in \Delta_{\rho}$, $\nu \in \Delta_{\sigma}$ such
that $D(\mu,\nu) < 1$.  This is equivalent to saying that all sets
that are of measure $1$ according to $\mu$ are of nonzero measure
according to to $\nu$.  In order to satisfy
Eq.~\eqref{eq:POEM:Distinct}, $\text{Pr}(E|M,\lambda)$ must be equal
to $1$ on a set $\Omega$ that is of measure one according to $\mu$,
and thus by ontological indistinctness $\nu(\Omega) > 0$.  However,
from Eq.~\eqref{eq:POEM:Distinct2} we have
\begin{align}
0 & = \int_{\Lambda} \text{Pr}(E|M,\lambda) d\nu(\lambda) \\
& \geq \int_{\Omega} \text{Pr}(E|M,\lambda) d\nu(\lambda)\\
& = \int_{\Omega} d\nu(\lambda) = \nu(\Omega) > 0,
\end{align}
which is a contradiction.
\end{proof}

Although operationally distinguishable quantum states are always
ontologically distinct, this argument does not imply that arbitrary
pairs of states are.  If all pairs of pure states are nonetheless
ontologically distinct then the ontological model is $\psi$-ontic.
\begin{definition}
\label{def:POEM:PON}
An ontological model is $\psi$\emph{-ontic} if all pairs of pure
quantum states $\Proj{\psi}, \Proj{\phi} \in \mathcal{P}$,
$\Proj{\psi} \neq \Proj{\phi}$, are ontologically distinct.
Otherwise the model is $\psi$\emph{-epistemic}.
\end{definition}

This definition captures the idea that quantum states are real in an
ontological model if the probability measures corresponding to
distinct pure states do not overlap.  Defining a $\psi$-epistemic
model to be the negation of this is extremely permissive.  For
example, a model in which only a single pair of pure states have any
overlap, all other pairs being ontologically distinct, would be
$\psi$-epistemic according to this definition.  The $\psi$-epistemic
explanations of quantum phenomena discussed in \S\ref{Spek} would not
apply to such a model.  The $\psi$-ontology theorems discussed in
\hyperref[SPON]{Part II} aim to rule out even this permissive notion of a
$\psi$-epistemic model, so they obviously rule out any less permissive
definition as well.  Nevertheless, because $\psi$-ontology theorems
require questionable auxiliary assumptions, it is still interesting to
consider stronger notions of what it means for the quantum state to be
epistemic.  This is done in in \S\ref{ME} and \hyperref[Beyond]{Part III}.

Note that a $\psi$-ontic model is not necessarily
\emph{$\psi$-complete}, where the latter means that the wavefunction,
and only the wavefunction, is the ontic state.  $\psi$-complete models
are obviously $\psi$-ontic, but in general $\psi$-ontic models may
involve other variables in addition to the wavefunction.  A lot of
confusion may be avoided by clearly separating the notions of
$\psi$-ontic and $\psi$-complete models, particularly since much of
the literature on hidden variable theories focuses on the question of
whether quantum theory is complete, i.e.\ it is assumed that the
wavefunction is real a priori and the only question is whether
anything else needs to be added to it.  Obviously, such a framework is
not suited to discussing the question of whether the wavefunction must
be real in the first place.

Of the examples considered in \S\ref{EOM}, Spekkens' toy theory and
the Kochen--Specker model are $\psi$-epistemic, and the
Beltrametti--Bugajski and Bell models are $\psi$-ontic.  In Spekkens'
theory, the ontic state space is finite, so the integral in
Eq.~\eqref{eq:POEM:intvar} can be performed with respect to the
counting measure, which yields the sum
\begin{equation}
D(\mu,\nu) = \frac{1}{2} \sum_{\lambda} |\mu(\lambda) - \nu(\lambda)
|.
\end{equation}
For toy bit, $\lambda$ takes the values $(+,+), (+,-), (-,+)$ and
$(-,-)$.  If $\mu$ is the $\RKet{x+}$ measure, corresponding to the
quantum state $\Proj{x+}$, and $\nu$ is $\RKet{y+}$, corresponding to
$\Proj{y+}$ then we have
\begin{equation}
D(\mu,\nu) = \frac{1}{2} \left ( \left | \frac{1}{2}-\frac{1}{2}
\right | + \left | \frac{1}{2}-0 \right | +
\left | 0-\frac{1}{2} \right | + \left | 0-0 \right | \right ) =
\frac{1}{2}.
\end{equation}
Thus, the model is $\psi$-epistemic because $D(\mu,\nu) < 1$, as we
expect for overlapping measures.  Appendix~\ref{App:BMKS} provides a
proof that the Kochen--Specker model is $\psi$-epistemic by showing
that it satisfies the stronger notion of being maximally
$\psi$-epistemic to be discussed in \S\ref{ME}.

The Beltrametti--Bugajski model is in fact $\psi$-complete, because its
ontic state space consists of the set of pure states and preparing a
given pure state causes the corresponding ontic state to be occupied
with certainty.  It follows that the model is $\psi$-ontic because
each pure state corresponds to a distinct point measure.  Thus, for
$\Proj{\psi} \neq \Proj{\phi}$ the corresponding measures are
$\delta_{\psi}$ and $\delta_{\phi}$.  The set $\left
\{\Proj{\psi} \right \}$ is measure one according to
$\delta_{\psi}$ and measure zero according to
$\delta_{\phi}$.  Similarly, the Bell model is $\psi$-ontic,
where now the set $\left \{\Proj{\psi} \right \} \times [0,1]$ is
measure one for the measure corresponding to $\Proj{\psi}$ and measure
zero for the measure corresponding to $\Proj{\phi}$.  It is, however,
not $\psi$-complete because of the additional component $[0,1]$ of the
ontic state space.

\section{Implications of $\psi$-ontology}

\label{Imp}

Before discussing $\psi$-ontology theorems, it is worth pausing to
consider some of their implications.  One of the most interesting
things about establishing the reality of the quantum state within the
ontological models framework is that it would imply a lot of existing
no-go results as simple consequences.  Thus, even if you are a
neo-Copenhagenist who rejects the ontological models framework
outright, you should still be interested in $\psi$-ontology theorems
as potentially the most powerful class of no-go results that we
currently have.  Additionally, the ontological models framework can be
thought of as an attempt to simulate quantum theory using classical
resources.  Just as Bell's Theorem has implications for the difference
between quantum and classical communication complexity, and for device
independent quantum cryptography \cite{Brunner2013}, $\psi$-ontology
theorems might become a useful tool in quantum information theory.

The main implications of $\psi$-ontology are illustrated in
Fig.~\ref{fig:Imp:implications}, and each of them is discussed in this
section.  There are two main strands of implications, one based on the
size of the ontic state space and one based on results related to
contextuality.  The results of the contextuality strand can
alternatively be derived as consequences of the Kochen--Specker
Theorem.  Since $\psi$-ontology is our main concern here, the
connection to the Kochen--Specker theorem is discussed in
Appendix~\ref{App:KS}.

\begin{figure}[t!]
\centering
\includegraphics[width=75mm]{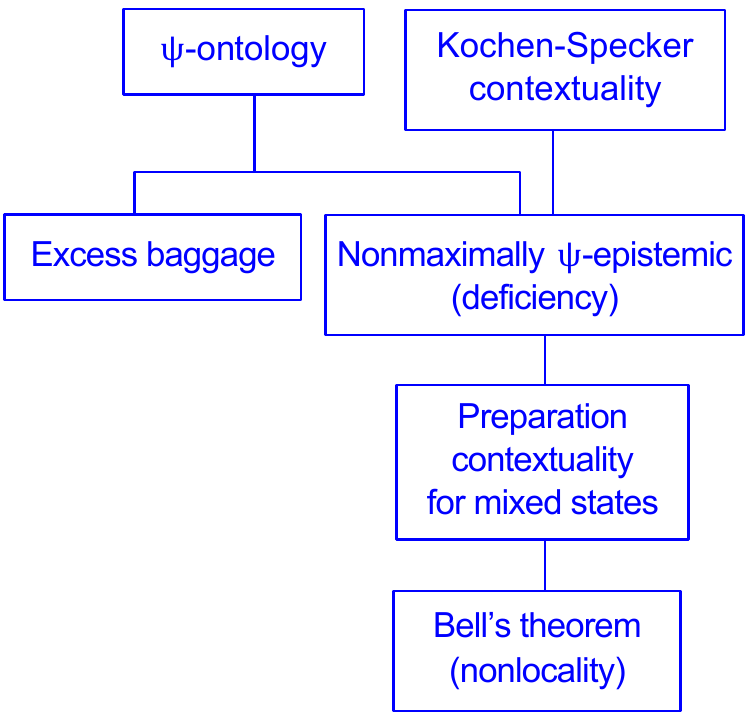}
\caption{\color[HTML]{0000FF}{\label{fig:Imp:implications}Hierarchy of constraints on ontological models of quantum theory.  Properties that appear higher in the diagram imply those that they are connected to lower down. Non maximal $\psi$-epistemicity can be derived from either $\psi$-ontology or Kochen--Specker contextuality.}}
\end{figure}

\subsection{Excess baggage}

\label{EBT}

The idea of excess baggage has already been discussed in
\S\ref{Excess} as the tension between the infinite amount of
information required to specify the quantum state of a qubit and the
fact that it can only be used to reliably transmit a single bit.
Similarly, the number of parameters required to specify a quantum
state scales exponentially with the number of systems, but the amount
of information that can be transmitted scales only linearly.

Excess baggage theorems show that this tension cannot be avoided in an
ontological model by providing lower bounds on the size of the ontic
state space required to reproduce quantum theory.  The first such
result, from which the ``excess baggage'' terminology originates, was
due to Hardy, who proved that an infinite number of ontic states are
required to reproduce the predictions of any quantum system, even just
a qubit \cite{Hardy2004}.  Subsequently, Montina showed that the
number of parameters required to specify an ontic state must scale
exponentially with the number of systems \cite{Montina2008,
Montina2011}.

If it can be proved that an ontological model of quantum theory must
be $\psi$-ontic then it follows immediately that the ontic state space
must be uncountably infinite and that the number of parameters
required to specify an ontic state must scale exponentially with the
number of systems.  This is because, in a $\psi$-ontic model, there
must be at least as many ontic states as there are quantum states.

\subsection{Maximally $\psi$-epistemic models}

\label{ME}

The concept of a maximally $\psi$-epistemic model was introduced by
Maroney as a stronger notion of what it means for an ontological model
to be $\psi$-epistemic \cite{Maroney2012}.  An equivalent concept was
introduced earlier by Harrigan and Rudolph \cite{Harrigan2007}, under
the terminology ``non deficient model''.  As you might imagine from
the name, a $\psi$-ontic model cannot be maximally $\psi$-epistemic,
but, even without $\psi$-ontology, maximally $\psi$-epistemic models
are ruled out by the Kochen--Specker theorem, as discussed in
Appendix~\ref{App:KS}.

The real interest in the concept of a maximally $\psi$-epistemic model
is that it enables one to devise measures of the extent to which a
model is $\psi$-epistemic, and hence to go beyond the sharp dichotomy
between $\psi$-epistemic and $\psi$-ontic implied by the definitions
adopted so far \cite{Maroney2012, Barrett2013, Leifer2014,
Branciard2014}.  This aspect is discussed in \S\ref{Lips}, but the
concept of a maximally $\psi$-epistemic model is also useful as a
stepping stone between $\psi$-ontology and the next step in the
contextuality strand, which is preparation contextuality.

The basic idea is that, in order to justify the $\psi$-epistemic
explanation of indistinguishability, we need more than just that the
probability measures corresponding to two different quantum states,
$\Proj{\psi}$ and $\Proj{\phi}$, should have nonzero overlap.
Ideally, when measuring a system prepared in the state $\Proj{\psi}$,
all of the probability of obtaining the outcome corresponding to the
projector $\Proj{\phi}$ should be accounted for by the overlap region
between the probability measures corresponding to the two states.
Since we are comparing the overlap of $\Proj{\psi}$ and $\Proj{\phi}$,
considered as states, with the probability of obtaining $\Proj{\phi}$
as a measurement outcome, this only makes sense if we are dealing with
a PM fragment for which, for every pure state $\Proj{\phi} \in
\mathcal{P}$, there exists a measurement $M \in \mathcal{M}$ such that
$\Proj{\phi} \in \mathcal{M}$.  Therefore, this is a standing
assumption for the remainder of this section.

\begin{definition}
\label{def:ME:ME}
An ontological model of a PM fragment is \emph{maximally
$\psi$-epistemic} if, for every pair of pure states $\Proj{\psi},
\Proj{\phi} \in \mathcal{P}$, for all $\mu \in \Delta_{\psi}$, $\nu
\in \Delta_{\phi}$, $M \in \mathcal{M}$ with $\Proj{\phi} \in M$,
and $\text{Pr} \in \Xi_M$,
\begin{equation}
\label{eq:ME:ME}
\int_{\Omega} \text{Pr}(\phi|M,\lambda) d\mu(\lambda) =
\int_{\Lambda} \text{Pr}(\phi|M,\lambda) d\mu(\lambda),
\end{equation}
for all sets $\Omega \in \Sigma$ such that $\nu(\Omega) = 1$.
\end{definition}

To unpack this definition a little, let $\mu \in \Delta_{\psi}$, let
$\nu \in \Delta_{\phi}$, let $M$ be a measurement that contains
$\Proj{\phi}$, and let $\text{Pr} \in \Xi_M$.  Suppose $\Omega$ is a
set of measure one according to $\nu$ and $\Omega^{\prime}$ is a set
of measure one according to $\mu$.  Then, Eq.~\eqref{eq:ME:ME} is
equivalent to
\begin{equation}
\label{eq:ME:ME2}
\int_{\Omega \cap \Omega^{\prime}} \text{Pr}(\phi|M,\lambda) d\mu(\lambda) =
\int_{\Lambda} \text{Pr}(\phi|M,\lambda) d\mu(\lambda).
\end{equation}
Whatever one might mean by the overlap region between $\mu$ and $\nu$,
it should be the intersection of a set that is measure one according
to $\mu$ with a set that is measure one according to $\nu$, and since
Eq.~\eqref{eq:ME:ME2} must hold for all such sets, it guarantees that
the probability of obtaining outcome $\Proj{\Phi}$ when measuring a
system prepared in the state $\Proj{\psi}$ is entirely accounted for
by any such region.

Note that, instead of Eq.~\eqref{eq:ME:ME}, previous works
\cite{Harrigan2007, Maroney2012, Maroney2012a, Leifer2013c} imposed
the requirement that
\begin{equation}
\label{eq:ME:MEold}
\int_{\Lambda_{q}} \text{Pr}(\phi|M,\lambda)p(\lambda) d\lambda =
\int_{\Lambda} \text{Pr}(\phi|M,\lambda)p(\lambda) d\lambda,
\end{equation}
where $d\lambda$ is a measure that dominates $\mu$ and $\nu$, $p$ and
$q$ are densities that represent them, and $\Lambda_{q} = \{\lambda
\in \Lambda | q(\lambda) > 0\}$ is the support of the probability
density $q$.  However, the probability measure $\nu$ does not
correspond to a unique density $q$, since densities that differ on a
set of measure zero according to $d \lambda$ represent the same
probability measure.  This means that $\Lambda_q$ is not uniquely
specified by $\nu$, since setting the values of $q$ to zero on a
measure zero set would change $\Lambda_q$ but not $\nu$.  Therefore,
to use this definition, one has to imagine that the ontological model
specifies a particular density representation for each measure, rather
than just the measure itself.  On its own, this older definition is
also not strong enough to entail that a $\psi$-ontic model cannot be
maximally $\psi$-epistemic.  To remedy this, one can adopt a standing
assumption that all probability measures appearing in the model are
dominated by, i.e.\ absolutely continuous with respect to, a canonical
measure $d\lambda$, as was done explicitly in \cite{Leifer2013c}, and
this must be regarded as an implicit assumption in previous works
\cite{Harrigan2007, Maroney2012, Maroney2012a}.  However, we do not
really want to make this assumption, because some models of interest
do not satisfy it.  For example, an uncountable set of point measures
on distinct points is not dominated by any measure, and these occur in
the Beltrametti--Bugajski and Bell models.

It is straightforward to verify that Spekkens' toy theory is maximally
$\psi$-epistemic and a proof that the Kochen--Specker model is
maximally $\psi$-epistemic is given in Appendix~\ref{App:BMKS}.  On
the other hand, a $\psi$-ontic model that reproduces the quantum
predictions cannot be maximally $\psi$-epistemic, so the
Beltrametti--Bugajski and Bell models are not maximally
$\psi$-epistemic.

\begin{theorem}
\label{thm:ME:MEPE}
Consider a PM fragment that contains at least one pair of
nonorthogonal pure states.  If an ontological model of this fragment
that reproduces the quantum predictions is maximally
$\psi$-epistemic then it is $\psi$-epistemic.
\end{theorem}
\begin{proof}
Assume the ontological model is $\psi$-ontic.  Then, for every pair
of nonorthogonal pure states $\Proj{\psi}, \Proj{\phi} \in
\mathcal{P}$, every pair $\mu \in \Delta_{\psi}$, $\nu \in
\Delta_{\phi}$ satisfies $D(\mu,\nu) = 1$, which is equivalent to
saying that there exists a set $\Omega$ that is measure $0$
according to $\mu$ and measure $1$ according to $\nu$.  However, if
the ontological model is maximally $\psi$-epistemic then, for all
$\mu \in \Delta_{\psi}$, $\nu \in \Delta_{\phi}$, $M \in
\mathcal{M}$ containing $\Proj{\phi}$, and $\text{Pr} \in \Xi_M$,
\begin{equation}
\int_{\Omega} \text{Pr}(\phi|M,\lambda) d\mu(\lambda) =
\int_{\Lambda} \text{Pr}(\phi|M,\lambda) d\mu(\lambda).
\end{equation}
This implies that,
\begin{align}
\mu(\Omega) & = \int_{\Omega} d\mu(\lambda) \\
& \geq \int_{\Omega}
\text{Pr}(\phi|M,\lambda) d \mu(\lambda) \\
& = \int_{\Lambda} \text{Pr}(\phi|M,\lambda) d \mu(\lambda) \\
& = \Tr{\Proj{\phi}\Proj{\psi}} > 0,
\end{align}
where the second line follows from the fact that
$\text{Pr}(\phi|M,\lambda) \in [0,1]$ and the fourth from the fact
that the ontological model reproduces the quantum predictions.
Thus, $\Omega$ cannot be of measure zero according to $\mu$, which
contradicts the assertion that the model is $\psi$-ontic.
\end{proof}

\subsection{Preparation contextuality}

\label{PC}

In \cite{Spekkens2005}, Spekkens introduced an operational approach to
contextuality.  In addition to generalizing the usual Kochen--Specker
approach to contextuality of measurements, as discussed in
Appendix~\ref{App:KS}, Spekkens introduced a notion of contextuality
for preparations.  In Spekkens' approach, the fundamental idea of
noncontextuality is that things that are operationally equivalent,
i.e.\ always lead to the same observable probabilities, should be
represented in the same way in an ontological model.  Preparation
noncontextuality is just this idea applied to states rather than
measurements.  This section shows that, when extended to mixed states,
non maximally $\psi$-epistemic models must be preparation contextual,
a result that was first pointed out in \cite{Leifer2013c}.

\begin{definition}
In an ontological model of a PM fragment, a state $\rho \in
\mathcal{P}$ is \emph{preparation noncontextual} if $\Delta_{\rho}$
only contains a single measure.  Otherwise $\rho$ is
\emph{preparation contextual}.  Similarly, the ontological model
itself is preparation noncontextual if every $\rho \in \mathcal{P}$
is preparation noncontextual, and otherwise it is preparation
contextual.
\end{definition}

Whenever two preparation procedures result in the same quantum state
$\rho$, there is no measurement that can distinguish between them
because, according to quantum theory, all of the outcome probabilities
for every measurement are exactly the same.  Thus, the general
principle of noncontextuality implies that all methods of preparing
$\rho$ should result in the same probability measure.

Note that much of the literature on $\psi$-ontology implicitly assumes
preparation noncontextuality for pure states by associating a unique
measure $\mu_{\psi}$ with each pure state $\Proj{\psi}$ under
consideration.  This assumption is fairly harmless for a couple of
reasons.  Firstly, the necessity of preparation contextuality has so
far only been established for mixed states, so assuming a unique
measure for pure states is not ruled out by any existing no-go
theorem.  Secondly, adapting results to preparation contextual models
is usually just a matter of modifying the definitions in a fairly
obvious way that does not require proofs to be modified substantively.
Nevertheless, I prefer to take the possibility of preparation
contextuality into account explicitly because it is oddly asymmetric
to allow for contextual measurements but not contextual states, and
because it allows results to be proved under weaker assumptions.

In order to prove preparation contextuality for mixed states, a
further assumption is required about how convex combinations of
quantum states should be represented in ontological models.  To
understand this, let $\rho = \sum_j p_j \sigma_j$ be a convex
decomposition of a mixed state $\rho$ into other (pure or mixed)
states $\sigma_j$, i.e. $0 \leq p_j \leq 1$ and $\sum_j p_j = 1$.  One
method of preparing $\rho$ is to generate a classical random variable
that takes value $j$ with probability $p_j$ (e.g.\ by flipping coins,
throwing dice, or any other suitable method), prepare the state
$\sigma_j$ if the value of the classical variable is $j$, and then
discard and forget the value of the classical variable.

In an ontological model, it is reasonable to assume that, in the above
mixing procedure, if $\mu_j$ is the probability measure corresponding
to the method used to prepare $\sigma_j$, then $\sum_j p_j \mu_j$ is
the measure corresponding to preparing $\rho$ by this mixing
procedure.  This is because the randomness used to generate the
classical variable could come from a source that is completely
independent of the system under investigation, in which case we would
expect that the ontic state of the system only depends on it via its
effect on which of the $\sigma_j$ is prepared.  More formally,

\begin{definition}
An ontological model \emph{respects convex decompositions} of a
state $\rho \in \mathcal{P}$ if, for every set of states
$\{\sigma_j\} \subseteq \mathcal{P}$ such that $\rho = \sum_j p_j
\sigma_j \in \mathcal{P}$ for some coefficients $p_j$ satisfying $0
\leq p_j \leq 1$ and $\sum_j p_j = 1$, for all possible choices of
$\mu_j \in \Delta_{\sigma_j}$, the measure $\sum_j p_j \mu_j$ is in
$\Delta_{\rho}$.  The ontological model \emph{respects convexity} if
it respects convex decompositions of every $\rho \in \mathcal{P}$.
\end{definition}

For the purpose of connecting non maximally $\psi$-epistemic models
with preparation contextuality, recall that an ontological model is
non maximally $\psi$-epistemic iff there are two nonorthogonal pure
states $\Proj{\psi}, \Proj{\phi} \in \mathcal{P}$, such that, for some
$\mu \in \Delta_{\psi}$, $\nu \in \Delta_{\phi}$, $M \in \mathcal{M}$
containing $\Proj{\phi}$, and $\text{Pr} \in \Xi_M$ there exists an
$\Omega \in \Sigma$ such that $\nu(\Omega) = 1$ but
\begin{equation}
\label{eq:PC:nME}
\int_{\Omega} \text{Pr}(\phi|M,\lambda)d\mu(\lambda) < \int_{\Lambda}
\text{Pr}(\phi|M,\lambda)d\mu(\lambda).
\end{equation}

Preparation contextuality can now be derived from non maximal
$\psi$-epistemicity via the following theorem.

\begin{theorem}
\label{prop:PC:PC}
Consider an ontological model of a PM fragment that reproduces the
quantum predictions, respects convexity and is non maximally
$\psi$-epistemic.

Assume that, for some pair of pure states $\Proj{\psi}, \Proj{\phi}
\in \mathcal{P}$ that satisfy Eq.~\eqref{eq:PC:nME}, the states
$\sigma^{\psi^{\perp}} = \frac{1}{d-1} \left ( I - \Proj{\psi}
\right )$, $\sigma^{\phi^{\perp}} = \frac{1}{d-1} \left ( I -
\Proj{\phi} \right )$ and the maximally mixed state $I/d$ are also
in $\mathcal{P}$.  Then, $I/d$ is preparation contextual in the
ontological model.
\end{theorem}
\begin{proof}
The maximally mixed state has convex decompositions
\begin{align}
\frac{I}{d} & = \frac{1}{d} \Proj{\psi} + \frac{d-1}{d}
\sigma^{\psi^{\perp}} \\
& = \frac{1}{d} \Proj{\phi} + \frac{d-1}{d} \sigma^{\phi^{\perp}}.
\end{align}
Since the ontological and quantum models respect convexity, for
every $\mu \in \Delta_{\psi}$ and $\mu^{\perp} \in
\Delta_{\sigma^{\psi^{\perp}}}$, the measure $\frac{1}{d} \mu +
\frac{d-1}{d}\mu^{\perp}$ is in $\Delta_{I/d}$.  Similarly, for
every $\nu \in \Delta_{\phi}$ and $\nu^{\perp} \in
\Delta_{\sigma^{\phi^{\perp}}}$, the measure $\frac{1}{d} \nu +
\frac{d-1}{d}\nu^{\perp}$ is in $\Delta_{I/d}$.  Assume that $I/d$
is preparation noncontextual, so that $\Delta_{I/d}$ only contains
one probability measure $\zeta$.  Then, we have
\begin{align}
\zeta & = \frac{1}{d} \mu + \frac{d-1}{d} \mu^{\perp} \\
& = \frac{1}{d} \nu + \frac{d-1}{d} \nu^{\perp}.
\end{align}

Now, since $\Proj{\psi}$ and $\Proj{\phi}$ satisfy
Eq.~\eqref{eq:PC:nME}, and the model reproduces the quantum
predictions, there must exist a $\mu \in \Delta_{\psi}$, a $\nu \in
\Delta_{\phi}$, an $M \in \mathcal{M}$ containing $\Proj{\phi}$, a
$\text{Pr} \in \Xi_M$, and a set $\Omega \in \Sigma$ such that
$\nu(\Omega) = 1$, but
\begin{equation}
\int_{\Omega} \text{Pr}(\phi|M,\lambda) d\mu(\lambda) <
\Tr{\Proj{\phi}\Proj{\psi}}.
\end{equation}
This means that, in order to reproduce the quantum predictions,
there must be a set $\tilde{\Omega} \subseteq \Lambda \backslash
\Omega$ such that $\text{Pr}(\phi|M,\lambda)$ is nonzero everywhere on
$\tilde{\Omega}$ and $\mu(\tilde{\Omega}) > 0$.  It is also the case
that $\nu(\tilde{\Omega}) = 0$ because $\tilde{\Omega}$ is a subset
of $\Lambda \backslash \Omega$ and $\Omega$ is measure one for
$\nu$.

Again, since the ontological model reproduces the quantum
predictions, we must have
\begin{equation}
\int_{\Omega} \text{Pr}(\phi|M,\lambda) d\nu^{\perp}(\lambda) =
\Tr{\Proj{\phi}\sigma^{\phi^{\perp}}} = 0.
\end{equation}
This implies that $\nu^{\perp}(\tilde{\Omega}) = 0$ because
$\text{Pr}(\phi|M,\lambda)$ is nonzero everywhere on $\tilde{\Omega}$.
Therefore we have
\begin{align}
\zeta(\tilde{\Omega}) & = \frac{1}{d} \nu(\tilde{\Omega}) +
\frac{d-1}{d} \nu^{\perp}(\tilde{\Omega}) \\
& = 0,
\end{align}
but also
\begin{align}
\zeta(\tilde{\Omega}) & = \frac{1}{d} \mu(\tilde{\Omega}) +
\frac{d-1}{d} \mu^{\perp}(\tilde{\Omega}) \\
& \geq \frac{1}{d} \mu(\tilde{\Omega}) \\
& > 0,
\end{align}
which is a contradiction.  Therefore, $I/d$ must be preparation
contextual.
\end{proof}

Similar arguments can be used to establish the preparation
contextuality of other mixed states.  If $\sigma$ and $\tau$ are states
such that $\Tr{\Proj{\phi}\sigma} = 0$ and
\begin{equation}
p \Proj{\psi} + (1-p)\tau = q \Proj{\phi} + (1-q) \sigma,
\end{equation}
for some $0 \leq p,q \leq 1$, then the state $\rho = p \Proj{\psi} +
(1-p)\tau$ must be preparation contextual by the same argument given
above.  Banik et.\ al.\ have recently used this method to show that
all mixed states must be preparation contextual in Hilbert spaces of
dimension $d \geq 3$ \cite{Banik2014}.  They also proved the same
result for all Hilbert space dimensions, but via a method that does
not make use of the connection with maximal $\psi$-epistemicity.  

\subsection{Bell's Theorem}

\label{Bell}

The purpose of this section is to argue that being committed to
$\psi$-ontology implies the conclusion of Bell's theorem, i.e.\ any
$\psi$-ontic model must fail to satisfy Bell's locality condition.  It
is not my intention to argue that the implication from $\psi$-ontology
theorems to Bell's theorem is the best or most elegant way of proving
Bell's theorem.  The argument is quite technical compared to other
proofs of Bell's theorem and $\psi$-ontology theorems are based on
auxiliary assumptions that are not required to prove Bell's theorem.
Instead, the aim is to compare the relative strength of the
conclusions of $\psi$-ontology theorems with those of Bell's theorem.

What we will actually do in this section is to argue that the kind of
preparation contextuality that follows from the failure of maximal
$\psi$-epistemicity implies the conclusion of Bell's theorem.  This
means that $\psi$-ontology implies the same, because it can be used to
establish the failure of maximal $\psi$-epistemicity.  In principle,
one can also view the result of this section as a full-blown proof of
Bell's theorem that does not rely on $\psi$-ontology, since, as
established in Appendix~\ref{App:KS}, the Kochen--Specker theorem can
alternatively be used to establish the failure of maximal
$\psi$-epistemicity without assuming the reality of the quantum state.
However, this would be quite a convoluted way of deriving Bell's
theorem compared to the standard proofs.

From a historical perspective, the Einstein-Podolsky-Rosen argument
\cite{Einstein1935} could be read as an argument that $\psi$-complete
models are nonlocal, but a $\psi$-ontic model may involve other
variables in addition to the wavefunction.  That nonlocality follows
from the weaker assumption of $\psi$-ontology is implicit in an
argument that Einstein made in correspondence with Schr{\"o}dinger.
This was first pointed out by Harrigan and Spekkens in
\cite{Harrigan2010}, where they also provide a more formal version of
the argument.  However, both the Einstein and Harrigan-Spekkens
arguments employ an additional assumption known as \emph{Einstein
separability}, which is discussed in more detail in \S\ref{MPIP}.
Einstein separability is not actually required to prove Bell's theorem
\cite{Henson2013}, so I do not make this assumption here.  This, in
addition to taking care of measure-theoretic issues, makes the
argument presented here a bit more involved than Harrigan and
Spekkens' treatment.  However, once this is done, weakening their
assumption of $\psi$-ontology to the failure of maximal
$\psi$-epistemicity does not significantly change the structure of the
argument.  As presented here, the argument is based on an idea due to
Barrett \cite{Barrett2004} that has appeared in the literature in an
informal version \cite{Leifer2013c}.

The basic idea runs as follows.  Suppose Alice and Bob share a
composite system with Hilbert space $\Hilb[A] \otimes \Hilb[B]$,
$\Hilb[A]=\Hilb[B]=\mathbb{C}^d$, prepared in the entangled state
$\Proj{\Phi^+}_{AB}$, where
\begin{equation}
\Ket{\Phi^+}_{AB} = \frac{1}{\sqrt{d}} \sum_{j=0}^{d-1} \Ket{j}_A \otimes
\Ket{j}_B,
\end{equation}
and $\left \{ \Ket{j} \right \}_{j=0}^{d-1}$ is an orthonormal basis
for $\mathbb{C}^d$.  If Alice performs a measurement described by the
projectors $\left \{ \Proj{\psi}_A^T, I_A - \Proj{\psi}_A^T \right
\}$, where $^T$ denotes transpose in the $\Ket{j}$ basis and $I_A$ is
the identity operator on $\Hilb[A]$, then, with probability $1/d$ she
obtains the $\Proj{\psi}_A^T$ outcome and the state of Bob's system
gets updated to $\Proj{\psi}_B$, and with probability $(d-1)/d$ she
gets the $I_A - \Proj{\psi}_A^T$ outcome and the state of Bob's system
gets updated to $\sigma^{\psi{\perp}}_B$.  If we allow Alice to
postselect on her measurement outcome, then we can regard this as a
method of preparing Bob's system in the states $\Proj{\psi}_B$ and
$\sigma^{\psi^{\perp}}_B$.  If she subsequently forgets and discards
her measurement outcome then we can also regard this as a method of
preparing Bob's system in the maximally mixed state by mixing
$\Proj{\psi}_B$ with $\sigma^{\psi^{\perp}}_B$.  Similar remarks apply
if Alice measures $\{\Proj{\phi}^T_A, I_A - \Proj{\phi}^T_A\}$, which,
by postselection, prepares Bob's system in the state $\Proj{\phi}_B$
or $\sigma^{\phi^{\perp}}_B$ and, upon forgetting the measurement
outcome, prepares it in the maximally mixed state by mixing
$\Proj{\phi}_B$ with $\sigma^{\phi^{\perp}}_B$.

Now, given an ontological model for the composite system that
satisfies Bell's locality condition, we can use it to derive an
ontological model for these preparation procedures on Bob's system
alone.  This model must be preparation noncontextual for the state
$I/d$ because locality implies that the probability of getting a
particular ontic state for the derived model on Bob's system cannot
depend on whether Alice chooses to measure $\left \{ \Proj{\psi}^T_A,
I_A - \Proj{\psi}^T_A \right \}$ or $\left \{ \Proj{\phi}^T_A, I_A -
\Proj{\phi}^T_A \right \}$.  However, in order to construct a model
for all possible preparations and measurements of this form that
reproduces the quantum predictions, preparation noncontextuality
must fail for some such pair because a maximally $\psi$-epistemic
model is impossible.  Therefore, Bell's locality condition must fail.
The remainder of this section formalizes this argument.

The type of PM fragment relevant to Bell's Theorem involves a
bipartite system and local measurements.  This is formalized in the
concept of a product measurement fragment.
\begin{definition}
\label{def:Bell:PM}
A \emph{product measurement fragment} $\mathfrak{F}_{AB} = \langle
\Hilb[A] \otimes \Hilb[B], \mathcal{P}_{AB}, \mathcal{M}_A \times
\mathcal{M}_B \rangle$ is a PM fragment on a composite system with
Hilbert space $\Hilb[A] \otimes \Hilb[B]$, where $\mathcal{M}_A$ is
a set of POVMs on $\Hilb[A]$, $\mathcal{M}_B$ is a set of POVMs on
$\Hilb[B]$ and $\mathcal{M}_A \times \mathcal{M}_B$ denotes the set
of all POVMs of the form $(M_A,M_B) = \{E_j \otimes F_k\}$, where
$M_A = \{E_j\} \in \mathcal{M}_A$ and $M_B = \{F_k\} \in
\mathcal{M}_B$.
\end{definition}

To prove Bell's Theorem, we only need consider a single bipartite
state $\rho_{AB}$, so for the remainder of this section we assume that
$\mathcal{P}_{AB}$ just contains this single state.  Most of the
concepts can be generalized to fragments with more than one state in
an obvious way.

The first step is to reinterpret experiments on the composite system
as prepare-and-measure experiments on Bob's system alone.  At the
operational level, when the system is prepared in the state
$\rho_{AB}$, Alice makes a measurement $M_A$ on system $A$ and Bob
makes a measurement $M_B$ on system $B$, quantum theory predicts the
outcome probabilities
\begin{equation}
\label{eq:Bell:joint}
\text{Prob}(E,F|M_A,M_B) =  \Tr[AB]{E \otimes F
\rho_{AB}},
\end{equation}
where $E \in M_A$ and $F \in M_B$.

These probabilities can be rewritten in terms of states and
measurements on $\Hilb[B]$ alone by imagining that Alice performs her
measurement before Bob performs his, and considering how the state of
Bob's system gets updated in light of Alice's measurement result.
Upon learning that the outcome of $M_A$ was $E$, the state of Bob's
system gets updated to
\begin{equation}
\label{eq:Bell:condstate}
\rho_{B|E} = \frac{\Tr[A]{E \otimes I_B
\rho_{AB}}}{\text{Prob}(E|M_A)},
\end{equation}
where $\text{Prob}(E|M_A) = \Tr[A]{E \rho_{A}}$ is the probability
of Alice's measurement outcome and $\rho_A = \Tr[B]{\rho_{AB}}$.
Eq.~\eqref{eq:Bell:joint} can then be rewritten as
\begin{equation}
\text{Prob}(E,F|M_A,M_B) = \text{Prob}(E|M_A) \Tr[B]{F
\rho_{B|E}}.
\end{equation}

By postselecting on obtaining the outcome $E$, Alice's measurement can
be viewed as a method of preparing Bob's system in the state
$\rho_{B|E}$.  In doing so, $E$ gets reinterpreted as specifying a
preparation procedure for system $B$, instead of a measurement on $A$,
so we now want to condition the probabilities on $E$.  This gives
\begin{equation}
\text{Prob}(F|M_A,E,M_B) =
\frac{\text{Prob}(E,F|M_A,M_B)}{\text{Prob}(E|M_A)}
= \Tr[B]{F \rho_{B|E}}. \label{eq:Bell:cond}
\end{equation}

We also want to consider the case where Alice forgets her measurement
outcome outcome.  This prepares Bob's system in the state $\rho_B =
\Tr[A]{\rho_{AB}}$ and we have, for every $M_A \in \mathcal{M}_A$,
\begin{equation}
\rho_B = \sum_{E \in M_A} \text{Prob}(E|M_A) \rho_{B|E}.
\end{equation}

We can now define the conditional fragment on Bob's system consisting
of all the states that Alice can prepare on his system in this way.
\begin{definition}
Let $\mathfrak{F}_{AB} = \langle \Hilb[A] \otimes \Hilb[B],
\mathcal{P}_{AB}, \mathcal{M}_A \times \mathcal{M}_B \rangle$ be a
product measurement fragment, where $\mathcal{P}_{AB}$ contains a
single state $\rho_{AB}$.  The \emph{conditional} fragment on $B$
given $A$ is $\mathfrak{F}_{B|A} = \langle \Hilb[B],
\mathcal{P}_{B|A}, \mathcal{M}_B \rangle$, where $\mathcal{P}_{B|A}$
consists of all states of the form $\rho_{B|E}$ as given by
Eq.~\eqref{eq:Bell:condstate} where $E \in M_A$ for some $M_A \in
\mathcal{M}_A$, as well as the state $\rho_B = \Tr[A]{\rho_{AB}}$.
\end{definition}

Typically, there are several different methods of preparing the same
quantum state and of implementing the same measurement in a PM
fragment.  These different methods need not correspond to the same
probability measures nor the same conditional probability
distributions in an ontological model.  However, the derivation of
Bell's theorem is greatly simplified if we focus attention on one
specific method of preparing $\rho_{AB}$ and one one method of
implementing each of the measurements $M_A \in \mathcal{M}_A$ and $M_B
\in \mathcal{M}_B$.  Then, we can assume that $\rho_{AB}$ is
represented by a unique measure $\mu$ and that each $(M_A,M_B)$ is
represented by a unique conditional probability distribution
$\text{Pr}(E,F|M_A,M_B,\lambda)$, where $E \in M_A$ and $F \in M_B$.
This does not weaken the conclusion because the result will hold
regardless of which method of implementation is chosen.  Therefore,
for the remainder of this section, we assume a unique measure and
conditional probabilities for ontological models of the product
measurement fragment.

Given such an ontological model of $\mathfrak{F}_{AB}$, it is natural
to think that an ontological model for the conditional PM fragment
$\mathfrak{F}_{B|A}$ might be derivable from it.  However, in general,
an ontological model for $\mathfrak{F}_{AB}$ cannot be cleanly
separated into a part that depends on system $A$ and a part that
depends on system $B$ unless it satisfies Bell's condition of local
causality.  To see this, suppose we have an ontological model of
$\mathfrak{F}_{AB}$.  If Alice chooses to make the measurement $M_A$
and Bob chooses to make the measurement $M_B$ then in general, not
assuming locality, this is represented by a conditional probability
distribution $\text{Pr}(E, F|M_A,M_B,\lambda)$, where $E \in M_A$
and $F \in M_B$.  The model then predicts that the observed
probability of obtaining the outcomes $E$ and $F$ is given by
\begin{equation}
\text{Pr}(E,F|M_A,M_B) = \int_{\Lambda}
\text{Pr}(E,F|M_A,M_B,\lambda) d\mu(\lambda), \label{eq:Bell:ontjoint}
\end{equation}
where $\mu$ is the measure representing $\rho_{AB}$.

At the quantum level, the conditional probability
$\text{Prob}(F|M_A,E,M_B)$ could be rewritten in terms of a state on
$\Hilb[B]$ that depends on $E$ alone and the measurement operator $F$,
which is also defined on $\Hilb[B]$.  The analogue at the ontological
level would be to absorb the dependence on $M_A$ and $E$ into the
measure $\mu$ and be left with a conditional probability distribution
that only depends on $M_B$ and $F$.  Unfortunately, the conditional
probability distribution $\text{Pr}(E, F | M_A,M_B,\lambda)$ need not
factor cleanly into a term that only depends on $E$ and $M_A$ and a
term that only depends on $F$ and $M_B$.  In order to obtain such a
decomposition, a locality assumption is required.

\begin{definition}
An ontological model of a product measurement fragment
$\mathfrak{F}_{AB} = \langle \Hilb[A] \otimes \Hilb[B],
\mathcal{P}_{AB}, \mathcal{M}_A \times \mathcal{M}_B \rangle$ is
\emph{Bell local} if, for all $M_A \in \mathcal{M}_A$, $M_B \in
\mathcal{M}_B$, $E \in M_A$, $F \in M_B$, $\lambda \in \Lambda$,
\begin{equation}
\label{eq:Bell:local}
\text{Pr}(E,F|M_A,M_B,\lambda) =
\text{Pr}(E|M_A,\lambda)\text{Pr}(F|M_B,\lambda)
\end{equation}
where $\text{Pr}(E|M_A,\lambda)$ is a valid conditional
probability distribution over $M_A$ and $\text{Pr}(F|M_B,\lambda)$
is a valid conditional probability distribution over $M_B$.
\end{definition}

In a Bell local model of $\mathfrak{F}_{AB}$, measures conditional on
$M_A$ and $E \in M_A$ can be defined as follows,
\begin{equation}
\label{eq:Bell:condprobs}
\mu_{M_A,E}(\Omega) = \frac{1}{\text{Pr}(E|M_A)}\int_{\Omega}
\text{Pr}(E|M_A,\lambda) d\mu(\lambda),
\end{equation}
where $\text{Pr}(E|M_A) = \int_{\Lambda}
\text{Pr}(E|M_A,\lambda)d\mu(\lambda)$ is the observed probability of
obtaining the outcome $E$ for the measurement $M_A$ according to the
ontological model.  Eq.~\eqref{eq:Bell:local} can then be
rewritten as
\begin{equation}
\text{Pr}(E,F|M_A,M_B,) = \text{Pr}(E|M_A)\int_{\Lambda} \text{Pr}(F|M_B,\lambda)
d\mu_{M_A,E}(\lambda),
\end{equation}
and conditioning on $E$ gives
\begin{equation}
\text{Pr}(F|M_A,E,M_B) = \int_{\Lambda} \text{Pr}(F|M_B,\lambda)
d\mu_{M_A,E}(\lambda),
\end{equation}
which is the analogue of Eq.~\eqref{eq:Bell:cond} at the ontological
level.

Further, in the case where Alice does not report her measurement
result we have,
\begin{equation}
\text{Pr}(F|M_A,M_B) = \int_{\Lambda} \text{Pr}(F|M_B,\lambda)
d\mu(\lambda),
\end{equation}

Thus, an ontological model for the conditional theory can be defined
as follows.
\begin{definition}
\label{def:Bell:cond}
Let $\mathfrak{F}_{AB} = \langle \Hilb[A] \otimes \Hilb[B],
\mathcal{P}_{AB}, \mathcal{M}_A \times \mathcal{M}_B \rangle$ be a
product measurement fragment, where $\mathcal{P}_{AB}$ contains a
single state $\rho_{AB}$, and consider a Bell local ontological
model of it in which $\rho_{AB}$ is represented by a unique measure
$\mu$, each $M_A \in \mathcal{M}_A$ is represented by a unique
conditional probability distribution $\text{Pr}(E|M_A,\lambda)$, and
each $M_B \in \mathcal{M}_B$ is represented by a unique conditional
probability distribution $\text{Pr}(F|M_B,\lambda)$.  Then, the
\emph{conditional ontological model} of the conditional fragment
$\mathfrak{F}_{B|A}$ is defined by
\begin{itemize}
\item Placing the measure $\mu_{M_A,E}$ in $\Delta_{\rho_{B|E}}$,
where
\begin{equation}
\mu_{M_A,E}(\Omega) = \frac{1}{\text{Pr}_A(E|M_A)}\int_{\Omega}
\text{Pr}(E|M_A,\lambda) d\mu(\lambda).
\end{equation}
Note that there may be more than one measure in
$\Delta_{\rho_{B|E}}$ if $E$ occurs in more than one of the POVMs
in $\mathcal{M}_A$.
\item Associating the unique measure $\mu$ with the state $\rho_B =
\Tr[A]{\rho_{AB}}$.
\item Associating the unique conditional probability distribution
$\text{Pr}(F|M_B,\lambda)$ as defined in the Bell local model
with each $M_B \in \mathcal{M}_B$.
\end{itemize}
\end{definition}
It is straightforward to see that if the original ontological model
reproduces the quantum predictions for $\mathfrak{F}_{AB}$ then the
conditional ontological model reproduces the quantum predictions of
$\mathfrak{F}_{B|A}$.

We now specialize to the product measurement fragment of interest for
Bell's Theorem.  This is $\mathfrak{F}_{AB} = \left \langle \Hilb[A]
\otimes \Hilb[B], \mathcal{P}_{AB}, \mathcal{M}_A \times
\mathcal{M}_B \right \rangle$, where $\Hilb[A] = \Hilb[B] =
\mathbb{C}^d$, $\mathcal{P}_{AB} = \left \{ \Proj{\Phi^+}_{AB} \right
\}$, $\mathcal{M}_A$ consists of all projective measurements of the
form $M_A = \{E_0,E_1\}$ where $E_0 = \Proj{\psi}$, $E_1 = I_A -
\Proj{\psi}$ for some $\Ket{\psi} \in \mathbb{C}^d$ and
$\mathcal{M}_B$ consists of all measurements on $\Hilb[B]$.

Given the available measurements in $\mathcal{M}_A$, $\rho_{B|E_0}$ is
always a pure state and in fact Alice can prepare Bob's system in any
pure state by choosing an appropriate $E_0$, i.e. $\Proj{\psi}^T_A$ if
she wants to prepare $\Proj{\psi}_B$.

\begin{theorem}
\label{prop:Bell:PNC}
Consider a Bell local model for the product measurement fragment
$\mathfrak{F}_{AB}$ defined above that reproduces the quantum
predictions for Alice's measurements, i.e.\ for all $M_A \in
\mathcal{M}_A$, $E \in M_A$,
\begin{equation}
\int_{\Lambda} \mathrm{Pr}(E|M_A,\lambda)d\mu(\lambda) =
\mathrm{Tr}_{AB} \left ( E \otimes I_B \Proj{\Phi^+}_{AB} \right
).
\end{equation}
Then, the conditional ontological model of the conditional fragment
$\mathfrak{F}_{B|A}$ is preparation noncontextual and respects
convex decompositions.
\end{theorem}
\begin{proof}
Given the choice of $\Proj{\Phi^+}_{AB}$, there is only one POVM
$M_A = \{E_0,E_1\}$ that can lead to a given conditional state
$\rho_{B|E_j}$.  The possible conditional states on Bob's side are
$\rho_{B|E_0} = \Proj{\psi}_B$ and $\rho_{B|E_1} =
\sigma^{\psi^{\perp}}_B = \frac{1}{d-1}(I_B - \Proj{\psi}_B)$, which
occur when Alice makes the measurement $\{\Proj{\psi}^T_A, I_A -
\Proj{\psi}^T_A\}$, and so each of Alice's measurements determine a
unique pair of conditional states for Bob.  Therefore, each
$\Delta_{\rho_{B|E_j}}$ contains only a single measure, so the model
is preparation noncontextual for these states.  Similarly, by
construction, there is a unique measure associated with $I/d$.

The only convex combinations that appear in $\mathfrak{F}_{B|A}$ are
those of the form
\begin{equation}
\frac{I_B}{d} = \sum_{j =0}^1 \text{Prob}(E_j|M_A) \rho_{B|E_j},
\end{equation}
where $\text{Prob}(E_j|M_A) = \Tr[AB]{E_j \otimes I_B
\Proj{\Phi^+}_{AB}}$.  Therefore, the result follows if
\begin{equation}
\mu = \sum_{j=0}^1 \text{Prob}(E_j|M_A) \mu_{M_A,E_j}.
\end{equation}
To see this, note that $\sum_j \text{Pr}(E_j|M_A,\lambda) = 1$ for all
$\lambda$, so, for $\Omega \in \Sigma$,
\begin{align}
\mu(\Omega) & = \int_{\Omega} d\mu(\lambda) \\
& = \int_{\Omega} \left [ \sum_j \text{Pr}(E_j|M_A,\lambda) \right
] d\mu(\lambda) \\
& = \sum_j \int_{\Omega} \text{Pr}(E_j|M_A,\lambda) d\mu(\lambda)
\\
& = \sum_j \text{Pr}(E_j|M_A)
\frac{1}{\text{Pr}(E_j|M_A)} \int_{\Omega}
\text{Pr}(E_j|M_A,\lambda) d\mu(\lambda) \\
& = \sum_j \text{Pr}(E_j|M_A) \mu_{M_A,E_j}(\Omega),
\end{align}
where $\text{Pr}(E_j|M_A) = \int_{\Lambda}
\text{Pr}(E_j|M_A,\lambda)d\mu(\lambda)$ is the probability of Alice
obtaining outcome $E_j$ when measuring $M_A$ according to the
ontological model.  By assumption, the model reproduces the quantum
predictions for Alice's measurements, so $\text{Pr}(E_j|M_A) =
\text{Prob}(E_j|M_A)$, from which the result follows.
\end{proof}

Theorem~\ref{prop:Bell:PNC} is enough to show that a $\psi$-ontic
model of the above product measurement fragment that reproduces the
quantum predictions cannot be Bell local.  To see this note that
Theorem~\ref{thm:ME:MEPE} shows that a $\psi$-ontic model cannot be
maximally $\psi$-epistemic and hence, by Theorem~\ref{prop:PC:PC}, if
convexity is respected then preparation noncontextuality must fail for
the maximally mixed state.  However, a Bell local model would give
rise to a preparation noncontextual ontological model for the
conditional theory that respects convexity.  Since the conditional
fragment contains the states needed for Theorem~\ref{prop:PC:PC}, this
implies that Bell locality must fail.

Note that, instead of relying on $\psi$-ontology one could use the
Kochen--Specker theorem to rule out maximally $\psi$-epistemic theories
for $d \geq 3$, as described in Appendix~\ref{App:KS}.  Then,
Theorem~\ref{prop:Bell:PNC} can be viewed as a full proof of Bell's
theorem that does not assume $\psi$-ontology, although this is quite
complicated compared to the standard proofs.

\section*{Part II. $\psi$-ontology theorems\label{SPON}}

Having understood the implications of $\psi$-ontology, it is finally
time to look at the extent to which it can be established.  The first
$\psi$-ontology theorem that was discovered is due to Pusey, Barrett
and Rudolph. Of the available results, it is the most widely discussed
in the literature and I think it makes the strongest case for
$\psi$-ontology, so it receives the longest treatment in
\S\ref{Pusey--Barrett--Rudolph}, including a discussion of
generalizations and criticisms.  Following this, I discuss two further
$\psi$-ontology theorems: Hardy's Theorem in \S\ref{Hardy}, and the
Colbeck--Renner Theorem in \S\ref{CR}.

Before getting to these theorems, I discuss the concept of
\emph{antidistinguishability} in \S\ref{Anti} because it plays a key
role in both the Pusey--Barrett--Rudolph and Hardy's theorems, as well as in some
of the results discussed in \hyperref[Beyond]{Part III}. Both Hardy's Theorem and the Colbeck--Renner Theorem involve assumptions about how dynamics
are represented in ontological models, so this is discussed in
\S\ref{Dyn} after the discussion of Pusey--Barrett--Rudolph is completed.

\section{Antidistinguishability}

\label{Anti}

According to Theorem~\ref{prop:POEM:distinguish}, sets of
operationally distinguishable states are pairwise ontologically
distinct in an ontological model.  However, $\psi$-ontology requires
that all pairs of pure states must be ontologically distinct, and most
pairs of pure states are not orthogonal.  Therefore, to prove
$\psi$-ontology, it is useful to consider a weaker concept.

\begin{definition}
A finite set of quantum states $\{\rho_j\}_{j=1}^n \subseteq
\mathcal{P}$ is \emph{antidistinguishable} if there exists a POVM $M
= \left \{ E_j \right \}_{j=1}^n \in \mathcal{M}$ such that, for
each $j$,
\begin{equation}
\Tr{E_j \rho_j} = 0.
\end{equation}
\end{definition}

Recall that the definition of distinguishability states that there
should exist a POVM $M = \{E_j\}_{j=1}^n \in \mathcal{M}$ such that
$\Tr{E_k\rho_j} = \delta_{jk}$ for all $j$ and $k$.  This is
equivalent to requiring $\Tr{E_j \rho_j} = 1$ for all $j$ by the
following argument.  Firstly, $\Tr{E_j\rho_j} = 1$ is obviously a
special case of $\Tr{E_k \rho_j} = \delta_{jk}$.  For the converse
direction note that if there were a $k \neq j$ such that $\Tr{E_k
\rho_j} > 0$ then $\Tr{\left ( E_j + E_k \right ) \rho_j} > 1$, but
this cannot be the case if $\rho_j$ is a density operator and
$\{E_j\}_{j=1}^n$ is a POVM because then we must have $1 = \Tr{\rho_j}
= \Tr{\sum_m E_m \rho_j} \geq \Tr{\left ( E_j + E_k \right ) \rho_j}$.
Therefore, the difference between distinguishability and
antidistinguishability is simply the replacement of $\Tr{E_j \rho_j} =
1$ with $\Tr{E_j \rho_j} = 0$.  To understand what this means, suppose
that one of the states in a set $\mathcal{D}$ is prepared and you do
not know which.  If $\mathcal{D}$ is distinguishable then there is a
measurement for which each outcome identifies that a specific member
of $\mathcal{D}$ was definitely the state prepared.  On the other
hand, if $\mathcal{D}$ is antidistinguishable then there is a
measurement for which each outcome identifies that a specific member
of $\mathcal{D}$ was definitely not the state prepared.  In this
sense, the two concepts are opposites.

Antidistinguishability is a weaker property than distinguishability
because, if the measurement outcome $j$ tells us that $\rho_j$ was
definitely prepared, then it also tells us that $\rho_k$ was
definitely not prepared for all $k \neq j$.  Therefore, a
distinguishing measurement can be converted to an antidistinguishing
one just by permuting the outcome indices in such a way that none of
them are left invariant, e.g.\ if there are $n$ POVM elements then
this can be done by mapping $j \rightarrow j+1$, where $n+1$ is
identified with $1$.  For sets of three or more states,
antidistinguishability is strictly weaker than distinguishability, as
there exist sets of nonorthogonal quantum states that can be
antidistinguished (see Fig.~\ref{fig:Anti:merc} for an example).  For
a set of two states, antidistinguishability is equivalent to
distinguishability because if you know that one of the states was not
the one prepared then you know that it must have been the other one.

\begin{figure}[t!]
\centering
\includegraphics[width=80mm]{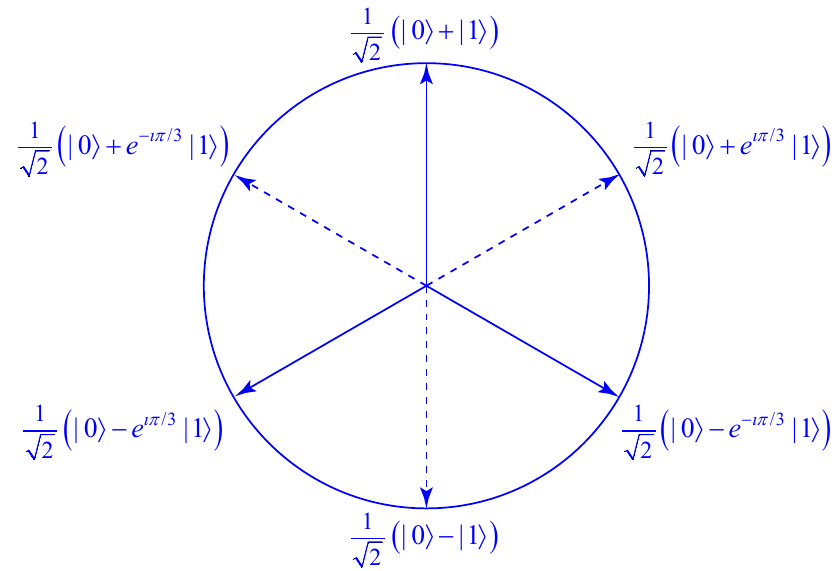}
\caption{\color[HTML]{0000FF}{\label{fig:Anti:merc}A set of nonorthogonal states that can
be antidistinguished, and the POVM that antidistinguishes them.
The diagram represents the equator of the Bloch sphere and the
three antidistinguishable states are shown by the solid arrows.
The dashed arrows show three states that are each orthogonal to
one of the states in the antidistinguishable set. A POVM can be
formed by scaling the projectors onto the three dashed states by a
factor of $2/3$.  By virtue of the orthogonality relations, this
is an antidistinguishing POVM.}}
\end{figure}

Like distinguishability, antidistinguishability has consequences for
the overlaps of probability measures in ontological models.  However,
whilst distinguishability constrains the pairwise overlaps,
antidistinguishability only constrains the $n$-way overlaps, where $n$
is the number of states in the antidistinguishable set.  Since the
variational distance only applies to two measures, a more general
notion of overlap is needed to capture this.

\begin{definition}
\label{def:Anti:partition}
A \emph{countable partition} $\{\Omega_k\}$ of a measurable space
$(\Lambda, \Sigma)$ is an at most countable set of disjoint subsets
$\Omega_k \in \Sigma$, $\Omega_k \cap \Omega_r = \emptyset$ for $k
\neq r$, such that $\cup_k \Omega_k = \Lambda$.  $T(\Lambda)$
denotes the set of countable partitions of $(\Lambda,\Sigma)$.
\end{definition}

\begin{definition}
\label{def:Anti:overlap}
The \emph{overlap} of an at most countable set $\{\mu_j\}$ of
probability measures on a sample space $(\Lambda, \Sigma) $ is given
by
\begin{equation}
\label{eq:Anti:overlap}
L \left ( \left \{ \mu_j \right \} \right ) = \inf_{\{\Omega_k\}
\in T(\Lambda)} \left ( \sum_k \min_j \left (
\mu_j(\Omega_k) \right ) \right ).
\end{equation}
For a set $\{\mu,nu\}$ of two measures, we write $L(\mu,\nu)$ for
$L(\{\mu,nu\})$.
\end{definition}

It can be shown that if $m$ is a measure that dominates $\{\mu_j\}$
then the overlap is given by
\begin{equation}
\label{eq:Anti:domform}
L \left ( \left \{ \mu_j \right \}\right ) = \int_{\Lambda} \min_j
\left (p_j(\lambda) \right ) dm(\lambda),
\end{equation}
where $p_j$ is a density that represents $\mu_j$ with respect to $m$,
i.e. $\mu_j(\Omega) = \int_{\Omega} p_j(\lambda)dm(\lambda)$.  For a
finite set of measures, a dominating measure always exists.  This form
of the overlap is often most useful for calculations.

Note that, if $\{\mu_j\}$ contains $n$ measures then the optimization
in Eq.~\eqref{eq:Anti:overlap} can be restricted to partitions that
contain $n$ subsets.  This is because, if a partition contains $m > n$
subsets, then the same $\mu_j$ is bound to be the minimum for more
than one $\Omega_k$ and these sets can then be replaced by their union
without affecting the result.  Therefore, the overlap is given by
\begin{equation}
L(\{\mu_j\}) = \min_{\{\Omega_k\} \in T_n(\Lambda)} \left ( \sum_{k=1}^n
\min_j \left ( \mu_j(\Omega_k) \right ) \right ),
\end{equation}
where $T_n(\Lambda)$ is the set of $n$-fold partitions of $\Lambda$.
In fact, this can be simplified further to
\begin{equation}
\label{eq:Anti:npart}
L(\{\mu_j\}) = \min_{\{\Omega_j\} \in T_n(\Lambda)} \left ( \sum_j
\mu_j(\Omega_j) \right ),
\end{equation}
because the minimizing partition can always be relabeled such that
$k$ gets mapped to the $j$ that minimizes $\mu_j(\Omega_k)$.

The overlap has the following operational interpretation.  Suppose
that a system is prepared according to one of $n$ preparation
procedures $P_j$, where $P_j$ corresponds to the measure $\mu_j$, each
case having equal a priori probability $1/n$.  You are then told the
actual value of $\lambda$ and your task is to announce a $j$ between
$1$ and $n$ such that $P_j$ was \emph{not} the procedure used to
prepare the system, i.e.\ you want to antidistinguish the probability
measures.  If you adopt a deterministic strategy then the best you can
do is to partition the space into $n$ subsets $\Omega_j$ and announce
that $P_j$ was not the preparation procedure used when $\lambda \in
\Omega_j$.  The probability of failure when using this strategy is
$\frac{1}{n} \sum_{j=1}^n \mu_j(\Omega_j)$ so the minimal probability
of failure is $\frac{1}{n} L(\{\mu_j\})$.  By convexity, a
probabilistic strategy cannot do any better than this.  Therefore,
$L(\{\mu_j\}) = 0$ corresponds to the case where the probability
measures can be antidistinguished perfectly.

The overlap generalizes the variational distance in the following
sense.

\begin{theorem}
\label{prop:Anti:varoverlap}
For two probability measures $\mu, \nu$ on $(\Lambda, \Sigma)$,
$L(\mu,\nu) = 1 - D(\mu,\nu)$, where $D$ is the variational
distance.
\end{theorem}
\begin{proof}
From Eq.~\eqref{eq:Anti:npart},
\begin{align}
L(\mu,\nu) & = \inf_{\Omega \in \Sigma} \left ( \mu(\Omega) +
\nu(\Lambda \backslash \Omega) \right ) \\
& = \inf_{\Omega \in \Sigma} \left ( \mu(\Omega) + 1 - \nu(\Omega)
\right ) \\
& = 1 - \sup_{\Omega \in \Sigma} \left ( \nu(\Omega) - \mu(\Omega)
\right ) \\
& = 1 - \sup_{\Omega \in \Sigma} \left | \mu(\Omega) - \nu(\Omega)
\right | \\
& = 1 - D(\mu,\nu).
\end{align}
\end{proof}

Because of this, the condition for two quantum states $\rho$ and
$\sigma$ to be ontologically distinct in an ontological model can be
restated as: for all $\mu \in \Delta_{\rho}$ and $\nu \in
\Delta_{\sigma}$, $L(\mu,\nu) = 0$.

The reason for introducing the more general $n$-way overlaps is that,
being a weaker concept than distinguishability, antidistinguishability
only places constraints on the $n$-way overlaps rather than on
pairwise overlaps.  To understand the intuition behind these
constraints, consider the case of a finite ontic state space and
suppose that there is an ontic state $\lambda$ that is assigned
nonzero probability by every measure $\mu_j$ in a set of probability
measures $\{\mu_j\}_{j=1}^n$ corresponding to an antidistinguishable
set of preparations $\{\rho_j\}_{j=1}^n$.  Then, $\lambda$ must assign
zero probability to each of the outcomes of the antidistinguishing
measurement.  However, the probabilities that $\lambda$ assigns to the
measurement outcomes must form a probability distribution, so they
cannot all be zero.  Hence, no such $\lambda$ can exist.

In order to formulate this argument rigorously, we have to deal with
measure zero sets.  For this reason, it is better to formulate the
argument in terms of the overlap.

\begin{theorem}
\label{thm:Anti:antiover}
Let $\{\rho_j\}_{j=1}^n$ be an antidistinguishable set of states in
a PM fragment.  Then, in any ontological model that reproduces the
quantum predictions,
\begin{equation}
L(\{\mu_j\}) = 0.
\end{equation}
for every possible choice of probability measures $\mu_j \in
\Delta_{\rho_j}$.
\end{theorem}
\begin{proof}
Let $M = \{E_j\}$ be the antidistinguishing POVM, let $\mu_j \in
\Delta_{\rho_j}$, and let $m$ be a measure that dominates
$\{\mu_j\}$ such that, for $\Omega \in \Sigma$, $\mu_j(\Omega) =
\int_{\Omega} p_j(\lambda) dm(\lambda)$ for some densities $p_j$.
In order to reproduce the predictions of $M$, for every $\text{Pr}
\in \Xi_M$, the ontological model must satisfy
\begin{eqnarray}
\int_{\Lambda} \text{Pr}(E_j|M,\lambda) d \mu_j(\lambda) & = &
\int_{\Lambda} \text{Pr}(E_j|M,\lambda) p_j(\lambda) dm(\lambda)\nonumber\\
& = & \Tr{E_j \rho_j} = 0.
\end{eqnarray}
However,
\begin{eqnarray}
&&\int_{\Lambda} \text{Pr}(E_j|M,\lambda) p_j(\lambda) dm(\lambda) \geq\nonumber\\
&&~\int_{\Lambda} \text{Pr}(E_j|M,\lambda) \min_k \left (p_k(\lambda)
\right ) dm(\lambda),
\end{eqnarray}
so
\begin{equation}
\int_{\Lambda} \text{Pr}(E_j|M,\lambda) \min_k \left (p_k(\lambda)
\right ) dm(\lambda) = 0.
\end{equation}
Now, sum both sides over $j$.  Because, $\sum_j
\text{Pr}(E_j|M,\lambda) = 1$ for every $\lambda \in \Lambda$, this
gives
\begin{equation}
\int_{\Lambda} \min_k \left ( p_k(\lambda) \right ) dm(\lambda) =
0,
\end{equation}
or in other words, by Eq.~\eqref{eq:Anti:domform}, $L \left ( \left
\{ \mu_k \right \} \right ) = 0$.
\end{proof}

Finally, a few remarks are in order about terminology.
Antidistinguishability was first introduced by Caves, Fuchs and Shack
in their study of the compatibility of quantum state assignments,
where it was called ``post-Peierls (PP) incompatibility''
\cite{Caves2002}.  The concept did not attract much attention until
its use in the Pusey--Barrett--Rudolph Theorem, but since then more systematic studies of
antidistinguishability have appeared \cite{Bandyopadhyay2013,
Perry2014}, where it goes under the name ``conclusive exclusion of
quantum states''.  Therefore, if you see the terms ``PP
incompatibility'' and ``conclusive exclusion'' in the literature then
you can rest assured that they mean the same thing as
antidistinguishability.  I think ``antidistinguishability'' is a more
suggestive name, but time will tell which terminology becomes
standard.

\section{The Pusey--Barrett--Rudolph Theorem}

\label{Pusey--Barrett--Rudolph}

The Pusey--Barrett--Rudolph Theorem is the most prominent $\psi$-ontology theorem.  It is
based on an assumption known as the Preparation Independence Postulate
(PIP), which essentially says that composite systems prepared in a
product state should be independent of one another.  It is this
postulate that allows bounds on $n$-way overlaps coming from
antidistinguishability to be leveraged into the bounds on pairwise
overlaps needed to establish $\psi$-ontology.  The PIP is presented in
its simplest form in \S\ref{PIP} in order to facilitate the simplest
possible proof of the main result in \S\ref{Main}.  However, the PIP
is the most controversial assumption of the Pusey--Barrett--Rudolph Theorem, so its
motivation is discussed in \S\ref{MPIP}, weakenings of it that still
allow the Pusey--Barrett--Rudolph Theorem to go through are discussed in \S\ref{WPusey--Barrett--Rudolph}, its
necessity for proving the theorem is established in \S\ref{NPIP}, and
criticism of it is discussed in \S\ref{CPIP}.  Finally, other
criticisms of the Pusey--Barrett--Rudolph Theorem, not directed at the PIP, are discussed
in \S\ref{OCPusey--Barrett--Rudolph}.

\subsection{The Preparation Independence Postulate}

\label{PIP}

The Preparation Independence Postulate (PIP) is an assumption about
how composite systems should be modeled when the subsystems are
prepared independently of one another.  Suppose Alice prepares her
system in a quantum state $\rho_A$ and Bob prepares his system in a
quantum state $\rho_B$, such that their joint state is a product
$\rho_A \otimes \rho_B$.  Suppose further that Alice has an
ontological model with ontic state space $(\Lambda_A,\Sigma_A)$ that
reproduces the quantum predictions for her system in isolation and
that $\mu_A$ is the probability measure that it assigns to her
preparation procedure.  Similarly, Bob has an ontological model for
his system with ontic state space $(\Lambda_B,\Sigma_B)$ that assigns
the measure $\mu_B$ to his preparation procedure.  Then, the PIP says
that there ought to be an ontological model for the joint system with
ontic state space $(\Lambda_A \times \Lambda_B, \Sigma_A \otimes
\Sigma_B)$ in which the product measure
\begin{equation}
\mu_A \times \mu_B(\Omega) = \int_{\Lambda_B}
\mu_A(\Omega_{\lambda_B}) d \mu_B(\lambda_B),
\end{equation}
reproduces the operational predictions for $\rho_A \otimes \rho_B$,
where $\Omega_{\lambda_B} = \{\lambda_A \in \Lambda_A |
(\lambda_A,\lambda_B) \in \Omega \}$.

This assumption looks superficially plausible, and indeed it is not
much commented upon in the Pusey--Barrett--Rudolph paper.  However, since it is the most
controversial assumption of the Pusey--Barrett--Rudolph Theorem, it is worth pausing to
define it a bit more rigorously.  Discussion of the motivation for the
PIP and criticism of it is deferred until \S\ref{MPIP}--\S\ref{CPIP}.

To start with, we need to understand how fragments for subsystems
can be combined to form fragments for composite systems.
\begin{definition}
Let $\mathfrak{F}_A = \langle \Hilb[A], \mathcal{P}_A, \mathcal{M}_A
\rangle$ and $\mathfrak{F}_B = \langle \Hilb[B], \mathcal{P}_B,
\mathcal{M}_B \rangle$ be PM fragments.  The \emph{direct product}
of $\mathfrak{F}_A$ and $\mathfrak{F}_B$ is
\begin{equation}
\mathfrak{F}_A \times \mathfrak{F}_B = \langle \Hilb[A] \otimes
\Hilb[B], \mathcal{P}_A \times \mathcal{P}_B, \mathcal{M}_A \times
\mathcal{M}_B \rangle,
\end{equation}
where
\begin{itemize}
\item $\mathcal{P}_A \times \mathcal{P}_B$ consists of all product
states $\rho_A \otimes \rho_B$ with $\rho_A \in \mathcal{P}_A$ and
$\rho_B \in \mathcal{P}_B$.
\item $\mathcal{M}_A \times \mathcal{M}_B$ consists of all POVMs of
the form $\{E_j \otimes F_k\}$ with $\{E_j\} \in \mathcal{M}_A$
and $\{F_k\} \in \mathcal{M}_B$.
\end{itemize}

A \emph{general product} (or just product) of $\mathfrak{F}_A$ and
$\mathfrak{F}_B$ is any PM fragment of the form
\begin{equation}
\mathfrak{F}_{AB} = \langle \Hilb[A] \otimes \Hilb[B],
\mathcal{P}_{AB}, \mathcal{M}_{AB} \rangle,
\end{equation}
where
\begin{itemize}
\item $\mathcal{P}_A \times \mathcal{P}_B \subseteq
\mathcal{P}_{AB}$.
\item $\mathcal{M}_A \times \mathcal{M}_B \subseteq
\mathcal{M}_{AB}$.
\end{itemize}

$\mathfrak{F}_A$ and $\mathfrak{F}_B$ are called the \emph{factors}
of $\mathfrak{F}_{AB}$.
\end{definition}

A direct product fragment represents a situation in which there are
two separated systems and we can only do separate prepare-and-measure
experiments on them individually.  It only allows product states and
local measurements, so there is no possibility of entanglement or even
classical correlation between the two systems.  A general product
fragment contains all the states and measurements that are in the
direct product, but it may, in addition, include some extra correlated
or entangled states and measurements on the composite system.

Any two ontological models $\Theta_A = (\Lambda_A, \Sigma_A, \Delta_A,
\Xi_A)$ and $\Theta_B = (\Lambda_B, \Sigma_B, \Delta_B, \Xi_B)$ for
$\mathfrak{F}_A$ and $\mathfrak{F}_B$ can be combined to form an
ontological model $\Theta_{AB} = (\Lambda_{AB}, \Sigma_{AB},
\Delta_{AB}, \Xi_{AB})$ for the direct product $\mathfrak{F}_A \times
\mathfrak{F}_B$.  To do so, set
\begin{itemize}
\item $\Lambda_{AB} = \Lambda_A \times \Lambda_B$ and $\Sigma_{AB} =
\Sigma_A \otimes \Sigma_B$.
\item For all $\rho_A \in \mathcal{P}_A$, $\rho_B \in \mathcal{P}_B$,
let $\Delta_{AB}[\rho_A \otimes \rho_B] = \Delta_A[\rho_A] \times
\Delta_B[\rho_B]$ be the set of measures of the form $\mu_{AB} =
\mu_A \times \mu_B$, where $\mu_A \in \Delta_A[\rho_A]$ and $\mu_B
\in \Delta_B[\rho_B]$.
\item  For all $M_A \in \mathcal{M}_A$, $M_B \in \mathcal{M}_B$, let
$\Xi_{AB}[M_A,M_B]$ be the set of conditional probability
distributions of the form
$\text{Pr}_{AB}(E \otimes F|M_A,M_B\lambda_A,\lambda_B) =
\text{Pr}_A(E|M_A,\lambda_A)\text{Pr}_B(F|M_B,\lambda_B)$,
where $\text{Pr}_A \in \Xi_A[M_A]$ and $\text{Pr}_B \in
\Xi_B[M_B]$.
\end{itemize}
Then, we have, for $\mu_{AB} \in \Delta_{AB}[\rho_A \otimes \rho_B]$,
\begin{align}
& \text{Pr}_{AB}(E,F|\rho_A \otimes \rho_B,M_A,M_B)\nonumber\\
& = \int_{\Lambda_{AB}} \text{Pr}_{AB}(E \otimes F
|M_A,M_B,\lambda_A,\lambda_B) d\mu_{AB}(\lambda_A,
\lambda_B) \\
& = \int_{\Lambda_A} \text{Pr}_A(E|M_A,\lambda_A) d \mu_{A}
(\lambda_A) \nonumber\\& \quad\quad\quad \times \int_{\Lambda_B} \text{Pr}_B(F|M_B,\lambda_B) d
\mu_{B}(\lambda_B) \\
& = \text{Pr}_A(E|\rho_A,M_A)\text{Pr}_B(F|\rho_B,M_B).
\end{align}
From this, it follows that if $\Theta_A$ and $\Theta_B$ reproduce the
quantum predictions for $\mathfrak{F}_A$ and $\mathfrak{F}_B$
respectively, then $\Theta_{AB}$ reproduces the quantum predictions
for $\mathfrak{F}_A \times \mathfrak{F}_B$.

The key point is that, in order to model a direct product fragment, we
only need to take the Cartesian product of the subsystem ontic state
spaces and we only need to use product measures.  The PIP states that
this should remain true even if the product fragment is expanded to
include entangled measurements (but still no entangled states).  A
little more terminology is required to state this formally.

\begin{definition}
Let $\mathfrak{F}_A = \langle \Hilb[A], \mathcal{P}_A, \mathcal{M}_B
\rangle$ and $\mathfrak{F}_B = \langle \Hilb[B], \mathcal{P}_B,
\mathcal{M}_B \rangle$ be PM fragments.  A \emph{product state
fragment} is a product fragment $\mathfrak{F}_{AB} = \langle
\Hilb[A] \otimes \Hilb[B], \mathcal{P}_{AB}, \mathcal{M}_{AB}
\rangle$ where $\mathcal{P}_{AB} = \mathcal{P}_A \times
\mathcal{P}_B$, i.e.\ it only includes product states, but it may
have entangled measurements.
\end{definition}

\begin{definition}
A pair of ontological models $\Theta_A = (\Lambda_A, \Sigma_A,
\Delta_A, \Xi_A)$, $\Theta_B = (\Lambda_B, \Sigma_B, \Delta_B,
\Xi_B)$ for PM fragments $\mathfrak{F}_A = \left \langle \Hilb[A],
\mathcal{P}_A, \mathcal{M}_A \right \rangle$ and $\mathfrak{F}_B =
\left \langle \Hilb[B], \mathcal{P}_B, \mathcal{M}_B \right \rangle$
are \emph{compatible} with an ontological model $\Theta_{AB} =
(\Lambda_{AB}, \Sigma_{AB}, \Delta_{AB}, \Xi_{AB})$ for a product
state fragment $\mathfrak{F}_{AB} = \left \langle \Hilb[A] \otimes
\Hilb[B], \mathcal{P}_A \times \mathcal{P}_B, \mathcal{M}_{AB}
\right \rangle$ if $\Theta_{AB}$ satisfies
\begin{itemize}
\item The Cartesian Product Assumption (CPA):
\begin{align}
\Lambda_{AB} & = \Lambda_A \times \Lambda_B, & \Sigma_{AB} & =
\Sigma_A \otimes \Sigma_B.
\end{align}
\item The No-Correlation Assumption (NCA): For all $\rho_A \in
\mathcal{P}_A$, $\rho_B \in \mathcal{P}_B$, $\Delta_{AB}[\rho_A
\otimes \rho_B] = \Delta_A[\rho_A] \times \Delta_B[\rho_B]$ is the
set of measures of the form $\mu_A \times \mu_B$, where $\mu_A \in
\Delta_A[\rho_A]$ and $\mu_B \in \Delta_B[\rho_B]$.  This means
that $\mu_{AB} \in \Delta_{AB}[\rho_A \otimes \rho_B]$ is of the
form
\begin{equation}
\mu_{AB}(\Omega) = \int_{\Lambda_B} \mu_A(\Omega_{\lambda_B}) d
\mu_B(\lambda_B),
\end{equation}
where $\Omega_{\lambda_B} = \{\lambda_A \in \Lambda_A |
(\lambda_A,\Lambda_B) \in \Omega\}$, $\mu_A \in \Delta_A[\rho_A]$
and $\mu_B \in \Delta_B[\rho_B]$.
\end{itemize}
\end{definition}

The CPA says that, when modeling independent local preparations,
there are no additional properties of the joint system that are not
derived from the properties of the individual systems.  In other words,
there are no ``genuinely nonlocal properties'' that are irrelevant for
local measurements but come into play when making joint measurements
of the two systems.  The NCA says that, further, local preparations
can be modeled as product measures.

\begin{definition}
A pair of ontological models $\Theta_A$ and $\Theta_B$ for PM
fragments $\mathfrak{F}_A$ and $\mathfrak{F}_B$ that reproduce the
quantum predictions satisfy the \emph{Preparation Independence
Postulate (PIP)} if, for all product state fragments
$\mathfrak{F}_{AB}$ with $\mathfrak{F}_A$ and $\mathfrak{F}_B$ as
factors, there exists an ontological model $\Theta_{AB}$ that is
compatible with $\Theta_A$ and $\Theta_B$ and reproduces the quantum
predictions.
\end{definition}

Note that, it is very important that the PIP only applies to product
state fragments and not to general product fragments.  In other words,
it does not have any implications for how entangled states should be
represented in an ontological model.  Suppose we have an ontological
model $(\Lambda_{AB}, \Sigma_{AB}, \Delta_{AB}, \Xi_{AB})$ for a
product state fragment $\mathfrak{F}_{AB}$ that satisfies the PIP and
we want to extend it to a more general product fragment
$\mathfrak{F}^{\prime}_{AB}$ that differs from $\mathfrak{F}_{AB}$ by
the addition of some entangled states.  Then, we have to allow for the
fact that the extended ontological model $(\Lambda^{\prime}_{AB},
\Sigma^{\prime}_{AB}, \Delta^{\prime}_{AB}, \Xi^{\prime}_{AB})$ might
have extra ontic states that have zero measure for all product states,
but nonzero measure for some entangled state.  Therefore, all we can
conclude about the ontic state space of the extended model from the
PIP is that $\Lambda_A \times \Lambda_B \subseteq
\Lambda^{\prime}_{AB}$, that $\Sigma_A \otimes \Sigma_B$ is a
$\sigma$-subalgebra of $\Sigma^{\prime}_{AB}$, and that independent
local preparations have measure zero outside the $\Lambda_A \times
\Lambda_B$ subset, but there may be other ontic states that come into
play when we prepare entangled states.  Similarly, probability
measures that induce correlations between $\Lambda_A$ and $\Lambda_B$
may come into play when we prepare entangled states.

The reason this is important is that $\psi$-ontic models do not
satisfy $\Lambda^{\prime}_{AB} = \Lambda_A \times \Lambda_B$ and it
would be poor form to prove a $\psi$-ontology theorem only to find
that $\psi$-ontic models are also ruled out by its assumptions.  To
understand this, consider the Beltrametti--Bugajski model for two
Hilbert spaces $\Hilb[A]$, $\Hilb[B]$ and for the composite system
$\Hilb[A] \otimes \Hilb[B]$.  For the two systems individually, we
have the two projective Hilbert spaces consisting of all pure states
$\Proj{\psi}_A$ and $\Proj{\phi}_B$ as the ontic state spaces.  If we
only want to model a product state fragment then we only need states
of the form $\Proj{\psi}_A \otimes \Proj{\phi}_B$, and the set of such
states is isomorphic to the Cartesian product of the two projective
Hilbert spaces.  Similarly, the point measure on $\Proj{\psi}_A
\otimes \Proj{\phi}_B$ can be written as the product of a point
measure at $\Proj{\psi}_A$ with a point measure at $\Proj{\psi}_B$.
Therefore, the model satisfies the PIP\@.  On the other hand, if we
want to model entangled states then we need additional ontic states
$\Proj{\psi}_{AB}$ that cannot be written as products, and the point
measure at $\Proj{\psi}_{AB}$ cannot be written as a product of
measures on the ontic state spaces of the factors.  Therefore,
modeling entangled states requires additional ontic states and
probability measures, but this is not a violation of the PIP because
the PIP only applies to product state fragments.

For ease of exposition, we have so far confined attention to
composites of two subsystems.  The Pusey--Barrett--Rudolph Theorem actually requires the
generalization to $n$ systems, but only in the case of $n$ identical
factors.  The required generalization should be obvious, but for
completeness here are the definitions.

\begin{definition}
Let $\mathfrak{F} = \langle \Hilb, \mathcal{P}, \mathcal{M} \rangle$
be a PM fragment.  An \emph{$n$ fold product state fragment} with
factors $\mathfrak{F}$ is any fragment $\mathfrak{F}^n$ of the form
\begin{equation}
\mathfrak{F}^n = \langle \Hilb^{\otimes n},
\mathcal{P}^{\times n}, \mathcal{M}^{\prime} \rangle,
\end{equation}
where $\mathcal{P}^{\times n}$ consists of states of the form
$\otimes_{j=1}^n \rho_j$ for $\rho_j \in \mathcal{P}$, and
$\mathcal{M}^{\times n} \subseteq \mathcal{M}^{\prime}$ where
$\mathcal{M}^{\times n}$ consists of all POVMs of the form
$\{\otimes_{j=1}^n E^{(j)}_{k_j}\}$ for $\{E^{(j)}_{k_j}\} \in
\mathcal{M}$.
\end{definition}

\begin{definition}
An ontological model $\Theta = (\Lambda, \Sigma, \Delta, \Xi)$ of a
fragment $\mathfrak{F}$ is \emph{compatible} with an ontological
model $\Theta^n = (\Lambda^{\prime}, \Sigma^{\prime},
\Delta^{\prime}, \Xi^{\prime})$ for an $n$-fold product state
fragment $\mathfrak{F}^n$ with factors $\mathfrak{F}$ if it
satisfies:
\begin{itemize}
\item The Cartesian Product Assumption (CPA):
\begin{align}
\Lambda^{\prime} & = \Lambda^{\times n}, & \Sigma^{\prime} & =
\Sigma^{\otimes n}.
\end{align}
\item The No-Correlation Assumption (NCA):
$\Delta^{\prime}_{\otimes_{j=1}^n \rho_j}$ consists of all product
measures of the form $\times_{j=1}^n\mu_j$, where $\mu_j \in
\Delta_{\rho_j}$.
\end{itemize}
\end{definition}

\begin{definition}
An ontological model for a PM fragment $\mathfrak{F}$ that
reproduces the quantum predictions satisfies the PIP if it is
compatible with an ontological model for any $n$-fold product state
fragment $\mathfrak{F}^n$ with factors $\mathfrak{F}$ that
reproduces the quantum predictions.
\end{definition}

\subsection{The main result}

\label{Main}

Recalling Definition~\ref{def:POEM:PON} of a $\psi$-ontic model, we
are now in a position to prove the main result of the Pusey--Barrett--Rudolph paper.

\begin{theorem}[The Pusey--Barrett--Rudolph Theorem]
\label{thm:Main:Pusey--Barrett--Rudolph}
Let $\mathfrak{F} = \langle \Hilb, \mathcal{P}, \mathcal{M} \rangle$
be a PM fragment where $\mathcal{P}$ contains all
pure states on $\Hilb$.  Then, any ontological model of
$\mathfrak{F}$ that reproduces the quantum predictions and satisfies
the PIP is $\psi$-ontic.
\end{theorem}

The proof strategy adopted here is due to Moseley \cite{Moseley2014},
and is simpler than that of the original Pusey--Barrett--Rudolph paper.  It is based on
two lemmas.  The first establishes a connection between
antidistinguishability of product states and ontological distinctness,
and the second establishes that the required antidistinguishability
holds.  Because they are quite technical, it is worth pausing to
sketch a special case.

\begin{example}
\label{exa:Main:Pusey--Barrett--Rudolph}
Let $\left \{ \Ket{0}, \Ket{1} \right \}$ be a basis for
$\mathbb{C}^2$ and consider the two states $\Proj{0}$ and
$\Proj{+}$, where $\Ket{+} = \frac{1}{\sqrt{2}} \left ( \Ket{0} +
\Ket{1} \right )$.  These two states are not orthogonal, so we
cannot use distinguishability to establish their ontological
distinctness.  Since there are only two of them, they cannot be
antidistinguishable either.  However, now consider two copies of the
system and the four states $\Proj{0}\otimes \Proj{0}$, $\Proj{0}
\otimes \Proj{+}$, $\Proj{+} \otimes \Proj{0}$ and $\Proj{+} \otimes
\Proj{+}$.  It turns out that these four states are
antidistinguishable.  The measurement required to antidistinguish
them is entangled and consists of the projectors onto the following
four orthonormal vectors.
\begin{align}
\Ket{\phi_{00}} & = \frac{1}{\sqrt{2}} \left ( \Ket{0}
\Ket{1} + \Ket{1} \Ket{0} \right) \\
\Ket{\phi_{01}} & = \frac{1}{\sqrt{2}} \left ( \Ket{0}
\Ket{-} + \Ket{1} \Ket{+} \right) \\
\Ket{\phi_{10}} & = \frac{1}{\sqrt{2}} \left ( \Ket{+}
\Ket{1} + \Ket{-} \Ket{0} \right) \\
\Ket{\phi_{11}} & = \frac{1}{\sqrt{2}} \left ( \Ket{+}
\Ket{-} + \Ket{-} \Ket{+} \right),
\end{align}
where $\Ket{-} = \frac{1}{\sqrt{2}} \left ( \Ket{0} - \Ket{1}
\right )$.  It is easy to check that $\Ket{\phi_{00}}$ is orthogonal to
$\Ket{0}\otimes \Ket{0}$, $\Ket{\phi_{01}}$ is orthogonal to
$\Ket{0} \otimes \Ket{+}$, $\Ket{\phi_{10}}$ is orthogonal to
$\Ket{+} \otimes \Ket{0}$, and $\Ket{\phi_{11}}$ is orthogonal to
$\Ket{+} \otimes \Ket{+}$, as required.

The second lemma shows that a similar thing happens for any pair of
pure states $\Proj{\psi_0}, \Proj{\psi_1}$ for which
$\Tr{\Proj{\psi_0}\Proj{\psi_1}} \leq \frac{1}{2}$.  Specifically,
the set of states
\begin{equation}
\left \{ \Proj{\psi_{0}} \otimes \Proj{\psi_{0}}, \Proj{\psi_{0}}
\otimes \Proj{\psi_{1}}, \Proj{\psi_{1}} \otimes
\Proj{\psi_{0}}, \Proj{\psi_{1}} \otimes \Proj{\psi_{1}} \right
\},
\end{equation}
is always antidistinguishable.

Returning to the example, by Theorem~\ref{thm:Anti:antiover},
antidistinguishability implies that, in an ontological model, the
probability measures corresponding to $\Proj{0}\otimes \Proj{0}$,
$\Proj{0} \otimes \Proj{+}$, $\Proj{+} \otimes \Proj{0}$ and
$\Proj{+} \otimes \Proj{+}$ can have no four-way overlap.  However,
if the PIP is assumed, this implies that $\Proj{0}$ and $\Proj{+}$
must be ontologically distinct.  As a rough sketch of this argument,
consider the case of a finite ontic state space, let $\mu_{0}$ and
$\mu_{+}$ be a pair of measures associated with $\Proj{0}$ and
$\Proj{+}$ in the ontological model, and assume that there is some
$\lambda$ to which they both assign nonzero probability.  According
to the PIP, the four product states must be represented by measures
of the form $\mu_{0} \times \mu_{0}$, $\mu_{0} \times \mu_{+}$,
$\mu_{+} \times \mu_{0}$, and $\mu_{+} \times \mu_{+}$.  Each of
these assigns a nonzero probability to $(\lambda, \lambda)$, since
both $\mu_{0}$ and $\mu_{+}$ assign nonzero probability to
$\lambda$, so there is a common ontic state to which all four
measures assign nonzero probability.  However,
antidistinguishability implies that there can be no four-way overlap
between the measures so there is a contradiction, and hence
$\Proj{0}$ and $\Proj{+}$ must be ontologically distinct.

The first lemma is a more formal and slightly more general version
of this argument.  The generalization is needed to adapt these ideas
to the case where $\Tr{\Proj{\psi_0}\Proj{\psi_1}} > \frac{1}{2}$.
\end{example}

\begin{lemma}
\label{lem:Main:prod}
Let $\mathfrak{F} = \langle \Hilb, \mathcal{P}, \mathcal{M} \rangle$
be a PM fragment with ontological model $\Theta = (\Lambda, \Sigma,
\Delta, \Xi)$ that reproduces the quantum predictions, and let
$\rho_0, \rho_1 \in \mathcal{P}$.  Let $\mathfrak{F}^n = \langle
\Hilb^{\otimes n}, \mathcal{P}^{\times n}, \mathcal{M}^{\prime}
\rangle$ be an $n$-fold product state fragment with factors
$\mathfrak{F}$ and let $\Theta^n = (\Lambda^{\times n},
\Sigma^{\otimes n}, \Delta^{\prime}, \Xi^{\prime})$ be an
ontological model of it that is compatible with $\Theta$.

For $\bm{j} \in \{0,1\}^n$, let
\begin{equation}
\rho_{\bm{j}} = \rho_{j_1} \otimes \rho_{j_2} \otimes \ldots
\otimes \rho_{j_n},
\end{equation}
where $\bm{j} = (j_1,j_2,\ldots,j_n)$.

If $L \left ( \left \{ \mu_{\bm{j}}\right \}_{\bm{j} \in \{0,1\}^n}
\right ) = 0$ for all possible choices of $\mu_{\bm{j}} \in
\Delta^{\prime}_{\rho_{\bm{j}}}$ then $\rho_0$ and $\rho_1$ are
ontologically distinct in $\Theta$.
\end{lemma}
\begin{proof}
Consider any choice of $\mu_j \in \Delta_{\rho_j}$.  By the NCA,
\begin{equation}
\mu_{\bm{j}} = \mu_{j_1} \times \mu_{j_2} \times \cdots \times \mu_{j_n},
\end{equation}
is in $\Delta^{\prime}_{\rho_{\bm{j}}}$ and hence, by assumption $L
\left ( \left \{ \mu_{\bm{j}}\right \}_{\bm{j} \in \{0,1\}^n}\right
) = 0$.  Let $m$ be a measure that dominates $\mu_0$ and $\mu_1$,
and let $p_0$ and $p_1$ be corresponding densities. Then,
\begin{eqnarray}
&&L \left ( \left \{ \mu_{\bm{j}}\right \}_{\bm{j} \in \{0,1\}^n}\right ) =
\int_{\Lambda_1} dm(\lambda_1) \times \int_{\Lambda_2} dm(\lambda_2)
\times\ldots\nonumber\\ &&\quad\quad \times \int_{\Lambda_n} dm(\lambda_n)
\min_{\bm{j} \in \{0,1\}^n} \left [ \prod_{k = 1}^n p_{j_k}(\lambda_k)
\right ].
\end{eqnarray}
Here, the subscripts are just being used to identify the different
copies of $\Lambda$, so $\Lambda_1 = \Lambda_2 = \ldots = \Lambda_n
= \Lambda$ and $\lambda_k \in \Lambda_k$.  Since each term
$p_{j_k}(\lambda)$ in the product is positive, the minimum is
attained by minimizing each term individually.  We therefore have
\begin{align}
L \left ( \left \{ \mu_{\bm{j}}\right \}_{\bm{j} \in \{0,1\}} \right ) &
= \left ( \int_{\Lambda}dm(\lambda) \min_{j \in \{0,1\}} \left
[ p_{j}(\lambda) \right ] \right )^n \\
& = \left [ L \left ( \mu_0, \mu_1 \right ) \right ]^n,
\end{align}
and since $L \left ( \left \{ \mu_{\bm{j}}\right \}_{\bm{j} \in
\{0,1\}^n} \right ) = 0$, this implies that $L \left ( \mu_0,
\mu_1\right ) = 0$ because it is a positive function.  Since this
holds for any choice of $\mu_j \in \Delta_{\rho_j}$, this implies
that $\rho_0$ and $\rho_1$ are ontologically distinct in $\Theta$.
\end{proof}

\begin{lemma}
\label{lem:Main:anti}
Let $\mathfrak{F} = \langle \Hilb, \mathcal{P}, \mathcal{M} \rangle$
be a PM fragment and let $\Proj{\psi_0}, \Proj{\psi_1} \in
\mathcal{P}$ be pure states that satisfy $\mathrm{Tr} \left (
\Proj{\phi}\Proj{\psi} \right ) \leq \frac{1}{2}$.  Then there
exists a two fold product state fragment with factors $\mathfrak{F}$
in which the set of states
\begin{equation}
\mathcal{D} = \left \{\Proj{\psi_0} \otimes \Proj{\psi_0},
\Proj{\psi_0} \otimes \Proj{\psi_1}, \Proj{\psi_1} \otimes
\Proj{\psi_0}, \Proj{\psi_1} \otimes \Proj{\psi_1} \right \}
\end{equation}
is antidistinguishable.
\end{lemma}
\begin{proof}
The result follows if there exists a POVM
$\{E_{00},E_{01},E_{10},E_{11}\}$ on $\Hilb \otimes \Hilb$ that
antidistinguishes $\mathcal{D}$, i.e.
\begin{equation}
\label{eq:Main:anti}
\Tr{E_{jk} \Proj{\psi_j} \otimes \Proj{\psi_k}} = 0
\end{equation}
Let $\Ket{\psi_0}$ and $\Ket{\psi_1}$ be vector representatives of
$\Proj{\psi_0}$ and $\Proj{\psi_1}$ such that
$\BraKet{\psi_0}{\psi_1}$ is real and positive.  Such vectors can
always be found by appropriate choice of global phases.  Then we
have $\BraKet{\psi_0}{\psi_1} \leq \frac{1}{\sqrt{2}}$ so we can
write $\BraKet{\psi_0}{\psi_1} = \sin \vartheta$, where $0 \leq
\vartheta \leq \frac{\pi}{4}$.  We can then choose an orthonormal
basis $\left \{ \Ket{0}, \Ket{1} \right \}$ for the two dimensional
subspace spanned by $\Ket{\psi_0}$ and $\Ket{\psi_1}$ such that
\begin{align}
\Ket{\psi_0} & = \Ket{0}, & \Ket{\psi_1} = \sin \vartheta \Ket{0}
+ \cos \vartheta \Ket{1}.
\end{align}
Note that the POVM elements $E_{jk}$ only need to be defined on this
subspace because the projector onto its orthogonal complement can be
added on to any one of them without changing the measurement
probabilities given in Eq.~\eqref{eq:Main:anti}.  Hence, the special
case $\vartheta = \frac{\pi}{4}$ is just the case dealt with in
Example~\ref{exa:Main:Pusey--Barrett--Rudolph}, for which the projectors $E_{jk} =
\Proj{\Phi_{jk}}$ provide the required antidistinguishing POVM.

The basic idea to deal with the case $\vartheta < \frac{\pi}{4}$ is
to note that in this case $\Proj{\psi_0}$ and $\Proj{\psi_1}$ are
more distinguishable than $\Proj{0}$ and $\Proj{+}$, and there is
always a physically allowed transformation to map a pair of states
to any pair of less distinguishable states.  The antidistinguishing
measurement therefore consists of performing the transformation
that maps $\Proj{\psi_0}$ to $\Proj{0}$, i.e.\ leaves it invariant,
and maps $\Proj{\psi_1}$ to $\Proj{+}$, followed by performing the
projective measurement $\left \{ \Proj{\Phi_{jk}} \right \}_{j,k \in
\{0,1\}}$.  Taken together, these two steps define the required
POVM.

More explicitly, consider the linear map $\mathcal{E}$ on the set of
density operators on $\Hilb$ defined by
\begin{equation}
\mathcal{E}(\rho) = M_0 \rho M_0^{\dagger} + M_1 \rho M_1^{\dagger},
\end{equation}
where the operators $M_0$ and $M_1$ are given by
\begin{align}
& M_0 = \KetBra{0}{0} + \tan \vartheta \KetBra{1}{1},\\
& M_1 = \sqrt{\frac{1 - \tan^2 \vartheta}{2}}\left ( \Ket{0} + \Ket{1} \right ) \Bra{1}.
\end{align}
This is an example of a completely-positive, trace-preserving map,
which is a dynamically allowed transformation of the set of quantum
states.  It does not matter if you do not know this formalism, since
for present purposes all that matters is that
\begin{align}
\label{eq:Main:incip}
\mathcal{E}(\Proj{\psi}) & = \Proj{0}, & \mathcal{E}(\Proj{\phi})
= \Proj{+},
\end{align}
which can easily be verified, and that, when $0 \leq \vartheta \leq
\frac{\pi}{4}$,
\begin{equation}
\label{eq:Main:idsum}
M_0^{\dagger}M_0 + M_1^{\dagger}M_1 = I_{\Hilb},
\end{equation}
where $I_{\Hilb}$ is the identity operator on $\Hilb$, which is
again easily verified.

Next, define the operators
\begin{equation}
E_{jk} = \sum_{r,s = 0}^1 M_r^{\dagger} \otimes M_s^{\dagger}
\Proj{\Phi_{jk}} M_r \otimes M_s.
\end{equation}
These operators form the required antidistinguishing POVM\@.  To see
this, we first have to check that $\{E_{jk}\}_{j,k \in \{0,1\}}$ is
a POVM and then that Eq.~\eqref{eq:Main:anti} is satisfied.  To see
that the operators $E_{jk}$ are positive note that any operator of
the form $N^{\dagger}N$ is positive and hence setting $N =
\Proj{\Phi_{jk}} M_r \otimes M_s$ shows that
\begin{equation}
M_r^{\dagger} \otimes M_s^{\dagger} \Proj{\Phi_{jk}} \Proj{\Phi_{jk}}
M_r \otimes M_s = M_r^{\dagger} \otimes M_s^{\dagger}
\Proj{\Phi_{jk}} M_r \otimes M_s,
\end{equation}
is positive, where we have used the fact that $\Proj{\Phi_{jk}}$ is
a projector, and hence idempotent.  Further, a sum of positive
operators is positive, and hence the $E_{jk}$'s are positive.  We
next need to check that they sum to the identity on $\Hilb \otimes
\Hilb$.  This follows from the fact that $\sum_{jk} \Proj{\Phi_{jk}}
= I_{\Hilb \otimes \Hilb}$ together with Eq.~\eqref{eq:Main:idsum}
as follows.
\begin{align}
\sum_{j,k = 0}^1 E_{jk} & = \sum_{j,k=}^1 \sum_{r,s = 0}^1
M_r^{\dagger} \otimes
M_s^{\dagger} \Proj{\Phi_{jk}} M_r \otimes M_s \\
& = \sum_{r,s = 0}^1 M_r^{\dagger} \otimes M_s^{\dagger} \left (
\sum_{j,k=0}^1\Proj{\Phi_{jk}}\right ) M_r
\otimes M_s \\
& = \sum_{r,s = 0}^1 M_r^{\dagger} \otimes M_s^{\dagger} I_{\Hilb
\otimes \Hilb} M_r \otimes M_s \\
& = \sum_{r,s = 0}^1 M_r^{\dagger} \otimes M_s^{\dagger} M_r
\otimes M_s \\
& = \sum_{r,s = 0}^1 M_r^{\dagger}M_r \otimes M_s^{\dagger} M_s \\
& = \left ( \sum_{r=0}^1 M_r^{\dagger}M_r \right ) \otimes \left
( \sum_{s=0}^1 M_s^{\dagger}M_s \right ) \\
& = I_{\Hilb} \otimes I_{\Hilb} \\
& = I_{\Hilb \otimes \Hilb}.
\end{align}

Finally, the $E_{jk}$'s satisfy Eq.~\eqref{eq:Main:anti} by virtue
of Eq.~\eqref{eq:Main:incip} and the cyclic property of the trace.
For example, $\Tr{E_{00} \Proj{\psi_0} \otimes \Proj{\psi_0}}$ is zero by
the following argument.
\begin{align}
&\Tr{E_{00} \Proj{\psi_0} \otimes \Proj{\psi_0}} \nonumber\\
& = \Tr{\sum_{r,s
= 0}^1 M_r^{\dagger} \otimes M_s^{\dagger} \Proj{\Phi_{00}}
M_r \otimes
M_s \Proj{\psi_0} \otimes \Proj{\psi_0}} \\
& = \Tr{\Proj{\Phi_{00}} \left [ \sum_{r=0}^1 M_r \Proj{\psi}
M_r^{\dagger} \right ] \otimes \left [ \sum_{s=0}^1 M_s
\Proj{\psi_0} M_s^{\dagger} \right ]} \\
& = \Tr{\Proj{\Phi_{00}} \mathcal{E}(\Proj{\psi_0}) \otimes
\mathcal{E}(\Proj{\psi_0})} \\
& = \Tr{\Proj{\Phi_{00}} \Proj{0} \otimes \Proj{0}} = 0.
\end{align}
The other antidistinguishability conditions follow from similar
arguments.
\end{proof}

\begin{proof}[Proof of Theorem~\ref{thm:Main:Pusey--Barrett--Rudolph}]
Consider any pair of pure states $\Proj{\psi_0}, \Proj{\psi_1} \in
\mathcal{P}$ and an ontological model $\Theta = (\Lambda, \Sigma,
\Delta, \Xi)$ of $\mathfrak{F}$ that reproduces the quantum
predictions and satisfies the PIP\@.  If
$\Tr{\Proj{\psi_0}\Proj{\psi_1}} \leq \frac{1}{2}$ then, by
Lemma~\ref{lem:Main:anti}, there exists a two fold product state
fragment $\mathfrak{F}^2$ in which the states
\begin{equation}
\left \{ \Proj{\psi_0} \otimes \Proj{\psi_0}, \Proj{\psi_0}
\otimes \Proj{\psi_1}, \Proj{\psi_1} \otimes \Proj{\psi_0},
\Proj{\psi_1} \otimes \Proj{\psi_1}\right \}
\end{equation}
are antidistinguishable.  Consider an ontological model $\Theta^2 =
(\Lambda \times \Lambda, \Sigma \otimes \Sigma, \Delta^{\prime},
\Xi^{\prime})$ of $\mathfrak{F}^2$ that is compatible with $\Theta$.
Then, by Theorem~\ref{thm:Anti:antiover} for any choice of measures
$\mu_{jk} \in \Delta^{\prime}_{\Proj{\psi_j}\otimes\Proj{\psi_k}}$,
\begin{equation}
L \left ( \left \{ \mu_{jk} \right \}_{j,k \in \{0,1\}} \right
) = 0.
\end{equation}
Then, by Lemma~\ref{lem:Main:prod}, $\Proj{\psi_0}$ and
$\Proj{\psi_1}$ must be ontologically distinct in any ontological
model of $\mathfrak{F}$ that satisfies the PIP and reproduces the
quantum predictions.

It remains to deal with the case where $\frac{1}{2}<\Tr{\Proj{\psi_0}\Proj{\psi_1}}<1$. For this, note that
\begin{equation}
\Tr{\Proj{\psi_0}^{\otimes n} \Proj{\psi_1}^{\otimes n}} =
\Tr{\Proj{\psi_0} \Proj{\psi_1}}^n.
\end{equation}
Because of this, we can always find an $n$ such that
\begin{equation}
\Tr{\Proj{\psi_0}^{\otimes n} \Proj{\psi_1}^{\otimes n}} \leq
\frac{1}{2}.
\end{equation}
Now, consider an $n$-fold product state fragment $\mathfrak{F}^n$.
This must contain the states $\Proj{\psi_0}^{\otimes n}$ and
$\Proj{\psi_1}^{\otimes n}$.  For any ontological model $\Theta =
(\Lambda, \Sigma, \Delta, \Xi)$ of $\mathfrak{F}$ that satisfies the
PIP, we can construct a compatible ontological model $\Theta^n =
(\Lambda^{\times n}, \Sigma^{\otimes n}, \Delta^{\prime},
\Xi^{\prime})$ of $\mathfrak{F}^n$.  This model will also satisfy
the PIP (with respect to taking products of $\mathfrak{F}^n$ with
itself), since a model compatible with $\mathfrak{F}$ is necessarily
also compatible with $\mathfrak{F}^n$.  Therefore, we can apply the
argument used in the $\Tr{\Proj{\psi_0}\Proj{\psi_1}} \leq
\frac{1}{2}$ case to the states $\Proj{\psi_0}^{\otimes n}$ and
$\Proj{\psi_1}^{\otimes n}$ and deduce that $\Proj{\psi_0}^{\otimes
n}$ and $\Proj{\psi_1}^{\otimes n}$ must be ontologically distinct
in $\Theta^n$.  This means that
\begin{equation}
L \left ( \left \{ \mu_{\bm{0}}, \mu_{\bm{1}} \right \}
\right ) = 0,
\end{equation}
for any choice of $\mu_{\bm{0}} \in
\Delta^{\prime}_{\Proj{\psi_0}^{\otimes n}}$ and $\mu_{\bm{1}} \in
\Delta^{\prime}_{\Proj{\psi_1}^{\otimes n}}$, where $\bm{0}$ denotes
the $n$-bit string $(0,0,\ldots,0)$ and $\bm{1}$ denotes the $n$-bit
string $(1,1,\ldots,1)$.  Now note that
\begin{equation}
L \left ( \left \{ \mu_{\bm{0}}, \mu_{\bm{1}} \right \} \right )
\geq L \left ( \left \{ \mu_{\bm{j}}\right \}_{\bm{j} \in
\{0,1\}^n}\right ),
\end{equation}
where $\mu_{\bm{j}} \in \Delta_{\psi_{\bm{j}}}$ and
\begin{equation}
\Proj{\psi_{\bm{j}}} = \Proj{\psi_{j_1}} \otimes
\Proj{\psi_{j_2}} \otimes \cdots \otimes \Proj{\psi_{j_n}}.
\end{equation}
This is because computing $L \left ( \left \{ \mu_{\bm{j}}\right
\}_{\bm{j} \in \{0,1\}^n}\right )$ involves minimizing over a
larger set than computing $ L \left ( \left \{ \mu_{\bm{0}},
\mu_{\bm{1}} \right \} \right )$.  Since $L$ is a positive
function, we therefore have that $L \left ( \left \{
\mu_{\bm{j}}\right \}_{\bm{j} \in \{0,1\}^n}\right ) = 0$ and
hence, by Lemma~\ref{lem:Main:prod}, $\Proj{\psi_0}$ and
$\Proj{\psi_1}$ must be ontologically distinct.
\end{proof}

\subsection{Motivation for the PIP}

\label{MPIP}

The PIP says that ontological models for subsystems should be
compatible with models for product states on the composite system.
Compatibility breaks down into two assumptions: the CPA, which says
that the ontic state space of the composite system should be a
Cartesian product of the ontic states spaces of the subsystems, and
the NCA, which says that product states should be modeled by product
measures.  We consider the motivations for each of these requirements
in turn.

The CPA is a weakening of a kind of locality assumption known as
\emph{Einstein separability}.  The terminology is due to Howard
\cite{Howard1985}, who defines it as the idea that ``two spatially
separated systems possess their own separate real states''.  Einstein
formulated this idea in the context of his arguments for the
incompleteness of quantum theory, and expressed it as follows.

\begin{quotation}
If one asks what is characteristic of the realm of physical ideas
independently of the quantum-theory, then above all the following
attracts our attention: the concepts of physics refer to a real
external world, i.e.\ ideas are posited of things that claim a `real
existence' independent of the perceiving subject (bodies, fields,
etc.), and these ideas are, on the other hand, brought into as
secure a relationship as possible with sense impressions.  Moreover,
it is characteristic of these physical things that they are
conceived of as being arranged in a space-time continuum.  Further,
it appears to be essential for this arrangement of the things
introduced in physics that, at a specific time, these things claim
an existence independent of one another, insofar as these things
`lie in different parts of space'.  Without such an assumption of
the mutually independent existence (the `being-thus') of spatially
distant things, an assumption which originates in everyday thought,
physical thought in the sense familiar to us would not be possible.
Nor does one see how physical laws could be formulated and tested
without such a clean separation.  Field theory has carried out this
principle to the extreme, in that it localizes within infinitely
small (four-dimensional) space-elements the elementary things
existing independently of one another that it takes as basic, as
well as the elementary laws it postulates for them.  --- Albert
Einstein \cite{Einstein1948}.
\end{quotation}

The Pusey--Barrett--Rudolph argument only makes use of product state preparations, so we
can always imagine that the individual systems are prepared very far
away from each other and are only later brought together to perform
the entangled measurement.  If implemented this way, the individual
systems would occupy spatially separated regions at their point of
preparation and so, according to separability, they ought to possess
their own separate real states.  Implicit in this is the idea that
there are no inherently global joint properties of the composite
system that are not determined by the properties of the individual
systems.  In the language of ontological models, this means that ontic
state spaces should compose according to the Cartesian product.

Einstein thought that separability should always hold, regardless of
whether we are preparing product states or entangled states.  In light
of Bell's Theorem, the case for such a general separability assumption
is significantly weakened.  Separability is not actually required to
prove Bell's Theorem \cite{Henson2013}, but, if we are contemplating a
world in which the effects of measurement can be transmitted
instantaneously across space, then it makes sense to also contemplate
a world in which there are inherently global properties as well.
Additionally, since we are in the business of proving $\psi$-ontology
theorems and the quantum state of an entangled system would, if ontic,
be such an inherently global property, we had better not introduce any
assumptions that rule them out.  For this reason, the CPA is
restricted to product state preparations.  It says that separability
should hold, not in general, but only when systems are prepared in
product states.

The motivation for assuming separability for product states is that we
generally think that experiments on separated systems are independent
of one another.  It should not be necessary to gather our system here
on Earth together with one on Mars in order to determine all of the
parameters relevant to our Earth-bound experiment.  Of course, when
performing experiments involving systems on Earth that are correlated
with those on Mars, what happens on Mars is very relevant, but the CPA
only applies to product states.  If we allow that genuinely global
properties may be relevant even to an isolated system then we open up
a Pandora's box.  It could well be necessary to gather together every
system in the universe in order to determine all the parameters that
are relevant for our system here on Earth.

At this point, it is worth noting that the assumptions behind no-go
theorems are often designed to mirror operational features of quantum
theory at the ontological level.  This is perhaps clearest in the case
of contextuality.  Preparation noncontextuality says that if there is
no difference between two preparation procedures in terms of the
observable statistics they predict, i.e.\ they are represented by the
same quantum state, then there should be no difference between them at
the ontological level either, i.e.\ they should be represented by the
same probability measure over the ontic states.  If this assumption
does not hold then an explanatory gap is opened because, if two
preparations are represented by different probability measures, then
one would generally expect to be able to pick up this difference in
the observed statistics.  To prevent this from happening, the way that
measurements are represented by conditional probability distributions
has to be fine-tuned so that the difference between the two
probability measures is washed out by averaging.  This fine-tuning has
no natural explanation within the ontological model, other than that
it is needed in order to reproduce the quantum predictions.
Similarly, on an operational level, local measurements on entangled
states cannot be used to send signals, so it makes sense to demand
locality when modeling them ontologically.  Otherwise, we would have
to explain why the purported nonlocal influences cannot be used to
signal.  Proofs of preparation contextuality and Bell's Theorem
therefore expose explanatory gaps in the ontological models framework.

In this vein, separability can be motivated by the operational
principle known as \emph{local tomography}.  This says that the state
of a joint system, and hence the probabilities it predicts for any
measurement, is completely determined by the statistics of local
measurements made on the subsystems.  For example, the state of a
two-qubit system is completely determined by the statistics of
measurements of the Pauli observables $\sigma_x$, $\sigma_y$, and
$\sigma_z$ made on each qubit individually and their correlations.
There are no parameters of the joint state that require bringing the
two subsystems together and making a joint measurement in order to
determine them.  Therefore, it makes sense to posit an underlying
ontology that does not involve genuinely global properties either.
Without this assumption, we would have to explain why the genuinely
global properties do not lead to observable parameters that can only
be determined by bringing the systems together.

Of course, the whole point of no-go theorems, and the reason they are
surprising, is that they show that such operationally motivated
assumptions are not actually viable at the ontological level, at least
within the usual framework.  They each imply fine-tunings for which we
currently have no good explanation.  My point here is just that
separability bears a family resemblance to the assumptions behind
other no-go theorems, so if you think that preparation
noncontextuality and Bell locality have intuitive appeal then the same
should hold for separability as well.  If we find, as we have, that
some of these assumptions are not actually viable then it makes sense
to explore the consequences of weakened versions of them, which
maintain some of the appeal but are not yet ruled out.  From this
perspective the CPA, as a weakening of separability, is a reasonable
thing to posit.

The NCA can similarly be motivated as mirroring an operational feature
of quantum theory at the ontological level.  When two systems are
prepared in a product state they are completely uncorrelated from each
other.  No quantum measurement will ever reveal any pre-existing
correlation.  Therefore, it makes sense to think that the systems are
uncorrelated at the ontological level as well.  If not, then it is
puzzling that this correlation does not show up in measurement
statistics.

Perhaps the CPA and NCA are not so self-evidently true that they must
never be questioned, but they are at least as solid as the assumptions
that go into other no-go theorems.  Therefore, the Pusey--Barrett--Rudolph Theorem
presents us with a dilemma.  Either we must give up the CPA, the NPA,
or $\psi$-epistemicism, and each choice opens up an explanatory gap.
I think this is at least as interesting as the dilemmas posed by
Bell's Theorem and contextuality.

\subsection{Weakening the assumptions}

\label{WPusey--Barrett--Rudolph}

The assumptions of the Pusey--Barrett--Rudolph Theorem have so far been presented in their
most intuitively accessible form.  However, it is possible to weaken
them somewhat without affecting the conclusion.  This has been pointed
out by Hall \cite{Hall2011} and by Schlosshauer and Fine
\cite{Schlosshauer2012}.  The cost of doing this is that the weakened
assumptions are less clearly motivated by operational properties of
quantum theory.  Nevertheless, it is interesting to identify the
weakest set of assumptions under which the theorem can be proved.

First of all, note that the only quantum predictions used in the Pusey--Barrett--Rudolph
argument involve measurement outcomes that are assigned probability
zero by some quantum state via antidistinguishability.  It would make
no difference to the argument if ontological models were only required
to reproduce the probability zero predictions of quantum theory,
instead of requiring that they reproduce all of the quantum
predictions exactly.  Thus, the requirement of reproducing the quantum
predictions may be replaced by the following.

\begin{definition}
An ontological model $\Theta = (\Lambda, \Sigma, \Delta, \Xi)$ of a
PM fragment $\mathfrak{F} = \langle \Hilb, \mathcal{P}, \mathcal{M}
\rangle$ \emph{reproduces the quantum preclusions} if, for all $\rho
\in \mathcal{P}$, $M \in \mathcal{M}$ and $E \in M$ such that $\Tr{E
\rho} = 0$, each $\mu \in \Delta_{\rho}$ and $\text{Pr} \in \Xi_M$
satisfies
\begin{equation}
\int_{\Lambda} \text{Pr}(E|M,\lambda) d \mu(\lambda) = 0.
\end{equation}
\end{definition}

Secondly, in the Pusey--Barrett--Rudolph argument, the PIP is only used in
Lemma~\ref{lem:Main:prod}, where it licenses the following inference.
Given a pair of quantum states $\rho_0$ and $\rho_1$, if the product
states $\{\rho_{j_1} \otimes \rho_{j_2} \otimes \cdots \otimes
\rho_{j_n}\}_{\bm{j} \in \{0,1\}^n}$ correspond to nonoverlapping
measures $\mu_{\bm{j}}$, i.e. $L \left ( \left \{ \mu_{\bm{j}}\right
\}_{\bm{j} \in \{0,1\}^n} \right ) = 0$, then the measures $\mu_0$
and $\mu_1$ corresponding to $\rho_0$ and $\rho_1$ should also not
overlap, i.e. $L(\mu_0,\mu_1) = 0$.  Therefore, if we want to avoid
using the PIP, then we can simply make this assumption directly.

\begin{definition}
An ontological model $\Theta = (\Lambda, \Sigma, \Delta, \Xi)$ for a
PM fragment $\mathfrak{F} = \langle \Hilb, \mathcal{P}, \mathcal{M}
\rangle $ is \emph{compact} with respect to an ontological model
$\Theta^n = (\Lambda^{\prime}, \Sigma^{\prime}, \Delta^{\prime},
\Xi^{\prime})$ for the $n$-fold product state fragment
$\mathfrak{F}^n = \langle \Hilb^{\otimes n}, \mathcal{P}^{\times n},
\mathcal{M}^{\prime} \rangle$ with factors $\mathfrak{F}$ if,
whenever $L \left ( \left \{ \mu_{\bm{j}} \right \}_{\bm{j} \in
\{0,1\}^n} \right ) = 0$ for all choices of $\mu_{\bm{j}} \in
\Delta^{\prime}_{\rho_{\bm{j}}}$ (where $\rho_{\bm{j}} =
\otimes_{k=1}^n \rho_{j_k}$), then $L(\mu_0,\mu_1) = 0$ for all
choices of $\mu_j \in \Delta_{\rho_j}$.
\end{definition}

\begin{definition}
An ontological model for a fragment $\mathfrak{F}$ that reproduces
the quantum preclusions satisfies \emph{compactness} if it is
compact with respect to an ontological model that reproduces the
quantum preclusions for any $n$-fold product state fragment
$\mathfrak{F}^n$ with factors $\mathfrak{F}$.
\end{definition}

The condition of compactness is originally due to Schlosshauer and
Fine \cite{Schlosshauer2012}, and is a slight generalization of a
condition that Hall called ``compatibility'' \cite{Hall2011}.  Here, I
have presented it in a slightly more rigorous form in order to take
care of measure zero issues and to allow for preparation
contextuality.  In this form, the meaning of compactness may be
somewhat obscure, so it is helpful to consider the special case of a
finite ontic state space.  In this case, when $n=2$, compactness says
is that, if a measure $\mu_0$ corresponding to $\rho_0$ assigns a
nonzero probability to the ontic state $\lambda$ and a measure $\mu_1$
corresponding to $\rho_1$ likewise assigns a nonzero probability to
$\lambda$, then there ought to be probability measures $\mu_{00}$,
$\mu_{01}$, $\mu_{10}$ and $\mu_{11}$ representing the states $\rho_0
\otimes \rho_0$, $\rho_0 \otimes \rho_1$, $\rho_1 \otimes \rho_0$ and
$\rho_1 \otimes \rho_1$ of the composite system that all assign
nonzero probability to some ontic state $\lambda^{\prime}$.  The
advantage of this formulation is that it does not assume that the
ontic state space of the composite system has a Cartesian product
structure because it does not specify how $\lambda^{\prime}$, which is
an ontic state of the composite, is related to $\lambda$, which is an
ontic state of the subsystem.  It also does not rule out the
possibility of correlations between the two systems at the ontological
level.

However, the motivation for assuming compactness is somewhat obscure.
Unlike the PIP, it is not a weakened form of separability, and has no
obvious operational motivation.  It is doubtful that one would come up
with compactness without having the idea of Cartesian products and
product measures in mind in the first place.  It amounts to simply
assuming that Lemma~\ref{lem:Main:prod} is true by fiat.  Without any
examples of natural models that satisfy compactness but not the PIP,
it is not clear why one would make such an assumption.  Life would be
very simple if we always just raised the lemmas needed to prove
theorems into assumptions instead of proving them.  Nevertheless, the
Pusey--Barrett--Rudolph Theorem can be stated in the more general form.

\begin{theorem}[The ``Generalized'' Pusey--Barrett--Rudolph Theorem]
\label{thm:Main:Pusey--Barrett--RudolphGen}
Let $\mathfrak{F} = \langle \Hilb, \mathcal{P}, \mathcal{M} \rangle$
be a PM fragment where $\mathcal{P}$ contains all
pure states on $\Hilb$.  Then, any ontological model of
$\mathfrak{F}$ that reproduces the quantum preclusions and satisfies
compactness is $\psi$-ontic.
\end{theorem}

The proof is just the same as the original Pusey--Barrett--Rudolph Theorem, except for the
removal of Lemma~\ref{lem:Main:prod}.

\subsection{Necessity of the PIP}

\label{NPIP}

The PIP is the most substantive assumption that goes into the
Pusey--Barrett--Rudolph Theorem.  All of the rest of its assumptions
are standard for the ontological models framework and are common to
the vast majority of no-go theorems for realist models of quantum
theory, such as Bell's Theorem and the Kochen--Specker Theorem.
Therefore, any criticism of Pusey--Barrett--Rudolph that is not
directed against the PIP could equally well be directed against these
other results.  Whilst I do not wish to belittle such criticisms, and
will deal with them in \S\ref{OCPusey--Barrett--Rudolph}, it is more
important to deal with objections that are specific to the
Pusey--Barrett--Rudolph Theorem, and hence to the PIP\@.  Before doing
so, we should check that the PIP is a necessary assumption, since if
$\psi$-ontology can be established without the PIP then the question
of whether to accept it is moot.

The Kochen--Specker model shows that the PIP is necessary for the case
of a qubit.  However, the case of a qubit is rather special and the
existence of this model does not rule out the possibility that
ontological models for systems of larger dimension might have to be
$\psi$-ontic.  Lewis, Jennings, Barrett and Rudolph (LJBR) showed that
it is in fact possible to construct a $\psi$-epistemic model for
measurements in orthonormal bases in all dimensions \cite{Lewis2012}.
This was based on an idea due to Rudolph, as discussed in Morris'
master's thesis \cite{Morris2009}.  Subsequently, Aaronson, Bouland,
Chua and Lowther (ABCL) provided an alternative construction based on
the same basic idea \cite{Aaronson2013}.  Of course, none of these
models satisfy the PIP\@.  Whilst the LJBR model is $\psi$-epistemic,
the ABCL construction additionally shows that, for any pair
$\Proj{\psi}$, $\Proj{\phi}$ of nonorthogonal states, a
$\psi$-epistemic model can be constructed in which measures
corresponding to $\Proj{\psi}$ and $\Proj{\phi}$ have nonzero overlap.
Since the main idea is the same, we focus on the ABCL result here.

Both LJBR and ABCL start from the Bell model described in
Example~\ref{exa:EOM:Bell}.  This is a $\psi$-ontic model, but by a
simple modification it can be made $\psi$-epistemic.  Recall that, in
the Bell model, the ontic state space is $\Lambda = \Lambda_1 \times
\Lambda_2$, where $\Lambda_1$ is the set of pure states on
$\mathbb{C}^d$ and $\Lambda_2$ is the unit interval.  A state
$\Proj{\psi}$ is represented by the product measure $\mu = \mu_1
\times \mu_2$, where $\mu_1 = \delta_{\psi}$ is the point
measure at $\Proj{\psi}$ on $\Lambda_1$ and $\mu_2$ is the uniform
measure on $\Lambda_2 = [0,1]$.  For each $[\lambda_1] \in \Lambda_1$,
the response functions are set by dividing up the interval $[0,1]$
into subsets of length $\Tr{\Proj{\phi_j}\Proj{\lambda_1}}$ that give
the $\Proj{\phi_j}$ outcome with certainty.  In
Example~\ref{exa:EOM:Bell} this was done simply by dividing $[0,1]$ up
into consecutive intervals of the form
\begin{equation}
\sum_{k=0}^{j-1} \Tr{\Proj{\phi_k}\Proj{\lambda_1}} \leq \lambda_2
< \sum_{k=0}^j \Tr{\Proj{\phi_k}\Proj{\lambda_1}},
\end{equation}
as illustrated in Fig.~\ref{fig:EOM:Bell}.  However, this division
could be done in a different way, providing the total length of the
subset that gives $\Proj{\phi_j}$ with certainty remains equal to
$\Tr{\Proj{\phi_j}\Proj{\lambda_1}}$.  The first step in converting
the Bell model into a $\psi$-epistemic model is to change this
division.

Following ABCL, let $\Proj{a}$ and $\Proj{b}$ be two nonorthogonal
states that we would like to make ontologically indistinct.  Given a
measurement $M = \left \{ \Proj{\phi_j} \right \}_{j=0}^{d-1}$, let
$\sigma$ be a permutation of $(0,1,\ldots,d-1)$ such that
\begin{align}
&\min \left ( \Tr{\Proj{\phi_{\sigma(0)}}\Proj{a}},
\Tr{\Proj{\phi_{\sigma(0)}}\Proj{b}} \right )\nonumber\\
& \geq \min \left ( \Tr{\Proj{\phi_{\sigma(1)}}\Proj{a}},
\Tr{\Proj{\phi_{\sigma(1)}}\Proj{b}} \right ) \geq \ldots \\
& \geq \min \left ( \Tr{\Proj{\phi_{\sigma(d-1)}}\Proj{a}},
\Tr{\Proj{\phi_{\sigma(d-1)}}\Proj{b}} \right ).
\end{align}
Now, the unit interval is divided up in exactly the same way as
before, except with respect to the permuted indices instead of the
original ordering of outcomes i.e.\ the outcome
$\Proj{\phi_{\sigma(j)}}$ is obtained with certainty if
\begin{equation}
\sum_{k=0}^{j-1} \Tr{\Proj{\phi_{\sigma(k)}}\Proj{\lambda_1}} \leq
\lambda_2 < \sum_{k=0}^j
\Tr{\Proj{\phi_{\sigma(k)}}\Proj{\lambda_1}}.
\end{equation}

So far, the model is still $\psi$-ontic because we have only modified
the conditional probability distributions representing measurements,
so the probability measures representing states still have their point
measure terms.  The next step is to show that there is an $\epsilon >
0$ such that all ontic states in the set
\begin{equation}
\Omega_{a,b} = \{\Proj{a},\Proj{b}\} \times [0,\epsilon],
\end{equation}
always return the $\Proj{\phi_{\sigma(0)}}$ outcome in any
measurement.  Thus, any weight that a probability measure assigns to
this region of the ontic state space can be redistributed in an
arbitrary way without affecting the observed probabilities, and by
doing so the probability measures associated with $\Proj{a}$ and
$\Proj{b}$ can be made to overlap.

To prove this, let $\Ket{a}$, $\Ket{b}$ and $\Ket{\phi_j}$ be vector
representatives of $\Proj{a}$, $\Proj{b}$ and $\Proj{\phi_j}$.  It is
sufficient to show that there exists an $\epsilon > 0$ such that, for
all measurements $M = \left \{ \Proj{\phi_j}\right \}_{j=0}^{d-1}$, there
exists a $j$ such that $\left | \BraKet{\phi_j}{a}\right | > \epsilon$
and $\left | \BraKet{\phi_j}{b}\right | > \epsilon$.  Then, whatever
$\Proj{\phi_{\sigma(0)}}$ is, it must satisfy $\left |
\BraKet{\phi_{\sigma(0)}}{a}\right | > \epsilon$ and $\left |
\BraKet{\phi_{\sigma(0)}}{b}\right | > \epsilon$ because $\sigma$
orders the outcomes in decreasing order of $\min \left ( \left |
\BraKet{\phi_j}{a}\right |, \left | \BraKet{\phi_j}{b}\right |
\right )$.  To show this, note that
\begin{align}
\left | \BraKet{a}{b} \right | & = \left |
\Bra{a} \left ( \sum_{j=0}^{d-1} \KetBra{\phi_j}{\phi_j} \right )
\Ket{b} \right | \\
& \leq \sum_{j=0}^{d-1} \left | \BraKet{a}{\phi_j}\BraKet{\phi_j}{b}
\right | \\
& \leq \sum_{j=0}^{d-1} \left | \BraKet{a}{\phi_j} \right | \left |
\BraKet{\phi_j}{b} \right |,
\end{align}
where we have used the fact that $\sum_{j=0}^{d-1}
\KetBra{\phi_j}{\phi_j}$ is a resolution of the identity and the
triangle inequality.  The largest term in the sum must be larger than
the average, so we have
\begin{equation}
\max_j  \left ( \left | \BraKet{a}{\phi_j} \right | \times \left |
\BraKet{\phi_j}{b} \right |\right ) \geq \frac{\left | \BraKet{a}{b}
\right |}{d},
\end{equation}
but since $\left | \BraKet{a}{\phi_j} \right |$ and $\left |
\BraKet{\phi_j}{b} \right |$ are between $0$ and $1$, this means
\begin{equation}
\max_j  \left ( \min \left ( \left | \BraKet{a}{\phi_j} \right |, \left |
\BraKet{\phi_j}{b} \right | \right ) \right ) \geq \frac{\left |
\BraKet{a}{b} \right |}{d},
\end{equation}
and since $\Proj{a}$ and $\Proj{b}$ are nonorthogonal, $\left |
\BraKet{a}{b} \right |/d > 0$, so setting $\epsilon = \left |
\BraKet{a}{b} \right |/d$ gives the desired result.

Now, we can replace the measure associated with $\Proj{a}$ with
\begin{equation}
\mu(\Omega) = \int_{\Lambda_2} \mu_{\lambda_2}(\Omega_{\lambda_2})
d\mu_2(\lambda_2),
\end{equation}
where as before $\mu_2$ is the uniform measure on $\Lambda_2$ and
$\Omega_{\lambda_2} = \left \{ \Proj{\lambda_1} \in \Lambda_1 \middle
| (\Proj{\lambda_1}, \lambda_2) \in \Omega \right \}$, but now
$\mu_{\lambda_2}$ is a measure on $\Lambda_1$ that depends on
$\lambda_2$ via
\begin{equation}
\mu_{\lambda_2} = \begin{cases} \delta_{a} & \text{if} \,\,
\lambda_2 > \epsilon \\ \epsilon \mu_{\Omega_{a,b}} & \text{if}
\,\, \lambda_2 \leq \epsilon, \end{cases}
\end{equation}
where $\mu_{\Omega_{a,b}}$ is an arbitrary measure on $\Omega_{a,b}$
which, for example, could be the uniform measure.  Similarly for
$\Proj{b}$, with $\delta_{a}$ replaced by $\delta_{b}$.  Since both
measures now agree on $\Omega_{a,b}$ and assign nonzero probability to
it, the model is $\psi$-epistemic.

This model is only $\psi$-epistemic in a fairly weak sense, since only
a single pair of pure states has any overlap.  However, the same idea
can be used as the basis for constructing a model in which every pair
of nonorthogonal pure states has overlap (see \S\ref{Pair} for an
outline of the construction).

The Pusey--Barrett--Rudolph Theorem implies that the ABCL model must necessarily fail to
satisfy the PIP\@.  To see this, we have to show that, given an ABCL
model for $\mathbb{C}^d$, there is some $n$ for which there is no
model for product states on the composite system $\left ( \mathbb{C}^d
\right )^{\otimes n}$ with which it is compatible.  Given that the
ABCL construction can be used to make any single pair of pure states
ontologically indistinct, we can apply it to the states $\Proj{z+}$
and $\Proj{x+}$ in $\mathbb{C}^2$ from Example~\ref{exa:Main:Pusey--Barrett--Rudolph}.
Then, compatibility implies that the model for product states on
$\mathbb{C}^2 \otimes \mathbb{C}^2$ would have to have nonzero
four-way overlap for the states $\Proj{z+} \otimes \Proj{z+}$,
$\Proj{z+} \otimes \Proj{x+}$, $\Proj{x+}\otimes\Proj{z+}$ and
$\Proj{x+}\otimes\Proj{x+}$, since these are represented by products
of the measures representing $\Proj{z+}$ and $\Proj{x+}$ in the ABCL
model, which have nonzero pairwise overlap.  However, nonzero overlap
for these four states is ruled out by their antidistinguishability, so
no such model can exist.  Similarly, whichever pair of quantum states
we choose to make ontologically indistinct, the Pusey--Barrett--Rudolph Theorem rules out
compatibility with models for the composite system for some $n$.

\subsection{Criticism of the PIP}

\label{CPIP}

Since the PIP is necessary for proving the Pusey--Barrett--Rudolph Theorem, it is
important to settle the question of whether it is reasonable.  In this
section, two criticisms of the PIP are discussed, directed at the CPA
and NCA respectively.  In \S\ref{CCPA}, a weakening of the PIP due to
Emerson, Serbin, Sutherland and Veitch (ESSV) \cite{Emerson2013} is
discussed and we show that, under this weakened assumption,
$\psi$-epistemic models can always be constructed.  This is still a
useful criticism, as it highlights the fact that giving up the CPA
need not lead to global parameters becoming relevant to local
measurements, so the spirit of separability can perhaps be preserved.
However, the weakened PIP is not really a viable constraint on how
systems should compose, as the construction shows that it can be
satisfied trivially.  \S\ref{CNCA} discusses a criticism due to Hall
\cite{Hall2011} that the NCA does not follow from causality
considerations of the type considered in Bell's Theorem.  Whilst this
argument is formally correct, I think it misses the point as the NCA
is not motivated by causality in the first place.

\subsubsection{Criticism of the CPA}

\label{CCPA}

Emerson, Serbin, Sutherland and Veitch (ESSV) \cite{Emerson2013} have
proposed a weakening of the PIP, aimed at criticizing the CPA, and
have shown how a $\psi$-epistemic model satisfying their weakened PIP
can reproduce the predictions of Example~\ref{exa:Main:Pusey--Barrett--Rudolph}.  Before
getting into the details, it is worth examining what we should expect
of a weakened PIP assumption.

The PIP is a constraint on how subsystems should be composed at the
ontological level.  The CPA encodes the idea that, when preparing two
independent systems in a product state, there should be no genuinely
global properties that are not reducible to local properties of the
individual systems.  It is motivated by the idea that, if this
condition were violated, then it might be necessary to bring the two
independent systems together in order to determine all of the
parameters that are relevant even for just local experiments on one of
the systems.  Any proposal for weakening the PIP by introducing global
properties should therefore be evaluated on the following criteria.
\begin{itemize}
\item It ought to demonstrate that the motivation for the CPA is
misguided by showing how the existence of global properties need not
interfere with local experiments.
\item It ought to provide a viable constraint on how subsystems should
be composed at the ontological level.
\end{itemize}

In my view, the proposal of ESSV satisfies the first criterion, but
not the second.  It does a good job of criticizing the motivation for
the CPA, but is not, in itself, a viable constraint on how subsystems
should be composed.  Explanation of what I mean by a viable constraint
on subsystem composition is deferred until after explaining the details
of the ESSV proposal.

The basic idea is as follows.  Recall that, in order to model direct
product fragments, we can always get away with taking the Cartesian
product of ontic state spaces and using product measures to represent
states.  Therefore, the only time global properties would necessarily
have to play a role when modeling product state fragments is when a
joint measurement is made, e.g.\ a measurement in an entangled basis.
Such measurements already involve an interaction between the two
systems at the operational level, so invoking global properties to
explain such experiments does not seem to open up a huge gap between
the structure of the ontological model and the operational structure
of the experiment.  The argument that it might be necessary to bring a
system on Earth together with a system on Mars in order to
characterize all the parameters relevant to experiments on Earth does
not hold water if the genuinely global properties do not have any
effect on the outcomes of local measurements.  It would still be
possible to work with separate systems completely independently of
each another, in blissful ignorance of the global properties, until we
decide to do an experiment that necessarily involves bringing the
systems together.  With this in mind, ESSV's weakened PIP runs as
follows.

\begin{definition}
A pair of ontological models $\Theta_A = (\Lambda_A, \Sigma_A,
\Delta_A, \Xi_A)$, $\Theta_B = (\Lambda_B, \Sigma_B, \Delta_B, \Xi_B)$
for fragments $\mathfrak{F}_A = \langle \Hilb[A], \mathcal{P}_A,
\mathcal{M}_A \rangle$ and $\mathfrak{F}_B = \langle \Hilb[B],
\mathcal{P}_B, \mathcal{M}_B \rangle$ are \emph{weakly compatible}
with an ontological model $\Theta_{AB} = (\Lambda_{AB}, \Sigma_{AB},
\Delta_{AB}, \Xi_{AB})$ for a product state fragment
$\mathfrak{F}_{AB} = \langle \Hilb[A] \otimes \Hilb[B],
\mathcal{P}_A \times \mathcal{P}_B, \mathcal{M}_{AB} \rangle$ if
$\Theta_{AB}$ satisfies
\begin{itemize}
\item The Weak Cartesian Product Assumption (WCPA):
\begin{align}
\Lambda_{AB} & = \Lambda_A \times \Lambda_B \times
\Lambda_{\text{NL}}, & \Sigma_{AB} & = \Sigma_A \otimes \Sigma_B \otimes
\Sigma_{\text{NL}},
\end{align}
where $\Lambda_{\text{NL}}$ represents some global degrees of
freedom not reducible to properties of system $A$ and system $B$
alone, and $\Sigma_{\text{NL}}$ is a $\sigma$-algebra over
$\Lambda_{\text{NL}}$.
\item The Weak No-Correlation Assumption (WNCA): $\forall \rho_A \in
\mathcal{P}_A, \rho_B \in \mathcal{P}_B$, all $\mu_{AB} \in
\Delta_{AB}[\rho_A \otimes \rho_B]$ satisfy
\begin{equation}
\mu_{AB}(\Omega_{AB} \times \Lambda_{\text{NL}}) = \mu_A
\times \mu_B(\Omega_{AB}),
\end{equation}
for all $\Omega_{AB} \in \Sigma_A \otimes \Sigma_B$ and for some
$\mu_A \in \Delta_A(\rho_A)$ and $\mu_B \in \Delta_B(\rho_B)$,
i.e.\ the marginal measure on $\Lambda_A \times \Lambda_B$
satisfies the NCA.
\end{itemize}
\end{definition}

In general, the ontic state space for a composite system need not
break down neatly into a product of local ontic state spaces and
global properties.  The WCPA therefore seems like a relatively mild
generalization of the CPA\@.  The WNCA then says that, if we only have
access to $\lambda_A$ and $\lambda_B$, then the probability measures
corresponding to product states look just like product measures, even
though they may in fact be correlated via the third variable
$\lambda_{\text{NL}}$.  Further, not only do they look like product
measures, they look just like the product of measures that the
ontological models $\Theta_A$ and $\Theta_B$ would assign to the two
systems individually.  If we assume that $\Theta_A$ and $\Theta_B$ are
adequate for modeling $\mathfrak{F}_A$ and $\mathfrak{F}_B$
individually, it follows that local measurements can be modeled by
products of the conditional probability distributions $\text{Pr}_A$
and $\text{Pr}_B$ assigned by these models, so the outcomes of local
measurements need only depend on $\lambda_A$ and $\lambda_B$.
Therefore, under local measurements, this type of model is
indistinguishable from one that satisfies the PIP\@.  It is only under
joint measurements that there is a difference, since the outcomes of
these may depend on $\lambda_{\text{NL}}$.

\begin{definition}
A pair of ontological models $\Theta_A$ and $\Theta_B$ for fragments
$\mathfrak{F}_A$ and $\mathfrak{F}_B$ that reproduce the quantum
predictions satisfy the \emph{Weak Preparation Independence
Postulate (WPIP)} if, for all product state fragments
$\mathfrak{F}_{AB}$ with $\mathfrak{F}_A$ and $\mathfrak{F}_B$ as
factors, there exists an ontological model $\Theta_{AB}$ that is
weakly compatible with $\Theta_A$ and $\Theta_B$ and reproduces the
quantum predictions.
\end{definition}

Before discussing this further, we prove that the WPIP places no
constraints at all on the local models $\Theta_A$ and $\Theta_B$.  The
proof idea is due to Matt Pusey and Terry Rudolph \cite{Pusey2013}.

\begin{theorem}
\label{thm:CCPA:reprod}
Let $\mathfrak{F}_A = \langle \Hilb[A], \mathcal{P}_A, \mathcal{M}_A
\rangle$ and $\mathfrak{F}_B = \langle \Hilb[B], \mathcal{P}_B,
\mathcal{M}_B \rangle$ be PM fragments and let $\Theta_A =
(\Lambda_A, \Sigma_A, \Delta_A, \Xi_A )$ and $\Theta_B = (\Lambda_B,
\Sigma_B, \Delta_B, \Xi_B)$ be any ontological models of them that
reproduce the quantum predictions.  Then, $\Theta_A$ and $\Theta_B$
satisfy the WPIP.
\end{theorem}
\begin{proof}
Let $\mathfrak{F}_{AB} = \langle \Hilb[A] \otimes \Hilb[B],
\mathcal{P}_A \times \mathcal{P}_B, \mathcal{M}_{AB} \rangle$ be a
product state fragment with factors $\mathfrak{F}_A$ and
$\mathfrak{F}_B$.  We need to show that there is an ontological
model $\Theta_{AB}$ of $\mathfrak{F}_{AB}$ satisfying the following
three requirements.
\begin{itemize}
\item $\Theta_{AB}$ satisfies the WCPA, i.e.\ it is of the form
$\Theta_{AB} = (\Lambda_A \times \Lambda_B \times
\Lambda_{\text{NL}}, \Sigma_A \otimes \Sigma_B \otimes
\Sigma_{NL}, \Delta_{AB}, \Xi_{AB})$.
\item $\Theta_{AB}$ satisfies the WNCA.
\item $\Theta_{AB}$ reproduces the quantum predictions.
\end{itemize}

To satisfy the first requirement, we first need to establish what
the measurable space $(\Lambda_{\text{NL}}, \Sigma_{\text{NL}})$
should be.  For this, consider an arbitrary ontological model
$\Theta_{\text{NL}} = ( \Lambda_{\text{NL}}, \Sigma_{\text{NL}},
\Delta_{\text{NL}}, \Xi_{\text{NL}})$ of $\mathfrak{F}_{AB}$ that
reproduces the quantum predictions.  Such an ontological model
always exists.  For example, we could use the Beltrametti--Bugajski
model extended to mixed states.  Now, just take
$(\Lambda_{\text{NL}}, \Sigma_{\text{NL}})$ to be the ontic state
space of this model.  With this, we have that the model satisfies
the WCPA.

For the second requirement, take the sets $\Delta_{AB}[\rho_A
\otimes \rho_B]$ to consist of all measures of the form $\mu_A
\times \mu_B \times \mu_{\text{NL}}$, where $\mu_A \in
\Delta_A[\rho_A]$, $\mu_B \in \Delta_B[\rho_B]$ and $\mu_{\text{NL}}
\in \Delta_{\text{NL}}[\rho_A \otimes \rho_B]$.  This assignment
satisfies the WNCA.

It remains to show that the model can be made to reproduce the
quantum predictions.  There are two cases to consider: local
measurements for which $(M_A,M_B) \in \mathcal{M}_A \times
\mathcal{M}_B$, and joint measurements $M_{AB} \in \mathcal{M}_{AB}$
that are not in $\mathcal{M}_A \times \mathcal{M}_B$.

For the local measurements, the sets $\Xi_{AB}[M_A,M_B]$ are chosen
to consist of all conditional probability distributions of the form
\begin{align}
&\text{Pr}_{AB}(E,F|M_A,M_B,\lambda_A,\lambda_B,\lambda_{\text{NL}})\nonumber\\
&\quad = \text{Pr}_A(E|M_A,\lambda_A) \text{Pr}_B(F|M_B,\lambda_B),
\end{align}
where $E \in M_A$ and $F \in M_B$, and where $\text{Pr}_A \in
\Xi_A[M_A]$ and $\text{Pr}_B \in \Xi_B[M_B]$.  In other words, the
nonlocal variable is ignored and we just use the conditional
probabilities from the two local ontological models.  This
reproduces the quantum predictions by virtue of the fact that
$\Theta_A$ and $\Theta_B$ do.

For the nonlocal measurements, the sets $\Xi_{AB}[M_{AB}]$ are
chosen to consist of all conditional probability distributions of
the form
\begin{equation}
\text{Pr}_{AB}(E|M_{AB},\lambda_A,\lambda_B,\lambda_{\text{NL}})
= \text{Pr}_{\text{NL}}(E|M_{AB},\lambda_{\text{NL}}),
\end{equation}
where $E \in M_{AB}$, and where $\text{Pr}_{\text{NL}} \in
\Xi_{\text{NL}}[M_{AB}]$.  This reproduces the quantum predictions
by virtue of the fact that $\Theta_{\text{NL}}$ does.
\end{proof}

If the PIP is replaced with the WPIP, the implications for
$\psi$-ontology are as follows.  Consider the situation in which the
factors $\mathfrak{F}_A$ and $\mathfrak{F}_B$ consist of all pure
states and measurements in orthonormal bases on a finite dimensional
Hilbert space of dimension $d$.  The ABCL construction shows that
$\psi$-epistemic models always exist for these factors.  Since any
pair of ontological models satisfy the WPIP, ABCL models do in
particular, so the WPIP cannot be used to prove a $\psi$-ontology
theorem.  One might object that we should demand that the weakly
compatible model of $\mathfrak{F}_{AB}$ should also be
$\psi$-epistemic, although that was not part of the definition of the
WPIP\@.  However, if we restrict attention to fragments
$\mathfrak{F}_{AB}$ that consist of measurements in orthonormal bases
then this can be achieved by using an ABCL model for dimension $d^2$
as $\Theta_{\text{NL}}$.  Since we can make an arbitrary pair of
states ontologically indistinct in an ABCL model, if $\Proj{\psi_0}_A$
and $\Proj{\psi_1}_A$ are ontologically indistinct in $\Theta_A$ and
$\Proj{\phi_0}_B$ and $\Proj{\phi_1}_B$ are ontologically indistinct
in $\Theta_B$ then we can choose $\Proj{\psi_0}_A \otimes
\Proj{\phi_0}_B$ and $\Proj{\psi_1}_A \otimes \Proj{\phi_1}_B$ to be
ontologically indistinct in $\Theta_{\text{NL}}$ and then they will
remain so in $\Theta_{AB}$.  One might make the further objection that
all four states $\Proj{\psi_0}_A \otimes \Proj{\phi_0}_B$,
$\Proj{\psi_0}_A \otimes \Proj{\phi_1}_B$, $\Proj{\psi_1}_A \otimes
\Proj{\phi_0}_B$ and $\Proj{\psi_1}_A \otimes \Proj{\phi_1}_B$ ought
to be pairwise ontologically indistinct, as they would be in the model
of the direct product fragment discussed in \S\ref{PIP}, but this can
also be arranged by using the extension of the ABCL model to be
discussed in \S\ref{Pair}, in which every pair of pure states can be
made ontologically indistinct.  Finally, it should be evident that the
construction used in proving Theorem~\ref{thm:CCPA:reprod} can be
iterated to larger numbers of subsystems.  We simply introduce a new
$\Theta_{\text{NL}}$ for the joint system every time a new subsystem
is added.

However, the proof of Theorem~\ref{thm:CCPA:reprod} also illustrates
why I do not consider the WPIP to be a viable constraint on subsystem
composition.  All we have done is to load an ontological model
$\Theta_{\text{NL}}$ for the joint system, that could potentially
reproduce the quantum predictions for both local and joint
measurements all by itself, onto the nonlocal variable
$\lambda_{\text{NL}}$.  When we make a local measurement, this model
is completely ignored and we use the local models $\Theta_A$ and
$\Theta_B$ instead, but when we make a joint measurement we do the
opposite, using the nonlocal model $\Theta_{\text{NL}}$ and completely
ignoring $\Theta_A$ and $\Theta_B$.  This is fully compatible with the
WPIP, but it is not clear why this model is better than just using
$\Theta_{\text{NL}}$ on its own in the first place.

I accept the criticism that dropping the CPA need not necessarily give
rise to global parameters becoming relevant to local experiments.
However, arguably, the form of the local ontological models should
place some constraints on what the model for the joint system should
look like.  The PIP, in its original form, does this, but the WPIP
does not, since it is compatible with just assigning an ontological
model to the nonlocal variable that is completely independent of the
local ontological models.  In response to this, I think one should
either drop the PIP entirely, or come up with something of
intermediate strength that places nontrivial constraints on the
ontological model for the joint system.  Without that, giving examples
of $\psi$-epistemic models that satisfy the WPIP does not, in my
opinion, strengthen the case against $\psi$-ontology.

Of course, ESSV do not have in mind the trivial sort of model
constructed here.  In their model of Example~\ref{exa:Main:Pusey--Barrett--Rudolph}, the
nonlocal variable does not contain enough information to reproduce the
quantum predictions all on its own, and the probability measures over
$\Lambda_A \times \Lambda_B \times \Lambda_{\text{NL}}$ are
nontrivially correlated.  The point is that ESSV have not done enough
to articulate the way in which their model differs from the trivial
construction presented here.  On this point, note that both the ESSV
model and the construction given here share the feature that the
amount of overlap between states is lower than it would be in the
model of the direct product fragment discussed in \S\ref{PIP}, due to
the probability measure on the extra factor $\Lambda_{\text{NL}}$.
One might therefore impose the additional requirement that the
reduction in overlap introduced by the nonlocal variable ought to be
as small as possible, which could potentially rule out the
construction presented here.  One might also demand that there should
be no redundancy in the model, in the sense that it should not be
possible to reproduce the quantum predictions for local measurements
using the nonlocal variable on its own.  Further work along these
lines is needed to determine whether there is a weakened version of
the PIP that can serve as a reasonable constraint on subsystem
composition.

\subsubsection{Criticism of the NCA}

\label{CNCA}

The PIP has also been criticized on other grounds by Hall
\cite{Hall2011} (this criticism is also mentioned by Schlosshauer and
Fine \cite{Schlosshauer2012}).  Hall objects to the NCA on the grounds
that, even if they are space-like separated, the events corresponding
to preparing two systems have a common past, so their lack of
correlation cannot be derived from causality by the same sort of
reasoning that Bell used to motivate his locality condition.  More
explicitly, suppose we have a model that does satisfy the CPA, so that
we can associate separate ontic states $\lambda_A$ and $\lambda_B$ to
Alice and Bob's systems.  Suppose that the intersection of the past
lightcones of Alice and Bob's preparation events also has some
physical properties described its own ontic state
$\lambda_{\text{past}}$.  For concreteness, suppose that all the ontic
state spaces are finite, i.e.\ the ontic state spaces of Alice's
system, Bob's system, and their common past are all finite.  The
argument does not depend on this, but it makes things conceptually
simpler.

Suppose that Alice decides to prepare the quantum state
$\Proj{\psi}_A$ and Bob decides to prepare the quantum state
$\Proj{\phi}_B$.  In general, the resulting $\lambda_A$ and
$\lambda_B$ might depend on both Alice and Bob's choices of
preparation and the ontic state $\lambda_{\text{past}}$ of their
common past, so the preparation procedure would be specified by
conditional probabilities
$\text{Pr}(\lambda_A,\lambda_B|\Proj{\psi}_A,\Proj{\phi}_B,\lambda_{\text{past}})$.
Given that the preparation procedures might occur at space-like
separation, it is not unreasonable to impose a factorization of
probabilities akin to Bell locality, which would read
\begin{align}
&\text{Pr}(\lambda_A,\lambda_B|\Proj{\psi}_A,\Proj{\phi}_B,\lambda_{\text{past}})\nonumber\\
&= \text{Pr}(\lambda_A|\Proj{\psi}_A,\lambda_{\text{past}})
\text{Pr}(\lambda_B|\Proj{\phi}_B, \lambda_{\text{past}}),
\end{align}
However, this is not enough to entail the NCA because the measure
$\mu$ corresponding to the preparation of $\Proj{\psi}_A \otimes
\Proj{\phi}_B$ would then be,
\begin{align}
&\mu(\lambda_A,\lambda_B) = \sum_{\lambda_{\text{past}}}
\text{Pr}(\lambda_A|\Proj{\psi}_A, \lambda_{\text{past}})\nonumber\\
&\quad\times\text{Pr}(\lambda_B|\Proj{\phi}_B, \lambda_{\text{past}})
\text{Pr}(\lambda_{\text{past}}),
\end{align}
where $\text{Pr}(\lambda_{\text{past}})$ is the prior distribution
over the variables in the common past.  This can induce correlations
between $\lambda_A$ and $\lambda_B$ due to their common dependence on
$\lambda_{\text{past}}$.

This argument is correct, but all it shows is that the NCA cannot be
regarded as a causality assumption akin to Bell locality.  It does not
imply that the NCA is necessarily unreasonable.  To infer that, one
would have to believe that the \emph{only} reasonable type of
assumption to make in a no-go theorem is one that follows from
Bell-type locality.  In contrast, \S\ref{MPIP} discussed a more
general strategy for positing reasonable assumptions, which is to look
at the operational structure of quantum theory and try to impose
similar structure at the ontological level.  In the case of the NCA,
we noted that a product state displays no correlations in any quantum
measurement, so it makes sense to posit that no such correlations
exist at the ontological level either.  It is this, and not Bell
locality, that is the motivation for the NCA\@.  Dropping the NCA in
light of the Pusey--Barrett--Rudolph Theorem is not an obviously crazy thing to do, but it
is the Pusey--Barrett--Rudolph Theorem itself rather than causality arguments that should
motivate this move.  The same reasoning applies to all other no-go
theorems that are not motivated by locality.  For example, I think it
is fairly clear that proofs of preparation contextuality are a good
reason for dropping the assumption of preparation noncontextuality,
but the fact that the latter cannot be derived from Bell locality is a
much less compelling reason for doing so.

\subsection{Other Criticisms of the Pusey--Barrett--Rudolph Theorem}

\label{OCPusey--Barrett--Rudolph}

Apart from criticism of the PIP, several other criticisms of the
Pusey--Barrett--Rudolph Theorem have been raised.
\S\ref{Crit:Response} discusses a criticism due to Drezet
\cite{Drezet2012, Drezet2012a} and Schlosshauer and Fine
\cite{Schlosshauer2012} based on the idea that the conditional
probability distributions representing measurements should depend on
the quantum state in addition to the ontic state.  I argue that this
criticism is simply a misunderstanding of what is meant by the term
``ontic state'' in the ontological models framework.  I then discuss
two further criticisms due to Schlosshauer and Fine
\cite{Schlosshauer2012}, the second of which has also been made by
Dutta et.\ al.\ \cite{Dutta2014}.  The first criticism, discussed in
\S\ref{Crit:SF1}, is a claim that the $\psi$-ontic/$\psi$-epistemic
distinction is merely conventional because one kind of model can be
converted into a structurally equivalent model of the other kind.  The
second criticism, discussed in \S\ref{Crit:SF2}, is based on the idea
that modeling detector inefficiencies offers a way out of the dilemma
posed by the Pusey--Barrett--Rudolph Theorem.  Whilst this is strictly
speaking true, it only concerns the practical implementation of
Pusey--Barrett--Rudolph-type experiments and has no impact on the
Pusey--Barrett--Rudolph Theorem as a structural result about quantum
theory itself, assuming ideal experiments.  If one accepts the PIP,
error analysis can be used to constrain this escape route in more
practical experiments.  Finally, in \S\ref{Crit:Copenhagen}, I discuss
criticisms that reject some aspect of the ontological models framework
outright.  These criticisms are mostly based on a neo-Copenhagen point
of view, and thus are easy to deal with as the Pusey--Barrett--Rudolph
Theorem was never intended to rule out such approaches.

Since the criticisms in this section are directed against the
ontological models framework in general, they are not specific to Pusey--Barrett--Rudolph
but could also be directed against Hardy's Theorem, the Colbeck--Renner
Theorem and the results discussed in \hyperref[Beyond]{Part III}.  I discuss
them here because they were made as responses to the Pusey--Barrett--Rudolph Theorem and
not these other results, but this is simply because the Pusey--Barrett--Rudolph Theorem
was the first and is still the most prominent $\psi$-ontology theorem.

\subsubsection{The conditional probabilities for measurements should
depend on the quantum state}

\label{Crit:Response}

Several authors have pointed out that, in the ontological models
framework, the conditional probability distribution
$\text{Pr}(E|M,\lambda)$ representing a measurement is assumed to be
independent of the quantum state $\rho$ that is prepared
\cite{Drezet2012, Drezet2012a, Schlosshauer2012}.  If a dependence
$\text{Pr}(E|\rho,M,\lambda)$ is allowed then the theorem can be
trivially evaded.  The conditional probabilities can simply be made
independent of the ontic state and and can just return the quantum
probabilities by setting $\text{Pr}(E|\rho,M,\lambda) = \Tr{E \rho}$,
since we then have, for any probability measure $\mu$,
\begin{align}
\int_{\Lambda} \text{Pr}(E|\rho,M,\lambda) d\mu(\lambda) & =
\Tr{E\rho} \int_{\Lambda} d\mu(\lambda) \\
& = \Tr{E\rho}.
\end{align}
The ontic state space and probability measures can then be anything at
all, so the model can trivially be made to satisfy the PIP\@.  For
example, $\mathbb{C}^d$ could be associated with the ontic state space
$\Lambda_d = \{1,2,\ldots,d\}$ and $\mathbb{C}^d \otimes
\mathbb{C}^{d^{\prime}}$ with the Cartesian product $\Lambda_d \times
\Lambda_{d^{\prime}}$, which is isomorphic to $\Lambda_{dd^{\prime}}$.
Then, the uniform measure can be used to represent all quantum states,
which makes the model trivially $\psi$-epistemic.

Of course, proponents of this view do not have this sort of model in
mind as a realistic candidate for describing quantum theory.  They
think the conditional probabilities should depend on the ontic states
in some way in addition to the quantum state.  The above model is just
intended to show how trivial the $\psi$-ontic/$\psi$-epistemic
distinction becomes when the conditional probabilities are allowed to
depend on the quantum state.

I think there are two intuitions behind this sort of objection.  The
first is based on elementary probability theory and the second on a
misunderstanding of how conventional hidden variable theories, such as
de Broglie--Bohm theory, are meant to fit into the ontological models
framework.  First of all, putting aside everything we know about the
ontological models framework for the moment, suppose we are interested
in some PM fragment $\mathfrak{F} = \langle \Hilb,
\mathcal{P}, \mathcal{M} \rangle$ and we have a theory for reproducing
its predictions that involves some sort of additional variable
$\lambda$ that takes values in a set $\Lambda$.  For the purposes of
this argument, assume that $\Lambda$, $\mathcal{P}$ and $\mathcal{M}$
are finite, so that elementary probability theory can be used.  This
restriction can easily be lifted, but dealing with the measure
theoretic complications would obscure the argument.

A quite general way that such a theory could be formulated is in terms
of a joint probability distribution $\text{Pr}(E,\rho,M,\lambda)$ over
all the variables involved, where $\text{Pr}(E,\rho,M,\lambda)$ is
specified for every $\rho \in \mathcal{P}$, $M \in \mathcal{M}$, $E
\in M$, $\lambda \in \Lambda$.  In order to determine whether such a
model reproduces the quantum predictions, we need to determine the
conditional probabilities $\text{Pr}(E|\rho,M)$, i.e.\ the probability
for an outcome given the choice of state and measurement, because
quantum theory tells us this should equal $\Tr{E\rho}$.  This can be
computed as follows.
\begin{align}
\text{Pr}(E|\rho, M) &= \sum_{\lambda} \text{Pr}(E,\lambda| \rho,
M)\nonumber\\
&= \sum_{\lambda} \text{Pr}(E|\rho,M,\lambda)
\text{Pr}(\lambda|\rho,M),
\end{align}
where the conditional probabilities are defined in terms of the joint
probability in the usual way and the second equality follows from the
law of total probability.

On the other hand, in the ontological models framework, the same
quantity would be computed as
\begin{equation}
\text{Pr}(E|\rho,M) = \sum_{\lambda} \text{Pr}(E|M,\lambda)
\mu(\lambda),
\end{equation}
where $\mu \in \Delta_\rho$ and the usual integral has been replaced
by a sum because $\Lambda$ is finite.  Comparing these two
expressions, in an ontological model we have
\begin{align}
\text{Pr}(E|\rho,M,\lambda) & =
\text{Pr}(E|M,\lambda) \label{eq:OCPusey--Barrett--Rudolph:response} \\
\text{Pr}(\lambda|\rho,M) & = \mu(\lambda). \label{eq:OCPusey--Barrett--Rudolph:state}
\end{align}

The left hand side of Eq.~\eqref{eq:OCPusey--Barrett--Rudolph:response} depends on $\rho$,
but the right hand side does not.  Similarly, the left hand side of
Eq.~\eqref{eq:OCPusey--Barrett--Rudolph:state} depends on $M$, but the right hand side
does not.  Despite appearances, both sides of
Eq.~\eqref{eq:OCPusey--Barrett--Rudolph:state} depend on $\rho$ because $\mu$ itself is
$\rho$ dependent given that it must be a member of $\Delta_{\rho}$.
Therefore, it seems that the ontological models framework implicitly
assumes that the following conditional independences hold.
\begin{align}
\text{Pr}(E|\rho,M,\lambda) & = \text{Pr}(E|M,\lambda) \\
\text{Pr}(\lambda|\rho,M) & = \text{Pr}(\lambda|\rho).
\end{align}

The second conditional independence, that $\lambda$ should not depend
on the choice of measurement, was noted in the context of the
Pusey--Barrett--Rudolph Theorem by Hall \cite{Hall2011}, but it is not
all that controversial.  It is usually justified by the idea that the
measurement setting is a ``free choice'' that may be chosen by the
experimenter long after the preparation is completed.  Things can be
set up such that the measurement choice is determined by something
that should be independent of the rest of the experiment, such as a
random number generator.  Theories in which dependence of $\lambda$ on
$M$ nevertheless still holds in the underlying ontology are often
called \emph{superdeterministic}.  There is a minority that seriously
advocates superdeterminism, but this loophole exists in almost all
no-go theorems for ontological models of quantum theory, e.g.\ it
applies to Bell's Theorem as well.  Therefore, if it is taken
seriously as a response to Pusey--Barrett--Rudolph then it must be
taken seriously for these other results as well.  One way in which
dependence on $M$ can happen is if there are retrocausal influences
that travel backwards in time.  Personally, I am not opposed to
developing retrocausal theories as a response to all the quantum no-go
theorems, although I am not convinced it is the right direction
either.  Nonetheless, this criticism is not specific to
Pusey--Barrett--Rudolph.

The first conditional independence, that the measurement outcome
should be independent of the quantum state given the ontic state and
the choice of measurement, is the one that is objected to more
frequently.  However, this is not really a substantive assumption, but
rather it is part of the very meaning of the term ``ontic state''.
The ontic state is supposed to comprise all the properties of the
system that exist in reality.  In addition to its own setting, the
response of the measurement device is only supposed to depend on those
properties of the system that exist in reality, so the ontic state is
the only information about the preparation procedure that it receives.
If there is an additional dependence on the quantum state then that
simply means that we have made an incorrect assertion about what the
ontic state actually is.  It must include all the information that is
required to determine the response of the measurement device.
Therefore, saying that the conditional probability distribution
describing the measurement should depend on the quantum state is
tantamount to saying that the quantum state is part of the ontic
state, and it is very easy to prove that the quantum state is ontic if
you assume that it is ontic from the outset.  In conclusion, I agree
with the critics that the scope of the Pusey--Barrett--Rudolph Theorem is restricted to
the case where this conditional independence holds, but this is part
of the definition of the term ``ontic state'', rather than something
that can be eliminated in order to arrive at a more general notion of
what it means for a model to be $\psi$-epistemic that still conveys
the same meaning.

I think that proponents of this objection have been misled by the way
in which hidden variable theories, and de Broglie--Bohm theory in
particular, have traditionally been presented.  It is often thought
that the aim of a hidden variable theory should be to restore
determinism, and so the problem of developing such a theory is often
phrased in terms of whether quantum theory is ``complete''.  The
terminology ``complete'' suggests taking the idea that elements of the
existing quantum formalism represent reality for granted, and only
asking whether anything else needs to be added to it.  It is obviously
critical to not take this point of view if the reality of the
wavefunction is the very thing under investigation.  The criticism of
Drezet \cite{Drezet2012, Drezet2012a} exemplifies this mistake.  He
suggests that de Broglie--Bohm theory is a counterexample to the Pusey--Barrett--Rudolph
argument.  In de Broglie--Bohm theory, particles have well-defined
positions which evolve deterministically, and the probability
distribution assigned to the particles is given by $|\psi(x)|^2$,
where $\psi(x)$ is the wavefunction.  Drezet's argument is that, if we
view the particle positions as the ontic states of the system, then
their distributions overlap for any pair of nonorthogonal states
because in this case the $|\psi(x)|^2$ distributions overlap.  He
claims that this makes the theory $\psi$-epistemic.

Now, in the conventional understanding of de Broglie--Bohm theory, the
wavefunction is understood to be part of the ontic state in addition
to the particle positions.  It is true that the particle positions are
in some sense more fundamental than the wavefunction, and they are
often called the ``primitive ontology'' \cite{Duerr1992,
Goldstein1998, Goldstein1998a} or the ``local beables'' of the
theory \cite{Bell2004}.  The particle positions are supposed to be the
things in the theory that provide a direct picture of what reality
looks like to us, e.g.\ when we observe the pointer of a measurement
device pointing to a specific value then it is the positions of the
particles that make up the pointer that determine this.  Nevertheless,
the wavefunction is still needed as part of the ontology because it
determines how the particles move via the guidance equation.  The
response of a measurement device to an interaction with a system it is
measuring depends on the wavefunction of the system as well as the
particle positions, so the wavefunction is still part of the ontic
state, even if it is in some sense less primitive than the particle
positions.  Of course, Drezet can get away with having only the
particle positions comprise the ontic state if he allows the
conditional probabilities representing a measurement to depend on the
wavefunction separately.  Indeed, as pointed out by Schlosshauer and
Fine \cite{Schlosshauer2012}, this is often how the probabilities are
written in de Broglie--Bohm theory and other hidden variable theories,
but this is because that framework was addressing the issue of
completeness, which assumes that the wavefunction is real, rather than
the question of whether the quantum state is real in the first place.
As I have argued, it is part of the definition of an ontic state that
it suffices to completely determine how a measurement device will
react to the system, so if your conditional probabilities for
measurement outcomes depend on the wavefunction then the wavefunction
is ontic and there is nothing left to prove.

\subsubsection{The $\psi$-ontic/epistemic distinction is conventional}

\label{Crit:SF1}

Fine and Schlosshauer \cite{Schlosshauer2012} claim that the
distinction between $\psi$-ontic and $\psi$-epistemic models is merely
conventional because a $\psi$-ontic model can be converted into a
``structurally equivalent'' model that is $\psi$-epistemic and vice
versa.  They do not actually define the term ``structurally
equivalent'', so it is perhaps best to look at the procedures they
propose for converting models.

Firstly, starting from a $\psi$-epistemic model it is trivial to
construct a $\psi$-ontic one.  Simply take the new ontic state space
to be the Cartesian product of the existing ontic state space with the
set of pure quantum states.  Then, for each pure state preparation,
take the existing probability measures and form the product with a
point measure on the same pure state in the new component of the ontic
state space.  These are the probability measures of the new model.
Finally, extend the conditional probabilities representing
measurements to the new ontic state space in the most trivial way, by
having them not depend on the new factor at all.  Because of the point
measures, the new model is $\psi$-ontic, but because the conditional
probabilities completely ignore this component of the ontic state
space, the model makes the exact same predictions as the original one.

The construction that Fine and Schlosshauer intend for converting a
$\psi$-ontic model into a $\psi$-epistemic model is less obvious.
They simply refer to the LJBR paper for this \cite{Lewis2012}, which
uses the same sort of construction as the ABCL model, discussed in
\S\ref{NPIP}.  Recall that the idea is to find regions of the ontic
state space associated with two different quantum states that make the
exact same predictions for all measurements.  One can then
redistribute any weight that probability measures assign to this
region such that measures associated with distinct quantum states now
match on this region, and hence they now overlap.  Fine and
Schlosshauer seem to think that such regions can always be found, but
this is not the case.  For example, in the Beltrametti--Bugajski model,
the ontic state space is just the set of pure quantum states and each
quantum state is represented by a different point measure on this
space.  The response functions simply return the quantum
probabilities, so each pair of ontic states makes different
predictions for some quantum measurement.  Therefore, there are no
regions of the ontic state space corresponding to distinct quantum
states that make identical predictions.  Thus, this construction
cannot be used to generate a $\psi$-epistemic model in this case.

To be fair, Fine and Schlosshauer confine their attention to
deterministic models, but they do not prove that appropriate regions
of the ontic state space can always be found even in this case.  It
might be interesting to investigate this, but nevertheless determinism
is deliberately not an assumption of the ontic models framework, so
the fact that there are non-deterministic models to which their
construction does not apply is enough to defeat the criticism.

Further, since Schlosshauer and Fine do not define what they mean by
``structurally equivalent'', it is not clear what their objection is
in the first place.  The only requirement I can infer from their paper
is that ``structurally equivalent'' models should make the same
predictions for all quantum measurements.  If this is really all that
the term means then any two models that reproduce the quantum
predictions would be structurally equivalent.  For example, the
Beltrametti--Bugajski model and de Broglie--Bohm theory would be
structurally equivalent by this criterion.  It is clear, however, that
they are not \emph{explanatorially} equivalent.  Beltrametti--Bugajski
is simply a more precise version of the orthodox interpretation of
quantum theory in which the quantum state, and only the quantum state,
is the state of reality.  This brings with it all the attendant
problems of measurement and the collapse of the wavefunction.  On the
other hand de Broglie--Bohm solves these problems by introducing
additional variables.  Whether or not you think it is plausible as a
fundamental theory, it does not have a measurement problem.
Therefore, reproducing the same set of predictions does not mean that
two theories are equally viable.

Of course, to some extent, Fine and Schlosshauer are just pointing out
that, without the PIP, both $\psi$-epistemic and $\psi$-ontic models
are possible.  However, this means that criticism of the Pusey--Barrett--Rudolph Theorem
should be directed specifically at the PIP, and not at whether the
$\psi$-ontic/$\psi$-epistemic distinction makes sense in the first
place.

\subsubsection{Detector inefficiency}

\label{Crit:SF2}

Both Schlosshauer and Fine \cite{Schlosshauer2012} and Dutta et.\ al.\
\cite{Dutta2014} have raised an objection to the Pusey--Barrett--Rudolph theorem based on
detector inefficiency.  Here is not the place to go over the fine
details of experimental error analysis, but it is worth taking a
little time to explain what a practical test of a $\psi$-ontology
theorem can be expected to show, so that we can deal with this more
easily.

Experimentally testing a mathematical theorem is a bizarre concept at
first sight, but what it really means is to test that the quantum
theoretical predictions hold in the experimental scenarios used to
prove the theorem, at least approximately.  In the case of the Pusey--Barrett--Rudolph
theorem, this means checking that the antidistinguishing measurements
used to prove the theorem really do antidistinguish the sets of states
that they are supposed to.  In practice, the experiments will not
reproduce the quantum predictions exactly.  The error analysis
therefore involves figuring out how the conclusions of the theorem
must be modified when the quantum predictions are only reproduced
approximately.

One possible source of error is detector inefficiency.  In the
theoretical treatment, we generally assume that, when performing a
measurement of a POVM $M = \{E_j\}$, one of the outcomes $E_j$
actually occurs.  In practice, the measuring device may sometimes
simply fail to register any outcome at all.  Thus, for a realistic
analysis, we should add an extra POVM element $E_{\text{null}}$ to
every POVM, representing the possibility of detector failure.  We will
assume that our detectors are reasonably efficient, so that there is
some small $\eta > 0$ such that $\Tr{E_{\text{null}} \rho} \leq \eta$
for every state $\rho$ prepared in the experiment.  Now, if a set of
states $\{\rho_j \}_{j=1}^n$ is antidistinguishable then this means
that there exists a measurement $M = \{E_j\}_{j=1}^n$ with exactly $n$
outcomes such that $\Tr{E_j \rho_j} = 0$.  With detector inefficiency,
we have to deal with the fact that there are $n+1$ outcomes in the
measurement that we actually perform and that all we can say about the
extra $E_{\text{null}}$ outcome is that $\Tr{E_{\text{null}} \rho_j}
\leq \eta$ for every $\rho_j$.

Of course, there are other sources of error as well.  Even if the
detectors do register an outcome, they may not fire with the exact
probabilities predicted by quantum theory due to environmental noise
and our inability to control the experimental apparatus with absolute
precision.  This sort of error was dealt with in the error analysis in
the original Pusey--Barrett--Rudolph paper \cite{Pusey2012}.  There is also a statistical
error arising from the fact that the experiment is only repeated a
finite number of times, so we are inferring probabilities from a
finite sample.  This can be dealt with by the standard techniques of
statistics.  In any case, the objection to the Pusey--Barrett--Rudolph theorem is based on
detector inefficiency, so we do not need to consider these other
sources of error here.

Given an antidistinguishable set of states $\{\rho_j\}_{j=1}^n$,
detector inefficiency has implications for what we can say about the
overlaps of the probability measures $\mu_j \in \Delta_{\rho_j}$.
Specifically, assuming the ontological model reproduces these
predictions, we now have
\begin{align}
\int_{\Lambda} \text{Pr}(E_j|M,\lambda) p_j(\lambda) dm(\lambda) & =
0 \\
\int_{\Lambda} \text{Pr}(E_{\text{null}}|M,\lambda) p_j(\lambda)
dm(\lambda) & \leq \eta,
\end{align}
where $m$ is a measure that dominates $\{\mu_j\}$ and the $p_j$'s are
corresponding densities.

From these equations, we can immediately infer that
\begin{align}
\int_{\Lambda} \text{Pr}(E_j|M,\lambda) \min_k \left [ p_k(\lambda)
\right ] dm(\lambda) & = 0 \\
\int_{\Lambda} \text{Pr}(E_{\text{null}}|M,\lambda) \min_k \left [
p_k(\lambda) \right ] dm(\lambda) & \leq \eta.
\end{align}
and since $\sum_j \text{Pr}(E_j|M,\lambda) +
\text{Pr}(E_{\text{null}}|M,\lambda) = 1$, we have
\begin{equation}
\int_{\Lambda} \min_k \left [ p_k(\lambda) \right ] dm(\lambda) \leq
\eta,
\end{equation}
or, in other words $L \left ( \left \{ \mu_j \right \}_{j=1}^n \right
) \leq \eta$.

The reason why $L \left ( \left \{ \mu_j \right \}_{j=1}^n \right )$
is no longer zero is fairly straightforward.  Given that there are
detector inefficiencies, we cannot rule out the possibility that there
are ontic states $\lambda \in \Lambda$ that always cause the outcome
$E_{\text{null}}$ to occur.  If the detector does not fire with
probability $\eta$, then one of these ontic states might be occupied
with probability $\eta$, regardless of which quantum state was
prepared, and this means that the corresponding probability measures
could have an overlap of up to $\eta$.

Both Fine and Schlosshauer \cite{Schlosshauer2012} and Dutta et.\ al.\
\cite{Dutta2014} point out that that, since $L \left ( \left \{ \mu_j
\right \}_{j=1}^n \right )$ is always nonzero regardless of how
small the detector inefficiency is, you can never get a definitive
confirmation of $\psi$-ontology from a practical experiment.  If the
aim of such an experiment is to definitively rule out the possibility
of a $\psi$-epistemic model, then any nonzero detector inefficiency
immediately makes this impossible.

However, it is important to remember that the definition of a
$\psi$-epistemic model is very permissive.  The $\psi$-epistemic
explanations of quantum phenomena such as the indistinguishability of
quantum states require much more than that the overlap should be
nonzero.  For example, the states $\Proj{0}$ and $\Proj{+}$ of
Example~\ref{exa:Main:Pusey--Barrett--Rudolph} can only be distinguished with probability
$3/4$ in quantum theory.  If the corresponding probability measures
have an overlap that is almost, but not exactly, zero, then there is
not enough overlap to explain why they cannot be distinguished much
better than this.  For this reason, we do not need to definitively
rule out all $\psi$-epistemic models to make $\psi$-epistemic
explanations seem implausible.  It is enough to show that the overlaps
are small compared to the indistinguishability of the quantum states,
and experiments should therefore be used to provide upper bounds on
these overlaps, rather than a definitive yes/no result.

I think both Fine and Schlosshauer \cite{Schlosshauer2012} and Dutta
et.\ al.\ \cite{Dutta2014} have been mislead by the way that detector
inefficiencies are dealt with in experimental tests of Bell's theorem.
In that case, one is looking for a definitive yes/no test of whether a
model satisfying Bell's locality condition can account for the
experimental probabilities and, in that case, there is indeed a sharp
detection efficiency above which such local models can be definitively
ruled out \cite{Eberhard1993}.  However, tests of $\psi$-ontology
theorems simply should not be thought of like this.  Instead of a
yes/no result, they yield a numerical upper bound on the degree of
overlap.  There is no specific value that this must take beyond which
$\psi$-epistemic models are definitively ruled out, but if it is
sufficiently close to zero then they should be regarded as
implausible.

In response to the perceived difficulty with detector inefficiency,
Dutta et.\ al.\ suggest that experimental tests of the Pusey--Barrett--Rudolph theorem
should instead be viewed as definitive yes/no experiments about the
possibility of a \emph{maximally}-$\psi$-epistemic theory, and they
derive the detection efficiencies required to do this.  However, in my
view, it is more informative to simply report the overlaps inferred
from the experiment rather than trying to do a definitive test of
whether they are below some fixed threshold.

Fine and Sclosshauer's version of the objection involves an additional
subtlety that stems from the fact that they prefer the compactness
condition to the PIP\@.  Recall that in the proof of the Pusey--Barrett--Rudolph theorem,
the antidistinguishing measurement acts on multiple copies of the
system prepared in a product state.  The PIP is needed to infer bounds
on the overlap of the measures for a single system from the bounds
obtained from the antidistinguishing joint measurement made on all of
the systems together.  The details can be found in the original Pusey--Barrett--Rudolph
paper \cite{Pusey2012} and I do not want to rehash them here.  The
main point is that this can be done assuming the PIP because the
probability measures of the joint system are products of those for the
individual systems.  With compactness, there is no necessary
connection between the amount of overlap of the probability measures
on the global system and the amount of overlap of the measures for the
individual subsystems.  Compactness just says that if the former is
zero then the latter must also be zero.  Because of this, compactness
does not allow for the derivation of error bounds.  However, we have
already argued that compactness is poorly motivated compared to the
PIP, so this is just another reason for preferring the PIP to
compactness.

Fine and Schlosshauer can also be read as making a prediction that a
certain amount of detector efficiency will necessarily occur when we
attempt to make the measurements needed to test the Pusey--Barrett--Rudolph theorem.  This
amounts to a prediction that the quantum theory can only be verified
to some finite accuracy in these experiments, so it is a prediction
that quantum theory will break down at some point.  Again, assuming
compactness, this only requires a nonzero detection inefficiency, but
under the PIP this inefficiency would have to be implausibly large for
some pairs of states, such as $\Proj{0}$ and $\Proj{+}$, in order to
allow for a viable $\psi$-epistemic model that supports things like
the $\psi$-epistemic explanation of indistinguishability.  In any
case, the idea that quantum theory would be violated in order to
preserve the possibility of $\psi$-epistemic explanations is one of
the more implausible reasons for thinking that quantum theory might
break down that I have heard.

\subsubsection{Rejecting the ontological models framework}

\label{Crit:Copenhagen}

Finally, several authors have objected to the adoption of the
ontological models framework wholesale \cite{Hofmann2011, Motl2011,
Griffiths2012}.  Usually, these objections come from those who adopt
neo-Copenhagen approaches, so I would say that this is just a
misunderstanding of the intended scope of the Pusey--Barrett--Rudolph Theorem, which was
never intended to rule out such interpretations.

For example, Griffiths \cite{Griffiths2012} wonders why anyone would
still be interested in the ontological models framework at all, given
that existing results like Bell's Theorem and the Kochen--Specker
Theorem already make it look quite implausible.  To this I would
respond that, to the extent that explicitly nonlocal and contextual
theories like de Broglie--Bohm theory, spontaneous collapse theories
and modal interpretations are currently taken seriously, the framework
is interesting because it encapsulates them and allows us to study
what other constraints must be satisfied by theories in this category.
Beyond that, the ontological models framework is interesting as a
model of how to simulate quantum systems with classical resources, so
even if the framework is without foundational significance, it is
still relevant to quantum information theory.  Of course, Griffiths
thinks he has a superior approach in the form of decoherent/consistent
histories, but to my mind the best way of understanding this approach
is either as a neo-Copenhagen interpretation or as a way of
formulating the branching structure in a many-worlds interpretation.
Both of these are beyond the intended scope of the ontological models
framework and the Pusey--Barrett--Rudolph Theorem.

Similarly, Hofmann \cite{Hofmann2011} thinks that the conclusion of
Pusey--Barrett--Rudolph can be avoided by allowing exotic probability theories, such as
those involving negative probabilities.  Such theories have a long
history in quantum theory, with the most famous example being the
Wigner function \cite{Wigner1932}.  However, exotic probability
theories are usually couched in neo-Copenhagen or operationalist
terms, i.e.\ it is fine to use whatever mathematical object you like
to represent unobservable things, so long as you always derive an
ordinary positive probability distribution for observable measurement
outcomes.  In the ontological models framework, the ontic state
$\lambda$ is supposed to be something that objectively exists
independently of the experimenter.  Although we may not know the exact
value of $\lambda$ and may only be able to detect coarse-grained
properties of it in our experiments, it is supposed to be knowable in
principle, even if only by a hypothetical super-quantum agent who is
not subject to the same limitations as us.  For such an agent, the
probability assigned to $\lambda$ has to have one of its conventional
meanings, e.g.\ in terms of frequencies, betting odds, etc., since it
is not conceptually different from any other probability.  Therefore,
only requiring probabilities to be positive for things that we can
observe is not good enough to constitute a realist interpretation.
One would need to specify what it means to assign a negative
probability to $\lambda$ for someone who can know $\lambda$ exactly
and, as far as I am aware, no such interpretation of negative
probabilities exists.

\section{Dynamics in ontological models}

\label{Dyn}

The remaining two $\psi$-ontology theorems---Hardy's Theorem and the
Colbeck--Renner Theorem---make use of assumptions about how dynamics
are represented in ontological models.  Two distinct scenarios are
relevant.  Firstly, we extend the notion of a PM fragment to include
the possibility of performing discrete unitary transformations between
preparation and measurement.  As with measurements, we do not want to
assume that this dynamics is deterministic at the ontological level,
so a unitary transformation is represented by a stochastic
transformation on the ontic state space.  This type of dynamics is
discussed in \S\ref{Dyn:With}.  An important property of stochastic
transformations is that they cannot increase the variational distance
between probability measures.  Therefore, if two quantum states are
ontologically indistinct then they remain ontologically indistinct
after applying a unitary transformation.

Using just unitary dynamics, it is possible to prove versions of both
the Hardy and Colbeck--Renner Theorems, but they fail to show that all
nonorthogonal pure states must be ontologically distinct.  Instead,
they show that, if the inner product of a pair of pure states is less
than some quantity that depends on the Hilbert space dimension, then
they must be ontologically distinct.

The trick to extending these results into full blown $\psi$-ontology
theorems for arbitrary dimensions is to consider a different type of
dynamical transformation.  This involves appending an ancillary system
in a fixed state to the system of interest and is discussed in
\S\ref{Dyn:Bet}.  This does not change the inner product of the
original pure states, but it does increase the dimension of the
Hilbert space.  The dimension of the ancillary system can then be
chosen so that the required inequality is satisfied in the larger
Hilbert space.  However, our definition of a fragment assumes that the
Hilbert space is fixed and adding an ancillary system changes the
Hilbert space.  Therefore, in the present framework, appending a
system should be viewed as a dynamical mapping between two
\emph{different} fragments and, at the ontological level, as a mapping
between two \emph{different} ontological models.  If appending an
ancilla is modeled by a stochastic transformation in the same sort of
way as unitary dynamics within a fixed fragment, then it also shares
the property that ontologically indistinct quantum states remain so
after such a transformation.

The property that dynamics preserve ontological distinctness, be
they unitary or the appending of ancillas, is the only assumption
about dynamics needed to prove the Colbeck--Renner Theorem.  Therefore,
we state this as an explicit assumption so that the Colbeck--Renner
argument can be formulated in terms of PM fragments with just this
additional assumption.  On the other hand, Hardy's Theorem involves an
additional assumption about dynamics, to be discussed in \S\ref{OI},
so it requires the extended notion of a fragment that includes unitary
transformations.

\subsection{Unitary dynamics}

\label{Dyn:With}

\begin{definition}
A \emph{prepare-measure-transform (PMT) fragment} of quantum theory
$\mathfrak{F} = \langle \Hilb, \mathcal{P}, \mathcal{M}, \mathcal{T}
\rangle$ consists of a Hilbert space $\Hilb$, a set $\mathcal{P}$ of
density operators on $\Hilb$, a set $\mathcal{M}$ of POVMs on
$\Hilb$, and a set $\mathcal{T}$ of unitary operators on $\Hilb$
that contains the identity.  Additionally, $\mathcal{P}$ is closed
under the action of $\mathcal{T}$, i.e.\ if $\rho \in \mathcal{P}$
and $U \in \mathcal{T}$ then $U \rho U^{\dagger} \in \mathcal{P}$.

The quantum probability of obtaining the outcome $E \in M$ of a
measurement $M \in \mathcal{M}$ when the state $\rho \in
\mathcal{P}$ is prepared and the transformation $U \in \mathcal{T}$
is applied between preparation and measurement is
\begin{equation}
\label{eq:Dyn:qprob}
\text{Prob} \left ( E \middle | \rho, M, U \right ) = \Tr{E U \rho
U^{\dagger}}.
\end{equation}
\end{definition}

The constraint that $\mathcal{P}$ should be closed under the action of
$\mathcal{T}$ is imposed because preparing the state $\rho$ followed
by implementing $U$ provides a method of preparing the state $U \rho
U^{\dagger}$.  The identity is assumed to be in $\mathcal{T}$ because
we want a PMT fragment to be an extension of a PM fragment, so it
should be possible to do nothing between preparation and measurement,
i.e.\ Eq.~\eqref{eq:Dyn:qprob} reduces to the probability rule for a
PM fragment when $U$ is the identity.  In addition, it is natural to
impose further consistency constraints. Firstly, it is common to
assume that if $U,V \in \mathcal{T}$ then $VU \in \mathcal{T}$,
because if we can implement $U$ and $V$ separately then we could apply
them one after the other.  This makes $\mathcal{T}$ a semigroup.
Similarly, it is usual to assume that $\mathcal{M}$ is closed under
the action of $\mathcal{T}$, i.e.\ if $\{E_j\} \in \mathcal{M}$ and $U
\in \mathcal{T}$ then $\{U^{\dagger} E_j U\} \in \mathcal{M}$.  This
is because applying $U$ before measuring $\{E_j\}$ is a method of
measuring $\{ U^{\dagger} E_j U \}$.  These additional constraints are
not imposed here because they are not required for the Hardy or
Colbeck--Renner Theorems.

In order to construct an ontological model for a PMT fragment, states
and measurements are represented on an ontic state space $(\Lambda,
\Sigma)$ as before.  The only novel issue is how to represent
transformations.  As with measurements, we do not want to assume that
$U \in \mathcal{T}$ acts deterministically on the ontic states, so in
general $U$ is represented by a Markov kernel $\gamma$ with source and
target both equal to $(\Lambda,\Sigma)$.  This means that $\gamma$ is
a measurable function that associates, to each $\lambda \in \Lambda$,
a probability measure $\gamma_{\lambda}$ on $(\Lambda,\Sigma)$.  For
$\Omega \in \Sigma$, $\gamma_{\lambda}(\Omega)$ is the conditional
probability that the ontic state will end up in $\Omega$ after the
transformation, given that it started in the ontic state $\lambda$.
For a finite ontic state space
$\gamma_{\lambda}(\{\lambda^{\prime}\})$ is the probability that the
dynamics causes $\lambda$ to make a transition to $\lambda^{\prime}$,
and thus a Markov kernel is just the measure theoretic generalization
of a transition matrix, which is used to model stochastic dynamics for
a system with finite state space.  If the system is assigned
probability measure $\mu$ before the transformation then afterward it
is assigned the measure $\nu$, where
\begin{equation}
\label{eq:Dyn:transform}
\nu(\Omega) = \int_{\Lambda} \gamma_{\lambda}(\Omega) d\mu(\lambda).
\end{equation}

As with states, there is the possibility that different methods of
implementing $U$ might lead to different Markov kernels, in which case
the model is \emph{transformation contextual} \cite{Spekkens2005}.  To
account for this, each $U \in \mathcal{T}$ is associated with a set
$\Gamma_{U}$ of Markov Kernels, rather than just one.

The consistency constraint that $\mathcal{P}$ is closed under
$\mathcal{T}$ also needs to be reflected at the ontological level.
Suppose $\rho, \sigma \in \mathcal{P}$ and $\sigma = U \rho
U^{\dagger}$ for some $U \in \mathcal{T}$.  For every $\mu \in
\Delta_{\rho}$ and $\gamma \in \Gamma_U$, it should be the case that
$\nu \in \Delta_{\sigma}$, where $\nu$ is given by
Eq.~\eqref{eq:Dyn:transform}.  This is because any method of preparing
$\rho$ and applying $U$ is a method of preparing $\sigma$.

\begin{definition}
\label{def:Dyn:OM}
An \emph{ontological model} $\Theta = (\Lambda, \Sigma, \Delta, \Xi,
\Gamma)$ of a PMT fragment $\mathfrak{F} = \langle \Hilb,
\mathcal{P}, \mathcal{M}, \mathcal{T} \rangle$ consists of
\begin{itemize}
\item A measurable space $(\Lambda, \Sigma)$, where $\Lambda$ is
called the \emph{ontic state space}.
\item A function $\Delta$ that maps each quantum state $\rho \in
\mathcal{P}$ to a set of probability measures $\Delta[\rho] =
\Delta_{\rho}$ on $(\Lambda,\Sigma)$.
\item A function $\Xi$ that maps each POVM $M \in \mathcal{M}$ to a
set of conditional probability distributions over $M$, $\Xi[M] =
\Xi_M$, i.e.\ each $\text{Pr} \in \Xi_M$ is a function from $M
\times \Lambda$ to $\mathbb{R}$ that is measurable as a function
of $\lambda \in \Lambda$ and satisfies, for all $\lambda \in
\Lambda$,
$\forall E \in M, \,\text{Pr}(E|M,\lambda) \geq 0$
and $\sum_{E \in M} \text{Pr}(E|M,\lambda) = 1$.
\item A function $\Gamma$ that maps each $U \in \mathcal{T}$ to a
set of Markov kernels $\Gamma(U) = \Gamma_U$, i.e.  $\gamma \in
\Gamma_U$ is a measurable function $\gamma:\lambda \rightarrow
\gamma_{\lambda}$ where $\gamma_{\lambda}$ is a probability
measure on $(\Lambda,\Sigma)$.
\end{itemize}

In addition, for every $\rho, \sigma \in \mathcal{P}$, $U \in
\mathcal{T}$ such that $\sigma = U \rho U^{\dagger}$, for every $\mu
\in \Delta_{\rho}$ and $\gamma \in \Gamma_U$, it must be the case
that $\nu \in \Delta_{\sigma}$, where
\begin{equation}
\nu(\Omega) = \int_{\Lambda} \gamma_{\lambda}(\Omega) d\mu(\lambda).
\end{equation}

The ontological model \emph{reproduces the quantum predictions} if,
for all $\rho \in \mathcal{P}$ and $M \in \mathcal{M}$, each $\mu
\in \Delta_{\rho}$ and $\text{Pr} \in \Xi_M$ satisfies
\begin{equation}
\label{eq:Dyn:Rep2}
\forall E \in M, \,\,\,\, \int_{\Lambda} \text{Pr}(E|M,\lambda) d
\mu(\lambda) = \Tr{E \rho}.
\end{equation}
\end{definition}

Note that we do not have to explicitly impose that, for all $\rho \in
\mathcal{P}$, $U \in \mathcal{T}$, each $\mu \in \Delta_{\rho}$,
$\gamma \in \Gamma_U$ and $\text{Pr} \in \Xi_M$ satisfies
\begin{equation}
\forall E \in M,\int_{\Lambda}\int_{\Lambda} \text{Pr}(E|M,\lambda^{\prime}) d
\gamma_{\lambda}(\lambda^{\prime}) d \mu(\lambda) = \Tr{E U \rho
U^{\dagger}},
\end{equation}
since this is implied by the consistency constraints on states and the
probability densities that represent them.

In fact, as in the Pusey--Barrett--Rudolph Theorem, Hardy's Theorem only depends on the
weaker requirement that the model reproduces the quantum preclusions,
which means that
\begin{equation}
\int_{\Lambda} \text{Pr}(E|M,\lambda) d \mu(\lambda) = 0,
\end{equation}
whenever $\Tr{E \rho} = 0$.

The following standard result is the basis of the most common
assumption about dynamics used to prove $\psi$-ontology theorems.
\begin{theorem}
\label{prop:Dyn:contract}
Let $\mu$ and $\nu$ be probability measures on a measurable space
$(\Lambda,\Sigma)$, let $\gamma$ be a Markov kernel with source
$(\Lambda,\Sigma)$ and target $(\Lambda^{\prime},\Sigma^{\prime})$,
and let $\mu^{\prime}$ and $\nu^{\prime}$ be the measures on
$(\Lambda^{\prime},\Sigma^{\prime})$ resulting from the action of
$\gamma$ on $\mu$ and $\nu$, i.e.\ for $\Omega^{\prime} \in
\Sigma^{\prime}$,
\begin{align}
\mu^{\prime}(\Omega^{\prime}) & = \int_{\Lambda}
\gamma_{\lambda}(\Omega^{\prime})d\mu(\lambda) \\
\nu^{\prime}(\Omega^{\prime}) & = \int_{\Lambda}
\gamma_{\lambda}(\Omega) d\nu(\lambda).
\end{align}
Then, $D(\mu^{\prime},\nu^{\prime}) \leq D(\mu,\nu)$, where $D$ is
the variational distance.
\end{theorem}
\begin{proof}
Let $m$ be a measure that dominates $\mu$ and $\nu$ and let $p$ and
$q$ be corresponding densities.  Then, for all $\Omega^{\prime} \in
\Sigma^{\prime}$
\begin{align}
\left | \mu^{\prime}(\Omega^{\prime}) -
\nu^{\prime}(\Omega^{\prime}) \right | & =  \left |
\int_{\Lambda} \gamma_{\lambda}(\Omega^{\prime}) \left (
p(\lambda) - q (\lambda) \right )
dm(\lambda) \right | \\
& \leq \int_{\Lambda} \left | \gamma_{\lambda}(\Omega^{\prime})
\left ( p(\lambda) - q (\lambda) \right ) \right |dm(\lambda),
\end{align}
where the second line follows from the triangle inequality.
However, since $0 \leq \gamma_{\lambda}(\Omega^{\prime}) \leq 1$, we
have
\begin{align}
&\int_{\Lambda} \left | \gamma_{\lambda}(\Omega^{\prime}) \left (
p(\lambda) - q (\lambda) \right ) \right |dm(\lambda) \leq\nonumber\\
&\quad \int_{\Lambda} \left | p(\lambda) - q(\lambda) \right |
dm(\lambda) = D(\mu,\nu).
\end{align}
Hence, for all $\Omega^{\prime} \in \Sigma^{\prime}$,
\begin{equation}
\left | \mu^{\prime}(\Omega^{\prime}) -
\nu^{\prime}(\Omega^{\prime}) \right | \leq D(\mu,\nu).
\end{equation}
Therefore, this also applies for the supremum of the left hand side,
so $D(\mu^{\prime},\nu^{\prime}) \leq D(\mu,\nu)$.
\end{proof}
\begin{corollary}
\label{cor:Dyn:contract}
Let $\langle \Hilb, \mathcal{P}, \mathcal{M}, \mathcal{T} \rangle$
be a PMT fragment and let $\Theta = (\Lambda, \Sigma, \Delta, \Xi,
\Gamma)$ be an ontological model of it.  If $\rho, \sigma \in
\mathcal{P}$ are ontologically indistinct then $U\rho U^{\dagger}$
and $U \sigma U^{\dagger}$ are also ontologically indistinct for all
$U \in \mathcal{T}$.
\end{corollary}
\begin{proof}
If $\rho$ and $\sigma$ are ontologically indistinct then there exist
$\mu \in \Delta_{\rho}$ and $\nu \in \Delta_{\sigma}$ such that
$D(\mu,\nu) < 1$.  However, for any $\gamma \in \Gamma_U$,
$\mu^{\prime} \in \Delta_{U \rho U^{\dagger}}$ and $\nu^{\prime} \in
\Delta_{U \sigma U^{\dagger}}$, where, for all $\Omega \in \Sigma$,
\begin{align}
\mu^{\prime}(\Omega) & = \int_{\Lambda}
\gamma_{\lambda}(\Omega)d\mu(\lambda) \\
\nu^{\prime}(\Omega) & = \int_{\Lambda} \gamma_{\lambda}(\Omega)
d\nu(\lambda).
\end{align}
Hence, by Theorem~\ref{prop:Dyn:contract},
$D(\mu^{\prime},\nu^{\prime}) \leq D(\mu,\nu) < 1$, so $U \rho
U^{\dagger}$ and $U \sigma U^{\dagger}$ are also ontologically
indistinct.
\end{proof}

If the PMT fragment of interest contains all unitary transformations
on $\Hilb$, then corollary~\ref{cor:Dyn:contract} implies that whether
or not two pure states are ontologically distinct depends only on
their inner product.  This is because, if $\Proj{\psi^{\prime}} =
U\Proj{\psi}U^{\dagger}$ and $\Proj{\phi^{\prime}} = U
\Proj{\phi}U^{\dagger}$, then $\Proj{\psi} =
U^{\dagger}\Proj{\psi^{\prime}}U$ and $\Proj{\phi} = U^{\dagger}
\Proj{\phi^{\prime}} U$, so applying Corollary~\ref{cor:Dyn:contract}
to both $U$ and $U^{\dagger}$ implies that $\Proj{\psi}$ and
$\Proj{\phi}$ are ontologically distinct iff $\Proj{\psi^{\prime}}$
and $\Proj{\phi^{\prime}}$ are.

In fact, Corollary~\ref{cor:Dyn:contract} is the only property of the
way that unitary dynamics is represented that is used in proving the
Colbeck--Renner Theorem.  Therefore, for the purposes of proving that
theorem, all of the detailed considerations about how unitary dynamics
are represented can be replaced by the following assumption.
\begin{definition}
\label{def:Dyn:uind}
Let $\mathfrak{F} = \langle \Hilb, \mathcal{P}, \mathcal{M} \rangle$
be a PM fragment and let $\Theta = (\Lambda, \Sigma, \Delta, \Xi)$
be an ontological model of it.  $\Theta$ \emph{preserves ontological
distinctness} with respect to a unitary operator $U$ if, for every
$\rho, \sigma \in \mathcal{P}$ that are ontologically distinct in
$\Theta$, $U \rho U^{\dagger}$ and $U \sigma U^{\dagger}$ are also
ontologically distinct in $\Theta$ whenever $U \rho U^{\dagger}, U
\sigma U^{\dagger} \in \mathcal{P}$.
\end{definition}
The benefit of this assumption is that, if we can show that
$\Proj{\psi}$ and $\Proj{\phi}$ are ontologically distinct then it
follows that all pairs $\Proj{\psi^{\prime}}$ and
$\Proj{\phi^{\prime}}$ where
$\Tr{\Proj{\phi^{\prime}}\Proj{\psi^{\prime}}} =
\Tr{\Proj{\phi}\Proj{\psi}}$ must also be ontologically distinct, so
to prove $\psi$-ontology we only need to prove that there exists a
pair of ontologically distinct states for every value of the inner
product.

The assumption that unitary dynamics preserves ontological
distinctness stands independently of how $U$ is represented in an
ontological model, but it is worth bearing in mind that the
representation of $U$ as a stochastic transformation is really what
motivates it.

\subsection{Appending ancillas}

\label{Dyn:Bet}

So far, we have considered unitary dynamics on a fixed Hilbert space.
More generally, the Hilbert space of the system may change during the
course of the experiment.  For example, the dimension may be reduced
if part of the system is absorbed into the environment.  For present
purposes, we only need to consider a very particular kind of change,
in which the experimenter appends an additional system in a fixed
quantum state $\tau$ to the system under investigation.  Specifically,
if the system is originally described by a state $\rho_A$ on a Hilbert
space $\Hilb[A]$, then after appending the ancilla it is described by
the state $\rho_A \otimes \tau_B$ on the Hilbert space $\Hilb[A]
\otimes \Hilb[B]$, where $\Hilb[B]$ is the Hilbert space of the
ancillary system.  In order to prove full-blown $\psi$-ontology
theorems for arbitrary dimensional Hilbert spaces, the following
assumption is used.

\begin{definition}
\label{def:Dyn:aind}
Let $\mathfrak{F}_A = \langle \Hilb[A], \mathcal{P}_A, \mathcal{M}_A
\rangle$ and $\mathfrak{F}_B = \langle \Hilb[B], \mathcal{P}_B,
\mathcal{M}_B \rangle$ be PM fragments and let $\mathfrak{F}_{AB} =
\left \langle \Hilb[A] \otimes \Hilb[B], \mathcal{P}_{AB} ,
\mathcal{M}_{AB} \right \rangle$ be a product fragment with
factors $\mathfrak{F}_A$ and $\mathfrak{F}_B$.  Let $\Theta_A =
(\Lambda_A, \Sigma_A, \Delta_A,\Xi_A)$ be an ontological model of
$\mathfrak{F}_A$ and let $\Theta_{AB} = (\Lambda_{AB}, \Sigma_{AB},
\Delta_{AB}, \Xi_{AB})$ be an ontological model of
$\mathfrak{F}_{AB}$.  $\Theta_{AB}$ \emph{preserves ontological
distinctness} with respect to $\Theta_A$ if, whenever $\rho_A
\otimes \tau_B$ and $\sigma_A \otimes \tau_B$ are ontologically
distinct in $\Theta_{AB}$ for some $\tau_B \in \mathcal{P}_B$,
$\rho_A, \sigma_A \in \mathcal{P}_A$ are ontologically distinct in
$\Theta_A$,
\end{definition}

The motivation for this is similar to that for the assumption of
preserving ontological distinctness with respect to unitary
transformations.  The action of appending an ancilla is a type of
dynamics, and so it should be represented by a stochastic
transformation at the ontological level.  The main difference from the
unitary case is that we are dealing with a transformation that changes
the underlying Hilbert space, and our definition of a PM fragment
assumes a fixed Hilbert space.  Therefore, appending an ancilla is a
map between two distinct fragments and so it is represented by a
stochastic map between two different ontological models.
Specifically, if $(\Lambda_A, \Sigma_A)$ is the ontic state space of
$\Theta_A$ and $(\Lambda_{AB},\Sigma_{AB})$ is the ontic state space
of $\Theta_{AB}$, then appending an ancilla would be represented by a
Markov kernel $\gamma$ with source $(\Lambda_A, \Sigma_A)$ and target
$(\Lambda_{AB},\Sigma_{AB})$.  By Theorem~\ref{prop:Dyn:contract},
this cannot increase the variational distance of the measures
representing quantum states and hence, by the same argument used to
prove Corollary~\ref{cor:Dyn:contract}, it preserves ontological
distinctness.

At first, it may seem surprising that appending ancillas should be
modeled at the ontological level in the same sort of way as unitary
dynamics, since they are very different types of operations.  As
additional motivation, note that the most general type of dynamics of
a quantum system is described by a Completely-Positive
Trace-Preserving (CPT) map and the more general claim is that CPT maps
should be represented by stochastic transformations (see
e.g. \cite{Spekkens2005} for a discussion of this).  Both unitary
dynamics and appending an ancilla are examples of CPT maps, so their
representation in terms of stochastic maps are special cases of this
general idea.

\section{Hardy's Theorem}

\label{Hardy}

Hardy has proven a $\psi$-ontology theorem \cite{Hardy2013} based on
an assumption about how dynamics should be represented in an
ontological model known as \emph{ontic indifference}.  This assumption
is rather unnatural for a $\psi$-epistemic theory, but nonetheless
Hardy's Theorem is of interest due to its close connection with the
argument for the reality of the wavefunction based on interference,
which was discussed in \S\ref{Int}.  Hardy's Theorem can be regarded
as the missing step in the inference from ``something must go through
both slits'' to ``that thing must be the wavefunction''.  The reason
for the failure of the argument from interference can therefore be
pinpointed more precisely as the failure of ontic indifference.

\S\ref{OI} describes the ontic indifference assumption and the way it
can be violated in a $\psi$-epistemic theory, as well as explaining
Hardy's motivation for introducing it.  \S\ref{HExa} presents a sketch
of a special case of Hardy's Theorem in terms of a simple Mach-Zehnder
interferometry experiment.  This helps to clarify the relation to the
argument from interference.  The full theorem and its proof are given
in \S\ref{HMain}.

\subsection{Ontic Indifference}

\label{OI}

To understand Hardy's assumption, first consider an ontological model
with a finite ontic state space.  Ontic indifference then says that,
if a pure state $\Proj{\psi}$ is invariant under the action of a
unitary $U$, i.e. $U \Proj{\psi} U^{\dagger} = \Proj{\psi}$, then
there should be a method of implementing $U$ such that every ontic
state that get assigned a nonzero probability by $\Proj{\psi}$ is
left invariant.  As usual, this needs to be modified to accommodate
the general measure-theoretic case.

\begin{definition}
Let $\mathfrak{F} = \langle \Hilb, \mathcal{P}, \mathcal{M},
\mathcal{T} \rangle$ be a PMT fragment and let $\Theta = \left (
\Lambda, \Sigma, \Delta, \Xi, \Gamma \right )$ be an ontological
model of it.  A state $\rho \in \mathcal{P}$ satisfies \emph{ontic
indifference} in $\Theta$ if, for every $U \in \mathcal{T}$ such that
$U \rho U^{\dagger} = \rho$, for every $\mu \in \Delta_{\rho}$ there
exists a $\gamma \in \Gamma_U$ and a set $\Omega \in \Sigma$ such
that $\mu(\Omega) = 1$ and $\gamma_{\lambda}(\Omega^{\prime}) =
\delta_{\lambda}(\Omega^{\prime})$ for all $\Omega^{\prime} \in
\Sigma, \Omega^{\prime} \subseteq \Omega$, where $\delta_{\lambda}$
is the point measure at $\lambda$.

The model $\Theta$ satisfies \emph{ontic indifference} if every pure
state $\Proj{\psi} \in \mathcal{P}$ satisfies ontic indifference in
$\Theta$.  It satisfies \emph{restricted ontic indifference} if there
exists a pure state $\Proj{\psi} \in \mathcal{P}$ that satisfies
ontic indifference in $\Theta$.
\end{definition}

To see why ontic indifference is suspect from a $\psi$-epistemic point
of view, it suffices to consider a model with a finite ontic state
space $\Lambda$.  It is then unclear why ontic indifference should
hold for a quantum state that is represented by a probability measure
with support on more than one ontic state, and there obviously must be
such states in a $\psi$-epistemic model.  For example, suppose that
$\Proj{\psi}$ is represented by the uniform distribution over
$\Lambda$.  Then, any permutation of the ontic states leaves this
distribution invariant and hence could potentially represent the
dynamics of a unitary that leaves $\Proj{\psi}$ invariant without
contradicting the quantum predictions.  More generally, the measures
corresponding to $\Proj{\psi}$ need only be fixed points of the
stochastic transformations representing the unitaries that leave them
invariant, and even this is stronger than required, since the
stochastic transformations could also map between different members of
$\Delta_{\psi}$.

In fact, permutations that leave the epistemic states invariant is
precisely how unitary dynamics are represented in Spekkens' toy
theory.  Consider the state $\Proj{z+}$, which is invariant under the
action of the unitary operator $\sigma_z$.  In Spekkens' theory,
$\Proj{z+}$ is represented by the distribution $\RKet{z+}$ that has
equal support on $(+,+)$ and $(-,-)$ and is zero elsewhere on the
ontic state space.  The transformation $\sigma_z$ can be represented
by the permutation
\begin{align}
(+,+) & \rightarrow (-,-) \\
(+,-) & \rightarrow (-,+) \\
(-,+) & \rightarrow (+,-) \\
(-,-) & \rightarrow (+,+).
\end{align}
This leaves $\RKet{z+}$ invariant but does not satisfy ontic
indifference because it swaps the two states $(+,+)$ and $(-,-)$ in
the support of $\RKet{z+}$.  It is straightforward to check that this
permutation acts appropriately on all the other distributions as well.
For example, $\Proj{x+}$ gets mapped to $\Proj{x-}$ under $\sigma_z$,
and, in the toy theory, $\RKet{x+}$ has equal support on $(+,+)$ and
$(+,-)$ and is zero elsewhere, whereas $\RKet{x-}$ has equal support
on $(-,+)$ and $(-,-)$ and is zero elsewhere.  Since the permutation
maps $(+,+)$ to $(-,-)$ and $(+,-)$ to $(-,+)$, it maps $\RKet{x+}$ to
$\RKet{x-}$ as required.

Given that Spekkens' toy theory is an archetypal example of a
$\psi$-epistemic theory, the fact that it does not satisfy ontic
indifference is evidence that ontic indifference is not a reasonable
assumption for a $\psi$-ontology theorem.  Nevertheless, Hardy does
provide a motivation for it based on locality, which reveals an
interesting connection to the argument from interference.

Hardy's motivation runs as follows.  Consider a single photon which
can be in one of $d$ different spatial modes, labeled $x_0, x_1,
\ldots, x_{d-1}$.  For example, the modes might represent the arms of
an interferometer, as depicted in Fig.~\ref{fig:Hardy:dmach}. The
state in which the photon is in mode $x_j$ is written as $\Proj{x_j}$.
Suppose that, at the ontological level, $\Proj{x_j}$ is to be thought
of as a state in which the photon is literally \emph{in} mode $x_j$ so
it corresponds to a situation in which there is literally nothing
relevant to the behavior of the photon located in any of the other
modes.
\begin{figure}[t!]
\centering
\includegraphics[width=75mm]{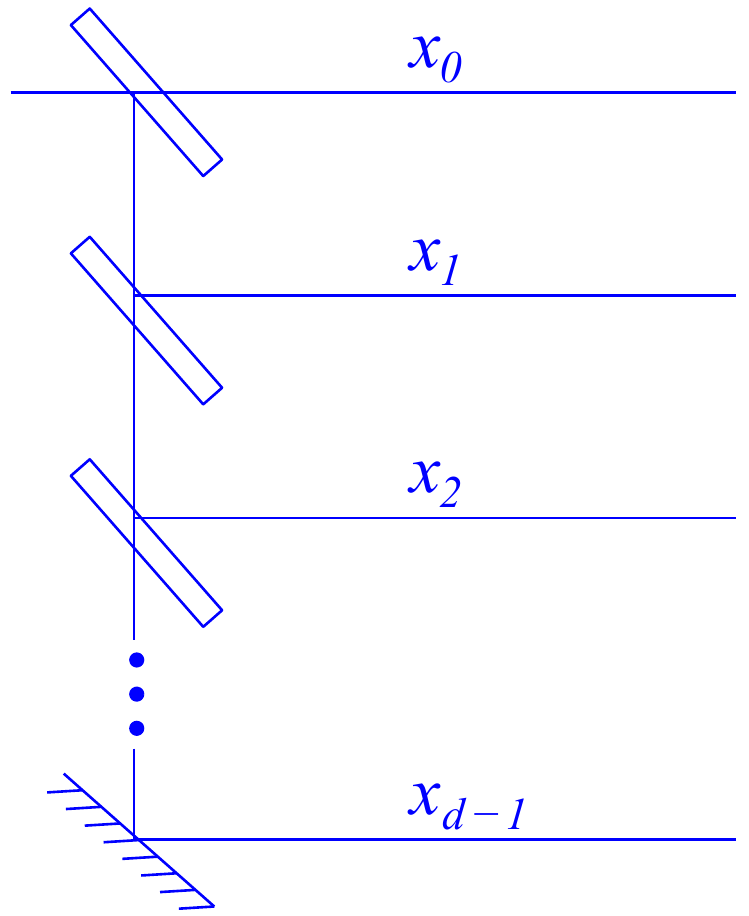}
\caption{\color[HTML]{0000FF}{\label{fig:Hardy:dmach}One way of instantiating $d$ spatial
modes in an interferometer. A single photon is passed through the
first beamsplitter from the left and is repeatedly split at $d-1$
beamsplitters, with the $(d-1)$th mode reflected by the mirror at
the bottom.}}
\end{figure}
To model this situation, the ontic state space is assumed to decompose
into sets of ontic states $\Lambda^{(j)}$, which are localized at each
mode, along with possibly some additional degrees of freedom
$\Lambda_{\text{NL}}$ which are not so localized.  The total ontic
state space is then assumed to be $\Lambda = \left (
\oplus_{j=0}^{d-1} \Lambda^{(j)} \right ) \oplus
\Lambda_{\text{NL}}$.  As a concrete example, each mode might have
exactly one ontic state, with the ontic state corresponding to mode
$x_j$ labeled by the integer $j$.  The ontic state space would then
be $\{0, 1, \cdots,d-1, \cdots\} = \{0\} \oplus \{1\} \oplus \cdots
\oplus \{d-1\} \oplus \cdots$, where the second $\cdots$ is to leave
room for possible additional ontic states not localized in a mode.
More generally, the structure would be similar, but there could be
multiple ontic states corresponding to each mode.  Then, the state
$\Proj{x_j}$ represents a situation in which only ontic states in
$\Lambda^{(j)}$ can be occupied, so $\Lambda^{(j)}$ would be a measure
one set according to any probability measure in $\Delta_{x_j}$.

Now, any unitary that leaves $\Proj{x_0}$ invariant can be implemented
in such a way that it only involves manipulating modes
$x_1$--$x_{d-1}$, say by adding phase shifters to them or combining
them at beamsplitters.  If implemented in this way, it makes sense, by
locality, to think that this would have no effect on ontic states
localized at $x_0$ and hence ontic indifference would be satisfied for
$\Proj{x_0}$.

This argument for ontic indifference is based on locality, so it only
works for states like $\Proj{x_j}$, which are spatially localized.
Fortunately, Hardy's Theorem can be proved under restricted ontic
indifference, which only requires that ontic indifference should hold
for a single pure state, so, in the case of our photon modes setup,
this can always be chosen to be $\Proj{x_0}$.  Hence the assumption
need never be applied to a state that is not localized.

Even if the argument from locality were sound, we should be skeptical
of imposing locality requirements on ontological models because of
Bell's Theorem.  However, the argument can be evaded without even
giving up on locality, since it is really the assumption that
$\Proj{x_j}$ corresponds to a set of ontic states localized at $x_j$
that is at fault.  Moving to a Fock space description, in which
$\Proj{n}_{x_j}$ is the state in which there are $n$ photons in mode
$x_j$, makes this objection clearer.  In Fock space, the state
$\Proj{x_0}$ is written as $\Proj{1}_{x_0} \otimes \Proj{0}_{x_1}
\otimes \cdots \otimes \Proj{0}_{x_{d-1}}$, which explicitly shows that
modes $x_1$--$x_{d-1}$ are in their vacuum state.  From quantum field
theory, we know that the vacuum is not a featureless void, but has
some sort of structure.  Therefore, it makes sense that, at the
ontological level, there might be more than one ontic state associated
with the vacuum, and a transformation that does not affect things
localized at $x_0$ might still act nontrivially on these vacuum ontic
states.

As an explicit example of this, we can construct a ``second
quantized'' version of Spekkens' toy theory, in which we allow each
mode to have at most one photon.  The states $\Proj{0}_{x_j}$,
$\Proj{1}_{x_j}$, $\Proj{+}_{x_j}$, $\Proj{-}_{x_j}$,
$\Proj{+i}_{x_j}$ and $\Proj{-i}_{x_j}$, where
\begin{align}
\Ket{\pm} & = \frac{1}{\sqrt{2}} \left (\Ket{0} \pm \Ket{1} \right ) \\
\Ket{\pm i} & = \frac{1}{\sqrt{2}} \left (\Ket{0} \pm i\Ket{1} \right ),
\end{align}
are isomorphic to the states $\Proj{z+}$, $\Proj{z-}$, $\Proj{x+}$,
$\Proj{x-}$, $\Proj{y+}$ and $\Proj{y-}$ of a spin-$1/2$ particle.
Therefore, since Spekkens' toy theory provides a model for these
spin-$1/2$ states under measurements in the corresponding bases, it
also provides a model for the corresponding photon states and
measurements.  In this model, the ontic state space for a mode is
$\Lambda^{(j)} = \{(+,+)_{x_j},(+,-)_{x_j},(-,+)_{x_j},(-,-)_{x_j}\}$
and the ontic state spaces for modes compose via the Cartesian product
rather than the direct sum.  So, for example, in the case of two modes
$x_0$ and $x_1$, the total ontic state space would be $\Lambda =
\Lambda^{(0)} \times \Lambda^{(1)}$, which is
\begin{eqnarray}
&\{\left[(+,+)_{x_0},(+,+)_{x_1}\right], \left[(+,+)_{x_0},(+,-)_{x_1}\right],\nonumber\\
&~\left[(+,+)_{x_0},(-,+)_{x_1}\right], \left[(+,+)_{x_0},(-,-)_{x_1}\right], \nonumber\\
&~\left[(+,-)_{x_0},(+,+)_{x_1}\right], \left[(+,-)_{x_0},(+,-)_{x_1}\right],\nonumber\\
&~\left[(+,-)_{x_0},(-,+)_{x_1}\right],\left[(+,-)_{x_0},(-,-)_{x_1}\right],\nonumber\\
&~\left[(-,+)_{x_0},(+,+)_{x_1}\right], \left[(-,+)_{x_0},(+,-)_{x_1}\right],\nonumber\\
&~\left[(-,+)_{x_0},(-,+)_{x_1}\right],\left[(-,+)_{x_0},(-,-)_{x_1}\right],\nonumber\\
&~\left[(-,-)_{x_0},(+,+)_{x_1}\right],\left[(-,-)_{x_0},(+,-)_{x_1}\right],\nonumber\\
&~\left[(-,-)_{x_0},(-,+)_{x_1}\right],\left[(-,-)_{x_0},(-,-)_{x_1}\right] \}.
\end{eqnarray}
The vacuum state $\Proj{0}_{x_j}$ is represented by a equal mixture of
$(+,+)_{x_j}$ and $(-,-)_{x_j}$, just like the $\Proj{z+}$ state in
the spin-$1/2$ case, and the state $\Proj{1}_{x_j}$ with one photon is
represented by an equal mixture of $(+,-)_{x_j}$ and $(-,+)_{x_j}$,
just like the $\Proj{z-}$ state in the spin-$1/2$ case.  The state
$\Proj{x_0} = \Proj{1}_{x_0} \otimes \Proj{0}_{x_1}$ is then
represented by the product of these two distributions, which is an
equal mixture of $((+,-)_{x_0},(+,+)_{x_1})$,
$((+,-)_{x_0},(-,-)_{x_1})$, $((-,+)_{x_0},(+,+)_{x_1})$ and
$((-,+)_{x_0},(-,-)_{x_1})$.  A transformation acting locally on $x_1$
can then switch the states $(+,+)_{x_1}$ and $(-,-)_{x_1}$, in
violation of ontic indifference, whilst leaving the distribution
invariant.  Thus, the assumption that local ontic state spaces compose
according to the direct sum, rather than locality per se., is the main
problem with Hardy's argument for ontic indifference.  Since we expect
vacuum states to have structure, this is not a good assumption on
which to build a $\psi$-ontology theorem.

Nevertheless, the direct sum construction that motivates ontic
indifference is closely related to the argument from interference
discussed in \S\ref{Int}.  To see this connection, it is helpful to
outline a simple special case of Hardy's Theorem.

\subsection{An example}

\label{HExa}

Consider a Mach-Zehnder interferometer as depicted in
Fig.~\ref{fig:Hardy:mach}.  We will outline the argument that ontic
indifference implies that the state $\Proj{x_0}$, representing a
photon in the upper arm of the interferometer, must be ontologically
distinct from $\Proj{\psi}$, where $\Ket{\psi} = \frac{1}{\sqrt{2}}
\left ( \Ket{x_0} + \Ket{x_1} \right )$, representing an equal
superposition of both paths.  The argument in this section assumes a
finite ontic state space.  A measure-theoretic argument is given for
the general case in the next section.

\begin{figure}[t!]
\centering
\includegraphics[width=85mm]{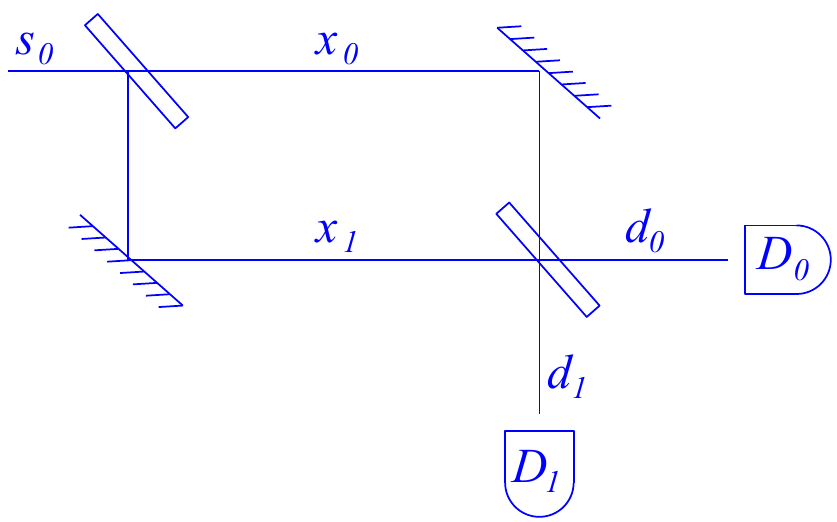}
\caption{\color[HTML]{0000FF}{\label{fig:Hardy:mach}The Mach-Zehnder interferometer used
in the example of Hardy's Theorem.}}
\end{figure}

On passing through the second beamsplitter, $\Ket{x_0}$ gets mapped to
$\frac{1}{\sqrt{2}} \left ( \Ket{d_0} + \Ket{d_1} \right )$ and
$\Ket{x_1}$ gets mapped to $\frac{1}{\sqrt{2}} \left ( \Ket{d_0} -
\Ket{d_1} \right )$, from which we infer that $\Ket{\psi}$ gets
mapped to $\Ket{d_0}$.  Hence, if the state $\Proj{x_0}$ is prepared
then the detectors $D_0$ and $D_1$ will fire with $50/50$ probability
and if the state $\Proj{\psi}$ is prepared then $D_0$ will fire with
certainty.  For the practically inclined, note that the state
$\Proj{\psi}$ can be prepared by passing a photon from the source
$s_0$ through the first beamsplitter as in Fig.~\ref{fig:Hardy:mach}
and the state $\Proj{x_0}$ can be prepared by removing the first
beamsplitter and passing a photon from $s_0$ directly into the upper
arm.

The outcome of this experiment can be altered by placing a $\pi$ phase
shifter in the lower arm of the interferometer, e.g.\ by altering its
path length by half a wavelength relative to the upper arm.  This
leaves $\Ket{x_0}$ invariant, but maps $\Ket{\psi}$ to $\Ket{\phi} =
\frac{1}{\sqrt{2}} \left ( \Ket{x_0} - \Ket{x_1} \right )$, which will
now cause the detector $D_1$ to fire with certainty.  This is a
typical example of interference, replacing constructive interference
($D_0$ fires without the phase shifter) with destructive interference
($D_1$ fires with the phase shifter).

In an ontological model, assume that there is some ontic state
$\lambda$ that is assigned nonzero probability by both $\Proj{x_0}$
and $\Proj{\psi}$.  We will show that, under the assumption of ontic
indifference, this leads to a contradiction.

Since $\Proj{\psi}$ and $\Proj{\phi}$ are orthogonal, by
Theorem~\ref{prop:POEM:distinguish} they must be ontologically
distinct.  This can also be seen more directly because every ontic
state assigned nonzero probability by $\Proj{\psi}$ must cause $D_0$
to fire with certainty, whereas every ontic state assigned nonzero
probability by $\Proj{\phi}$ must cause $D_1$ to fire with certainty.
Therefore, $\Proj{\phi}$ must assign zero probability to the ontic
state $\lambda$.

Now consider the action of the unitary $U$ that adds a $\pi$ phase
shift to the mode $x_1$.  Since this leaves $\Proj{x_0}$ invariant, the
assumption of ontic indifference implies that $\lambda$ must be left
invariant.  However, since $U$ maps $\Proj{\psi}$ to $\Proj{\phi}$, it
must also cause $\lambda$ to transition to an ontic state that causes
$D_1$ to fire with certainty, and no such ontic state has nonzero
probability according to $\Proj{\psi}$.  Therefore $\lambda$ cannot be
left invariant, which is a contradiction.  We conclude that there can
be no $\lambda$ assigned nonzero probability by both $\Proj{x_0}$ and
$\Proj{\psi}$.

By considering Hardy's motivation for ontic indifference, this can be
connected to the argument from interference as follows.  Assume that
the ontic state space decomposes into a direct sum $\Lambda =
\Lambda^{(0)} \oplus \Lambda^{(1)} \oplus \Lambda_{\text{NL}}$.  The
states in $\Lambda^{(0)}$ are localized in the upper arm of the
interferometer, representing the photon definitely taking the upper
path, and the states in $\Lambda^{(1)}$ are localized in the lower arm
of the interferometer, representing the photon definitely taking the
lower path.  If we assume that the placement of a $\pi$ phase shifter
in the lower arm has no effect on the ontic states in $\Lambda^{(0)}$
then we can deduce that $\Proj{\psi}$ cannot assign nonzero
probability to states in $\Lambda^{(0)}$ by the same argument as
above.  Switching the roles of $x_0$ and $x_1$, so that the phase
shifter is now placed or not placed in $x_0$, we can deduce the same
for $\Lambda^{(1)}$.  Therefore, $\Proj{\psi}$ must assign all its
weight to $\Lambda_{\text{NL}}$, which is not localized to either arm.
You could say that this represents a situation in which the photon
``travels along both paths'', but it would be more accurate to say
that it only has global degrees of freedom, not localized on either
path.  Further, if we assume that $\Proj{x_0}$ and $\Proj{x_1}$ are
entirely supported on $\Lambda^{(0)}$ and $\Lambda^{(1)}$
respectively, then the three quantum states $\Proj{\psi}$,
$\Proj{x_0}$ and $\Proj{x_1}$ must be ontologically distinct.  This is
a bit further along the road to proving $\psi$-ontology than the more
naive version of the argument from interference.  Hardy's result then
shows that the same sort of assumption can be used to prove a full
blown $\psi$-ontology theorem, which seems to vindicate the intuition
behind the argument from interference.

Of course, if we give up the idea that the ontic state space must
decompose as a direct sum of local state spaces plus some global
degrees of freedom, then there is no motivation for ontic
indifference, so the above argument does not go through.  In fact, the
second quantized version of Spekkens' toy theory is able to account
for almost all the qualitative features of Mach-Zehnder interferometry
in a local and $\psi$-epistemic way \cite{Martin2014}.  Thus, any
claim that these experiments entail the reality of the wavefunction or
violate locality must involve implicitly assuming something like ontic
indifference.  The most important aspect of Hardy's analysis is to
bring this assumption to the fore more clearly.

\subsection{The main result}

\label{HMain}

We are now in a position to state and prove Hardy's main result.  The
proof strategy used here is due to Patra, Pironio and Massar
\cite{Patra2013a} and is conceptually simpler than Hardy's original
proof.

\begin{theorem}
\label{thm:Hardy:hardy}
Let $\mathcal{P}$ be the set of pure states on $\mathbb{C}^d$ for $d
\geq 3$.  Then, there exists a set of measurements $\mathcal{M}$ and
a set of unitaries $\mathcal{T}$ such that any ontological model of
the PMT fragment $\mathfrak{F} = \langle \mathbb{C}^d, \mathcal{P},
\mathcal{M}, \mathcal{T} \rangle$ that reproduces the quantum
preclusions and satisfies restricted ontic indifference has the
property that any pair of states $\Proj{\psi}, \Proj{\phi} \in
\mathcal{P}$ satisfying $\Tr{\Proj{\phi}\Proj{\psi}} \leq (d-1)/d$
must be ontologically distinct.
\end{theorem}

Before proving this theorem, note that, if we form a product of
$\mathfrak{F}$ with another fragment $\mathfrak{F}^{\prime}$ and
assume that there exits an ontological model of the larger fragment
that reproduces the quantum preclusions and preserves ontological
distinctness with respect to our ontological model of $\mathfrak{F}$,
then, for any pair $\Proj{\psi}, \Proj{\phi} \in \mathcal{P}$, we can
always choose the dimensionality of the Hilbert space of
$\mathfrak{F}^{\prime}$ such that $\Tr{\Proj{\phi} \otimes \Proj{0}
  \Proj{\psi} \otimes \Proj{0}} \leq (d-1)/d$ is satisfied, where
$\Proj{0}$ is any fixed pure state of $\mathfrak{F}^{\prime}$.
Hardy's Theorem implies that $\Proj{\psi} \otimes \Proj{0}$ and
$\Proj{\phi} \otimes \Proj{0}$ must be ontologically distinct, and
then the assumption that appending ancillas preserves ontological
distinctness implies that $\Proj{\psi}$ and $\Proj{\phi}$ must be
ontologically distinct in the original model.  Since this applies to
any pair of pure states, the original ontological model must be
$\psi$-ontic.  The case $d=2$ is dealt with in a similar way.

The proof relies on the following three lemmas.
\begin{lemma}
\label{lem:Hardy:overlap}
Let $\mathfrak{F} = \langle \Hilb, \mathcal{P}, \mathcal{M},
\mathcal{T} \rangle$ be a PMT fragment, let $\Theta = (\Lambda, \Sigma,
\Delta, \Xi, \Gamma)$ be an ontological model of it, let
$\Proj{\psi}$ be a state that satisfies ontic indifference in $\Theta$,
let $\{U_j\} \subseteq \mathcal{T}$ be an at most countable set of
unitaries that leave $\Proj{\psi}$ invariant and suppose that
$\Proj{\phi}$ is ontologically indistinct from $\Proj{\psi}$ in
$\Theta$.  Then there exist measures $\mu_j \in \Delta_{U_j
\Proj{\phi}U_j^{\dagger}}$, such that $L(\{\mu_j\}) > 0$, where
$L$ is the overlap.
\end{lemma}
\begin{proof}
Let $\nu \in \Delta_{\psi}$ and suppose that $\Omega_j$ is a measure
one set left invariant by some $\gamma^j \in \Gamma_{U_j}$,
i.e. $\gamma^j_{\lambda}(\Omega_j^{\prime}) =
\delta_{\lambda}(\Omega_j^{\prime})$ for $\lambda \in \Omega_j$ and
$\Omega_j^{\prime} \subseteq \Omega_j$.  Then, $\Omega = \cap_j
\Omega_j$ is also measure one according to $\nu$ because it is the
intersection of an at most countable set of measure one sets, and
$\Omega$ is left invariant by each $\gamma^j$.  Since $\Proj{\psi}$
and $\Proj{\phi}$ are ontologically indistinct, there exists a $\mu
\in \Delta_{\phi}$ such that $\mu(\Omega) > 0$.  Now, the measure
$\mu_j$, defined by $\mu_j(\Omega^{\prime}) = \int_{\Lambda}
\gamma^j_{\lambda}(\Omega^{\prime}) d\mu(\lambda)$ for
$\Omega^{\prime} \in \Sigma$, is in $\Delta_{U_j
\Proj{\phi}U_j^{\dagger}}$, and $\mu_j(\Omega^{\prime}) =
\mu(\Omega^{\prime})$ for $\Omega^{\prime} \subseteq \Omega$.

Now, consider a partition of $\Lambda$ into sets $\Omega^k$.  Then,
\begin{align}
&\sum_k \min_j \left [ \mu_j(\Omega^k )\right ] \nonumber\\
& = \sum_k \min_j
\left [ \mu_j \left ( \Omega^k \cap \Omega \right ) + \mu_j \left
( \Omega^k \cap \left ( \Lambda \backslash \Omega \right )
\right ) \right ] \\
& \geq \sum_k \left ( \min_j \left [ \mu_j \left ( \Omega^k \cap
\Omega \right ) \right ] + \min_j \left [ \mu_j \left (
\Omega^k \cap \left ( \Lambda
\backslash \Omega \right ) \right ) \right ] \right ) \\
& \geq \sum_k \min_j \left [ \mu_j \left ( \Omega^k \cap \Omega
\right ) \right ] \\
& = \sum_k \mu(\Omega^k \cap \Omega) = \mu(\Omega) > 0.
\end{align}
Since this holds for all partitions it must also hold for the
infimum, and hence $L(\{\mu_j\}) > 0$.
\end{proof}

The next lemma is also used in the proof of the Colbeck--Renner
Theorem.
\begin{lemma}
\label{lem:Hardy:inner}
Let $\left \{ \Ket{j} \right \}_{j=0}^{d-1}$ be a set of $d \geq 3$
orthonormal vectors in $\mathbb{C}^{d^{\prime}}$ for $d^{\prime}
\geq d$, i.e.\ not necessarily a complete basis, and consider the
vectors
\begin{equation}
\Ket{\psi} = \frac{1}{\sqrt{d}} \sum_{j=0}^{d-1} \Ket{j},
\end{equation}
and
\begin{equation}
\Ket{\phi} = \frac{1}{\sqrt{d-1}} \sum_{j=1}^{d-1} e^{\imath
\varphi_j} \Ket{j}.
\end{equation}
Then, the phases $\varphi_j$ may be chosen such that
$\BraKet{\phi}{\psi}$ is real and takes any value between $0$ and
$\sqrt{\frac{d-1}{d}}$.
\end{lemma}
\begin{proof}
First consider the case where $d$ is odd.  In this case, set
$\varphi_j = \varphi$ for $1 \leq j \leq \frac{d-1}{2}$ and
$\varphi_j = -\varphi$ for $\frac{d+1}{2} \leq j \leq {d-1}$.  Then,
$\BraKet{\phi}{\psi} = \sqrt{\frac{d-1}{d}} \cos \varphi$, which may
take any value between $0$ and $\sqrt{\frac{d-1}{d}}$ by varying
$\varphi$ from $0$ to $\pi/2$.

When $d$ is even, set $\varphi_1 = 0$, $\varphi_j = \phi$ for $2
\leq j \leq \frac{d}{2}$, and $\varphi_j = -\phi$ for $\frac{d}{2} +
1 \leq j \leq {d-1}$. Then, $\BraKet{\phi}{\psi} =
\sqrt{\frac{1}{d(d-1)}} (1 + (d-2)\cos \varphi)$.  This is equal to
$\sqrt{\frac{d-1}{d}}$ for $\varphi = 0$ and is equal to $0$ for
$\varphi = \pi - \cos^{-1} \left ( \frac{1}{d-2} \right )$ and
varies continuously in between, so again any positive real number
between $0$ and $\sqrt{\frac{d-1}{d}}$ can be obtained.
\end{proof}

\begin{lemma}
\label{lem:Hardy:qstates}
Let $\Proj{\psi}$ and $\Proj{\phi}$ be two nonorthogonal states on
$\mathbb{C}^d$ for $d \geq 3$ such that $\Tr{\Proj{\phi}\Proj{\psi}}
\leq (d-1)/d$.  Then, there exist unitaries $\{U_j\}_{j=0}^{d-1}$
such that $U_j \Proj{\psi} U_j^{\dagger} = \Proj{\psi}$ and such
that $\{\Proj{\phi_j}\}$ is antidistinguishable, where $
\Proj{\phi_j} = U_j \Proj{\phi} U_j^{\dagger}$.
\end{lemma}
\begin{proof}
Let $\Ket{\psi}$ and $\Ket{\phi}$ be vector representatives of
$\Proj{\psi}$ and $\Proj{\phi}$.  Then,
\begin{equation}
\Ket{\phi} = \alpha \Ket{\psi} + \beta \Ket{\psi^{\perp}},
\end{equation}
where $\alpha = \BraKet{\phi}{\psi}$ and
$\BraKet{\psi^{\perp}}{\psi} = 0$.  By appropriate choice of the
global phase of $\Ket{\psi}$, it is possible to make $\alpha$ real
and positive.

Let $\left \{ \Ket{j} \right \}_{j=0}^{d-1}$ be a basis in which
$\Ket{\psi}$ is represented as $\Ket{\psi} = \frac{1}{\sqrt{d}}
\sum_{j=0}^{d-1} \Ket{j}$ and consider the vector
\begin{equation}
\Ket{\phi_0} = \frac{1}{\sqrt{d-1}}\sum_{j=0}^{d-1} e^{\imath \varphi_j} \Ket{j}.
\end{equation}
By Lemma~\ref{lem:Hardy:inner}, it is possible to choose the phases
$\varphi_j$ such that $\BraKet{\phi_0}{\psi} = \alpha$, since
$\alpha$ is between $0$ and $\sqrt{\frac{d-1}{d}}$.

Making this choice of phases, it follows that
\begin{equation}
\Ket{\phi_0} = \alpha \Ket{\psi} + \beta \Ket{\psi^{\perp}_0},
\end{equation}
where $\BraKet{\psi^{\perp}_0}{\psi} = 0$.  Let $U$ be a unitary
that acts as $U \Ket{\psi} = \Ket{\psi}$ and $U \Ket{\psi^{\perp}} =
\Ket{\psi^{\perp}_0}$.  It follows that $\Ket{\phi_0} = U
\Ket{\phi}$.

Now, consider the unitary $V$, where $V\Ket{j} = \Ket{j + 1 \mod
d}$, which leaves $\Ket{\psi}$ invariant.  Then, $U_j = V^jU$ also
leaves $\Ket{\psi}$ invariant for $j = 0,1,\ldots,d-1$ and the
vectors $\Ket{\phi_j} = U_j\Ket{\psi}$ satisfy $\BraKet{j}{\phi_j} =
0$, since, by construction, $\BraKet{0}{\phi_0} = 0$.  Thus, the
states $\left \{ \Proj{\phi_j} \right \}_{j=0}^{d-1}$ are
antidistinguished by the measurement $\left \{\Proj{j} \right
\}_{j=0}^{d-1}$ as required.
\end{proof}

\begin{proof}[Proof of Theorem~\ref{thm:Hardy:hardy}]
Let $\Theta = (\Lambda, \Sigma, \Delta, \Xi, \Gamma)$ be an ontological
model of the PMT fragment $\mathfrak{F}$ to be constructed.  For the
purposes of contradiction, assume that two states $\Proj{\psi},
\Proj{\phi} \in \mathcal{P}$ satisfying $\Tr{\Proj{\phi}\Proj{\psi}}
\leq (d-1)/d$ are ontologically indistinct in $\Theta$ and let
$\Proj{\psi^{\prime}}$ be a state that satisfies ontic indifference
in $\Theta$.  Without loss of generality, we may assume that
$\Proj{\psi} = \Proj{\psi^{\prime}}$.  If not, let $U$ be a unitary
such that $\Proj{\psi^{\prime}} = U \Proj{\psi} U^{\dagger}$ and
assume $U \in \mathcal{T}$.  Then the states $\Proj{\psi^{\prime}}$
and $\Proj{\phi^{\prime}} = U \Proj{\phi} U^{\dagger}$ satisfy
$\Tr{\Proj{\phi^{\prime}}\Proj{\psi^{\prime}}} \leq (d-1)/d$ because
the inner product is preserved under unitary transformations.  Also,
by Corollary~\ref{cor:Dyn:contract}, if $\Proj{\psi}$ and
$\Proj{\phi}$ are ontologically indistinct then so are
$\Proj{\psi^{\prime}}$ and $\Proj{\phi^{\prime}}$ so it is
sufficient to derive a contradiction for these two states.

By Lemma~\ref{lem:Hardy:qstates}, there exist unitaries
$\{U_j\}_{j=0}^{d-1}$ that leave $\Proj{\psi}$ invariant such that
$\left \{ U_j \Proj{\phi}U_j^{\dagger} \right \}_{j=0}^{d-1}$ is
antidistinguishable.  Assume that each $U_j$ is in $\mathcal{T}$ and
that the antidistinguishing measurement is in $\mathcal{M}$.  Then,
by Theorem~\ref{thm:Anti:antiover}, for every set of probability
measures $\{\mu_j\}_{j=0}^{d-1}$ with $\mu_j \in \Delta_{U_j
\Proj{\phi}U_j^{\dagger}}$, $L(\{\mu_j\}) = 0$, but by
Lemma~\ref{lem:Hardy:overlap} there must exist a set such that
$L(\{\mu_j\}) > 0$.  Thus, we have a contradiction so $\Proj{\psi}$
and $\Proj{\phi}$ must be ontologically distinct.
\end{proof}

\section{The Colbeck--Renner Theorem}

\label{CR}

Our final $\psi$-ontology theorem is due to Colbeck and Renner
\cite{Colbeck2013a}.  They originally proved a $\psi$-ontology theorem
as a byproduct of a broader no-go theorem that aimed to rule out any
ontological model of quantum theory in which more precise predictions
can be made given knowledge of the ontic state than if you only know
the quantum state \cite{Colbeck2011, Colbeck2012a, Colbeck2012}.  This
broader theorem is quite tricky to understand, both conceptually and
technically, and its assumptions have been criticized
\cite{Ghirardi2013, Ghirardi2013a, Ghirardi2013b}.  Fortunately,
Colbeck and Renner have recently put forward a simplified argument for
$\psi$-ontology based on similar ideas that bypasses much of the
complexity \cite{Colbeck2013a}.  I present my own variation of the
proof here.

The Colbeck--Renner Theorem is based on a locality assumption known as
\emph{parameter independence}.  The condition of Bell locality,
discussed in \S\ref{Bell}, can be decomposed as the conjunction of
parameter independence and another assumption called \emph{outcome
independence}.  Roughly speaking, parameter independence says that
the measurement outcome at Bob's side should not depend on Alice's
choice of measurement, and outcome independence says that the
measurement outcome at Bob's side should not depend on Alice's
measurement outcome.  For a typical Bell inequality experiment, it is
in principle possible to reproduce the quantum predictions by
violating outcome independence and leaving parameter independence
intact, so the Colbeck--Renner Theorem rules out the class of
$\psi$-epistemic models that opt for this resolution to the conundrum
presented by Bell's Theorem.  Nevertheless, Bell's Theorem should lead
us to be skeptical of imposing locality assumptions on ontological
models.  Indeed, many notable realist approaches to quantum theory,
such as de Broglie--Bohm theory, violate parameter independence.  As
well as ruling out $\psi$-epistemic models, parameter independence
rules out a wide class of viable $\psi$-ontic models, so it is perhaps
an unreasonably strong assumption for a $\psi$-ontology theorem.

It should be noted that Colbeck and Renner claim that their theorem is
not directly based on parameter independence, but rather on a broader
assumption that they call ``free choice'', aimed at capturing the idea
that certain variables, such as the settings of preparation and
measurement devices, are chosen freely by the experimenter and not
constrained by other parameters of the system.  Parameter independence
is just one consequence of their free choice assumption.  However,
whether their assumption truly captures the meaning of making a free
choice is one of the most controversial aspects of the Colbeck--Renner
argument.  My own opinion is that, whilst some aspects of free choice
are captured by the Colbeck--Renner assumption, it is at least partly a
locality assumption as well because it implies parameter independence.
Therefore, I prefer to avoid the free choice controversy and just
assume parameter independence directly.

Understanding the Colbeck--Renner Theorem requires quite a bit of
background material, so the discussion is broken down into four
sections.  \S\ref{CR:PI} gives a formal definition of parameter
independence and shows how, given an ontological model for a bipartite
system, it can be used to derive an ontological model for the reduced
states on a subsystem.  \S\ref{CR:Chain} introduces chained Bell
measurements, which are the main technical tool used in proving the
Colbeck--Renner Theorem.  \S\ref{CR:Equi} discusses the equiprobability
theorem, which shows that, in a model that satisfies parameter
independence, all ontic states associated with a maximally entangled
state must assign equal probability to all the outcomes of some local
measurement in an orthonormal basis.  Finally, \S\ref{CR:Main}
leverages the equiprobability theorem to prove the main result.

\subsection{Parameter Independence}

\label{CR:PI}

As in Bell's Theorem, the fragment of interest for the Colbeck--Renner
Theorem is a product measurement fragment, consisting of states and
local measurements on a bipartite system (see
Definition~\ref{def:Bell:PM}).  \S\ref{Bell} discussed the concept of
a conditional fragment on Bob's system that results from viewing
Alice's measurements as preparation procedures for states on Bob's
system.  Given an ontological model for a product measurement fragment
that satisfies Bell locality, it is possible to derive a model for the
conditional fragment from it.  Here, we are concerned with a smaller
fragment that results from tracing out one of the systems, but not
post-selecting on measurement outcomes.

\begin{definition}
Let $\mathfrak{F}_{AB} = \langle \Hilb[A] \otimes \Hilb[B],
\mathcal{P}_{AB}, \mathcal{M}_A \times \mathcal{M}_B \rangle$ be a
product measurement fragment.  The \emph{marginal fragment on $A$}
is $\mathfrak{F}_A = \langle \Hilb[A], \mathcal{P}_A, \mathcal{M}_A
\rangle$, where $\mathcal{P}_A$ consists of all states of the form
$\rho_A = \Tr[B]{\rho_{AB}}$ for $\rho_{AB} \in \mathcal{P}_{AB}$.
The marginal fragment on $B$ is defined similarly, with the roles
of $A$ and $B$ interchanged.
\end{definition}

In general, an ontological model $\Theta_{AB} = (\Lambda_{AB},
\Sigma_{AB}, \Delta_{AB}, \Xi_{AB})$ for a product measurement
fragment $\mathfrak{F}_{AB} = \langle \Hilb[A] \otimes \Hilb[B],
\mathcal{P}_{AB}, \mathcal{M}_A \times \mathcal{M}_B \rangle$
specifies a set of conditional probability distributions
$\text{Pr}_{AB} \in \Xi_{AB}[M_A,\allowbreak M_B]$ for every $M_A \in
\mathcal{M}_A$ and $M_B \in \mathcal{M}_B$.  These distributions are
of the form $\text{Pr}_{AB}(E,F|M_A,M_B,\lambda)$, where $E$ varies
over $M_A$ and $F$ varies over $M_B$.  If the model reproduces the
quantum predictions, then for all $\rho_{AB} \in \mathcal{P}_{AB}$,
$M_A \in \mathcal{M}_A$, $M_B \in \mathcal{M}_B$, $E \in M_A$ and $F
\in M_B$,
\begin{equation}
\int_{\Lambda_{AB}} \text{Pr}_{AB}(E,F|M_A,M_B,\lambda)
d\mu(\lambda) = \Tr[AB]{E \otimes F \rho_{AB}},
\end{equation}
for every $\mu \in \Delta_{AB}[\rho]$ and $\text{Pr}_{AB} \in
\Xi_{AB}[M_A,M_B]$.  Summing both sides of this equation over $F \in
M_B$ gives
\begin{equation}
\label{eq:CR:margreprod}
\int_{\Lambda_{AB}} \text{Pr}_{AB}(E|M_A,M_B,\lambda) d\mu(\lambda) =
\Tr[A]{E \rho_A},
\end{equation}
where $\rho_A = \Tr[B]{\rho_{AB}}$ and
$\text{Pr}_{AB}(E|M_A,M_B,\lambda) = \sum_{F \in M_B}
\text{Pr}_{AB}(E,F|M_A,M_B,\lambda)$.  For a fixed $M_B$,
$\text{Pr}_{AB}(E|M_A,M_B,\lambda)$ is a valid conditional probability
distribution over $M_A$ that could be used in an ontological model of
the marginal fragment.  Further, if $\mu$ is used to represent
$\rho_A$, this choice of conditional probability distributions
reproduces the quantum predictions.  The only difficulty is that
$\text{Pr}_{AB}(E|M_A,M_B,\lambda)$ depends on $M_B$ as well as $M_A$,
so there is an ambiguity about which choice of $M_B$ should be used to
model the marginal fragment.  The assumption of parameter independence
removes this ambiguity.

\begin{definition}
An ontological model $\Theta_{AB} = (\Lambda_{AB}, \Sigma_{AB},
\Delta_{AB}, \xi_{AB})$ of a product measurement fragment
$\mathfrak{F}_{AB} = \langle \Hilb[A] \otimes \Hilb[B],
\mathcal{P}_{AB}, \mathcal{M}_A \times \mathcal{M}_B \rangle$
satisfies \emph{parameter independence} if, for all $M_A \in
\mathcal{M}_A$, $M_B \in \mathcal{M}_B$ and $\text{Pr}_{AB} \in
\Xi_{AB}[M_A,M_B]$,
\begin{align}
\text{Pr}_{AB}(E|M_A,M_B,\lambda) & =
\text{Pr}_{AB}(E|M_A,\lambda) \\
\text{Pr}_{AB}(F|M_A,M_B,\lambda) & =
\text{Pr}_{AB}(F|M_B,\lambda).
\end{align}
\end{definition}

Parameter independence is a locality assumption that says that Alice's
measurement outcome should not depend on Bob's choice of measurement
and vice versa.  It is a weaker requirement than Bell locality, as it
can be shown that the conjunction of parameter independence and
\emph{outcome independence}, which reads
\begin{align}
\text{Pr}_{AB}(E|M_A,M_B,F,\lambda) & = \text{Pr}_{AB}(E|M_A,M_B,\lambda) \\
\text{Pr}_{AB}(F|M_A,E,M_B,\lambda) & = \text{Pr}_{AB}(F|M_A,M_B,\lambda),
\end{align}
is equivalent to Bell locality \cite{Jarrett1984, Shimony1993}.  For
present purposes, the point of this is that, whilst the full Bell
locality condition is needed to derive a well-defined ontological
model for conditional fragments, parameter independence is sufficient
to have unambiguous models for marginal fragments.

\begin{definition}
Let $\mathfrak{F}_{AB} = \langle \Hilb[A] \otimes \Hilb[B],
\mathcal{P}_{AB}, \mathcal{M}_A \times \mathcal{M}_B \rangle$ be a
product measurement fragment and let $\Theta_{AB} = ( \Lambda_{AB},
\Sigma_{AB}, \Delta_{AB}, \Xi_{AB} )$ be an ontological model of it
that satisfies parameter independence.  The \emph{marginal model on
$A$}, $\Theta_A = (\Lambda_{AB}, \Sigma_{AB}, \Delta_A, \Xi_A)$ is an
ontological model for the marginal fragment $\mathfrak{F}_A =
\langle \Hilb[A], \mathcal{P}_A, \mathcal{M}_A \rangle$, where, for
$\rho_A \in \mathcal{P}_A$,
\begin{equation}
\Delta_A[\rho_A] =  \cup_{\{\rho_{AB} \in
\mathcal{P}_{AB}|\Tr[B]{\rho_{AB}} = \rho_A\}} \Delta_{AB}[\rho_{AB}],
\end{equation}
and $\Xi_A[M_A]$ consists of all conditional probability
distributions of the form
\begin{equation}
\text{Pr}_A(E|M_A, \lambda) = \sum_{F \in M_B}
\text{Pr}_{AB}(E,F|M_A,M_B, \lambda),
\end{equation}
for all $\text{Pr}_{AB} \in \Xi_{AB}[M_A,M_B]$.  Note that, by
parameter independence, these distributions are independent of
$M_B$.
\end{definition}

Eq.~\eqref{eq:CR:margreprod} implies that if $\Theta_{AB}$ reproduces the
quantum predictions for $\mathfrak{F}_{AB}$ then the marginal model
$\Theta_A$ reproduces the quantum predictions for the marginal fragment
$\mathfrak{F}_A$.

\subsection{Chained Bell measurements}

\label{CR:Chain}

The main engine of the proof of the Colbeck--Renner Theorem is the
\emph{chained Bell measurements}.  These were originally developed in
the context of a proof of Bell's Theorem that involved chaining
together several tests of the Clauser--Horne--Shimony--Holt inequality
\cite{Clauser1969}, resulting in the chained Bell inequalities
\cite{Braunstein1989, Braunstein1989a}, which explains the
terminology.

We are interested in the product measurement fragment $\mathfrak{F} =
\langle \Hilb[A] \otimes \Hilb[B], \mathcal{P}_{AB}, \mathcal{M}_A
\times \mathcal{M}_B \rangle$, where $\Hilb[A] = \Hilb[B] =
\mathbb{C}^2$ and both $\mathcal{M}_A$ and $\mathcal{M}_B$ contain $N$
measurements in orthonormal bases.  For later convenience, Alice's
measurements are labeled by even integers and Bob's by the odd
integers, so we have $M_A^a \in \mathcal{M}_A$ for $a \in E_N =
\{0,2,\ldots,2N-2\}$ and $M_A^b \in \mathcal{M}_B$ for $b \in O_N =
\{1,3,\ldots,2N-1\}$.  The measurements are described by pairs of
orthogonal rank-1 projectors $M_A^a =
\{\Proj{\phi^a_0}_A,\Proj{\phi^a_1}_A\}$, $M_B^b =
\{\Proj{\phi^b_0}_B,\Proj{\phi^b_1}_B\}$ onto the orthonormal bases
$\left \{ \Ket{\phi^a_0}_A, \Ket{\phi^a_1}_A \right \}$ and $\left \{
\Ket{\phi^b_0}_B, \Ket{\phi^b_1}_B \right \}$.  If the state
$\rho_{AB}$ is prepared, Alice measures $M_A^a$, and Bob measures
$M_B^b$, then the outcome probabilities are
\begin{equation}
\text{Prob} \left ( \Proj{\phi^a_j},
\Proj{\phi^b_k}|M_A^a,M_B^b,\rho_{AB}\right ) =
\Tr[AB]{\Proj{\phi^a_j} \otimes \Proj{\phi^b_k} \rho_{AB}}.
\end{equation}
For later convenience, we introduce two random variables $X$ and $Y$
that take values $0$ and $1$ to represent Alice and Bob's outcomes and
write
\begin{align}
&\text{Prob} (X = j,Y = k|a,b,\rho_{AB}) \nonumber\\
&\quad= \text{Prob} \left (
\Proj{\phi^a_j}, \Proj{\phi^b_k}|M_A^a,M_B^b,\rho_{AB}\right ).
\end{align}

Now consider the following correlation measure for these outcome
probabilities.
\begin{align}
\label{eq:CR:corr}
&I_N(\rho_{AB}) = \text{Prob}(X=Y|0,2N-1,\rho_{AB}) \nonumber\\
&\quad+ \sum_{\left
\{(a,b) \middle | a \in E_N, b \in O_N, |a-b| = 1 \right \}}
\text{Prob}(X \neq Y|a,b,\rho_{AB}).
\end{align}
The structure of this measure is illustrated in
Fig.~\ref{fig:CR:corr}. Note that smaller values of this measure
indicate that the measurement outcomes are more highly correlated.

Now, let $\left \{ \Ket{0}, \Ket{1} \right \}$ be an orthonormal basis
for $\mathbb{C}^2$, and consider the quantum state
$\Proj{\Phi^+}_{AB}$, where
\begin{equation}
\Ket{\Phi^+}_{AB} = \frac{1}{\sqrt{2}} \left ( \Ket{0}_A \otimes \Ket{0}_B
+ \Ket{1}_A \otimes \Ket{1}_B \right ).
\end{equation}
Let $\vartheta^a_j = \left ( \frac{a}{2N} + j\right ) \pi$, and
suppose that the bases for Alice and Bob's measurements are given by
\begin{align}
\Ket{\phi^a_j}_A & = \cos \left ( \frac{\vartheta^a_j}{2} \right )
\Ket{0}_A + \sin \left ( \frac{\vartheta^a_j}{2} \right
) \Ket{1}_A \\
\Ket{\phi^b_k}_B & = \cos \left ( \frac{\vartheta^b_k}{2} \right )
\Ket{0}_A + \sin \left ( \frac{\vartheta^b_k}{2} \right ) \Ket{1}_B.
\end{align}
The $j,k=0$ outcomes of these measurements are illustrated on the
Bloch sphere in Fig.~\ref{fig:CR:chained}.
\begin{figure}[t!]
\centering
\includegraphics[width=85mm]{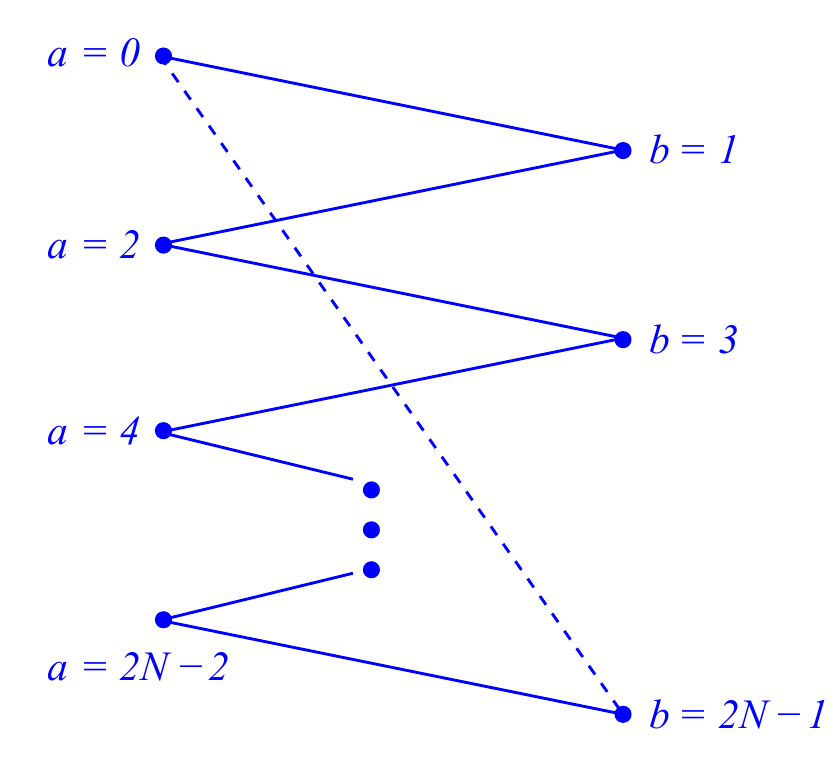}
\caption{\color[HTML]{0000FF}{\label{fig:CR:corr}The structure of the correlation measure
given in Eq.~\eqref{eq:CR:corr}.  A solid line between $a$ and $b$
represents the probability that the outcomes of $M^A_a$ and
$M^A_b$ are different and a dotted line represents the probability
that they are the same.  The correlation measure is then the sum
of the terms represented by each line.}}
\end{figure}

With these particular choices, the outcome probabilities are
\begin{align}
&\text{Prob} \left ( X = j,Y = k|a,b,\Proj{\Phi^+}_{AB}\right ) \nonumber\\
&\quad= \frac{1}{2}
\cos^2 \left ( \left ( \frac{a-b}{2N} + j - k \right ) \frac{\pi}{2}
\right ),
\end{align}
from which it is straightforward to compute that
\begin{equation}
\label{eq:CR:corval}
I_N(\Proj{\Phi^+_{AB}}) = 2N \sin^2 \frac{\pi}{4N} \leq
\frac{\pi^2}{8N} ,
\end{equation}
so that the correlation measure tends to zero as $N \rightarrow
\infty$.

Without going through the whole calculation, it is fairly easy to see
why this should be the case.  For the state $\Proj{\Phi^+}_{AB}$, when
Alice performs the measurement $M_A^a$ and obtains the outcome $j$,
Bob's state gets updated to $\Proj{\phi^a_j}_B$.  The condition $|a-b|
= 1$ in the sum in Eq.~\eqref{eq:CR:corr} means that, in the limit of
large $N$, the terms in this sum involve basis states for Bob that are
very close to Alice's basis states.  Thus, since Bob's system
collapses to the basis state that Alice obtained in her measurement,
the probability of getting the same outcome is very high, and hence
$P(X \neq Y|a,b,\Proj{\Phi^+}_{AB})$ is close to zero.  Similarly, the
state $\Proj{\phi_j^{2N-1}}_B$ is almost orthogonal to
$\Proj{\phi_j^{0}}_B$, so the probability of getting equal outcomes in
the first term is also close to zero.
\begin{figure}[t!]
\centering
\includegraphics[width=85mm]{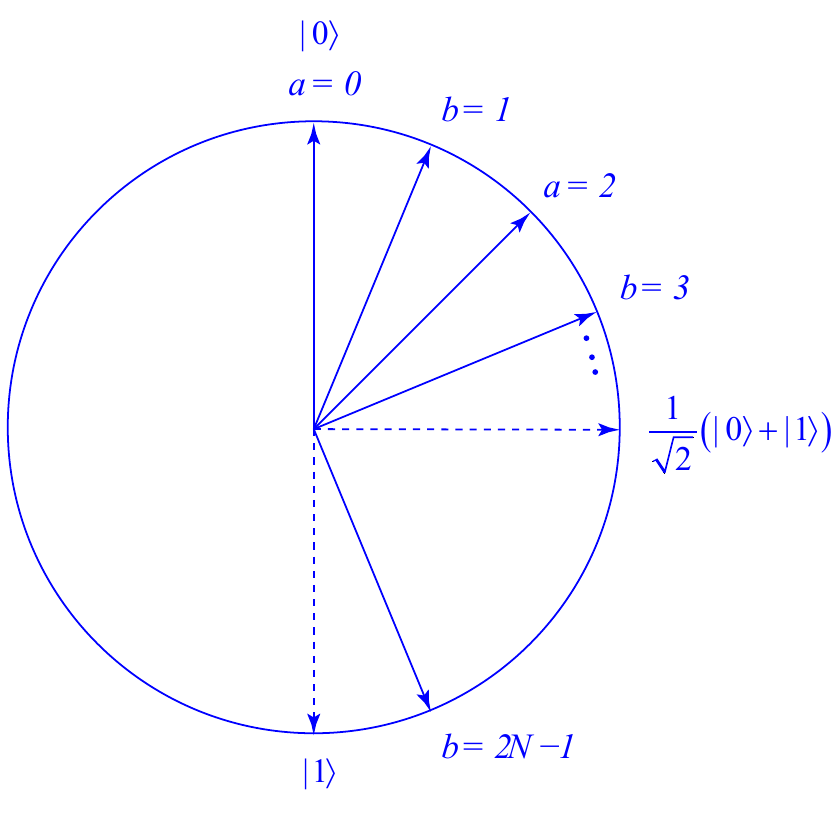}
\caption{\color[HTML]{0000FF}{\label{fig:CR:chained}Alice and Bob's measurement bases
represented on a great circle of the Bloch sphere.}}
\end{figure}

\subsection{The equiprobability theorem}

\label{CR:Equi}

The aim of this section is to prove the main technical result that the
Colbeck--Renner Theorem relies upon.  We first deal with the case where
Alice and Bob's Hilbert spaces are two dimensional, and then extend
this to higher dimensions.

\begin{theorem}
\label{thm:CR:qubit}
Let $\mathfrak{F}_{AB} = \langle \Hilb[A] \otimes \Hilb[B],
\mathcal{P}_{AB}, \mathcal{M}_A \times \mathcal{M}_B \rangle$ be a
product measurement fragment where $\Hilb[A] = \Hilb[B] =
\mathbb{C}^2$, $\mathcal{P}_{AB}$ contains the maximally entangled
state $\Proj{\Phi^+}_{AB}$, where
\begin{equation}
\Ket{\Phi^+}_{AB} = \frac{1}{\sqrt{2}} \left ( \Ket{0}_A \otimes
\Ket{0}_B  + \Ket{1}_A \otimes \Ket{1}_B \right ),
\end{equation}
for some orthonormal basis $\left \{ \Ket{0}, \Ket{1} \right \}$,
$\mathcal{M}_A$ contains measurements in every orthonormal basis on
$\Hilb[A]$ and $\mathcal{M}_B$ contains measurements in every
orthonormal basis on $\Hilb[B]$.

Let $\Theta_{AB} = (\Lambda_{AB}, \Sigma_{AB}, \Delta_{AB}, \Xi_{AB})$
be an ontological model of $\mathfrak{F}_{AB}$ that reproduces the
quantum predictions and satisfies parameter independence.

Then, for any $\mu \in \Delta_{AB}\left [ \Proj{\Phi^+}_{AB} \right
]$, there exists a set $\Omega \in \Sigma_{AB}$ such that
$\mu(\Omega) = 1$ and, for the measurement $M_A = \left
\{\Proj{0}_A, \Proj{1}_A \right \}$, all $\mathrm{Pr}_A \in
\Xi_A[M_A]$ in the marginal model on $A$ satisfy $\mathrm{Pr}_A
\left ( \Proj{0}_A \middle | M_A,\lambda \right ) = \mathrm{Pr}_A
\left ( \Proj{1}_A \middle | M_A,\lambda \right ) = \frac{1}{2}$ for
$\lambda \in \Omega$.
\end{theorem}

Before proving this theorem, it is instructive to see how it can be
used to establish ontological distinctness for some pairs of states.
Consider the state $\Proj{0}_A \otimes \Proj{0}_B$.  To reproduce the
quantum predictions, the marginal model on $A$ must predict that the
$\Proj{1}_A$ outcome has zero probability.  This means that any $\nu
\in \Delta_{AB}[\Proj{0}_A \otimes \Proj{0}_B]$ must assign zero
measure to any set $\Omega \in \Sigma_{AB}$ for which
$\text{Pr}_A(\Proj{1}_A|M_A,\lambda)$ is nonzero for some $\text{Pr}_A
\in \Xi_A[M_A]$.  However, according to Theorem~\ref{thm:CR:qubit},
for any $\mu \in \Delta_{AB} \left [ \Proj{\Phi^+}_{AB} \right ]$
there is such a set that is measure one according to $\mu$.  Thus, for
any $\mu \in \Delta_{AB} \left [ \Proj{\Phi^+}_{AB} \right ]$ and any
$\nu \in \Delta_{AB} \left [ \Proj{0}_A \otimes \Proj{0}_B \right ]$,
there is a set that is measure $1$ according to $\mu$ and measure zero
according to $\nu$, so $\Proj{\Phi^+}_{AB}$ and $\Proj{0}_A \otimes
\Proj{0}_B$ are ontologically distinct.  The Colbeck--Renner Theorem
generalizes this argument to arbitrary pairs of pure states.

The proof of Theorem~\ref{thm:CR:qubit} relies on the following two
lemmas.

\begin{lemma}
\label{lem:CR:convex}
Let $X$ and $Y$ be random variables that take values in the set
$\{0,1,\ldots,d-1\}$ and let $P(X), P^{\prime}(X), Q(Y),
Q^{\prime}(Y)$ be probability distributions over them.  For $0 \leq
p \leq 1$, let
\begin{align}
P^{\prime\prime}(X) & = p P(X) + (1-p)P^{\prime}(X) \\
Q^{\prime\prime}(Y) & = p Q(Y) + (1-p)Q^{\prime}(Y).
\end{align}
Then,
\begin{align}
D \left (P^{\prime\prime}(X),Q^{\prime\prime}(Y) \right ) \leq &\quad p D
\left ( P(X),Q(Y) \right ) \nonumber\\
&+ (1-p)D \left (
P^{\prime}(X),Q^{\prime}(Y) \right ),
\end{align}
where
\begin{equation}
D \left ( P(X),Q(Y) \right ) = \frac{1}{2} \sum_{j=0}^{d-1} \left
| P(X=j) - Q(Y=j)\right |,
\end{equation}
is the variational distance.
\end{lemma}
\begin{proof}
  \begin{align}
	& D  ( P^{\prime\prime}(X), Q^{\prime\prime}(Y)  ) \nonumber\\
	& =
    \frac{1}{2} \sum_{j=0}^{d-1}  | p  ( P(X=j) -
        Q(Y=j) ) \nonumber\\
				&\quad\quad+ (1-p)  ( P^{\prime}(X=j) -
        Q^{\prime}(X=j) ) | \\
    & \leq \frac{1}{2} \sum_{j=0}^{d-1}  (  | p  (
          P(X=j) - Q(Y=j) )  | \nonumber\\
					&\quad\quad +  | (1-p)  (
          P^{\prime}(X=j) - Q^{\prime}(X=j) ) |  ) \\
  & = p \frac{1}{2} \sum_{j=0}^{d-1}  | P(X=j) - Q(Y=j) | \nonumber\\
	&\quad\quad +
  (1-p) \frac{1}{2} \sum_{j=0}^{d-1}  | P^{\prime}(X=j) -
    Q^{\prime}(Y=j) |  \\
  & = p D  ( P(X), Q(Y) ) + (1-p)  ( P^{\prime}(X),
    Q^{\prime}(Y) ),
 \end{align}
where the second step follows from the triangle inequality.
\end{proof}

\begin{lemma}
\label{lem:CR:couple}
Let $X$ and $Y$ be random variables that take values in the set
$\{0,1,\ldots,d-1\}$ and let $P(X)$ and $P(Y)$ be probability
distributions over them.  Let $P(X,Y)$ be a joint distribution such
that $P(X) = \sum_k P(X,Y=k)$ and $P(Y) = \sum_j P(X=j, Y)$.  Then,
\begin{equation}
D \left ( P(X),P(Y) \right ) \leq P(X \neq Y),
\end{equation}
where $D$ is the variational distance.
\end{lemma}
\begin{proof}
Let $P_{\neq}(X,Y)$ be the conditional probability distribution of
$X$ and $Y$ given that $X \neq Y$, let $P_{=}(X,Y)$ be the
conditional probability distribution of $X$ and $Y$ given that
$X=Y$, let $p_= = P(X=Y)$ and let $p_{\neq} = P(X \neq Y)$.  Then,
by the law of total probability
\begin{equation}
P(X,Y) = p_{\neq}P_{\neq}(X,Y)  + p_{=}P_{=}(X,Y).
\end{equation}
Taking the marginals of this equation gives
\begin{align}
P(X) & = p_{\neq}P_{\neq}(X) + p_{=}P_{=}(X) \\
P(Y) & = p_{\neq}P_{\neq}(Y) + p_{=}P_{=}(Y).
\end{align}
Then, by Lemma~\ref{lem:CR:convex},
\begin{align}
D(P(X),P(Y)) & \leq p_{\neq} D(P_{\neq}(X), P_{\neq}(Y)) \nonumber\\
&\quad\quad + p_{=}
D(P_{=}(X),P_{=}(Y))  \\
& \leq p_{\neq} \\
& = P(X \neq Y).
\end{align}
\end{proof}

\begin{proof}[Proof of Theorem~\ref{thm:CR:qubit}].  Consider the
chained Bell measurements $M^a_A = \left \{ \Proj{\phi^a_0}_A,
\Proj{\phi^a_1}_A \right \}$ and $M^b_B = \left \{
\Proj{\phi^b_0}_B,\allowbreak \Proj{\phi^b_1}_B\right \}$ defined
in \S\ref{CR:Chain} and note that $M^0_A = \left \{ \Proj{0}_A,
\Proj{1}_A \right \}$.  From Eq.~\eqref{eq:CR:corval}, the
correlation measure $I_N$ defined in Eq.~\eqref{eq:CR:corr}
satisfies $I_N \left ( \Proj{\Phi^+}_{AB} \right ) \leq
\frac{\pi^2}{8N}$.  For any fixed choice of $\text{Pr}_{AB} \in
\Xi[M_A^a,M_B^b]$ for each $M^a_A$ and $M^b_B$, we can also define a
similar correlation measure within the ontological model, given by
\begin{align}
\label{eq:CR:corront}
&I_N(\lambda) = P(X=Y|0,2N-1,\lambda) \nonumber\\
&\quad + \sum_{\left \{(a,b)
\middle | a \in E_N, b \in O_N, |a-b| = 1 \right \}} P(X \neq
Y|a,b,\lambda),
\end{align}
where
\begin{equation}
P(X=j,Y=k|a,b,\lambda) = \text{Pr}_{AB} \left ( \Proj{\phi^a_j}_A,
\Proj{\phi^b}_B \middle | M^a_A,M^b_B,\lambda \right ).
\end{equation}
Now, by parameter independence
\begin{equation}
\label{eq:CR:presponse1}
P(X = j|a,b,\lambda) = P(X = j|a,\lambda) =
\text{Pr}_A \left ( \Proj{\phi^a_j} \middle | M_A^a,\lambda \right ),
\end{equation}
and
\begin{equation}
\label{eq:CR:presponse2}
P(Y = k|a,b,\lambda) = P(Y = k|b,\lambda) \\ =
\text{Pr}_B \left ( \Proj{\phi^b_k} \middle | M_B, \lambda \right ),
\end{equation}
where $\text{Pr}_A$ and $\text{Pr}_B$ are the conditional
probability distributions of the marginal models on $A$ and $B$
respectively.

Let $\tilde{X} = (X + 1) \mod 2$, so that if $X = 0$ then $\tilde{X}
= 1$ and vice versa.  Then, $P(X=Y|0,2N-1,\lambda) = P(\tilde{X}
\neq Y|0,2N-1,\lambda)$.  Now, by parameter independence and
Lemma~\ref{lem:CR:couple},
\begin{align}
I_N(\lambda) & \geq D(P(\tilde{X}|0,\lambda), P(Y|2N-1,\lambda)) \nonumber\\
&\quad +
\sum_{\left \{(a,b) \middle | a \in E_N, b \in O_N, |a-b| = 1
\right \}}
D(P(X|a,\lambda), P(Y|b,\lambda)) \\
& \geq D(P(\tilde{X})|0,\lambda),
P(X|0,\lambda)), \label{eq:CR:ontcor}
\end{align}
where the second step follows from the triangle inequality for the
variational distance and the structure of the correlation measure
shown in Fig.~\ref{fig:CR:corr}.

Since the ontological model reproduces the quantum predictions we
have, for any $\mu \in \Delta_{AB} \left [ \Proj{\Phi^+}_{AB} \right
]$, $\int_{\Lambda} I_N(\lambda) d\mu(\lambda) = I_N \left (
\Proj{\Phi^+}_{AB} \right )$ and in the limit $N \rightarrow
\infty$, this is equal to $0$.  Therefore, using
Eq.~\eqref{eq:CR:ontcor} gives
\begin{equation}
\int_{\Lambda}
D(P(\tilde{X}|0,\lambda),P(X|0,\lambda))d\mu(\lambda) = 0.
\end{equation}
Using Eq.~\eqref{eq:CR:presponse1}, this is equivalent to
\begin{equation}
\int_{\Lambda}\left | \text{Pr}_A \left ( \Proj{0}_A \middle |
M^0_A, \lambda \right ) -
\text{Pr} \left ( \Proj{1}_A \middle | M^0_A, \lambda \right )
\right |d\mu(\lambda) = 0.
\end{equation}
This can only happen if there exists a set $\Omega \in \Sigma_{AB}$
such that $\mu(\Omega) = 1$ and $\text{Pr}_A \left ( \Proj{0}_A
\middle | M^0_A, \lambda \right ) = \text{Pr} \left ( \Proj{1}_A
\middle | M^0_A, \lambda \right )$ for $\lambda \in \Omega$.  But
since $\text{Pr}_A \left ( \Proj{0}_A \middle | M^0_A, \lambda
\right ) + \text{Pr} \left ( \Proj{1}_A \middle | M^0_A, \lambda
\right ) = 1$, this means that $\text{Pr}_A \left ( \Proj{0}_A
\middle | M^0_A, \lambda \right ) = \text{Pr} \left ( \Proj{1}_A
\middle | M^0_A, \lambda \right ) = \frac{1}{2}$ for $\lambda \in
\Omega$.
\end{proof}

The next step is to generalize this result to arbitrary dimensions.

\begin{theorem}
\label{thm:CR:qudit}
Let $\mathfrak{F}_{AB} = \langle \Hilb[A] \otimes \Hilb[B],
\mathcal{P}_{AB}, \mathcal{M}_A \times \mathcal{M}_B \rangle$ be a
product measurement fragment where $\Hilb[A] = \Hilb[B] =
\mathbb{C}^d$, $\mathcal{P}_{AB}$ contains the maximally entangled
state $\Proj{\Phi^+}_{AB}$, where
\begin{equation}
\Ket{\Phi^+}_{AB} = \frac{1}{\sqrt{d}} \sum_{j=0}^{d-1} \Ket{j}_A
\otimes \Ket{j}_B,
\end{equation}
for some orthonormal basis $\left \{ \Ket{j} \right \}_{j=0}^{d-1}$,
$\mathcal{M}_A$ contains measurements in every orthonormal basis on
$\Hilb[A]$, and $\mathcal{M}_B$ contains measurements in every
orthonormal basis on $\Hilb[B]$.

Let $\Theta_{AB} = (\Lambda_{AB}, \Sigma_{AB}, \Delta_{AB}, \Xi_{AB})$
be an ontological model of $\mathfrak{F}_{AB}$ that reproduces the
quantum predictions and satisfies parameter independence.

Then, for any $\mu \in \Delta_{AB} \left (\Proj{\Phi^+}_{AB} \right
)$, there exists a set $\Omega \in \Sigma_{AB}$ such that
$\mu(\Omega) = 1$ and, for the measurement $M_A = \left \{\Proj{j}_A
\right \}_{j=0}^{d-1}$, all $\mathrm{Pr}_A \in \Xi_A[M_A]$ in the
marginal model on $A$ satisfy $\mathrm{Pr}_A \left (\Proj{j}_A \middle
| M_A, \lambda \right ) = \frac{1}{d}$ for $\lambda \in \Omega$.
\end{theorem}

\begin{proof}
The chained Bell measurements can be extended to higher dimensions
by defining $\Proj{\phi^a_j}_A$ $\Proj{\phi^b_k}_B$ as before for
$j,k \in \{0,1\}$ and setting $\Proj{\phi^a_j}_A = \Proj{j}_A$,
$\Proj{\phi^b_k} = \Proj{k}_B$ for $2 \leq j,k \leq d-1$.  Now,
define two random variables $X$ and $Y$ taking values $0,1,\ldots
d-1$ and, as before, set
\begin{align}
&\text{Prob} \left (X = j, Y = k|a,b,\Proj{\Phi^+}_{AB} \right ) \nonumber\\
&\quad\quad=
\Tr[AB]{\Proj{\phi^a_j}_A \otimes \Proj{\phi^b_k}_B \Proj{\Phi^+}_{AB}}.
\end{align}
Now, condition these probabilities on the event $C = \{X \in \{0,1\}
\cap Y \in \{0,1\}\}$, i.e.\ define
\begin{align}
&\text{Prob} \left (X = j, Y = k|a,b,\Proj{\Phi^+}_{AB}, C \right )
\nonumber\\
&\quad= \frac{\text{Prob} \left (X = j, Y = k|a,b,\Proj{\Phi^+}_{AB}
\right )}{\sum_{j^{\prime},k^{\prime} = 0}^1 \text{Prob} \left
(X = j, Y = k|a,b,\Proj{\Phi^+}_{AB} \right )},
\end{align}
for $j,k \in \{0,1\}$.  Straightforward calculation shows that
$\text{Prob} \left (X = j, Y = k|a,b,\Proj{\Phi^+}_{AB}, C \right )$
is equal to the $d=2$ case of the unconditioned probability
distribution.

Similarly, at the ontological level, for $j,k \in \{0,1\}$, we can
define the conditional probability distribution
\begin{align}
&P(X=j,Y=k|a,b,\lambda,C) \nonumber\\
&\quad=
\frac{1}{N(\lambda,M_A^a,M_B^b)}\text{Pr}_{AB} \left ( \Proj{\phi^a_j} ,
\Proj{\phi^b_k} \middle | M_A^a,M_B^b, \lambda \right ),
\end{align}
where
\begin{equation}
N(\lambda,M_A^a,M_B^b) = \sum_{j,k = 0}^1
\text{Prob} \left ( \Proj{\phi^a_{j}}, \Proj{\phi^b_{k}} \middle |
M_A^a, M_B^b \lambda \right ).
\end{equation}
Given that the model reproduces the quantum predictions, we must
have
\begin{align}
&\text{Prob} \left (X = j, Y = k|a,b,\Proj{\Phi^+}_{AB}, C \right )
\nonumber\\
&= \int_{\Lambda_{AB}} \frac{1}{N(M_A^a,M_B^b,\lambda)}
\text{Pr}_{AB} \left ( \Proj{\phi^a_{j}} , \Proj{\phi^b_{k}} \middle
| M_A^a, M_B^b, \lambda \right ) d\mu(\lambda),
\end{align}
for $j,k \in \{0,1\}$ and $\mu \in \Delta_{AB} \left [
\Proj{\Phi^+}_{AB}\right ]$.  Then, Theorem~\ref{thm:CR:qubit}
implies that $\frac{1}{N(M_A,M_B,\lambda)}\text{Pr}_A \left (
\Proj{0}_A \middle | M_A, \lambda \right ) =
\frac{1}{N(M_A,M_B,\lambda)} \text{Pr}_A \left ( \Proj{1}_B \middle |
M_A, \lambda \right )$ for a measure one set $\Omega_{0,1}$
according to $\mu$ and hence $\text{Pr}_A \left ( \Proj{0}_A \middle |
M_A, \lambda \right ) = \text{Pr}_A \left ( \Proj{1} \middle | M_A,
\lambda \right )$ for $\lambda \in \Omega_{0,1}$.

Now, the same argument can be repeated by placing the nontrivial
part of the chained Bell measurements on a different subspace, e.g.\
on the subspace spanned by $\Ket{1}$ and $\Ket{2}$ rather than the
subspace spanned by $\Ket{0}$ and $\Ket{1}$ and conditioning on $\{X
\in \{1,2\} \cap Y \in \{1,2\}\}$ rather than $\{X \in \{0,1\} \cap
Y \in \{0,1\}\}$.  In total, this shows that, for all $j,k \in
\{0,1,\ldots,d-1\}$ there exists sets $\Omega_{j,k}$ that are
measure one according to $\mu$ and such that $\text{Pr} \left (
\Proj{j}_A \middle | M_A, \lambda \right ) = \text{Pr} \left (
\Proj{k}_A \middle | M_A, \lambda \right )$ for $\lambda \in
\Omega_{j,k}$.  Now, $\Omega = \cap_{j,k = 0}^{d-1} \Omega_{j,k}$ is
also measure one according to $\mu$ because it is the intersection
of a finite collection of measure one sets.  This implies that, for
every $j$ and $k$, we have $\text{Pr}_A \left ( \Proj{j}_A \middle |
M_A, \lambda \right ) = \text{Pr}_A \left ( \Proj{k}_A \middle |
M_A, \lambda \right )$ for $\lambda \in \Omega$.  Finally,
$\sum_{j=0}^{d-1} \text{Pr}_A \left ( \Proj{j}_A \middle | M_A,
\lambda \right ) = 1$ implies that $\text{Pr}_A \left ( \Proj{j}_A
\middle | M_A, \lambda \right ) = \frac{1}{d}$ for $\lambda \in
\Omega$.
\end{proof}

\subsection{The main result}

\label{CR:Main}

\begin{theorem}
\label{thm:CR:CR}
Let $\mathfrak{F}_{AB} = \langle \Hilb[A] \otimes \Hilb[B],
\mathcal{P}_{AB}, \mathcal{M}_A \times \mathcal{M}_B \rangle$ be a
product measurement fragment where $\Hilb[A] = \Hilb[B] =
\mathbb{C}^d$ for $d \geq 3$, $\mathcal{P}_{AB}$ contains all pure
states on $\Hilb[A] \otimes \Hilb[B]$, $\mathcal{M}_A$ contains all
measurements in orthonormal bases on $\Hilb[A]$, and $\mathcal{M}_B$
contains all measurements in orthonormal bases on $\Hilb[B]$.

Let $\Theta_{AB} = (\Lambda_{AB}, \Sigma_{AB}, \Delta_{AB},
\Xi_{AB})$ be an ontological model of $\mathfrak{F}_{AB}$ that
satisfies parameter independence, preserves ontological distinctness
with respect to all unitaries on $\Hilb[A] \otimes \Hilb[B]$, and
reproduces the quantum predictions.  Then, any pair of pure states
$\Proj{\psi}_{AB}, \Proj{\phi}_{AB} \in \mathcal{P}_{AB}$ that
satisfies $\Tr[AB]{\Proj{\phi}_{AB}\Proj{\psi}_{AB}} \leq
\frac{d-1}{d}$ is ontologically distinct in $\Theta_{AB}$.
\end{theorem}

Before proving this theorem note that it assumes that $\Theta_{AB}$
preserves ontological distinctness with respect to unitary
transformations (see Definition~\ref{def:Dyn:uind}).  Recall that this
can be derived from the requirement that unitary transformations
should be represented by stochastic transformations on the ontic state
space, which is a basic framework assumption of the ontological models
framework for PMT fragments as discussed in \S\ref{Dyn}.

Also, note that, as with Hardy's Theorem, assuming that appending
ancillas preserves ontological distinctness allows this to be
converted into a full blown $\psi$-ontology theorem.  In this case,
there are two issues to deal with.  The first is that the condition on
the inner product is not satisfied for all states and the second is
that not all fragments have the tensor product structure of a product
measurement fragment.  To deal with the first issue, consider a
fragment $\mathfrak{F}_A = \langle \Hilb[A], \mathcal{P}_A,
\mathcal{M}_A \rangle$ where $\Hilb[A] = \mathbb{C}^d$,
$\mathcal{P}_A$ contains all pure states on $\Hilb[A]$, and
$\mathcal{M}_A$ contains all measurements in orthonormal bases on
$\Hilb[A]$.  A given pair of pure states, $\Proj{\psi}_A$ and
$\Proj{\phi}_A$, either satisfy the condition $\Tr[A]{\Proj{\phi}_A
\Proj{\psi}_A} \leq \frac{d-1}{d}$ or they do not.  If they do not
then we can introduce a product fragment $\mathfrak{F}_{AA^{\prime}} =
\langle \Hilb[A] \otimes \Hilb[A^{\prime}], \mathcal{P}_{AA^{\prime}},
\mathcal{M}_{AA^{\prime}} \rangle$, where we again assume that
$\mathcal{P}_{AA^{\prime}}$ contains all pure states and
$\mathcal{M}_{AA^{\prime}}$ contains all measurements in orthonormal
bases.  The dimension of $\Hilb[A^{\prime}]$ is chosen such that
$\Proj{\psi}_{A} \otimes \Proj{0}_{A^{\prime}}$ and $\Proj{\phi}_A
\otimes \Proj{0}_{A^{\prime}}$ satisfy the inner product condition,
where $\Proj{0}_{A^{\prime}}$ is a fixed pure state.  If we assume
that an ontological model of $\mathfrak{F}_{AA^{\prime}}$ must
preserve ontological distinctness with respect to an ontological model
of $\mathfrak{F}_A$, then proving the ontological distinctness of
$\Proj{\psi}_A \otimes \Proj{0}_{A^{\prime}}$ and $\Proj{\phi}_A
\otimes \Proj{0}_{A^{\prime}}$ is enough to establish the ontological
distinctness of $\Proj{\psi}_A$ and $\Proj{\phi}_A$ in the original
model.

To deal with the issue of tensor product structure, first of all, to
avoid cluttered notation, relabel $AA^{\prime}$ as $A$.  Then, we can
take $\mathfrak{F}_B = \langle \Hilb[B], \mathcal{P}_B, \mathcal{M}_B
\rangle$ to be a second copy of $\mathfrak{F}_A$, i.e.\ it has the
same Hilbert space, states and measurements, but just a different
system label.  Then, form a product measurement fragment
$\mathfrak{F}_{AB} = \langle \Hilb[A] \otimes \Hilb[B],
\mathcal{P}_{AB}, \mathcal{M}_A \times \mathcal{M}_B \rangle$ with
factors $\mathfrak{F}_A$ and $\mathfrak{F}_B$, where
$\mathcal{P}_{AB}$ contains all pure states on $\Hilb[A] \otimes
\Hilb[B]$.  If $\Tr[A]{\Proj{\phi}_A\Proj{\psi}_A} \leq \frac{d-1}{d}$
then $\Tr[AB]{\Proj{\phi}_A \otimes \Proj{0}_B\Proj{\psi}_A \otimes
\Proj{0}_B} = \Tr[A]{\Proj{\phi}_A\Proj{\psi}_A} \leq \frac{d-1}{d}$
for any fixed pure state $\Proj{0}_B \in \mathcal{P}_B$.  Assuming
that an ontological model of $\mathfrak{F}_{AB}$ must preserve
ontological distinctness with respect to an ontological model of
$\mathfrak{F}_A$ means that proving ontological distinctness of
$\Proj{\psi}_A \otimes \Proj{0}_B$ and $\Proj{\phi}_A \otimes
\Proj{0}_B$ is enough to prove the ontological distinctness of
$\Proj{\psi}_A$ and $\Proj{\phi}_A$.  If we assume further that an
ontological model of $\mathfrak{F}_{AB}$ must satisfy the conditions
of the theorem, i.e.\ parameter independence, preserving ontological
distinctness with respect to unitary transformations, and reproducing
the quantum predictions, then any ontological model of
$\mathfrak{F}_A$ must be $\psi$-ontic.

\begin{proof}[Proof of Theorem~\ref{thm:CR:CR}]
By choosing the global phases appropriately, we can find vector
representatives, $\Ket{\psi}_{AB}$ and $\Ket{\phi}_{AB}$, of
$\Proj{\psi}_{AB}$ and $\Proj{\phi}_{AB}$ such that
$\BraKet{\psi}{\phi}_{AB}$ is real and positive.  Since
$\Tr{\Proj{\phi}_{AB}\Proj{\psi}_{AB}} \leq \frac{d-1}{d}$, we have
$0 \leq \BraKet{\phi}{\psi}_{AB} \leq \sqrt{\frac{d-1}{d}}$.

Let $\left \{ \Ket{j} \right \}_{j=0}^{d-1}$ be an orthonormal basis
for $\mathbb{C}^d$ and consider the vectors
\begin{equation}
\Ket{\Phi^+}_{AB} = \frac{1}{\sqrt{d}}\sum_{j=0}^{d-1} \Ket{j}_A
\otimes \Ket{j}_B,
\end{equation}
and
\begin{equation}
\Ket{\eta}_{AB} = \frac{1}{\sqrt{d-1}}\sum_{j=1}^{d-1} e^{i
\varphi_j} \Ket{j}_A \otimes \Ket{j}_B,
\end{equation}
By Lemma~\ref{lem:Hardy:inner}, the phases $\varphi_j$ can be chosen
in such a way that $\BraKet{\eta}{\Phi^+}_{AB} =
\BraKet{\phi}{\psi}_{AB}$.  Thus, there exists a unitary operator
$U$ such that $U \Ket{\psi}_{AB} = \Ket{\Phi^+}$ and
$U\Ket{\phi}_{AB} = \Ket{\eta}_{AB}$.  Given that $\Theta_{AB}$
preserves ontological distinctness with respect to unitaries, if
$\Proj{\Phi^+}_{AB}$ and $\Proj{\eta}_{AB}$ are ontologically
distinct in $\Theta_{AB}$ then $\Proj{\psi}_{AB}$ and
$\Proj{\phi}_{AB}$ must also be ontologically distinct.

Let $\mu \in \Delta_{AB} \left [ \Proj{\Phi^+}_{AB} \right ]$ and
$\nu \in \Delta_{AB} \left [ \Proj{\eta}_{AB} \right ]$.  Consider
the measurement $M_A = \left \{ \Proj{j}_A \right \}_{j=0}^{d-1}$.
By Theorem~\ref{thm:CR:qudit}, any conditional probability
distribution $\text{Pr}_A \in \Xi_A[M_A]$ must satisfy $\text{Pr}_A
\left ( \Proj{0}_A \middle | M_A, \lambda \right ) = 1/d$ on a set
$\Omega$ that is measure one according to $\mu$.  However, in order
to reproduce the quantum predictions, we must have
\begin{equation}
\int_{\Lambda_{AB}} \text{Pr}_A \left ( \Proj{0}_A \middle | M_A,
\lambda \right ) d\nu(\lambda) = 0,
\end{equation}
because $\Ket{\eta}_{AB}$ is orthogonal to $\Ket{0}_A$.  Thus,
$\text{Pr}_A \left ( \Proj{0}_A \middle | M_A, \lambda \right )$ can
only be nonzero on a set that is of measure zero according to $\nu$.
Hence, $\mu(\Omega) = 1$ and $\nu(\Omega) = 0$, so
$\Proj{\Phi^+}_{AB}$ and $\Proj{\eta}_{AB}$ are ontologically
distinct.
\end{proof}

\section*{Part III. Beyond the $\psi$-ontic/epistemic distinction\label{Beyond}}

\hyperref[SPON]{Part II} presented three theorems that aimed to rule out
$\psi$-epistemic ontological models.  In each case, auxiliary
assumptions were required to prove the theorem: the preparation
independence postulate for the Pusey--Barrett--Rudolph Theorem, ontic indifference for
Hardy's Theorem, and parameter independence for the Colbeck--Renner
Theorem.  Each of these auxiliary assumptions are question begging.
We also saw, in \S\ref{NPIP}, that $\psi$-epistemic theories are
possible if we do not make any auxiliary assumptions.

However, as noted in \S\ref{POEM}, the definition of a
$\psi$-epistemic model is highly permissive in that it only requires
probability measures corresponding to two states to have overlap in an
ontological model.  Most pairs of nonorthogonal states may still be
ontologically distinct and furthermore the overlaps that do exist may
be arbitrarily small.  Clearly, more than this is needed if the
$\psi$-epistemic explanations of quantum phenomena like the
indistinguishability of nonorthogonal states are to be viable.
Therefore, one might hope that viable $\psi$-epistemic interpretations
might be ruled out without making auxiliary assumptions by imposing
stronger requirements on the overlaps in ontological models.  This
part discusses four proposals for such requirements.

\S\ref{Pair} discusses the requirement that all pairs of nonorthogonal
states should correspond to measures that have nonzero overlap.
Aaronson et.\ al.\ \cite{Aaronson2013} have shown that such models
exist in all dimensions, but they have also shown that imposing
additional assumptions can rule them out.  One of these is that the
model should satisfy a fairly reasonable symmetry requirement, but the
other severely restricts the class of ontic state spaces. \S\ref{Cont}
discusses the requirement that ontological models should be
continuous, in the sense that quantum states with large inner product
should correspond to measures with large overlap.  Two different
notions of continuity have been discussed in the literature.  Patra,
Pironio and Massar have proven a theorem ruling out models that
satisfy a very strong notion of continuity \cite{Patra2013a}.
However, their notion of continuity is so strong that it rules out
many quite reasonable $\psi$-epistemic models.  This is
discussed\S\ref{PPM}.  \S\ref{Lips} discusses a version of Lipschitz
continuity, which says that the ratio of overlaps in an ontological
model to the quantum probabilities should be bounded.  This can be
used as the basis of measuring the degree to which a model is
$\psi$-epistemic by comparing the overlaps of measures in the
ontological model with the inner products of the corresponding quantum
states.  The idea here is that if the former is small for a pair of
states with large inner product then the $\psi$-epistemic explanation
of the indistinguishability of quantum states is not really viable.
Several results bounding these overlaps have recently been obtained
\cite{Maroney2012, Barrett2013, Leifer2014, Branciard2014},
culminating in a proof by Branciard \cite{Branciard2014} that
Lipschitz continuity is impossible in an ontological model that
reproduces the quantum predictions in Hilbert space dimension $\geq
4$.  Finally, \S\ref{Never} discusses the concept of a
\emph{sometimes} $\psi$-ontic model and its converse, a \emph{never}
$\psi$-ontic model.  In a \emph{sometimes} $\psi$-ontic model, for
each pure state there is a region of the ontic state space to which
its probability measures assign nonzero measure, but this region is
assigned zero measure by the probability measures corresponding to all
other pure states.  Thus, if the system happens to occupy an ontic
state in one of these regions then the quantum state that was prepared
can be deduced uniquely.  The system need not always occupy such an
ontic state though; hence the terminology ``sometimes'' $\psi$-ontic.
In the context of reproducing the predictions of all projective
measurements for all pure states on $\mathbb{C}^d$, all
$\psi$-epistemic models that have been proposed to date are sometimes
$\psi$-ontic for $d \geq 3$, but the question of whether this must
necessarily be the case is open.

Even if you think the additional assumptions of the $\psi$-ontology
theorems discussed in \hyperref[SPON]{Part II} are reasonable, there are still
important reasons to look for stronger notions of $\psi$-epistemicity
that might be ruled out without them.  In \S\ref{Imp}, we discussed
the implications of assuming that the quantum state is $\psi$-ontic,
ranging from excess baggage to preparation contextuality and
nonlocality.  However, if your claim of $\psi$-ontology is based on
the existing theorems then the implications from $\psi$-ontology to
these other features of ontological models inherit the auxiliary
assumptions made in the existing $\psi$-ontology theorems.  Since we
can already prove these things without such additional assumptions,
e.g.\ the standard proof of Bell's Theorem does not require assuming
the preparation independence postulate or making assumptions about how
dynamics is represented, the claim that $\psi$-ontology theorems are
the strongest known no-go theorems for ontological models is somewhat
weakened.  From this point of view, proving that models must be
sometimes $\psi$-ontic would be particularly interesting, since all of
the implications discussed in \S\ref{Imp} would follow from it.  This
is also discussed in \S\ref{Never}.

The discussion in this part is less detailed than in Parts~\ref{OED}
and \ref{SPON}.  This is because going beyond the
$\psi$-ontic/$\psi$-epistemic distinction is a relatively new
development and many of the known results are only provisional.  I
expect that some of the material in this section will be obsolete by
the time this review is published, so it is more important to outline
the main ideas than to give detailed proofs.

\section{Pairwise $\psi$-epistemic models}

\label{Pair}

One of the deficiencies in the definition of a $\psi$-epistemic model
is that it only requires that a single pair of pure states is
ontologically indistinct.  Clearly, in order to explain the
indistinguishability of every pair of nonorthogonal states in terms of
the overlap of the corresponding probability measures, all such pairs
should to be ontologically indistinct.  This is what is meant by a
pairwise $\psi$-epistemic model.

\begin{definition}
Let $\mathfrak{F} = \langle \Hilb, \mathcal{P}, \mathcal{M} \rangle$
be a PM fragment and let $\Theta = (\Lambda, \Sigma, \Delta, \Xi)$ be
an ontological model of it.  $\Theta$ is \emph{pairwise
$\psi$-epistemic} if all pairs of nonorthogonal pure states in
$\mathcal{P}$ are ontologically indistinct in $\Theta$.
\end{definition}

Aaronson, Bouland, Chua and Lowther (ABCL) have shown that pairwise
$\psi$-epistemic models exist for the fragment consisting of the set
of all pure states and all measurements in orthonormal bases in any
finite dimensional Hilbert space \cite{Aaronson2013}.  They also prove
a theorem ruling out such models by imposing some additional
assumptions, which generalize a class of models considered by Rudolph
\cite{Rudolph2006} who previously provided numerical evidence that
they could not reproduce quantum theory.  One of their assumptions is
a quite reasonable symmetry requirement, but the other is that the
ontic state space describing a system with Hilbert space
$\mathbb{C}^d$ is either the projective Hilbert space, i.e.\ the set
of pure states on $\mathbb{C}^d$, or the set of unitary operators on
$\mathbb{C}^d$.  This is obviously a very restricted setting, so their
theorem should be regarded as a provisional step towards ruling out
pairwise $\psi$-epistemic models by symmetry requirements.  Here, I
give a rough outline of their construction of a pairwise
$\psi$-epistemic theory in finite dimensional Hilbert spaces and give
a brief account of the symmetry assumption behind their theorem.

\S\ref{NPIP} described ABCL's construction of an ontological model
wherein any fixed pair of nonorthogonal pure states could be made
ontologically indistinct.  The key to constructing a pairwise
$\psi$-epistemic model is to find a way of mixing ontological models
together such that any states that are ontologically indistinct in the
original models remain so in the mixture.  Then, if one mixes
sufficiently many ABCL models the result will be pairwise
$\psi$-epistemic.  A mixture of two ontological models is defined as
follows.

\begin{definition}
Let $\mathfrak{F} = \langle \Hilb, \mathcal{P}, \mathcal{M} \rangle$
be a PM fragment and let $\Theta_1 = (\Lambda_1, \Sigma_1, \Delta^1,
\Xi^1)$ and $\Theta_2 = (\Lambda_2, \Sigma_2, \Delta^2, \allowbreak
\Xi^2)$ be ontological models of it.  The model $\Theta_3 =
(\Lambda_3, \Sigma_3, \Delta^3, \Xi^3) = p\Theta_1 + (1-p)\Theta_2$,
where $0 \leq p \leq 1$, is defined as follows.
\begin{itemize}
\item $\Lambda_3 = \Lambda_1 \oplus \Lambda_2$ and $\Sigma_3$
consists of all sets of the form $\Omega_1 \oplus \Omega_2$ for
$\Omega_1 \in \Sigma_1$, $\Omega_2 \in \Sigma_2$.
\item For all $\rho \in \mathcal{P}$, $\Delta^3_{\rho}$ consists of
all measures of the form
\begin{equation}
p \mu_1 + (1-p) \mu_2,
\end{equation}
for $\mu_1 \in \Delta^1_{\rho}$ and $\mu_2 \in \Delta^2_{\rho}$.
\item For all $M \in \mathcal{M}$, $\Xi^3_M$ consists of all
conditional probability distributions of the form
\begin{equation}
\label{eq:Pair:CondProb}
\text{Pr}_3 \left ( E \middle | M, \lambda \right ) =
\text{Pr}_1 \left ( E \middle | M, \lambda \right ) +
\text{Pr}_2 \left ( E \middle | M, \lambda \right ),
\end{equation}
for $\text{Pr}_1 \in \Xi^1_M$ and $\text{Pr}_2 \in \Xi^2_M$.
\end{itemize}
\end{definition}

A few words of clarification are in order here.  First of all, the
direct sum of two spaces is essentially just the two spaces placed
side-by-side.  More formally, $\Lambda_1 \oplus \Lambda_2$ is the
union of the set of objects of the form $(1,\lambda_1)$ with the set
of objects of the form $(2,\lambda_2)$.  This is different from the
union of $\Lambda_1$ and $\Lambda_2$, since $\Lambda_1$ and
$\Lambda_2$ might contain some common elements.  The addition of the
integer label is to ensure that we have two copies of such common
elements, with the label indicating whether it is the copy that
belongs to $\Lambda_1$ or to $\Lambda_2$.  However, to simplify
notation, I generally omit the integer label.  When $\mu_1$ is a
probability measure on $(\Lambda_1,\Sigma_1)$, I also denote by
$\mu_1$ the probability measure on $(\Lambda_3,\Sigma_3)$ that
satisfies $\mu_1(\Lambda_2) = 0$ and agrees on $\Sigma_1$.  Similarly,
when $\text{Pr}_1$ is a probability distribution conditioned on
$\Lambda_1$ the same notation denotes the function on $\Lambda_3$ that
agrees on $\Lambda_1$ and is zero on $\Lambda_2$.  Note, this is not a
valid conditional probability distribution because $\sum_{E \in M}
\text{Pr}_1 \left ( E \middle | M, \lambda_2 \right ) = 0$ for
$\lambda_2 \in \Lambda_2$, but when added to a probability
distribution conditioned on $\Lambda_2$, as in
Eq.~\eqref{eq:Pair:CondProb}, the result is a valid conditional
probability distribution.  Similar remarks apply under the exchange of
$1$ and $2$.  It is then straightforward to check that if $\Theta_1$
and $\Theta_2$ reproduce the quantum predictions then so does
$\Theta_3$.

\begin{theorem}
\label{the:Pair:ind}
Let $\mathfrak{F} = \langle \Hilb, \mathcal{P}, \mathcal{M} \rangle$
be a PM fragment and and let $\Theta_1 = (\Lambda_1, \Sigma_1,
\Delta^1, \Xi^1)$ and $\Theta_2 = (\Lambda_2, \Sigma_2, \Delta^2,
\Xi^2)$ be ontological models of it.  Let $\Theta_3 = p \Theta_1 +
(1-p)\Theta_2$ for some $0 < p < 1$.  If $\Proj{\psi}$ and
$\Proj{\phi}$ are ontologically indistinct in either $\Theta_1$ or
$\Theta_2$ then they are ontologically in distinct in $\Theta_3$.
\end{theorem}
\begin{proof}
Without loss of generality, assume that $\Proj{\psi}$ and
$\Proj{\phi}$ are ontologically indistinct in $\Theta_1$, since the
proof for $\Theta_2$ follows the same logic.  Then, there exist $\mu_1
\in \Delta^1_{\psi}$ and $\nu_1 \in \Delta^1_{\phi}$ such that
$c = D(\mu_1, \nu_1) < 1$, where $D$ is the variational distance.  Let
$\mu_2 \in \Delta^2_{\psi}$ and $\nu_2 \in \Delta^2_{\phi}$.  Then,
$\mu_3 = p \mu_1 + (1-p)\mu_2 \in \Delta^3_{\psi}$ and $\nu_3 = p
\nu_1 + (1-p) \nu_2 \in \Delta^3_{\phi}$.  Let $\Omega \in \Sigma_3$
and define $\Omega_1 = \Omega \cap \Lambda_1$ and $\Omega_2 = \Omega
\cap \Lambda_2$.  Then, for any measure $\mu$ on $(\Lambda_3,
\Sigma_3)$, $\mu(\Omega) = \mu(\Omega_1) + \mu(\Omega_2)$, but we
also have $\mu_1(\Omega_2) = \nu_1(\Omega_2) = \mu_2(\Omega_1) =
\nu_2(\Omega_1) = 0$.  Hence,
\begin{align}
&\left | \mu_3(\Omega) - \nu_3(\Omega) \right | \nonumber\\
& = \left |
p\mu_1(\Omega_1) + p\nu_1(\Omega_1) +
(1-p)\mu_2(\Omega_2) + (1-p)\nu_2(\Omega_2) \right | \\
& \leq p \left | \mu_1(\Omega_1) - \nu_1(\Omega_1) \right | +
(1-p) \left | \mu_2(\Omega_2) - \nu_2(\Omega_2) \right | \\
& \leq p c + (1-p) \\
& = 1 - (1-c)p,
\end{align}
and this is $< 1$ because $0 \leq c < 1$ and $p > 0$.  Since this
holds for any $\Omega \in \Sigma_3$ it also holds for the supremum.
Hence, $D(\mu_3,\nu_3) < 1$ and so $\Proj{\psi}$ and $\Proj{\phi}$
are ontologically indistinct in $\Theta_3$.
\end{proof}

The rough idea of the ABCL construction is now quite simple to
describe.  Let $\mathfrak{F} = \langle \mathbb{C}^d, \mathcal{P},
\mathcal{M} \rangle$ be the PM fragment where $\mathcal{P}$ consists
of all pure states on $\mathbb{C}^d$ and $\mathcal{M}$ consists of all
measurements in orthonormal bases.  In \S\ref{NPIP}, we showed that,
for any pair of nonorthogonal states $\Proj{\psi}, \Proj{\phi} \in
\mathcal{P}$, there exists an ontological model $\Theta_{\psi,\phi}$ in
which $\Proj{\psi}$ and $\Proj{\phi}$ are ontologically indistinct.
Therefore, by Theorem~\ref{the:Pair:ind}, the model $p \Theta_{\psi,\phi}
+ (1-p) \Theta_{\psi^{\prime},\phi^{\prime}}$ has both $\left (
\Proj{\psi},\Proj{\phi} \right )$ and $\left ( \Proj{\psi^{\prime}},
\Proj{\phi^{\prime}}\right )$ as ontologically indistinct pairs.
This construction can then be iterated in order to obtain a model in
which any finite set of pairs of nonorthogonal pure states are
ontologically distinct.

The remaining difficulty is to extend this construction to an infinite
number of pairs, i.e.\ all pairs of pure states on $\mathbb{C}^d$.
The details of this are quite technical but straightforward once one
understands the basic idea.  The interested reader should consult
\cite{Aaronson2013}.

ABCL go on to prove a theorem ruling out pairwise $\psi$-epistemic
models under an additional symmetry requirement.  They also impose
that the ontic state space for a quantum system with Hilbert space
$\mathbb{C}^d$ is either the projective Hilbert space of
$\mathbb{C}^d$ or the set of unitary operators on $\mathbb{C}^d$.  In
the more general context of arbitrary models, their symmetry
assumption runs as follows.
\begin{definition}
Let $\mathfrak{F} = \langle \Hilb, \mathcal{P}, \mathcal{M},
\mathcal{T} \rangle$ be a PMT fragment and let $\Theta = (\Lambda,
\Sigma, \Delta, \Xi, \Gamma)$ be an ontological model of it.  $\Theta$
satisfies \emph{ABCL symmetry} if, for all pairs $\Proj{\psi} \in
\mathcal{P}$ and $U \in \mathcal{T}$ such that $U \Proj{\psi}
U^{\dagger} = \Proj{\psi}$, there exists a $\mu \in \Delta_{\psi}$
and a $\gamma \in \Gamma_{U}$ such that
\begin{equation}
\mu(\Omega) = \int_{\Lambda} \gamma_{\lambda}(\Omega) d \mu(\lambda).
\end{equation}
\end{definition}
This says that, if a unitary leaves a pure state invariant, then there
should be a probability measure representing that state and a Markov
kernel representing the unitary such that the measure is invariant
under the stochastic map representing the unitary.  ABCL then show
that, for the fragment consisting of all pure states, all measurements
in orthonormal bases, and all unitaries on $\mathbb{C}^d$, if the
ontic state space is either the projective Hilbert space of
$\mathbb{C}^d$ or the set of unitary operators on $\mathbb{C}^d$ and
the dynamics is represented by the usual action of the unitary group
on these spaces, then no ABCL symmetric and pairwise $\psi$-epistemic
model exists.

Imposing additional symmetry requirements of the type suggested by
ABCL is a promising direction for $\psi$-ontology theorems, but their
setup is so restricted that one cannot really draw any firm
conclusions about ontological models in general from it.  The proof
itself is quite technical and not obviously generalizable to other
ontic state spaces, so the details are omitted here.

\section{Continuity}

\label{Cont}

The idea of continuity is that if two pure quantum states are close to
one another, in the sense of having large inner product, then there
should be probability measures representing them that have large
overlap.  Imposing continuity is not a bad idea, since presumably
models in which a small change in experimental conditions leads to a
large change at the ontological level would be quite contrived.  The
most basic notion of continuity runs as follows.

\begin{definition}
\label{def:Cont:cont}
Let $\mathfrak{F} = \langle \Hilb, \mathcal{P}, \mathcal{M} \rangle$
be a PM fragment where $\mathcal{P}$ consists of pure states and let
$\Theta = (\Lambda, \Sigma, \Delta, \Xi)$ be an ontological model of
it.  $\Theta$ is \emph{continuous} if, for all $\epsilon > 0$, there
exists a $\delta > 0$ such that if $\Tr{\Proj{\phi}\Proj{\psi}} > 1
- \delta$ for some pure states $\Proj{\psi}, \Proj{\phi} \in
\mathcal{P}$, then there exists $\mu \in \Delta_{\psi}$ and $\nu \in
\Delta_{\phi}$ such that $L(\mu,\nu) > 1 - \epsilon$.
\end{definition}

This notion of continuity is quite permissive because it imposes no
constraints on how $\delta$ should be related to $\epsilon$.  It is
easy to see that any pairwise $\psi$-epistemic model is continuous in
this sense, since all it requires is that the overlap of measures
corresponding to nonorthogonal states should be bounded away from
zero.  Since ABCL have shown that pairwise $\psi$-epistemic models
exist, this notion of continuity cannot be used to derive a no-go
theorem.

Therefore, it is interesting to look at stronger notions of continuity
that might be reasonable for a $\psi$-epistemic model.  In the
remainder of this section, two such notions are discussed.
\S\ref{PPM} discusses a notion of continuity due to Patra, Pironio and
Massar \cite{Patra2013a} and discusses a theorem that they proved to
rule such models out.  I argue that their definition of continuity is
unreasonable for a $\psi$-epistemic model.  \S\ref{Lips} discusses a
different notion based on Lipschitz continuity, which has been used in
a number of recent works to derive measures of the extent to which a
model is $\psi$-epistemic.  The best known bound shows that the ratio
of the overlap to the inner product must be zero for Hilbert spaces of
dimension $\geq 4$.

\subsection{The Patra--Pironio--Massar Theorem}

\label{PPM}

PPM's continuity assumption runs as follows.

\begin{definition}
Let $\mathfrak{F} = \langle \Hilb, \mathcal{P}, \mathcal{M} \rangle$
be a PM fragment, where $\mathcal{P}$ consists of pure states, let
$\delta > 0$ and, for $\Proj{\psi} \in \mathcal{P}$, let
$B^{\delta}_{\psi} = \left \{ \Proj{\phi} \in \mathcal{P} \middle |
\sqrt{\Tr{\Proj{\phi}\Proj{\psi}}} \geq 1 - \delta \right \}$ be
the ball of radius $\delta$ centered at $\Proj{\psi}$.  An
ontological model $\Theta = (\Lambda, \Sigma, \Delta, \Xi)$ of
$\mathfrak{F}$ is \emph{PPM-$\delta$-continuous} if, for every pure
state $\Proj{\psi} \in \mathcal{P}$ and every set $\{\Proj{\phi_j}\}
\subseteq \mathcal{P}$ where $\Proj{\phi_j} \in B^{\delta}_{\psi}$,
there exist probability measures $\mu_j \in \Delta_{\phi_j}$ such
that $L(\{\mu_j\}) > 0$.
\end{definition}

For the case of a finite ontic state space, the
PPM-$\delta$-continuity assumption says that, for every set of quantum
states that are sufficiently close some common state $\Proj{\psi}$,
there should be at least one ontic state to which they all assign
nonzero probability.  This is stronger than the basic notion of
continuity given in Definition~\ref{def:Cont:cont}, which only
constrains pairwise overlaps.

Based on their assumption, PPM prove the following theorem.
\begin{theorem}
\label{thm:PPM:PPM}
Let $\mathfrak{F} = \langle \mathbb{C}^d, \mathcal{P}, \mathcal{M}
\rangle$ be a PM fragment with $\left \{ \Proj{k} \right
\}_{k=0}^{d-1} \in \mathcal{M}$ for some orthonormal basis $\left \{
\Ket{k} \right \}_{k=0}^{d-1}$.  If $\Proj{\psi} \in \mathcal{P}$
and $\left \{ \Proj{\phi_j} \right \}_{j=0}^{d-1} \subseteq
\mathcal{P}$, where
\begin{align}
\Ket{\psi} & = \frac{1}{\sqrt{d}}\sum_{k=0}^{d-1} \Ket{k} \\
\Ket{\phi_j} & = \frac{1}{\sqrt{d-1}} \sum_{k \neq j} \Ket{k},
\end{align}
then there is no PPM-$\delta$-continuous ontological model of
$\mathfrak{F}$ that reproduces the quantum preclusions for $\delta
\geq 1 - \sqrt{(d-1)/d}$.
\end{theorem}

As with the Hardy and Colbeck--Renner Theorems, the constraint on
$\delta$ can be removed by assuming that adding ancillas preserves
ontological distinctness.

\begin{proof}[Proof of Theorem~\ref{thm:PPM:PPM}]
Let $\Theta = (\Lambda, \Sigma, \Delta, \Xi)$ be an ontological model
of $\mathfrak{F}$ and let $\mu_j \in \Delta_{\phi_j}$.  The states
$\Proj{\phi_j}$ all lie within $B^{\delta}_{\psi}$ for $\delta = 1 -
\sqrt{(d-1)/d}$.  However, the states $\Proj{\phi_j}$ are
antidistinguished by the measurement $\left \{ \Proj{k} \right
\}_{k=0}^{d-1}$ and hence, by Theorem~\ref{thm:Anti:antiover}, $L
\left ( \left \{ \mu_j \right \} \right ) = 0$.
\end{proof}

This proof is remarkably simple, but it does not mean much as
PPM-$\delta$-continuity is unreasonably strong.  In order to see the
problem, it is helpful to consider a simple example.  Therefore,
consider how the fragment $\mathfrak{F} = \langle \mathbb{C}^d,
\mathcal{P}, \mathcal{M} \rangle$ might be modeled, where
$\mathcal{P}$ just contains the states $\Proj{\psi}$, $\Proj{\phi_j}$
and $\mathcal{M}$ only contains the measurement $M =
\{\Proj{k}\}_{k=1}^d$.  A natural way of doing this is to use the
ontic state space $\Lambda = \{0,1,\ldots,d-1\}$, model $\Proj{\psi}$
by the uniform distribution, and model $\Proj{\phi_j}$ by the
distribution that is zero on $j$ and uniformly distributed over the
rest of the ontic state space.  Then, it is straightforward to see
that setting the conditional probabilities $\text{Pr} \left (\Proj{k}
\middle | M, j \right ) = \delta_{jk}$ yields a model that
reproduces the quantum predictions.  Furthermore, this model is
pairwise $\psi$-epistemic, and hence satisfies the more basic notion
of continuity, because $D(\mu_j,\mu_k) = 1/(d-1)$ for $j \neq k$.

From this, it is easy to see that the problem with
PPM-$\delta$-continuity comes from ontological states that are
assigned a small weight in the distribution.  If a distribution
corresponding to some quantum state assigns a small weight to all of
the ontic states in its support then one way of making a small change
to that distribution is to set the weight assigned to one of the ontic
states to zero and redistribute it over the other ontic states.  If
those distributions all happen to represent quantum states that are
close to the original one, and there is no reason why they should not,
then these quantum states will have no ontic state that is common to
all of them, and hence PPM-$\delta$-continuity will be violated.

In fact, the same criticism applies to classical models.  Suppose we
have a system consisting of a ball that can be in one of $d$ boxes,
labeled $0,1,\ldots,d-1$.  The system is prepared in one of the
following ways.  Either a box is chosen uniformly at random and the
ball is placed in that box, or first an integer $j$ from $0$ to $d-1$
is specified and box $j$ is removed, then the ball is placed uniformly
at random in one of the remaining boxes, before finally returning box
$j$ to its place.  Suppose you are given a table which specifies the
probabilities of finding the ball in each of the boxes for each of
these $d+1$ preparation procedures.  You are, however, not told
anything about where these probabilities come from, i.e.\ you are not
told about the ball and boxes but rather just given a table that tells
you that if preparation $j$ is made then the probability of getting
outcome $k$ is $0$ if $k=j$ and $1/(d-1)$ otherwise.  Your task is to
try and come up with a model of what is going on in reality in order
to generate these statistics; a task perfectly analogous to the
project of constructing an ontological model for some fragment of
quantum theory.  If you imposed upon yourself the analog of
PPM-$\delta$-continuity, that sets of preparations having close
operational predictions must have an ontic state in common, you would
never be able to come up with the ball and box model that is actually
generating the probabilities and, in fact, no model satisfying this
condition exists.  That, in a nutshell, is what is wrong with
PPM-$\delta$-continuity.  A condition that rules out a classical
probabilistic description of a classical probabilistic model is
clearly way to strong to prove anything meaningful about quantum
theory.

\subsection{Lipschitz continuity}

\label{Lips}

In order for the $\psi$-epistemic explanation of the
indistinguishability of quantum states to be viable, quantum states
that have large inner product should correspond to probability
measures with large overlap, but quantum states that are almost
orthogonal need have very little overlap at the ontic level.  In
between these extremes, an amount of overlap that scales with the
inner product in some way is needed.  This motivates the definition of
Lipschitz continuity.

\begin{definition}
Let $\mathfrak{F} = \langle \Hilb, \mathcal{P}, \mathcal{M} \rangle$
be a PM fragment where $\mathcal{P}$ consists of pure states and
where, for every $\Proj{\psi} \in \mathcal{P}$, there exists an $M
\in \mathcal{M}$ such that $\Proj{\psi} \in M$.  Let $\Theta =
(\Lambda, \Sigma, \Delta, \Xi)$ be an ontological model of
$\mathfrak{F}$.  $\Theta$ is Lipschitz continuous if there exists a $K
> 0$ such that, for all $\Proj{\psi}, \Proj{\phi} \in \mathcal{P}$,
\begin{equation}
L(\mu,\nu) \geq K \Tr{\Proj{\phi}\Proj{\psi}},
\end{equation}
for some $\mu \in \Delta_{\psi}$, $\nu \in \Delta_{\phi}$.
\end{definition}

The first bounds on $K$ were obtained by Maroney \cite{Maroney2012},
followed by Barrett et.\ al.\ \cite{Barrett2013} and myself
\cite{Leifer2014}.  However, Branciard \cite{Branciard2014} has
recently proved a theorem to the effect that, for Hilbert spaces of $d
\geq 4$, $K$ must be zero, or, in other words, Lipschitz continuous
models are not possible.  This could be regarded as a $\psi$-ontology
theorem in its own right, if we regard Lipschitz continuity as a
stronger definition of what it means for a model to be
$\psi$-epistemic.

However, Lipschitz continuity is quite a strong requirement.  It says
that there should be a fixed bound on the ratio of $L(\mu,\nu)$ to
$\Tr{\Proj{\phi}\Proj{\psi}}$, that is independent of the states
$\Proj{\psi}$ and $\Proj{\phi}$ under consideration.  Is this really
required to maintain the viability of $\psi$-epistemic explanations of
quantum phenomena, such as the indistinguishability of quantum states?

In any ontological model, regardless of whether it is $\psi$-epistemic
or $\psi$-ontic, there are two mechanisms that can account for the
indistinguishability of a pair of quantum states.  The first is
overlap of the probability measures corresponding to states, which I
have referred to as the $\psi$-epistemic explanation of
indistinguishability.  Even if you are told the exact ontic state
occupied by the system, your probability of correctly guessing whether
$\Proj{\psi}$ or $\Proj{\phi}$ was prepared is bounded by $\frac{1}{2}
\left ( 2 - L(\mu,\nu) \right )$, where $\mu \in \Delta_{\psi}$ and
$\nu \in \Delta_{\phi}$.  Hence, if $L(\mu,\nu) > 0$, the states are
necessarily indistinguishable.  However, there is also a second
mechanism.  We do not usually have complete knowledge about the ontic
state, but are instead only able to infer information about it via the
outcomes of quantum measurements.  The conditional probability
distributions $\text{Pr}(E|M,\lambda)$ corresponding to measurements
typically do not reveal the value of $\lambda$ exactly, but only give
coarse-grained information about it.  This will render quantum states
less distinguishable than they would be if we knew the exact ontic
state, so this mechanism can also explain part of their
indistinguishability.  In fact, in a $\psi$-ontic model, the
indistinguishability of quantum states must be explained entirely by
this second mechanism as there is no overlap of probability measures.
In general, there will be a tradeoff between the two effects, with
some portion of the indistinguishability of quantum states being
explained by overlap of probability measures, and some being explained
by the coarse-grained nature of measurements.  Clearly, some portion
of the indistinguishability must be explained by overlap in a
$\psi$-epistemic theory, as the theory would be $\psi$-ontic if there
were no overlap, but how much of it ought to be explained in this way?

It is easy to see that measurements must only reveal coarse-grained
information about $\lambda$ in certain types of ontological model,
regardless of whether they are $\psi$-epistemic or $\psi$-ontic.  For
example, in an outcome deterministic model, such as the Kochen--Specker
model, in which the probabilities $\text{Pr}(\phi|M,\lambda)$ for
measurements in orthonormal bases are all either $0$ or $1$, exact
knowledge of $\lambda$ would entail exact knowledge of the outcomes of
all measurements, which would violate the uncertainty principle and
hence cannot be compatible with reproducing the quantum predictions.
Further, recall that Spekkens' toy theory was derived from the
knowledge-balance principle, which states that at most half of the
information needed to specify the ontic state can be known at any
given time.  This also entails that measurements can only reveal
coarse-grained information about the ontic state.  Therefore, even in
an archetypal $\psi$-epistemic theory, which is derived entirely from
epistemic principles, measurements only reveal coarse-grained
information about $\lambda$.  For these reasons, it is to be expected
that both overlap of probability measures and the coarse-grained
nature of measurements will be present in a viable $\psi$-epistemic
theory.

Because of this, it is not obvious that there is some definite value
of the ratio $L(\mu,\nu) / \Tr{\Proj{\phi}\Proj{\psi}}$ below which
$\psi$-epistemic models ought to be discarded as implausible.
Further, it is also not obvious that this ratio should have a fixed
bound, independent of the choice of $\Proj{\psi}$ and $\Proj{\phi}$.
It is conceivable that the amount of indistinguishability explained by
overlap and the amount explained by the coarse-grained nature of
measurements could vary for different pairs of states, e.g.\ it might
vary with $\Tr{\Proj{\phi}\Proj{\psi}}$.  Therefore, it is interesting
to see what can be inferred from the results mentioned above
\cite{Maroney2012, Barrett2013, Leifer2014, Branciard2014} without
assuming Lipschitz continuity.

All of these papers consider some family of states $\mathcal{S}_n =
\{\Proj{\psi_j}\}_{j=1}^n$, i.e.\ there are sets of states with
increasing values of $n$.  Without Lipschitz continuity, the results
provide an upper bound on the ratio
\begin{equation}
R_n = \min_{\Proj{\psi}, \Proj{\phi} \in \mathcal{S}_n}
\frac{L(\mu_{\psi},\nu_{\phi})}{\Tr{\Proj{\phi}\Proj{\psi}}},
\end{equation}
for any choice of $\mu_{\psi} \in \Delta_{\psi}$, $\nu_{\phi} \in
\Delta_{\phi}$.  In other words, we can infer that there exists at
least one pair of states in the set for which the ratio is at least
this bad, but we can no longer infer that this holds for all states in
the Hilbert space as we could under Lipschitz continuity.  Given that
it is difficult to say just how much of the indistinguishability of
quantum states needs to be explained by overlap of probability
measures in a viable $\psi$-epistemic theory, this type of bound is
only really interesting if it shows that
$L(\mu_{\psi},\nu_{\phi})/\Tr{\Proj{\phi}\Proj{\psi}}$ must be close
to zero for a wide range of different pairs of states.  Since the
state families used in these constructions can be unitarily
transformed without changing the result, this means that
$L(\mu_{\psi},\nu_{\phi})/\Tr{\Proj{\phi}\Proj{\psi}}$ should be close
to zero for a wide range of values of $\Tr{\Proj{\phi}\Proj{\psi}}$.

The best bound so far obtained comes from Branciard's work
\cite{Branciard2014}, in which he shows that there exists a family of
states in any Hilbert space of dimension $\geq 4$, such that $R_n
\rightarrow 0$ as $n \rightarrow \infty$.  However, all of the
existing results share the feature that $\min_{\Proj{\psi},\Proj{\phi}
\in \mathcal{S}_n} \Tr{\Proj{\phi}\Proj{\psi}} \rightarrow 0$ as $n
\rightarrow \infty$ as well, so all we can really say without assuming
Lipschitz continuity is that the ratio is close to zero for pairs of
states that are almost distinguishable.  It is difficult to tell the
difference between states that are perfectly distinguishable and
states that are almost distinguishable in a practical experiment, so
it is arguable whether it matters that the $\psi$-epistemic
explanation of indistinguishability plays almost no role for such
states.  It is still an open possibility that the ratio might be much
higher for states with a larger inner product, so what is really
needed is to find families of states for which $R_n \rightarrow 0$ as
$n \rightarrow \infty$, but $\min_{\Proj{\psi},\Proj{\phi} \in
\mathcal{S}_n}\Tr{\Proj{\phi}\Proj{\psi}}$ takes on a wide range of
values.

\section{Never $\psi$-ontic models}

\label{Never}

To conclude, I would like to discuss one further strengthening of the
notion of a $\psi$-epistemic model, first introduced informally by
Montina \cite{Montina2012a}, about which very little is currently
known.  This is the notion of a never $\psi$-ontic model.

\begin{definition}
Let $\mathfrak{F} = \langle \Hilb, \mathcal{P}, \mathcal{M} \rangle$
be a PM fragment and let $\Theta = (\Lambda, \Sigma, \Delta, \Xi)$ be
an ontological model of it.  $\Theta$ is \emph{sometimes $\psi$-ontic}
if, for all pure states $\Proj{\psi} \in \mathcal{P}$, there exists
a $\Omega \in \Sigma$ such that $\mu(\Omega) > 0$ for some $\mu \in
\Delta_{\psi}$, but for all other pure states $\Proj{\phi} \in
\mathcal{P}$, $\Proj{\phi} \neq \Proj{\psi}$, every $\nu \in
\Delta_{\phi}$ has $\nu(\Omega) = 0$.  Otherwise the model is
\emph{never $\psi$-ontic}.
\end{definition}

Roughly speaking, in a sometimes $\psi$-ontic model, each pure state
has a special region of the ontic state space all to itself.  If you
know that the ontic state occupies this region then you can identify
the quantum state that was prepared with probability one.  However,
the ontic state need not always occupy such a region; hence the
terminology ``sometimes'' $\psi$-ontic.  In contrast, in a never
$\psi$-ontic model, every region of the ontic state space is shared
nontrivially by more than one quantum state.  Considerations about the
degree of overlap of pairs of states do not really bear on the
question of whether a never $\psi$-ontic model is possible, since
arbitrary overlaps may occur outside the special regions in a
sometimes $\psi$-ontic model.

The reason why the notion of a sometimes $\psi$-ontic model is
interesting is that all of the implications of $\psi$-ontology
discussed in \S\ref{Imp} can be derived from sometimes $\psi$-ontology
instead.  Whilst the fact that a maximally $\psi$-epistemic model is
impossible is enough to derive preparation contextuality and Bell's
Theorem, it does not imply excess baggage.  However, excess baggage
follows directly from sometimes $\psi$-ontology, since there must be
at least as many ontic states as there are pure quantum states if each
quantum state is to be assigned its own region of ontic state space.
Further, a sometimes $\psi$-ontic model cannot be maximally
$\psi$-epistemic, so we obtain all the implications of that as well.

\begin{theorem}
Let $\mathfrak{F} = \langle \Hilb, \mathcal{P}, \mathcal{M} \rangle$
be a PM fragment and let $\Theta = (\Lambda, \Sigma, \Delta, \Xi)$ be
an ontological model of it.  If $\Theta$ is sometimes $\psi$-ontic then
it is not maximally $\psi$-epistemic.
\end{theorem}
\begin{proof}
Assume that $\Theta$ is maximally $\psi$-epistemic.  Then, for all
$\Proj{\psi}, \Proj{\phi} \in \mathcal{P}$, all $\mu \in
\Delta_{\psi}$, all $\nu \in \Delta_{\phi}$, and all $\Omega \in
\Sigma$ such that $\nu(\Omega) = 1$,
\begin{equation}
\label{eq:Never}
\int_{\Omega} \text{Pr}(E|M,\lambda)d\mu(\lambda) = \int_{\Lambda}
\text{Pr}(E|M,\lambda)d\mu(\lambda),
\end{equation}
for all $M \in \mathcal{M}$, $\text{Pr} \in \Xi_M$.

Let $\Omega^{\prime}$ be the region uniquely assigned to
$\Proj{\psi}$ so that $\nu(\Omega^{\prime}) = 0$ and
$\mu(\Omega^{\prime}) > 0$.  Then $\Omega$ can be chosen such that
$\Omega \cap \Omega^{\prime} = \emptyset$.  Summing
Eq.~\eqref{eq:Never} over $E$ then gives $\mu(\Omega) = \mu(\Lambda)
= 1$.  However, $\mu(\Omega)$ cannot be equal to one because
$\mu(\Omega^{\prime}) > 0$ and $\Omega$ and $\Omega^{\prime}$ are
disjoint, so this would make $\mu(\Lambda) > 1$.  Hence, there is a
contradiction, so $\Theta$ cannot be maximally $\psi$-epistemic.
\end{proof}

The Kochen--Specker model is never $\psi$-ontic, but all the existing
$\psi$-epistemic models for Hilbert spaces of dimension $\geq 3$ are
sometimes $\psi$-ontic, so whether never $\psi$-ontic models exist in
all dimensions is an open question.  Since it is impossible to prove
$\psi$-ontology without auxiliary assumptions, sometimes
$\psi$-ontology is perhaps the strongest result about the reality of
the quantum state that we could hope to prove without such
assumptions.  For this reason, I consider this to be the most
important open question in this area.

\section{Discussion and conclusions}

\label{Conc}

\subsection{Summary}

This review article was divided into three parts.  In the first part,
I explained the distinction between ontic and epistemic
interpretations of the quantum state, outlined the pre-existing
arguments for epistemic and ontic interpretations, and explained how
many existing no-go theorems would follow from proving that the
quantum state is ontic.  The aim of this part was to convince you that
the $\psi$-ontic/epistemic distinction is interesting, that the
question is unresolved by qualitative arguments, and that it is the
sort of issue that should be addressed with the conceptual rigor that
Bell brought to nonlocality.

The second part discussed three existing $\psi$-ontology theorems: the
Pusey--Barrett--Rudolph Theorem, Hardy's Theorem, and the Colbeck--Renner Theorem.  Each of
these results involve auxiliary assumptions, of varying degrees of
reasonableness.  Hardy's ontic indifference assumption is not really
appropriate for a $\psi$-epistemic theory, but it does allow the flaw
in the argument for the reality of the wavefunction from interference
to be exposed more clearly.  Colbeck and Renner's assumption of
parameter independence is doubtful in light of Bell's Theorem.  Whilst
violations of Bell inequalities can arise from the failure of outcome
independence alone, many viable realist interpretations of quantum
theory, such as de Broglie--Bohm theory, violate parameter independence
instead, and one really wants the scope of a no-go theorem to include
as many viable realist interpretations as possible.  In my view, the
preparation independence postulate of the Pusey--Barrett--Rudolph Theorem is the best of
the bunch.  It is not completely unassailable, since weakening it
slightly to allow for genuinely nonlocal degrees of freedom does
potentially allow for a viable $\psi$-epistemic interpretation.
Nevertheless, the PIP is satisfied in theories like de Broglie--Bohm,
so, unlike the Colbeck--Renner Theorem, the Pusey--Barrett--Rudolph Theorem does explain
why the wavefunction must be real in theories of that type.

Without auxiliary assumptions, $\psi$-epistemic models do exist, so
the third part of this review discussed strengthenings of the notion
of a $\psi$-epistemic model that may be needed in order to make the
$\psi$-epistemic explanations of phenomena like the
indistinguishability of quantum states and the no cloning theorem go
through.  Whilst the case is not yet watertight, I expect that future
work on bounding the ratio of the overlap of probability measures to
the quantum probabilities of the states they represent will eventually
make $\psi$-epistemic explanations look implausible within the
ontological models framework.

\subsection{Open questions}

Part of the aim of this review was to provide the necessary background
for researchers who would like to begin working in this area.  In this
regard, several open questions were raised in the main text, and
solving some of them would help to put the non viability of
$\psi$-epistemic interpretations on a more secure footing.  For easy
reference, the following list collects all the open problems together
in one place, and adds a couple of new ones.  The section numbers
indicate where in the review you can find a more detailed description.

\begin{enumerate}
\item \S\ref{Crit:SF1}: Can a deterministic $\psi$-ontic model always
be converted into a $\psi$-epistemic one by finding regions of the
ontic state space associated with different pure states that make
the exact same predictions?  One of Schlosshauer and Fine's
criticisms of the Pusey--Barrett--Rudolph Theorem was based on the idea that this can
always be done.  It is not true for nondeterministic theories, but
it may be true for deterministic theories.
\item \S\ref{Pair}: Can reasonable symmetry requirements be used to
rule out pairwise $\psi$-epistemic models without restrictive
assumptions on the nature of the ontic state space?  The theorem of
Aaronson et.\ al.\ only applies if the ontic state space is
projective Hilbert space or the set of unitary operators, which is a
very restrictive class of models.
\item \S\ref{Lips}: Does the overlap of probability measures
representing nonorthogonal states have to tend to be close to zero
for pairs of states with a fixed inner product?  The existing
results only imply this for families of states that become
orthogonal in the $n \rightarrow \infty$ limit, so this only implies
that the $\psi$-epistemic explanation of indistinguishability is
implausible for states that are almost distinguishable.
\item \S\ref{Never}: Does a never $\psi$-ontic model exist for the
fragment consisting of all pure states and projective measurements
in $\mathbb{C}^d$ for $d \geq 3$?  In my view, this is the most
important open question, as proving that all models must be
sometimes $\psi$-ontic has the same implications as $\psi$-ontology.
\item Are better $\psi$-ontology results possible with POVMs?  The
known $\psi$-epistemic models, including the Kochen--Specker model
discussed in \S\ref{EOM} and the ABCL model discussed in
\S\ref{NPIP}, only reproduce quantum theory for measurements in
orthonormal bases.  Morris has shown that the Kochen--Specker model
cannot be extended to all POVMs \cite{Morris2009}.  Therefore, it
is still possible that there are no $\psi$-epistemic models that
reproduce the quantum predictions for all POVMs, and it may be
possible to prove this without auxiliary assumptions.  Further, most
of the proofs of $\psi$-ontology results presented here are based on
measurements in orthonormal bases, so it is possible that simpler
proofs and/or better overlap bounds could be derived using POVMs.
This is the case for the Pusey--Barrett--Rudolph theorem, as the simplified proof
presented here does involve POVMs.
\item Do $\psi$-ontology results have applications in quantum
information theory, beyond those of the implications discussed in
\S\ref{Imp}?  The ontological models framework is already regarded
as implausible by some physicists, but it is worth investigating
anyway because it provides a way of simulating quantum experiments
using classical resources.  Results derived in this framework often
go on to have applications in quantum information theory.  In this
regard, Montina proved an upper bound on the classical communication
complexity of simulating the identity channel for a qubit using the
Kochen--Specker model, which is maximally $\psi$-epistemic
\cite{Montina2012}.  It is possible that the nonexistence of such
models in higher dimensions could be used to prove lower bounds on
the communication complexity of the identity channel in higher
dimensions.  Conversely, Montina et.\ al.'s recent results on the
communication complexity of simulating the identity channel
\cite{Montina2011a, Montina2013, Montina2014, Montina2014a} might be
used to derive $\psi$-ontology results.  Some work has also been
done on deriving quantum advantages from antidistinguishability
\cite{Bandyopadhyay2013, Perry2014}.
\end{enumerate}

\subsection{Experiments}

\label{Conc:Exp}

Probably the most important issue not discussed so far in this review
is the question of whether the reality of the quantum state can be
established experimentally.  When it comes to experiments in the
foundations of quantum theory, tests of Bell's Theorem are somewhat of
a gold standard.  Violation of a Bell inequality rules out a wide
class of local realistic theories independently of the details of
quantum theory.  Consequently, the experimenter does not have to
assume much about how their experiment is represented within quantum
theory.  From their observed statistics they simply get a yes/no
answer as to whether local hidden variable theories are viable, modulo
the known loopholes.  In modern parlance, tests of Bell's Theorem are
\emph{device independent} \cite{Brunner2013}.

In contrast, a test of the reality of the quantum state would not be
device independent simply because the ``quantum state'' is the thing
we are testing the reality of, and that is a theory dependent notion.
Consequently, one has to assume that our quantum theoretical
description of the way that our preparation devices work is more or
less accurate, in the sense that they are approximately preparing the
quantum states the theory says they are, in order to test the existing
$\psi$-ontology results.  Therefore, it is desirable to have a more
theory independent notion of whether a given set of observed
statistics imply that the ``probabilistic state'', i.e.\ some
theory-independent generalization of the quantum state, must be real.
It is not obvious whether this can be done, but if it can then
experimental tests of $\psi$-ontology results would become much more
interesting.

Of course, one can still perform non device independent experimental
tests.  This amounts to trying to prepare the states, perform the
transformations, and make the measurements involved in a
$\psi$-ontology result and checking that the quantum predictions are
approximately upheld.  Due to experimental error, the agreement will
never be exact, but one can bound the overlap between probability
measures representing quantum states instead of showing that it must
be exactly zero.  For the special case of the Pusey--Barrett--Rudolph Theorem given in
Example~\ref{exa:Main:Pusey--Barrett--Rudolph}, this has been done using two ions in an
ion trap \cite{Nigg2012}.  However, the experimental result only shows
that the overlap in probability measures must be a little smaller than
the quantum probability, and not that it must be close to zero, and of
course, this only applies if one assumes the PIP.

A better approach to deriving overlap bounds is to test the results
described in \S\ref{Lips}.  For finite $n$, these can provide overlap
bounds for specific sets of states without auxiliary assumptions.
Based on Branciard's results \cite{Branciard2014}, a recent experiment
with single photons has given a bound on the ratio of overlap to the
quantum probability of $0.690 \pm 0.001$ for a set of $10$ states in a
Hilbert space of dimension $4$ \cite{Ringbauer2014}.  This is still
quite far from establishing the reality of the quantum state, for
which one would want to test many pairs of quantum states with a
variety of different inner products, and to show that the overlap
ratios are consistently close to zero.  Recall that, due to the fact
that part of the indistinguishability of quantum states may be
explained by the coarse-grained nature of measurements in an
ontological model, it is difficult to attach significance to overlap
ratios that are significantly larger than zero.  Further theoretical
work is required to derive feasible experiments that can achieve much
lower ratios.

Finally, the original version of the Colbeck--Renner Theorem, i.e.\ the
one that aimed to rule out ontological models with greater predictive
power than quantum theory, has been tested using a quantum optical
system \cite{Stuart2012} and there has also been an optical test of
the Patra--Pironio--Massar Theorem on PPM-continuous models
\cite{Patra2013}.

\subsection{Future directions}

Assuming that future results drive the final nails into the coffin of
$\psi$-epistemic explanations within the ontological models framework,
the final question I want to address is where to go next.  One option
is to embrace the reality of the wavefunction by adopting one of the
existing realist interpretations that fit into the ontological models
framework, e.g.\ de Broglie--Bohm theory, spontaneous collapse theories
or modal interpretations.  Another is to adopt a neo-Copenhagen
interpretation.

The first option is unappealing because adopting one of these realist
interpretations opens up an explanatory gap.  Namely, given that a
$\psi$-epistemic interpretation would provide compelling explanations
of a whole variety of quantum phenomena, it is puzzling that the
quantum state should nevertheless be ontic.  Of course, the existing
no-go results, such as Bell's Theorem, also imply explanatory gaps,
e.g.\ if a theory is explicitly nonlocal then why can this not be used
to send a superluminal signal?  For this reason, those inclined to the
$\psi$-epistemic view are likely to have rejected these
interpretations already on the basis of these other gaps.

Being neo-Copenhagen is always an option, but the merits of such an
move depend on the degree to which one believes that realism is
desirable.  This is not the place to get into the debate between
realism and antirealism and whether neo-Copenhagen views are
compatible with some weakened notion of realism.  Suffice to say that
the viability of these interpretations turns on issues that are far
deeper than the reality of the wavefunction.  For my part, I think
that if one denies the existence of an observer-independent reality
then it becomes very difficult to maintain a clear notion of
explanation at all.  Closing explanatory gaps by denying the need for
any explanation at all does not seem that appealing to me.

The only remaining option then is to adopt a realist interpretation
that does not fit in to the ontological models framework.  There are
several possibilities, most of them highly speculative.

\begin{itemize}
\item Many-worlds: The many-worlds interpretation \cite{Everett1957,
DeWitt1973, Wallace2012} is not contained within the ontological
models framework because the latter assumes that measurements yield
a single unique outcome.  Many-worlds is a commonplace retreat for
realists who want to avoid introducing nonlocality in the face of
Bell's Theorem.  However, the conventional many-worlds
interpretation is not a retreat for $\psi$-epistemicists because it
is based on the idea that the quantum state is a literal description
of reality.  Nevertheless, since $\psi$-ontology theorems do not
apply to theories that involve many-worlds, it is possible that
there is a viable $\psi$-epistemic interpretation that involves
them.  This is doubtless an unappealing option to both existing
many-worlds advocates and $\psi$-epistemicists.  For many-worlds
advocates, the reason to take multiple worlds seriously is because
they are encoded in the branching structure of the wavefunction.  If
you take away the reality of the wavefunction then you take away
their reason for believing in them.  Similarly, this is the only
real reason for taking many worlds seriously, so someone who does
not believe in the reality of the wavefunction is unlikely to have
found the existence of many-worlds plausible in the first place.
Nevertheless, it is logically possible that the universe could be
described by some structure that does not imply a unique
wavefunction, but does support the existence of many-worlds.
\item Histories approaches: Histories interpretations, such as
consistent/decoherent histories \cite{Griffiths2002} and Sorkin's
co-event formalism \cite{Sorkin2007a, Sorkin2007} are based on
taking the space-time picture provided by the path integral
seriously.  Consistent histories does not fit into the ontological
models framework as it decrees that not all sets of histories can be
assigned a probability and there may not be a unique description of
the universe in terms of a single history.  Griffiths argues that
the consistent histories interpretation is appropriately
$\psi$-epistemic \cite{Griffiths2013}.  In contrast, the co-event
formalism is based on a modified logic and probability theory.  I
explained my doubts about realist interpretations of exotic
probability theories in \S\ref{Crit:Copenhagen} and, in any case,
Wallden \cite{Wallden2013} provides evidence that a result analogous
to the Pusey--Barrett--Rudolph Theorem may hold in the coevent formalism.
\item Retrocausality: In the ontological models framework, it is
assumed that the probability measure representing a quantum state is
independent of the choice of future measurement setting.  If this
were not the case then the $\psi$-epistemic interpretation of
quantum phenomena could be maintained by having the measures
corresponding to antidistinguishable states have no overlap when the
antidistinguishing measurement is made, but nonzero overlap when
other measurements are made.  One way that dependence on the
measurement setting may occur is if there is a direct causal
influence, traveling backwards in time, from the measurement to the
preparation.  Several authors have argued that there are independent
reasons for adopting a retrocausal approach to quantum theory
\cite{CostadeBeauregard1976, Price2012, Price2013}, not least
because it might allow an appropriately local resolution to the
dilemma posed by Bell's Theorem.  The transactional interpretation
\cite{Cramer1986, Kastner2013} is explicitly retrocausal and the
two-state vector formalism \cite{Aharonov2002} can be read in a
retrocausal way.  However, both of these theories posit an ontic
wavefunction.  If we are to maintain $\psi$-epistemic explanations
then we instead need to look for retrocausal ontological models that
posit a deeper reality underlying quantum theory that does not
include the quantum state.
\item Relationalism: The ontological models framework assumes that
systems have their own intrinsic properties, encompassed by the
ontological state $\lambda$.  Relationalism posits that systems do
not have intrinsic properties, but only properties relative to other
systems.  The usual analogy is with the concept of position.  One
cannot talk about the position of a particle without setting down
some coordinate system, and this implicitly means that we are
measuring position with respect to some other physical system that
provides a reference frame.  There is an obvious commonality with
Everett's relative state approach \cite{Everett1957}, except that we
want a theory of this type that is $\psi$-epistemic.  Rovelli's
relational quantum mechanics \cite{Rovelli1996} is a theory of this
type in which the wavefunction is supposed to be epistemic, but he
defines relational properties in terms of the global wavefunction so
it is not clear how they are supposed to be determined if the
wavefunction is not real.
\end{itemize}

In conclusion, I think we should try to find a way of understanding
quantum theory that closes as many of the explanatory gaps opened by
no-go theorems as possible.  This is because an interpretation that
merely accommodates the known facts about quantum theory, rather than
explaining them, is unlikely to yield principles that can reliably
guide us towards future physical theories.  Since $\psi$-ontology
implies many of the existing no-go theorems, the gap opened by
$\psi$-ontology results should be taken at least as seriously as the
others.  This means that we should investigate the speculative roads
less travelled described above, in addition to others that we have not
thought of yet.  The chances that any one of them will bear fruit may
be slim, but the rewards if they do will be worth the effort.

\section*{Acknowledgements}

I would like to thank Jonathan Barrett, Sarah Croke, Joseph Emerson,
Lucien Hardy, Teiko Heinosaari, Ravi Kunjwal, Klaas Landsman, Peter
Lewis, Owen Maroney, Matt Pusey, Terry Rudolph, Max Schlosshauer and
Rob Spekkens for useful discussions, and Paul Merriam for a careful
proof reading.  This research was supported in part by the FQXi Large
Grant ``Time and the Structure of Quantum Theory'' and by Perimeter
Institute for Theoretical Physics.  Research at Perimeter Institute is
supported by the Government of Canada through Industry Canada and by
the Province of Ontario through the Ministry of Research and
Innovation.

\section*{Appendix}
\appendix

\section{The $\psi$-ontic/epistemic distinction and objective chance}

\label{App:Chance}

In this appendix, the question of whether a distinction equivalent to
the Bayesian distinction between ontic and epistemic states can be
made in theories that involve objective chance is addressed.

Firstly, note that whether or not this issue is relevant to the
reality of the wavefunction depends on exactly which probabilities are
deemed to be objective chances.  Most advocates of objective chance
support hybrid theories of probability in which there are epistemic
probabilities in addition to objective chances.  Objective chances are
invoked to explain lawlike probabilistic regularities, as they occur
in our theories of physics for example, but epistemic probabilities,
usually cashed out in Bayesian or frequentist terms, also exist in
order to make contact with statistics.  In the ontological models
framework for a PM fragment, probabilities enter into the description
in two places.  Firstly, quantum states correspond to probability
measures over the ontic state space and, secondly, measurements
correspond to conditional probability distributions over their
outcomes.  A fairly natural position would be to assert that the
probability measures corresponding to states are epistemic, whereas
the conditional probabilities associated to measurements are
objective.  The former represent our uncertainty about the true ontic
state, whereas the latter represent our uncertainty about the response
of the measurement device.  On a conventional view, it is the response
of the measurement device that represents a ``genuinely stochastic''
event in quantum theory, so it would make sense to only place the
objective chances here where they seem to be most needed.

Regardless of whether you agree with this account, the main point is
that the introduction of objective chances into quantum theory only
presents a problem for the ontic/epistemic distinction if the
probability measures representing states are deemed to be objective
chances.  This is because the definition of $\psi$-ontic and
$\psi$-epistemic models is given entirely in terms of the overlaps of
these probability measures and does not involve the measurement device
probabilities at all.  Therefore, if you are happy to put objective
chances only in the measuring device, then the introduction of such
chances does not present any new problem for the ontic/epistemic
distinction.  It is only if you think that the probability measures
representing states are objective chances that the distinction needs
to be reexamined.  Therefore, I assume that this is what is meant by
introducing objective chances into quantum theory in the remainder of
this appendix.

As in the frequency theory, we are interested in whether objective
chances can be viewed as intrinsic properties of individual systems,
or whether they must be defined with reference to the wider world
beyond the individual system, e.g.\ in terms of an ensemble of similar
systems or the conditions prevailing in the environment of the system.
An intrinsic property of a system is something like the charge of an
electron or the hidden variable state of a quantum system, if the
latter are presumed to exist.  In a realist approach to physics, these
correspond to ontological features of the system and they cannot be
changed just by changing our description of the world in which the
system is embedded.  Changing them requires an intervention in the
system itself.

Deciding whether or not objective chances are intrinsic properties is
difficult because there is no universally agreed upon theory of
objective chance.  Fortunately, the question only depends on a couple
of broad features of the theory.

Firstly, some authors posit that objective chance is compatible with
determinism so that, for example, the probabilities involved in
classical statistical mechanics can be viewed as objective chances.
Others think that objective chances only make sense if there is a
genuine stochasticity in nature, with quantum theory providing the
prime example of a theory that involves such genuine chance.  It
should be clear that objective chance cannot be intrinsic in any
theory that is compatible with determinism.  This is because, in a
deterministic theory, the intrinsic properties of an isolated system
determine its future uniquely, so they could only ever give rise to
objective chances of $0$ or $1$.  Therefore, to make such a theory
work, one has to refer either to an ensemble of systems, as in the
frequency theory, or to the conditions surrounding the system.  As an
example of the latter, one can imagine a specification of the way in
which a coin should be tossed such as ``a strong flip between thumb
and forefinger'' that is specific enough to license the assignment of
a fixed objective chance, but vague enough that it does not determine
the outcome of the toss uniquely.  Of course, it is questionable
whether such a notion makes sense, but the point is that, if one does
take this view, then probabilities are not intrinsic properties of
systems.  Either they refer to ensembles or they refer to the
conditions of interaction between a system and its environment, and
specifying these in different levels of detail would lead to
different probability assignments, just as in the frequency case.

However, many philosophers of objective chance, including Popper
\cite{Popper2011} and Lewis \cite{Lewis2011}, take the view that
objective chances only make sense in a genuinely stochastic universe.
In this case, one can require that the objective chances of an
experimental outcome are only the same if the intrinsic properties of
the system before the experiment are identical in all relevant
details.  The prime example of this would be to say that the objective
chances of obtaining a given outcome in a quantum measurement on two
different systems are the same iff the systems are described by the
same quantum state $\Ket{\psi}$ prior to the experiment.  This
presents more of a problem for the distinction we are trying to make,
since ostensibly the state $\Ket{\psi}$ only refers to the system
itself.

However, it should be noted that the quantum state example rests on
questionable assumptions about the interpretation of quantum theory.
It assumes that $\Ket{\psi}$ can be regarded as an intrinsic property
of the system and, further, that $\Ket{\psi}$ is a complete
description of the system.  In an operationalist approach, the
assertion that $\Ket{\psi}$ is an intrinsic property of a system is
denied.  Instead, it is a description of those facts about the device
that prepared the system that are relevant for predicting future
measurement outcomes.  In other words, it is a condensed description
of a set of knob settings, meter readings, etc.\ that refer to a piece
of experimental apparatus external to the system.  These are not
intrinsic properties of the system, so the relevant distinction is
maintained.

Of course, both Popper and Lewis intend a more realist interpretation
of quantum theory.  However, even if quantum states are intrinsic
properties of quantum systems, they need not be complete.  For
example, in de Broglie--Bohm theory one also needs to specify the
positions of particles.  The theory is deterministic when both the
quantum state and particle positions are specified, so we are back in
the position of having to define objective chances in a deterministic
theory, in which case they are not intrinsic.  In de Broglie--Bohm
theory, the statement that a system is described by a state
$\Ket{\psi}$ really means that it is part of an ensemble of systems
described by the quantum state $\Ket{\psi} \otimes \Ket{\psi} \otimes
\ldots$ and in which the particle positions are distributed according
to $\left | \psi(x) \right |^2$.  Since there is freedom to look at
subensembles where the positions are not distributed according to
$\left | \psi(x) \right |^2$, and these would allow prediction with
greater accuracy than quantum theory, systems described by the same
quantum state do not have the same intrinsic properties in all
relevant detail.

Nevertheless, although quantum theory is a prime motivation for
objective chances, theories of objective chance are usually
independent of the details of physics.  Thus, we can ignore the
quantum motivation and just look at the theories of chance actually
proposed to determine whether chances are intrinsic.  In this regard,
an important distinction is whether or not a theory of chance is
\emph{Humean}.  Roughly speaking, a Humean theory is one in which the
chances are defined in terms of the facts on the ground, i.e.\ facts
about the universe that could form part of our experience of it (see
\cite{Eagle2011} for a more precise definition).  This means that
objective chances cannot be defined in terms of things like
$\Ket{\psi}$, which do not form part of our experience.  Lewis
\cite{Lewis2011} was a proponent of Humean theories of chance and his
favoured \emph{best system} theory is really just a modification of
frequentism.  More specifically, he thought that chances were
specified via a tradeoff between accurately capturing the relative
frequencies and simplicity.  Thus, if a large number of coin flips is
performed several different times, and on each occasion the relative
frequency of heads obtained was close to, but not exactly, $1/2$,
varying in a seemingly random way, Lewis would say that this licenses
assigning an objective chance of $1/2$ to the coin flips.  This
differs from the frequency theory in that it does not demand that
probabilities are exactly the relative frequencies in some real or
hypothetical, finite or infinite, sequence of experiments.  Of course,
defining an objective tradeoff between predictability and simplicity
is difficult, but for present purposes all that matters is that again
objective chances are not intrinsic properties of systems, but are
instead defined with respect to an ensemble.  Likewise, I think that
any Humean theory of chance makes chance a non-intrinsic property
because chances would have to be defined in terms of observable facts,
and I do not see how this could be done without referring to ensembles
or to the surroundings of the system.

Therefore, the only chance theories that really pose a problem for the
distinction we are trying to make are those that are non-Humean, i.e.\
they posit that chances do not supervene on facts that could form part
of our experience.  Of these, the most prominent is Popper's
propensity theory \cite{Popper2011}.  Propensities are dispositional
properties, i.e.\ a system has a disposition to produce a certain
outcome in an experiment.  Propensity theories are broadly classified
as either long-run propensity theories or single-case propensity
theories.  In a long-run theory, a propensity is read as a disposition
to produce a certain relative frequency of outcomes in the long run.
Since this refers to an ensemble, again there is no problem
distinguishing this type of property from intrinsic properties of
individual systems.  On the other hand, single case propensities are
read as a disposition to produce a certain outcome in a single
experiment.  These are much more problematic for the distinction we
are trying to make, as they do not refer to entities other than the
individual system.  Thus, in a single case propensity theory, it may
not be possible to make a clean distinction that is analogous to the
Bayesian distinction between ontic and epistemic states.

Before concluding, note that many philosophers take a more laissez
faire approach to objective chances.  Whilst they believe that
objective chances exist, they do not commit to a specific theory and
are instead content to specify some rules that they must obey, such as
Lewis' principal principle \cite{Lewis2011}.  This is a perfectly
reasonable attitude to take, but it is unreasonable to expect that the
question of whether chances are intrinsic properties can be settled at
this level of generality.  Some further features of objective chance
would need to be specified, such as whether or not they ought to be
Humean.

In summary, most theories of objective chance seem to admit a
distinction equivalent to the Bayesian distinction between ontic and
epistemic states in that chances refer to non intrinsic facts about a
system.  The problematic theories are single-case, such as the
single-case propensity theory.  However, it seems a bit of a stretch
to adopt this theory in order to avoid investigating the question of
whether quantum states are ontic or epistemic, particularly since the
interpretation of quantum theory is a prime motivation for introducing
objective chances in the first place.

\section{The Kochen--Specker model}

\label{App:BMKS}

This appendix proves that the Kochen--Specker model of
Example~\ref{exa:EOM:KS} reproduces the quantum predictions and is
maximally $\psi$-epistemic.  Recall that, in the Kochen--Specker model,
the quantum state $\Proj{\psi}$ is represented by a unique probability
measure
\begin{equation}
\mu(\Omega) = \int_{\Omega} p(\vec{\lambda}) \sin \vartheta
d\vartheta d\varphi,
\end{equation}
over the Bloch sphere, where the density $p$ is given by
\begin{equation}
p(\vec{\lambda}) = \frac{1}{\pi} H \left ( \vec{\psi} \cdot
\vec{\lambda} \right ) \vec{\psi} \cdot \vec{\lambda},
\end{equation}
and $H$ is the Heaviside step function.  The conditional probability
distribution for a measurement $M = \left \{ \Proj{\phi},
\Proj{\phi^{\perp}} \right \}$ is given by
\begin{align}
\text{Pr}(\phi|M, \vec{\lambda}) & = H (\vec{\phi} \cdot
\vec{\lambda}) \\
\text{Pr}(\phi^{\perp}|M, \vec{\lambda}) & = 1 - \text{Pr}(\phi |M,
\vec{\lambda}).
\end{align}

To prove that this reproduces the quantum predictions, we need to show
that, for any pair of states $\Proj{\psi}$ and $\Proj{\phi}$,
\begin{equation}
\label{eq:BMKS:reprod}
\int_{\Lambda} \text{Pr}(\phi|M,\vec{\lambda}) d\mu(\vec{\lambda}) =
\Tr{\Proj{\phi}\Proj{\psi}}.
\end{equation}
The corresponding equation for $\Proj{\phi^{\perp}}$ will then be
automatically satisfied because
\begin{equation}
\int_{\Lambda} \text{Pr}(\phi^{\perp}|M,\vec{\lambda})
d\mu(\vec{\lambda}) = 1 -  \int_{\Lambda}
\text{Pr}(\phi|M,\vec{\lambda}) d\mu(\vec{\lambda}).
\end{equation}

To prove Eq.~\eqref{eq:BMKS:reprod}, it is convenient to choose a
parameterization of the Bloch sphere such that both $\vec{\psi}$ and
$\vec{\phi}$ lie on the equator.  We can further choose $\vec{\psi}$
to point along the $x$ axis so that $\vec{\psi} = (1,0,0)$ and
$\vec{\phi} = (\cos \varphi_{\phi}, \sin \varphi_{\phi}, 0)$ for some
angle $-\pi < \varphi_{\phi} \leq \pi$.  Using
Eq.~\eqref{eq:EOM:Bloch}, this means that the right hand side of
Eq.~\eqref{eq:BMKS:reprod} is
\begin{equation}
\Tr{\Proj{\phi}\Proj{\psi}} = \left | \BraKet{\phi}{\psi} \right
|^2 = \frac{1}{2}\left ( 1 + \cos(\varphi_{\phi}) \right ).
\end{equation}

Now, expanding the left hand side of Eq.~\eqref{eq:BMKS:reprod} gives
\begin{align}
\label{eq:BMKS:int}
&\int_{\Lambda} \text{Pr}(\phi|M,\vec{\lambda}) d\mu(\vec{\lambda}) \nonumber\\
&\quad\quad=
\frac{1}{\pi} \int_{\Lambda}  H (\vec{\phi} \cdot \vec{\lambda})
H \left ( \vec{\psi} \cdot \vec{\lambda} \right )
\vec{\psi} \cdot \vec{\lambda} \sin \vartheta d\vartheta d\varphi.
\end{align}
Since $\vec{\lambda} = (\sin \vartheta \cos \varphi, \sin \vartheta \sin
\varphi, \cos \vartheta)$, we have
\begin{align}
\vec{\psi} \cdot \vec{\lambda} & = \sin \vartheta \cos \varphi \\
\vec{\phi} \cdot \vec{\lambda} & = \sin \vartheta \cos \varphi \cos
\varphi_{\phi} + \sin \vartheta \sin \varphi \sin \varphi_{\phi} \\
& = \sin \vartheta \cos \left ( \varphi - \varphi_{\phi} \right ).
\end{align}
Due to he Heaviside step functions, we only need to integrate over the
region where both of these are positive.  This is the intersection of
$-\frac{\pi}{2} < \varphi < \frac{\pi}{2}$ and $-\frac{\pi}{2} +
\varphi_{\phi} < \varphi < \frac{\pi}{2} + \varphi_{\phi}$.  When
$\varphi_{\phi}$ is positive this is the interval $-\frac{\pi}{2} +
\varphi_{\phi} < \varphi < \frac{\pi}{2}$ and when $\varphi_{\phi}$ is
negative this is the interval $-\frac{\pi}{2} < \vartheta <
\frac{\pi}{2} + \varphi_{\phi}$.  Consider first the case where
$\varphi_{\phi}$ is positive.  Then Eq.~\eqref{eq:BMKS:int} reduces to
\begin{align}
\int_{\Lambda} \text{Pr}(\phi|M,\vec{\lambda}) d\mu(\vec{\lambda}) &
= \int_0^{\pi} \sin^2 \vartheta d \vartheta \int_{-\frac{\pi}{2} +
\varphi_{\phi}}^{\frac{\pi}{2}} \cos \varphi d \varphi  \\
& = \frac{1}{2} \left [ \sin \varphi \right ]_{-\frac{\pi}{2} +
\varphi_{\phi}}^{\frac{\pi}{2}} \\
& = \frac{1}{2} \left ( 1 + \sin \left ( \frac{\pi}{2} -
\varphi_{\phi}\right )\right ) \\
& = \frac{1}{2} \left ( 1 + \cos \varphi_{\phi} \right ),
\end{align}
as required.  The case where $\varphi_{\phi}$ is negative gives the
same result because $\sin$ is an odd function.

To prove that the model is maximally $\psi$-epistemic, let $\nu$ be
the probability measure associated with $\Proj{\phi}$, i.e.
\begin{equation}
\nu(\Omega) = \int_{\Omega} q(\vec{\lambda}) \sin \vartheta
d\vartheta d\varphi,
\end{equation}
where
\begin{equation}
q(\vec{\lambda}) = \frac{1}{\pi} H \left ( \vec{\phi} \cdot
\vec{\lambda} \right ) \vec{\phi} \cdot \vec{\lambda}.
\end{equation}
We then need to show that
\begin{equation}
\int_{\Omega} \text{Pr}(\phi|M,\vec{\lambda}) d\mu(\vec{\lambda}) =
\int_{\Lambda} \text{Pr}(\phi|M,\vec{\lambda}) d\mu(\vec{\lambda}),
\end{equation}
for any $\Omega$ such that $\nu(\Omega) = 1$.  Assume that
$\varphi_{\phi}$ is positive (the negative case follows the same
logic).  Let,
\begin{align}
\Omega_{\psi} & = \left \{ \vec{\lambda} \in \Lambda \middle | 0<
\vartheta < \pi, - \frac{\pi}{2} < \varphi <
\frac{\pi}{2} \right \} \\
\Omega_{\phi} & = \left \{ \vec{\lambda} \in \Lambda \middle | 0<
\vartheta < \pi, -\frac{\pi}{2} + \varphi_{\phi} < \varphi <
\frac{\pi}{2} + \varphi_{\phi} \right \}.
\end{align}
Note that, for any measurable set $\Omega$
\begin{equation}
\int_{\Omega} \text{Pr}(\phi|M,\vec{\lambda}) d\mu(\vec{\lambda})=
\int_{\Omega \cap \Omega_{\psi}} \text{Pr}(\phi|M,\vec{\lambda})
d\mu(\vec{\lambda}),
\end{equation}
because $p(\vec{\lambda})$ is zero outside $\Omega_{\psi}$.  Note also
that $\Omega_{\phi}$ is a measure one set according to $\nu$ because
$q(\vec{\lambda})$ is zero outside this set.  However, in proving that
the model reproduces the quantum predictions, we showed that
\begin{align}
\int_{\Lambda} \text{Pr}(\phi|M,\vec{\lambda}) d\mu(\vec{\lambda}) &=
\int_{\Omega_{\psi} \cap \Omega_{\phi}} \text{Pr}(\phi|M,\vec{\lambda})
d\mu(\vec{\lambda}) \\
&= \Tr{\Proj{\phi}\Proj{\psi}},
\end{align}
and hence
\begin{align}
\int_{\Omega_{\phi}} \text{Pr}(\phi|M,\vec{\lambda}) d\mu(\vec{\lambda})
&= \int_{\Omega_{\psi} \cap \Omega_{\phi}}
\text{Pr}(\phi|M,\vec{\lambda}) d\mu(\vec{\lambda}) \\
&= \int_{\Lambda}
\text{Pr}(\phi|M,\vec{\lambda}) d\mu(\vec{\lambda}),
\end{align}
so we have the required property for the special case of the set
$\Omega_{\phi}$.

Now let $\Omega$ be any other set that is of measure one according to
$\nu$.  We can write $\Omega$ as the union of two disjoint sets via
\begin{equation}
\Omega = \left (\Omega \cap \Omega_{\phi} \right) \cup \left ( \Omega
\backslash \Omega_{\phi} \right ).
\end{equation}
The set $\Omega \backslash \Omega_{\phi}$ is of measure zero according
to $\nu$ because $q(\vec{\lambda})$ is zero outside $\Omega{\phi}$.
This means that $\Omega \cap \Omega_{\phi}$ is of measure one.
Further, $\text{Pr}(\phi|M,\vec{\lambda})$ is also zero outside
$\Omega_{\phi}$ so
\begin{equation}
\int_{\Omega \backslash \Omega_{\phi}} \text{Pr}(\phi|M,\vec{\lambda})
d\mu(\vec{\lambda}) = 0.
\end{equation}
Therefore, we only need to show that
\begin{equation}
\int_{\Omega \cap \Omega_{\phi}} \text{Pr}(\phi|M,\vec{\lambda})
d\mu(\vec{\lambda}) =   \int_{\Omega_{\phi}} \text{Pr}(\phi|M,\vec{\lambda})
d\mu(\vec{\lambda}).
\end{equation}

To see this, note that $\mu$ and $\nu$ are absolutely continuous with
respect to one another on $\Omega_{\psi} \cap \Omega_{\phi}$.  Since,
$\Omega \cap \Omega_{\phi}$ is of measure one according to $\nu$,
$\Omega_{\phi} \backslash \left ( \Omega \cap \Omega_{\phi} \right )$
is of measure zero according to $\nu$ and hence, by absolute
continuity, $\Omega_{\psi} \cap \left ( \Omega_{\phi} \backslash \left
( \Omega \cap \Omega_{\phi} \right ) \right )$ is of measure zero
according to both $\nu$ and $\mu$.  Thus,
\begin{align}
& \int_{\Omega_{\phi}} \text{Pr}(\phi|M,\vec{\lambda}) d\mu(\vec{\lambda})\nonumber\\
&\quad = \int_{\Omega_{\psi} \cap \Omega_{\phi}}\text{Pr}(\phi|M,\vec{\lambda})
d\mu(\vec{\lambda}) \\
&\quad = \int_{\Omega_{\psi} \cap \Omega_{\phi} \cap \Omega}
\text{Pr}(\phi|M,\vec{\lambda}) d\mu(\vec{\lambda}) \nonumber\\
&\quad\quad\quad + \int_{\Omega _{\psi} \cap
\left (\Omega_{\phi} \backslash \left (\Omega \cap
\Omega_{\phi} \right ) \right )} \text{Pr}(\phi|M,\vec{\lambda})
d\mu(\vec{\lambda}) \\
&\quad = \int_{\Omega_{\psi} \cap \Omega_{\phi} \cap \Omega}
\text{Pr}(\phi|M,\vec{\lambda}) d\mu(\vec{\lambda}) + 0 \\
&\quad = \int_{\Omega \cap \Omega_{\phi}} \text{Pr}(\phi|M,\vec{\lambda})
d\mu(\vec{\lambda}),
\end{align}
as required.

\section{Kochen--Specker contextuality}

\label{App:KS}

Kochen--Specker contextuality is not directly related to
$\psi$-ontology, but many of the consequences of $\psi$-ontology can
alternatively be derived from it.

The notion of contextuality first arose in Kochen and Specker's
attempt to prove a no-go theorem for hidden variable theories
\cite{Kochen1967}.  Kochen and Specker's definition of contextuality
only deals with projective measurements, but this has since been
generalized and given a more operational spin by Spekkens
\cite{Spekkens2005}, and we follow his approach here.  The basic idea
is that, if two things are operationally equivalent in quantum theory,
i.e.\ if they always give rise to the exact same observable
probabilities, then they should be represented the same way in an
ontological model.  Applied to measurements, this is formally defined
as follows.

\begin{definition}
Let $\mathfrak{F} = \langle \Hilb, \mathcal{P}, \mathcal{M} \rangle$
be a PM fragment and let $\Theta = (\Lambda, \Sigma, \Delta, \Xi)$
be an ontological model of it.  $\Theta$ is \emph{measurement
noncontextual} if, for each $M \in \mathcal{M}$, $\Xi_M$ consists
of a single unique conditional probability distribution, and,
whenever there exists $M,M^{\prime} \in \mathcal{M}$ and a POVM
element $E$ such that $E \in M$ and $E \in M^{\prime}$, then, for
all $\lambda$,
\begin{equation}
\text{Pr}(E|M,\lambda) = \text{Pr}(E|M^{\prime},\lambda).
\end{equation}
Otherwise, $\Theta$ is \emph{measurement contextual}.
\end{definition}

If a measurement $M$ contains the POVM element $E$, then, for a system
prepared in the state $\rho$, the outcome $E$ always occurs with
probability $\Tr{E \rho}$, regardless of how the measurement is
implemented.  Additionally, if another measurement $M^{\prime}$ also
contains $E$ then the outcome $E$ still has the same quantum
probability $\Tr{E \rho}$ in the measurement $M^{\prime}$ as it did in
$M$.  Since there is nothing in the quantum predictions that
distinguishes $E$ occurring in these different contexts, a
noncontextual model should represent them all in the same way.  Note
that, since $\Xi_M$ is a singleton for every $M \in \mathcal{M}$ in a
noncontextual model, we can unambiguously refer to a unique
conditional probability distribution $\text{Pr}(E|M,\lambda)$
associated to each measurement.

The classic example of the same measurement operator occurring in more
than one POVM is to take two orthonormal bases, $\left \{\Ket{\phi_j}
\right \}_{j=0}^{d-1}$ and $\left \{ \Ket{\phi^{\prime}_j} \right
\}_{j=0}^{d-1}$, such that $\Ket{\phi_0} = \Ket{\phi^{\prime}_0}$.
Such a pair can be constructed from a unitary $U$ that leaves
$\Ket{\phi_0}$ invariant via $\Ket{\phi^{\prime}_j} = U \Ket{\phi_j}$.
Then, the two measurements $M = \left \{ \Proj{\phi_j}\right
\}_{j=0}^{d-1}$ and $M^{\prime} = \left \{
\Proj{\phi^{\prime}_j}\right \}_{j=0}^{d-1}$ share the common
projector $\Proj{\phi_0}$.  This can only happen nontrivially if $d
\geq 3$, since if $d=2$ and $U \Ket{\phi_0} = \Ket{\phi_0}$ then
$U\Ket{\phi_1}$ can only differ from $\Ket{\phi_1}$ by a global phase.
For non projective POVMs, nontrivial examples can be constructed for
$d=2$ as well (see \cite{Spekkens2005} for details), but we are only
concerned with the traditional Kochen--Specker notion of contextuality
here, which only applies to projective measurements.

Note that, given a POVM $M = \{E_0,E_1,E_2\}$, we can construct another
POVM $N = \{E_0,E_1 + E_2\}$ by coarse-graining the second and third
outcomes.  One way of implementing $N$ is to perform the measurement
$M$, but only record whether or not the zeroth outcome occurred,
i.e.\ lump outcomes $1$ and $2$ together into a single outcome.
Because of this, it is natural to assume that the conditional
probability distribution representing $N$ satisfies
\begin{equation}
\text{Pr}(E_0|N,\lambda) = \text{Pr}(E_0|M,\lambda),
\end{equation}
since the coarse-graining is just a classical post-processing of the
outcome of the measurement that happens after the measurement is made.
If this holds then measurement contextuality implies that the
conditional probability distribution representing a measurement
depends on more than just which POVM is measured.  In addition to
this, different methods of implementing the exact same POVM must also
sometimes be represented by different conditional probability
distributions.

To see this, consider again the two measurements $M = \left \{
\Proj{\phi_j} \right\}_{j=0}^{d-1}$ and $M^{\prime} = \left \{
\Proj{\phi_j^{\prime}} \right \}_{j=0}^{d-1}$ with $d=3$ and suppose
that
\begin{equation}
\text{Pr} \left ( \phi_0 \middle | M,\lambda \right ) \neq
\text{Pr} \left ( \phi_0 \middle | M^{\prime},\lambda \right ).
\end{equation}
Then, since $E^{\perp} = \Proj{\phi_1} + \Proj{\phi_2} =
\Proj{\phi^{\prime}_1} + \Proj{\phi^{\prime}_2}$, we have two ways of
implementing the coarse-grained measurement $N = \left \{
\Proj{\phi_0}, E^{\perp} \right \}$, either by measuring $M$ and
then coarse-graining over $\Proj{\phi_1}$ and $\Proj{\phi_2}$ or by
measuring $M^{\prime}$ and coarse-graining over
$\Proj{\phi_1^{\prime}}$ and $\Proj{\phi_2^{\prime}}$.  Both of these
procedures correspond to the same POVM, but the $\Proj{\phi_0}$
outcome is represented by $\text{Pr} \left ( \phi_0 \middle |
M,\lambda \right )$ in the first case and $\text{Pr} \left ( \phi_0
\middle | M^{\prime},\lambda \right )$ in the second.  Since these
are not equal, the same POVM is represented by two different
conditional probability distributions, depending on how the
measurement is implemented.  It is for this reason that, generally, a
measurement $M$ has to be represented by a set $\Xi_M$ of conditional
probability distributions rather than just a single one.

\begin{definition}
Let $\mathfrak{F} = \langle \Hilb, \mathcal{P}, \mathcal{M} \rangle$
be a PM fragment where $\mathcal{M}$ consists of projective
measurements.  An ontological model $\Theta = (\Lambda, \Sigma, \Delta,
\Xi)$ of $\mathfrak{F}$ is \emph{Kochen--Specker (KS) noncontextual}
if it is both:
\begin{itemize}
\item Outcome deterministic: for all $M \in \mathcal{M}, E \in M,
\lambda \in \Lambda$, every $\text{Pr} \in \Xi_M$ satisfies
\begin{equation}
\text{Pr}(E|M,\lambda) \in \{0,1\}.
\end{equation}
\item Measurement noncontextual.
\end{itemize}
Otherwise the model is \emph{KS contextual}.
\end{definition}

In other words, KS noncontextuality only applies to projective
measurements and is the combination of Spekkens' notion of measurement
noncontextuality with outcome determinism, i.e.\ the idea that the
ontic state should determine the outcome of a projective measurement
with certainty.  The Kochen--Specker Theorem \cite{Kochen1967}, and
other proofs of KS contextuality \cite{Mermin1990, Peres1991,
Clifton1993, Cabello1996, Klyachko2008, Cabello2010, Liang2011,
Cabello2014}, show that it is impossible to construct a KS
noncontextual model for all projective measurements in Hilbert spaces
of dimension $\geq 3$.

Note that a model may be KS contextual either by being measurement
contextual or by being nondeterministic.  It can of course be both,
but either one on its own is sufficient to reproduce the quantum
predictions.  The Beltrametti--Bugajski model is an example of a model
that is nondeterministic, but measurement noncontextual because the
conditional probabilities $\text{Pr}(E|M,\Proj{\lambda}) = \Tr{E
\Proj{\lambda}}$ just mimic the quantum probabilities, which only
depend on $E$, and not on which $M$ containing $E$ is measured nor on
how $M$ is implemented.  On the other hand, the Bell model is
deterministic, but measurement contextual because the way in which the
unit interval is divided up depends on all of the basis projectors in
the measurement and their ordering.

Unlike $\psi$-ontology, KS contextuality does not obviously imply
excess baggage, but many of the other implications of $\psi$-ontology
are implied by KS contextuality.  Therefore, one might wonder whether
$\psi$-ontology implies KS contextuality.  This is not the case, as
the Bell model in two dimensions is $\psi$-ontic but can be converted
into a KS noncontextual model by reordering the way in which the unit
interval is divided, e.g.\ by making the first interval always
correspond to the measurement outcome in the northern hemisphere of
the Bloch sphere.  Similarly, there are $\psi$-epistemic models that
are KS contextual, such as the ABCL model described in \S\ref{NPIP},
so $\psi$-ontology and KS contextuality are inequivalent constraints
on ontological models.

The following characterization of KS noncontextual models will be
useful in what follows.

\begin{definition}
\label{def:KS:cosupp}
Let $\mathfrak{F} = \langle \Hilb, \mathcal{P}, \mathcal{M} \rangle$
be a PM fragment and let $\Theta = (\Lambda, \Sigma, \Delta, \Xi)$
be an ontological model of it.  For every $M \in \mathcal{M}$, $E
\in M$, $\text{Pr} \in \Xi_M$, let
\begin{equation}
\Lambda^{(E,M,\text{Pr})} = \left \{ \lambda \in \Lambda \middle |
\text{Pr}(E|M,\lambda) = 1 \right \},
\end{equation}
and let
\begin{equation}
\Lambda^{(E,M)} = \cap_{\text{Pr} \in \Xi_M} \Lambda^{(E,M,\text{Pr})}.
\end{equation}
The \emph{cosupport} $\Lambda^E$ of a POVM element $E$ is then
\begin{equation}
\Lambda^E = \cap_{\{M \in \mathcal{M}|E \in M\}} \Lambda^{(E,M)}.
\end{equation}
\end{definition}

In other words, the cosupport of $E$ is the set of ontic states that
always return the outcome $E$ whenever a POVM that contains it is
measured, for all possible methods of implementing the measurement.

\begin{theorem}
\label{prop:KS:char}
Let $\mathfrak{F} = \langle \Hilb, \mathcal{P}, \mathcal{M} \rangle$
be a PM fragment, where $\mathcal{M}$ consists of projective
measurements.  If an ontological model $\Theta = (\Lambda, \Sigma,
\Delta, \Xi)$ of $\mathfrak{F}$ is KS noncontextual then, for all $M
\in \mathcal{M}$, $E \in M$, $\rho \in \mathcal{P}$, every $\mu \in
\Delta_{\rho}$ satisfies
\begin{equation}
\label{eq:KS:char}
\int_{\Lambda} \mathrm{Pr}(E|M,\lambda) d\mu(\lambda) = \mu(\Lambda^E).
\end{equation}
\end{theorem}
\begin{proof}
If $\Theta$ is KS noncontextual then there is only one conditional
probability distribution $\text{Pr} \in \Xi_M$ and
$\text{Pr}(E|M,\lambda)$ is independent of which measurement
containing $E$ is performed, so, for any $M$ containing $E$,
$\Lambda^E = \Lambda^{(E,M,\text{Pr})} = \left \{ \lambda \in
\Lambda \middle | \text{Pr}(E|M,\lambda) = 1 \right \}$.
Furthermore, $\text{Pr}(E|M,\lambda)$ is either equal to $0$ or $1$
everywhere, so we only need integrate over those $\lambda$ for which
it is equal to $1$.  This gives
\begin{align}
\int_{\Lambda} \text{Pr}(E|M,\lambda) d \mu(\lambda) & =
\int_{\Lambda^E} \text{Pr}(E|M,\lambda) d \mu(\lambda) \\
& = \int_{\Lambda^E} d \mu(\lambda) \\
& = \mu(\Lambda^E),
\end{align}
as required.
\end{proof}

The converse to this theorem does not hold, i.e.\ if
Eq.~\eqref{eq:KS:char} always holds then the model may still be KS
contextual.  This is because there may be subsets of the ontic state
space that are of measure zero according to all the probability
measures that occur in the model.  The values of the conditional
probability distributions on these sets are not constrained by
Eq.~\eqref{eq:KS:char}.  However, a partial converse can be obtained
by excising some of the measure zero sets from the ontic state space.
Since they are of measure zero, their removal does not substantively
affect the structure of the model.

\begin{definition}
Let $\mathfrak{F} = \langle \Hilb, \mathcal{P}, \mathcal{M} \rangle$
be a PM fragment.  An ontological model $\tilde{\Theta} =
(\tilde{\Lambda}, \tilde{\Sigma}, \tilde{\Delta}, \tilde{\Xi})$ of
$\mathfrak{F}$ is a \emph{measure zero revision} of another
ontological model $\Theta = (\Lambda, \Sigma, \Delta, \Xi)$ of
$\mathfrak{F}$ if
\begin{itemize}
\item The ontic state space $\tilde{\Lambda}$ differs from $\Lambda$
only by the removal of a measure zero set.  Formally,
$\tilde{\Lambda} \in \Sigma$ and for all $\rho \in \mathcal{P}$
and $\mu \in \Delta_{\rho}$, $\mu(\tilde{\Lambda}) = 1$.  Further,
$\tilde{\Sigma}$ is the sub-$\sigma$-algebra of $\Sigma$
containing all $\Omega \in \Sigma$ such that $\Omega \subseteq
\tilde{\Lambda}$.
\item The probability measures of $\tilde{\Theta}$ are the restrictions
of the probability measures of $\Theta$ to $(\tilde{\Lambda},
\tilde{\Sigma})$.  Formally, for all $\rho \in \mathcal{P}$, there
is a surjective mapping $f:\Delta_{\rho} \rightarrow
\tilde{\Delta}_{\rho}$ defined by
\begin{equation}
f[\mu](\Omega) = \mu(\Omega),
\end{equation}
for all $\Omega \in \tilde{\Sigma}$.
\item The conditional probability distributions of $\tilde{\Theta}$
are the restrictions of the conditional probability distributions
of $\Theta$ to $\tilde{\Lambda}$.  Formally, for all $M \in
\mathcal{M}$, there is a surjective mapping $g:\Xi_M \rightarrow
\tilde{\Xi}_M$ defined by,
\begin{equation}
g[\text{Pr}](E|M,\lambda) = \text{Pr}(E|M,\lambda),
\end{equation}
for all $E \in M$, $\lambda \in \tilde{\Lambda}$.
\end{itemize}
\end{definition}

Performing a measure zero revision does not represent a substantive
change in the following sense.  Imagine you are passively observing an
experimenter who is performing a sequence of prepare-and-measure
experiments in the fragment $\mathfrak{F}$.  Suppose that, in addition
to observing the choices of preparation and measurement that the
experimenter makes and the measurement outcomes she obtains, you also
get to see the exact ontic state $\lambda$ in every run of the
experiment.  You write down a long list of data consisting of the
preparation performed, the ontic state, the choice of measurement, and
the measurement outcome in each run of the experiment.  This data
would allow you to distinguish between an ontological model $\Theta$
and a measure zero revision of it $\tilde{\Theta}$ only if an ontic
state in $\Lambda \backslash \tilde{\Lambda}$ happens to be occupied
in at least one run of the experiment, since this can only happen in
$\Theta$ and the probabilistic predictions made by the two models are
otherwise exactly the same.  However, the probability of this
happening is zero because the set of ontic states removed has zero
probability according to all the probability measures in the theory.
To all intents and purposes then, the two models tell the same story
about reality.

A partial converse to Theorem~\ref{prop:KS:char} can be obtained
as follows.

\begin{theorem}
\label{prop:KS:conv}
Let $\mathfrak{F} = \langle \Hilb, \mathcal{P}, \mathcal{M} \rangle$
be a PM fragment where $\mathcal{M}$ is at most countable and
consists of projective measurements.  Let $\Theta = (\Lambda,
\Sigma, \Delta, \Xi)$ be an ontological model of $\mathfrak{F}$.
If, for all $M \in \mathcal{M}$, $E \in M$, $\mathrm{Pr} \in \Xi_M$,
$\rho \in \mathcal{P}$, every $\mu \in \Delta_{\rho}$ satisfies
\begin{equation}
\label{eq:KS:char2}
\int_{\Lambda} \mathrm{Pr}(E|M,\lambda) d\mu(\lambda) =
\mu(\Lambda^E),
\end{equation}
then there exists a measure zero revision of $\Theta$ that is KS
noncontextual.
\end{theorem}
\begin{proof}
The measure zero revision $\tilde{\Theta} = (\tilde{\Lambda},
\tilde{\Sigma}, \tilde{\Delta}, \tilde{\Xi})$ is constructed as
follows.  First, let $M = \{E_j\} \in \mathcal{M}$ and note that,
for all $\rho \in \mathcal{P}$ and $\text{Pr} \in \Xi_M$, every $\mu
\in \Delta_{\rho}$ satisfies
\begin{align}
1 & = \mu(\Lambda) = \int_{\Lambda} d \mu(\lambda) \\
& = \int_{\Lambda} \sum_j \text{Pr}(E_j|M,\lambda)
d \mu(\lambda) \\
& = \sum_j \int_{\Lambda} \text{Pr}(E_j|M,\lambda) d \mu(\lambda) \\
& = \sum_j \mu(\Lambda^{E_j}),
\end{align}
where the second line follows from $\sum_j \text{Pr}(E_j|M,\lambda)
= 1$ and the fourth from Eq.~\eqref{eq:KS:char2}.

The sets $\Lambda^{E_j}$ are disjoint because $\Lambda^{E_j}
\subseteq \Lambda^{(E_j,M,\text{Pr})}$ for every $M$ that contains
$E$ and every $\text{Pr} \in \Xi_M$, and if
$\text{Pr}(E_j|M,\lambda) = 1$ then it must be the case that
$\text{Pr}(E_k|M,\lambda) = 0$ for all $k \neq j$ in order to
satisfy $\sum_j \text{Pr}(E_j|M,\lambda) = 1$.  Therefore,
$\Lambda^M = \cup_j \Lambda^{E_j}$ is a set of measure one according
to $\mu$.  Now set $\tilde{\Lambda} = \cap_{M \in \mathcal{M}}
\Lambda^M$.  This is also measure one according to $\mu$, since it
is the intersection of an at most countable number of measure one
sets.  The probability measures in the revision are then obtained by
restriction to $\tilde{\Lambda}$.

For any measurement $M = \{E_j\} \in \mathcal{M}$ and $\text{Pr} \in
\Xi_M$, conditional probability distributions
$\tilde{\text{Pr}}(E_j|M,\lambda)$ obtained by restriction of
$\text{Pr}(E_j|M,\lambda)$ to $\tilde{\Lambda}$ are equal to $1$ on
$\Lambda^{E_j}$ and are zero elsewhere on $\tilde{\Lambda}$.  The
latter follows because $\tilde{\Lambda} \subseteq \cup_j
\Lambda^{E_j}$ and we already saw that $\text{Pr}(E_j|M,\lambda)$ is
equal to zero on $\Lambda^{E_k}$ for $k \neq j$.  Thus, the model is
outcome deterministic.  Further, it is measurement noncontextual
because $\Lambda^{E_j}$ does not depend on the measurement context.
Thus, the model is KS noncontextual.
\end{proof}

In \S\ref{Imp}, we showed that $\psi$-ontic models cannot be maximally
$\psi$-epistemic and that this implies preparation contextuality.  KS
contextuality can be used in place of $\psi$-ontology to derive these
results, in the following sense.  If we restrict attention to models
of measurements in complete orthonormal bases such that each
measurement outcome occurs in an at most countable set of contexts, a
maximally $\psi$-epistemic model has a measure zero revision that is KS
noncontextual.  This means that the Kochen--Specker theorem, which
shows that KS noncontextual models that reproduce the quantum
predictions are impossible, can be used in place of $\psi$-ontology to
derive these results.

\begin{theorem}
\label{prop:ME:char}
Let $\mathfrak{F} = \langle \Hilb, \mathcal{P}, \mathcal{M} \rangle$
be a PM fragment where $\mathcal{M}$ consists of measurements in
complete orthonormal bases and, for each rank-one projector
$\Proj{\phi}$, the set of $M \in \mathcal{M}$ that contains
$\Proj{\phi}$ is at most countable.  Let $\Theta = (\Lambda, \Sigma,
\Delta, \Xi)$ be an ontological model of $\mathfrak{F}$ that
reproduces the quantum predictions and is such that, for each $M \in
\mathcal{M}$, $\Xi_M$ is at most countable.  If $\Theta$ is
maximally $\psi$-epistemic then, for all pure states $\Proj{\psi},
\Proj{\phi} \in \mathcal{P}$, every $\mu \in \Delta_{\psi}$
satisfies
\begin{equation}
\label{eq:ME:KS}
\int_{\Lambda} \text{Pr}(\phi|M,\lambda) d \mu(\lambda) =
\mu(\Lambda^{\phi}),
\end{equation}
for every $M \in \mathcal{M}$ that contains $\Proj{\phi}$ and every
$\text{Pr} \in \Xi_M$, and where $\Lambda^{\phi}$ is the cosupport
of $\Proj{\phi}$.
\end{theorem}
\begin{proof}
Let $\nu \in \Delta_{\phi}$ and consider an $M \in \mathcal{M}$ that
contains $\Proj{\phi}$.  Then, in order to reproduce the quantum
predictions, for every $\text{Pr} \in \Xi_M$ we must have
\begin{equation}
\int_{\Lambda} \text{Pr}(\phi|M,\lambda) d\nu(\lambda) =
\Tr{\Proj{\phi}\Proj{\phi}} = 1.
\end{equation}
In order for this to be true, there must be a set $\Omega \in
\Sigma$ such that $\nu(\Omega) = 1$ and $\text{Pr}(\phi|M,\lambda)$
is equal to one on $\Omega$.  Any such set satisfies $\Omega
\subseteq \Lambda^{(\phi,M,\text{Pr})}$ and so
$\Lambda^{(\phi,M,\text{Pr})}$ is also a set of measure one
according to $\nu$.  Since this applies to any $\text{Pr} \in
\Xi_M$, $\Lambda^{(\phi,M)} = \cap_{\text{Pr} \in \Xi_M}
\Lambda^{(\phi,M,\text{Pr})}$ is also a set of measure one according
to $\nu$ because it is the intersection of an at most countable set
of measure one sets.  Finally, $\Lambda^{\phi} = \cap_{\{M \in
\mathcal{M}| \Proj{\phi} \in M\}} \Lambda^{(\phi,M)}$ is a measure
one according to $\nu$, since it too is the intersection of an at
most countable collection of measure one sets. Since
$\text{Pr}(\phi|M,\lambda)$ is equal to $1$ on $\Lambda^{\phi}$, we
have
\begin{equation}
\int_{\Lambda^{\phi}} \text{Pr}(\phi|M \lambda) d\mu(\lambda) =
\int_{\Lambda^{\phi}} d\mu(\lambda) = \mu(\Lambda^{\phi}).
\end{equation}
Thus, Eq.~\eqref{eq:ME:KS} is a special case of
Eq.~\eqref{eq:ME:ME}, so it must hold for a maximally
$\psi$-epistemic model.
\end{proof}

\begin{corollary}
\label{cor:ME:KS}
Let $\mathfrak{F} = \langle \Hilb, \mathcal{P}, \mathcal{M} \rangle$
be a PM fragment in which $\mathcal{M}$ is at most countable and
consists of measurements in complete orthonormal bases.  Suppose
that, for all $M \in \mathcal{M}$, $\Proj{\phi} \in M$, it is also
the case that $\Proj{\phi} \in \mathcal{P}$.  Then, any maximally
$\psi$-epistemic ontological model $\Theta =
(\Lambda,\Sigma,\Delta,\Xi)$ of $\mathfrak{F}$ in which $\Xi_M$ is
at most countable for each $M \in \mathcal{M}$ has a measure one
revision that is KS noncontextual.
\end{corollary}

The proof just consists of combining Theorem~\ref{prop:ME:char}
with Theorem~\ref{prop:KS:conv}.
\pagebreak
\balance

\end{document}